\documentclass[11pt,a4paper,reqno]{amsbook}

\usepackage[latin1]{inputenc}
\usepackage[T1]{fontenc}
\usepackage[left=3.8cm, right=3.8cm, top=3cm, bottom=3cm]{geometry}
\usepackage{graphicx}
\usepackage{amssymb}
\usepackage{amsthm}
\usepackage{amsmath}
\usepackage{hyperref}
\usepackage{physics2}
\usephysicsmodule{braket}
\usepackage{enumerate}
\usepackage{xcolor}
\usepackage[abs]{overpic}
\usepackage{pict2e}
\usepackage{changepage}
\usepackage{pdfpages}

\newtheorem{theorem}{Theorem}[chapter]
\newtheorem{lemma}[theorem]{Lemma}

\newtheorem{proposition}[theorem]{Proposition}
\newtheorem{corollary}[theorem]{Corollary}
\newtheorem{conjecture}[theorem]{Conjecture}
\newtheorem*{informaltheorem}{Theorem (Informal)}
\theoremstyle{definition}
\newtheorem{definition}[theorem]{Definition}
\newtheorem{remark}[theorem]{Remark}

\newcommand{\Hil}{\mathcal{H}}
\newcommand{\K}{\mathbb{K}}
\newcommand{\R}{\mathbb{R}}
\newcommand{\C}{\mathbb{C}}
\newcommand{\N}{\mathbb{N}}
\newcommand{\Q}{\mathbb Q}
\newcommand{\T}{\mathbb{T}}
\newcommand{\Cha}{\mathcal{C}}
\newcommand{\Lin}{\mathcal{L}}

\newcommand{\ketbrar}[1]{\ketbra{#1}{#1}}
\newcommand{\tr}{\mathrm{Tr}}
\newcommand{\mc}[1]{\mathcal{#1}}
\newcommand{\mb}[1]{\mathbb{#1}}

\newcommand{\subeq}{\textrm{`}\!\!\subseteq\!\!\text{'}}
\newcommand{\supeq}{\textrm{`}\!\!\supseteq\!\!\text{'}}
\newcommand{\eq}{\textrm{`}\!\!=\!\!\text{'}}
\newcommand{\no}{\textrm{`}\!\!\times\!\!\text{'}}

\definecolor{lighter}{RGB}{10.0,146.0,35.0}
\definecolor{darker}{RGB}{2.0,100.0,20.0}
\definecolor{sininen}{RGB}{89.5,101.8,238.1}
\definecolor{punainen}{RGB}{227.0,12.0,12.0}

\newcommand{\lighter}[1]{\textcolor{lighter}{#1}}
\newcommand{\darker}[1]{\textcolor{darker}{#1}}
\newcommand{\punainen}[1]{\textcolor{punainen}{#1}}
\newcommand{\sininen}[1]{\textcolor{sininen}{#1}}

\newcommand{\fii}{\varphi}
\newcommand{\tj}{\vartheta}
\newcommand{\V}{\mathcal{V}}
\DeclareMathOperator*{\esssup}{ess\,sup}
\def\<{\langle}
\def\>{\rangle}
\newcommand{\rleq}{\preceq}
\newcommand{\rgeq}{\succeq}

\makeatletter
\DeclareRobustCommand\bigop[1]{%
	\mathop{\vphantom{\sum}\mathpalette\bigop@{#1}}\slimits@
}
\newcommand{\bigop@}[2]{%
	\vcenter{%
		\sbox\z@{$#1\sum$}%
		\hbox{\resizebox{\ifx#1\displaystyle.9\fi\dimexpr\ht\z@+\dp\z@}{!}{$\m@th#2$}}%
	}%
}
\makeatother
\newcommand{\bigboxplus}{\DOTSB\bigop{\boxplus}}

\numberwithin{figure}{chapter}

\newcommand\invisiblechapter[1]{
    \phantomsection
    \addcontentsline{toc}{chapter}{#1}
    \chaptermark{#1}}

\begin{document}

\begin{titlepage}
    \vspace*{0.8cm}
    
    \centering
    
    {\Large\bfseries\textsc{Classical and Quantum Information Measures}}
    
    \vspace{5mm}
    
    {\Large\bfseries\textsc{Characterizing Convertibility:}}
    
    \vspace{5mm}
    
    {\Large\bfseries\textsc{An Axiomatic Approach}}

    \vspace{17mm}

    {\large\bfseries\textsc{Frits Verhagen}}

    \vspace{5mm}

    {\large\itshape\textsc{(M.Sc, Utrecht University)}}

    \vspace{17mm}

    {\large\textsc{A thesis submitted for the degree of}}

    \vspace{5mm}

    {\large\textsc{Doctor of Philosophy}}

    \vspace{17mm}

    {\large\textsc{Centre for Quantum Technologies}}

    \vspace{5mm}

    {\large\textsc{National University of Singapore}}

    \vspace{17mm}

    {\large\textsc{Main Advisor: Professor Marco Tomamichel}}

    \vspace{5mm}

    {\large\textsc{Co-Advisor: Dr. Erkka Haapasalo}}

    \vspace{17mm}

    {\large\textsc{Examiners:}}

    \vspace{5mm}

    {\large\textsc{Professor Vincent Y.F. Tan}}
    
    \vspace{5mm}

    {\large\textsc{Associate Professor Divesh Aggarwal}}
    
    \vspace{5mm}

    {\large\textsc{Associate Professor P\'eter Vrana,}\\
    \textsc{Budapest University of Technology and Economics}}

    \vspace{16mm}       

    {\large\textsc{2026}}
\end{titlepage}

\setcounter{page}{2}

\cleardoublepage

\invisiblechapter{Declaration}

\includepdf{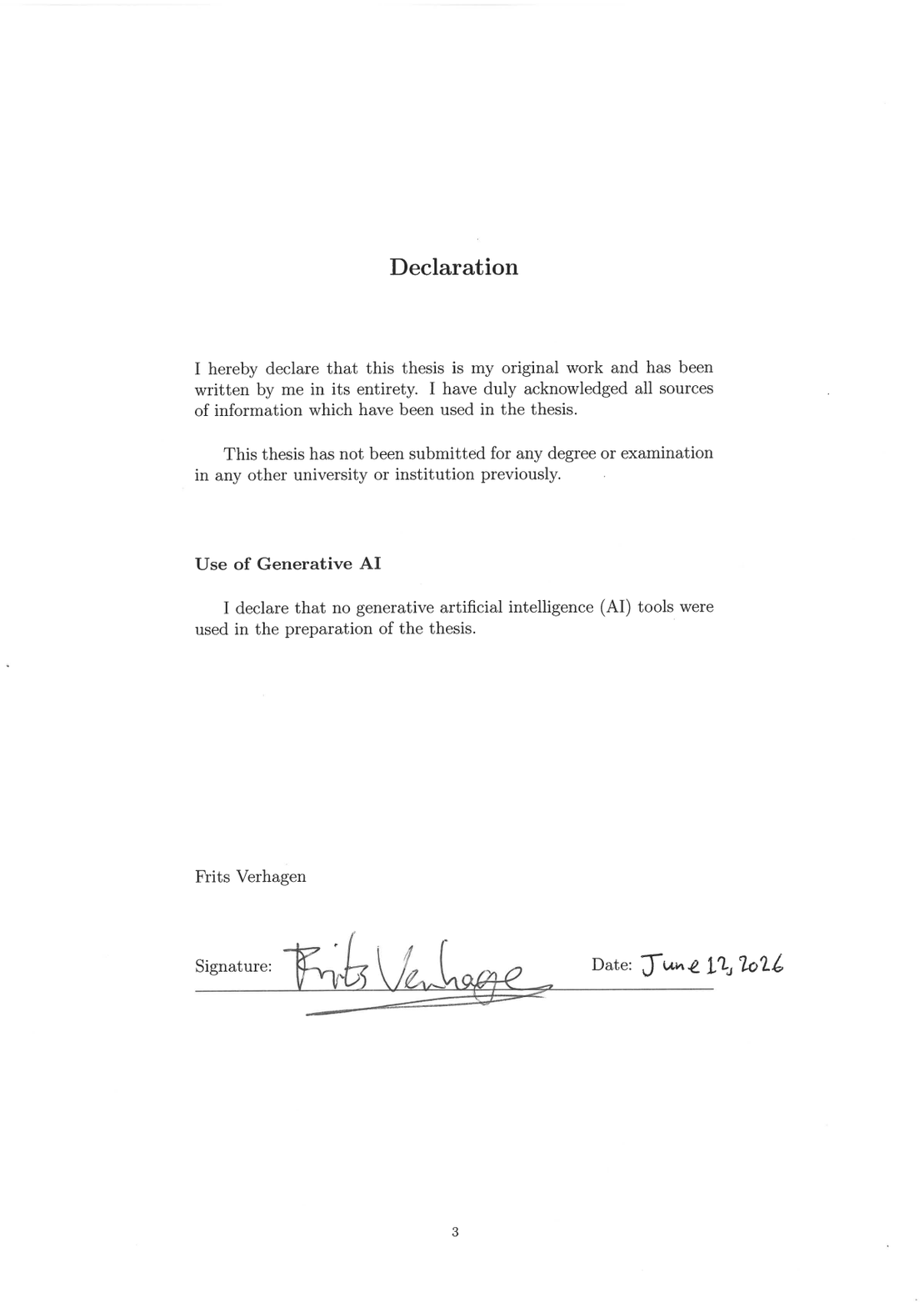}

\chapter*{Acknowledgments}

\begin{adjustwidth}{1cm}{1cm}
First of all, I want to thank Marco Tomamichel and Erkka Haapasalo, my PhD supervisor and co-supervisor. I really liked white board discussions with Marco where we would often have very useful insights, and was always amazed how well he remembered the details of my projects. I'm also very happy that he offered me the opportunity to join his group after I had already started my PhD. Erkka was always very enthusiastic to discuss new ideas with me, and I learned a lot about doing research by working with him closely. He was always willing to take the time to go through the details of my work and I benefited from his impressive grasp of math. 

I also want to thank everyone who contributed to my research by giving feedback or having discussions. Thanks to Jiang Zhang for helping me and others in the semiring group go through the proofs of the Vergleichsstellens\"atze. Now that my thesis is done, I look forward to continuing collaborating with people in the group and learning new things. 

Thanks as well to all reviewers of my publications and the examiners of this thesis for their interesting suggestions and ideas for improvement.  

My PhD was also made very enjoyable by the other PhD students, postdocs, and many interns and visitors in Marco's group. The atmosphere in the group and at the Centre for Quantum Technologies made my PhD a very social experience. Thanks to all the friends I met at CQT and Singapore, the last five years has been very memorable. 

Finally, I'd like to thank my parents and sister back in The Netherlands for always supporting and encouraging me during my time in Singapore! 
\end{adjustwidth}

\tableofcontents
	
\chapter*{Summary}

\begin{adjustwidth}{1cm}{1cm}
We take an axiomatic approach to studying multivariate divergences, which are measures of the information content or distinguishability of a tuple of classical or quantum states. Naturally, these divergences cannot increase after applying any physically allowed operation. Additionally, they should be additive when the states describe independent systems. The collection of all divergences satisfying these two axioms of monotonicity and additivity consists of the physically relevant measures of information. 

We study these divergences in the context of convertibility of tuples of states. More precisely, we ask if one tuple of states can be converted into another by a single channel if we have an arbitrarily large number of independent copies of the states, or if we allow the conversion to be aided by catalyst states. In these settings, convertibility is characterized by the ordering of an axiomatically defined family of divergences. In addition to monotonicity and additivity, a third axiom singles out this family of divergences, which can be shown to be minimal in certain settings. 

In this thesis, we identify this family of divergences explicitly for two scenarios, one classical and one quantum. In the classical setting, we consider tuples of $d\geq2$ finite probability distributions. We show that the multivariate divergences that characterize large-sample and catalytic majorization are generalizations of the bivariate R\'enyi divergences. Unlike previous studies, we do not require the distributions to have the same support, yielding interesting new divergences. In the quantum setting, we consider pairs of classical-quantum states with pure component states, and find that majorization is characterized by all $\alpha$-$z$ relative entropies for $\alpha<1$. This demonstrates the first operational interpretation of these relative entropies. In both settings, we also find results on majorization where the conversion is allowed to be approximate with vanishing error. 

Our methods utilize real-algebraic techniques recently developed by T. Fritz. This involves studying preordered semirings representing the tuples of states, and identifying certain monotone homomorphisms and derivations on them, which correspond to the divergences characterizing large-sample and catalytic majorization. 
\end{adjustwidth}

\chapter*{Publications}

This thesis is based on the following publications by the author: \\

\begin{itemize}
    \item F. Verhagen, M. Tomamichel, and E. Haapasalo,\\
    \emph{Matrix Majorization in Large Samples With Varying}
    \emph{Support Restrictions},\\
    IEEE Transactions on Information Theory, vol. 71, no. 9,\\
    pp. 6517-6545 (2025) \cite{verhagen2025}\\
    \item F. Verhagen, M. Tomamichel, and E. Haapasalo,\\
    \emph{Conditions for Large-Sample Majorization of Pairs of Flat States}\\
    \emph{in Terms of $\alpha$-$z$ Relative Entropies},\\
    Communications in Mathematical Physics, vol. 407, art. 157 (2026) \cite{verhagen2025alphaz}\\
\end{itemize}

\chapter{Introduction}\label{cha:introduction}

We investigate the connection between convertibility of tuples of quantum or classical states, and axiomatically defined divergences, which are measures of the distinguishability between states. Consider a tuple of $d\geq2$ states $\left(\rho^{(1)},\ldots,\rho^{(d)}\right)$ on a Hilbert space $\Hil_{\rm in}$ and another tuple $\left(\sigma^{(1)},\ldots,\sigma^{(d)}\right)$ on $\Hil_{\rm out}$. We say that the former tuple of states can be \emph{converted} to the latter if there exists a single quantum channel $\Cha$ that maps each $\rho^{(k)}$ to the corresponding $\sigma^{(k)}$, i.e. 
\begin{equation}\label{eq:convertible}
    \Cha\left(\rho^{(k)}\right)=\sigma^{(k)} 
\end{equation}
for all $k=1,\ldots,d$. One also says that the former tuple \emph{majorizes} the latter. Usually, we are interested in studying only specific types of states in a $d$-tuple: for instance, only classical states or quantum states of a specific form. We denote the tuples of interest by $S$. 

Intuitively, acting on a tuple of states by some (quantum) channel will make the output states more similar, hence harder to distinguish. Divergences allow us to quantify this distinguishability of states: a \emph{(multivariate) divergence} is a real-valued map $D:S\rightarrow\R$ on tuples of states, satisfying the axiom of \emph{monotonicity}, also known as the \emph{data-processing inequality}: for any channel $\Cha$ and any tuple $\left(\rho^{(1)},\ldots,\rho^{(d)}\right)\in S$ we require 
\begin{equation}
    D\left(\Cha\left(\rho^{(1)}\right),\ldots,\Cha\left(\rho^{(d)}\right)\right)\leq D\left(\rho^{(1)},\ldots,\rho^{(d)}\right). 
\end{equation}

The tensor product between the states in two tuples $\left(\rho^{(1)},\ldots,\rho^{(d)}\right)$ and $\left(\sigma^{(1)},\ldots,\sigma^{(d)}\right)$ in $S$ defines a new tuple on the composite system: 
\begin{equation}
    \left(\rho^{(1)}\otimes\sigma^{(1)},\ldots,\rho^{(d)}\otimes\sigma^{(d)}\right), 
\end{equation}
which we require to be in $S$, as well. In this thesis, we focus on the convertibility of tuples of states in the large-sample or the catalytic setting. For \emph{large-sample} convertibility, we ask whether for all $n$ large enough there exists a channel $\Cha_n$ such that for $k=1,\ldots,d$, 
\begin{equation}\label{eq:LS}
    \Cha_n\left(\left(\rho^{(k)}\right)^{\otimes n}\right)=\left(\sigma^{(k)}\right)^{\otimes n}. 
\end{equation}
Conversion between the two tuples may be impossible in the one-shot setting as in \eqref{eq:convertible}, but is potentially possible only for large number $n$ of copies of the states. For \emph{catalytic} convertibility, we ask if there exist so-called catalyst states $\left(\tau^{(1)},\ldots,\tau^{(d)}\right)$ such that for $k=1,\ldots,d$,
\begin{equation}\label{eq:cat}
    \Cha\left(\tau^{(k)}\otimes\rho^{(k)}\right)=\tau^{(k)}\otimes\sigma^{(k)}. 
\end{equation}
Here, the catalyst states aid in making the conversion possible, and are returned unaltered after applying the channel. 

In the context of large-sample and catalytic convertibility, where we consider tensor products of states, a second natural axiom one expects a divergence to satisfy is \emph{additivity under the tensor product}, i.e. 
\begin{equation}\label{eq:axiomproduct}
    D\left(\rho^{(1)}\otimes\sigma^{(1)},\ldots,\rho^{(d)}\otimes\sigma^{(d)}\right)=D\left(\rho^{(1)},\ldots,\rho^{(d)}\right)+D\left(\sigma^{(1)},\ldots,\sigma^{(d)}\right). 
\end{equation}

When $S$ consists of tuples of classical states, i.e. the $d$ states within a tuple commute with each other pairwise, a tuple can be written as 
\begin{equation}\label{eq:classical}
    \left(p^{(1)},\ldots,p^{(d)}\right), 
\end{equation}
where for each $k=1,\ldots,d$, $p^{(k)}$ is a finite probability distribution of some size $n\geq1$, i.e. $p^{(k)}\in\R_{\geq0}^n$ and $\sum_{i=1}^n(p^{(k)})_i=1$. The convertibility between such tuples is also known as matrix majorization. One can interpret these tuples as statistical experiments \cite{Blackwell51,Blackwell53}: here, $k=1,\ldots,d$ represents the $d$ different states a certain physical system can be in, and after performing a measurement with $n$ outcomes on the system, each $p^{(k)}$ represents the corresponding state's statistical distribution of outcomes. A channel applied to a tuple can now be seen as post-processing the measurements outcomes. In this way, divergences on such tuples quantify the information content of statistical experiments. For $d=2$, the study of convertibility of pairs of finite probability distributions $(p,q)$ is also known as relative majorization. A well-known family of divergences on such pairs are the R\'enyi divergences \cite{Renyi61}, parameterized by a real number $\alpha$: 
\begin{equation}\label{eq:Renyi}
    D_\alpha(p\|q):=\frac{1}{\alpha-1}\log{\sum_{i=1}^n (p_i)^\alpha (q_i)^{1-\alpha}}, 
\end{equation}
along with limits $\alpha\rightarrow0$, $\alpha\rightarrow1$ (the Kullback--Leibler divergence) and $\alpha\rightarrow\infty$, with some restrictions on the supports of the two distributions. The R\'enyi divergences find applications in information processing tasks, such as hypothesis testing: here, one samples from a large number of copies of either distribution $p$ or $q$, without knowing which one, and is asked to guess the correct distribution. The R\'enyi divergences characterize the rate, or error exponent, at which the guessing error vanishes when $n$ grows larger \cite{Csiszar1995}. 

When $S$ consists of pairs of quantum states $(\rho,\sigma)$, one rather large family of divergences are the $\alpha$-$z$ relative entropies $D_{\alpha,z}$ \cite{Jaksic2012,Audenaert_Datta_2015}, defined by 
\begin{equation}\label{eq:introductionalphazdivergence}
    D_{\alpha,z}(\rho\|\sigma):=\frac{1}{\alpha-1}\log\tr\left[\left(\sigma^\frac{1-\alpha}{2z}\rho^\frac{\alpha}{z}\sigma^\frac{1-\alpha}{2z}\right)^z\right], 
\end{equation}
as well as certain limits $\alpha\rightarrow1$ (the Umegaki relative entropy) and $z\rightarrow\infty$ of this expression, and with some restrictions on the supports of $\rho$, $\sigma$. Note that when $\rho$, $\sigma$ commute, then $D_{\alpha,z}$ reduces to the classical R\'enyi divergence $D_\alpha$ in \ref{eq:Renyi}, for any $z$. When $z=1$, the $D_{\alpha,z}$ reduce to the Petz relative entropies, which find an operational interpretation as error exponents in quantum hypothesis testing \cite{Mosonyi2011}. When $\alpha=z$, the $D_{\alpha,z}$ reduce to the sandwiched relative entropies, which appear in the strong converse bound in quantum hypothesis testing \cite{Mosonyi2013}. The Umegaki relative entropy appears in the quantum Stein's Lemma, also in the setting of hypothesis testing \cite{Hiai1991,Ogawa2000}. 

Consider now a general collection $S$ of tuples of states. If one tuple of states can be converted to another tuple either in the large-sample setting as in \eqref{eq:LS} or in the catalytic setting as in \eqref{eq:cat}, then the axioms of monotonicity and additivity under tensor products imply, for any divergence $D$, that 
\begin{equation}\label{eq:ordering}
    D\left(\rho^{(1)},\ldots,\rho^{(d)}\right)\geq D\left(\sigma^{(1)},\ldots,\sigma^{(d)}\right). 
\end{equation}
Hence, satisfying the inequality above for any divergence $D$ is a necessary condition for large-sample or catalytic convertibility to hold. In this thesis, we show that, essentially, this condition is also \emph{sufficient}. More precisely, if the inequalities in \eqref{eq:ordering} are strict, then large-sample and catalytic convertibility is possible. In this work, we will demonstrate this for different types of tuples of states $S$, both classical and quantum, establishing for each of these scenarios a result of the following form, that we state informally below. 

\begin{informaltheorem}
    Let $\left(\rho^{(1)},\ldots,\rho^{(d)}\right),\left(\sigma^{(1)},\ldots,\sigma^{(d)}\right)$ be two tuples of states in $S$. If for all divergences $D$, 
    \begin{equation}\label{eq:informaltheoremcompare}
        D\left(\rho^{(1)},\ldots,\rho^{(d)}\right)>D\left(\sigma^{(1)},\ldots,\sigma^{(d)}\right), 
    \end{equation}
    then 
    \begin{enumerate}[(a)]
        \item there exist a channel $\Cha$ and catalyst states $\left(\tau^{(1)},\ldots,\tau^{(d)}\right)$ such that\\
        for $k=1,\ldots,d$,         
        \begin{equation}
            \Cha\left(\tau^{(k)}\otimes\rho^{(k)}\right)=\tau^{(k)}\otimes\sigma^{(k)}, 
        \end{equation}
        \item for sufficiently large $n$, there exists a channel $\Cha_n$ such that for $k=1,\ldots,d$, 
        \begin{equation}
            \Cha_n\left(\left(\rho^{(k)}\right)^{\otimes n}\right)=\left(\sigma^{(k)}\right)^{\otimes n}. 
        \end{equation}
    \end{enumerate}
    Conversely, if either (a) or (b) holds, then also \eqref{eq:informaltheoremcompare} holds for all divergences $D$ with non-strict inequality $\geq$ instead of $>$. 
\end{informaltheorem}

Catalytic convertibility as in (a) and large-sample convertibility as in (b) are not equivalent in general. However, large-sample convertibility implies catalytic convertibility (see \cite{Duan2005}) by using catalyst states of the form 
\begin{equation}
    \tau^{(k)}=\frac{1}{n}\bigoplus_{\ell=0}^{n-1} \left(\rho^{(k)}\right)^{\otimes\ell}\otimes\left(\sigma^{(k)}\right)^{\otimes(n-1-\ell)}, 
\end{equation}
for $k=1,\ldots,d$. The converse implication does not always hold according to \cite{Feng_et_al_2006}. 

We will also study asymptotic convertibility with vanishing error. In the scenarios that we consider, a result along the same lines as in the informal theorem above holds. More precisely, for every error $\varepsilon>0$, there exist some $\sigma_\varepsilon^{(k)}$, $k=1,\ldots,d$, that are each $\varepsilon$ close to their exact counterpart according to some distance measure, such that in the large-sample setting, $n_\varepsilon$ tensor powers of each exact $\rho^{(k)}$ is converted to $n_\varepsilon$ tensor powers of the approximated $\sigma_\varepsilon^{(k)}$. Note that the power $n_\varepsilon$ depends on the error $\varepsilon$. Similarly, in the catalytic setting, for every $\varepsilon>0$ there exist catalyst states $\tau_\varepsilon^{(k)}$, depending on $\varepsilon$, and approximated target states $\sigma_\varepsilon^{(k)}$ such that $\tau_\varepsilon^{(k)}\otimes\rho^{(k)}$ converts to $\tau_\varepsilon^{(k)}\otimes\sigma_\varepsilon^{(k)}$ for each $k=1,\ldots,d$. In this case, the inequalities in \eqref{eq:informaltheoremcompare} are only required to be satisfied non-strictly. In this way, large-sample and catalytic majorization occurs in the asymptotic sense if and only if the divergences are ordered non-strictly. From a practical perspective, asymptotic convertibility is the more natural notion to consider, since obtaining the desired output states $\sigma^{(k)}$ up to a small error is usually sufficient. 

It turns out that it is not necessary to check \eqref{eq:informaltheoremcompare} for every divergence, but only those that satisfy an additional, third axiom, namely \emph{additivity under the direct sum} of states. Informally, this axiom can be stated as follows: a divergence $D$ must be of the form 
\begin{equation}\label{eq:thirdaxiom1}
    D\left(\rho^{(1)},\ldots,\rho^{(d)}\right)=c\log Q\left(\rho^{(1)},\ldots,\rho^{(d)}\right), 
\end{equation}
for some constant $c\in\R$ and a function $Q:S\rightarrow\R_{\geq0}$ which is additive under the direct sum, i.e. 
\begin{equation}\label{eq:thirdaxiom2}
    Q\left(\rho^{(1)}\oplus\sigma^{(1)},\ldots,\rho^{(d)}\oplus\sigma^{(d)}\right)=Q\left(\rho^{(1)},\ldots,\rho^{(d)}\right)+Q\left(\sigma^{(1)},\ldots,\sigma^{(d)}\right). 
\end{equation}
Note that taking the direct sum of two normalized states yields a state that is not normalized. Therefore, we allow our set of state tuples $S$ to include tuples of states with any normalization.

Denote by $\mc D$ the set of divergences satisfying all three axioms of monotonicity and additivity under tensor products and direct sums. The mathematical framework developed in \cite{fritz2022}, which we will use extensively in this thesis, shows that adding this third axiom of additivity under direct sums singles out those divergences characterizing large-sample and catalytic convertibility: \eqref{eq:informaltheoremcompare} needs only to hold for those divergences $D$ in $\mc D$. A recent paper \cite{Haapasalo2025} establishes that $\mc D$ is essentially the set of extreme points in the convex set of all divergences that satisfy at least the two axioms of monotonicity and tensor product additivity. 

In some scenarios, we will show that $\mc D$ is also the smallest set of divergences possible that give sufficient conditions \eqref{eq:informaltheoremcompare} for large-sample and catalytic convertibility. Namely, if one removes an open subset from $\mc D$, then satisfying the conditions \eqref{eq:informaltheoremcompare} for the remaining divergences is not sufficient for (a) and (b) in the informal theorem to hold. 

The optimal rate of converting a tuple $\left(\rho^{(1)},\ldots,\rho^{(d)}\right)$ to another tuple\\
$\left(\sigma^{(1)},\ldots,\sigma^{(d)}\right)$ in $S$ can be (informally) defined as the largest $r\in\R$ such that 
\begin{equation}
    \Cha\Big(\left(\rho^{(k)}\right)^{\otimes n}\Big)=\left(\sigma^{(k)}\right)^{\otimes rn}, 
\end{equation}
for some channel $\Cha$ and all $k=1,\ldots,d$, in the limit $n\rightarrow\infty$. We will prove the following expression for the optimal conversion rate in Section \ref{sec:optimalconversionrate}, in terms of the divergences $\mc D$: 
\begin{equation}
    \min_{D\in\mc D}\frac{D\left(\rho^{(1)},\ldots,\rho^{(d)}\right)}{D\left(\sigma^{(1)},\ldots,\sigma^{(d)}\right)}. 
\end{equation}

In this thesis, we will analyze two different scenarios, one classical and one quantum, of convertibility of tuples of states. In both cases, we will identify the precise set of divergences $\mc D$ characterizing large-sample and catalytic convertibility. For the classical scenario of matrix majorization in Chapter \ref{cha:matrix}, we consider $d$-tuples of finite probability distributions as in \eqref{eq:classical}. The set of divergences characterizing convertibility in this setting when all distributions $p^{(k)}$ have equal support was studied in \cite{farooq2024}. In this thesis, we generalize these results by weakening this strict assumption on the supports. Namely, we consider the case where we make no assumptions on the supports, other than having at least some common support among the $d$ distributions (Section \ref{Msec:findings}). We also consider another case, where the support of one distribution $p^{(d)}$ dominates, or contains, the supports of all other $p^{(k)}$, $k\neq d$. (Section \ref{Msec:Asymmetric}). This setting is applicable to convertibility of Gibbs-preserving states in quantum thermodynamics: here, we take $d=2$ and fix $\rho^{(2)}=\sigma^{(2)}$ to be the Gibbs state of a thermodynamic system, and ask if some state $\rho^{(1)}$ of this system can be converted to another $\sigma^{(1)}$. This is detailed in Section \ref{Msec:thermal}, where we derive conditions for convertibility in terms of generalized free energies. This establishes the main conclusion from \cite{Brandao2014}, while avoiding some technical issues in that paper's analysis. 

The divergences encountered in the setting of matrix majorization are of the form 
\begin{equation}\label{eq:introductionmatrixdivergence}
    D_{\underline{\alpha}}(P)=\frac{1}{\max_{1\leq k\leq d}\alpha_k-1}\log\sum_{i=1}^n \left(p^{(1)}_i\right)^{\alpha_1}\cdots\left(p^{(d)}_i\right)^{\alpha_d} 
\end{equation}
where $\underline{\alpha}\in\R^d$ satisfy $\alpha_1+\ldots+\alpha_d=1$, as well as certain point-wise limits in $\underline{\alpha}$. This was initially shown in \cite{farooq2024} for distributions of equal support. Note that these divergences reduce to the R\'enyi divergences of \eqref{eq:Renyi} when $d=2$. We will show that considering more general supports of the distributions gives rise to variations of the above divergences where only a subset of the distribution's outcomes $i=1,\ldots,n$ is summed over in \eqref{eq:introductionmatrixdivergence}. Additionally, the range that the vector $\underline{\alpha}$ is allowed to be in depends strongly on the precise relations of the supports of the $d$ distributions. 

The second scenario, analyzed in Chapter \ref{cha:flat}, concerns quantum states. Ideally, we would like to consider arbitrary tuples of two or more general finite-dimensional quantum states. However, identifying all relevant divergences is not straightforward, due to reasons we will now explain. In the case of matrix majorization, a tuple of probability distributions $\left(p^{(1)},\ldots,p^{(d)}\right)$ can be seen as $d$ commuting density operators, which may be trivially decomposed into $1$-dimensional Hilbert spaces using the direct sum: 
\begin{equation}
    p^{(k)}=\left(p^{(k)}\right)_1\oplus\ldots\oplus\left(p^{(k)}\right)_n. 
\end{equation}
The axiom of direct sum additivity given by \eqref{eq:thirdaxiom1} and \eqref{eq:thirdaxiom2} then allows us to reduce the problem of identifying all divergences to considering only tuples of $1$-dimensional (potentially non-normalized) states, which will be explained in detail in Chapter \ref{cha:matrix}. In the setting of a tuple of $d$ general quantum states, a decomposition into smaller-sized density matrices is not possible, in general, due to the non-commutativity of the states, even for the case of $d=2$. However, in some settings we may proceed in a way similar to the classical matrix majorization case. Namely, Jordan's Lemma \cite{Jordan1875} says that a pair of projections can be simultaneously block-diagonalized such that the component blocks are of rank at most 1: 

\begin{lemma}[Jordan's Lemma]
    Let $\Hil$ be a finite-dimensional linear space and $A,B$ projections in $\Hil$ (i.e. two operators on $\Hil$ such that $A^2=A=A^*$, $B^2=B=B^*$). Then $A,B$ can be simultaneously block-diagonalized with blocks of dimension $1$ or $2$. In the case of dimension $2$, the corresponding blocks in $A,B$ have rank $1$. More precisely, there exists a decomposition $\Hil=\bigoplus_{i=1}^n\Hil_i$, for some $n\geq1$, where the $\Hil_i$ are at most $2$ dimensional linear spaces, such that 
    \begin{equation}\label{eq:Jordanform}
        A=\bigoplus_{i=1}^n A_i,\quad B=\bigoplus_{i=1}^n B_i, 
    \end{equation}
    where the $A_i$, $B_i$ are projections of rank at most $1$ in $\Hil_i$. 
\end{lemma}

We call a finite-dimensional quantum state $\rho$ a \emph{flat state} if it has a single repeated eigenvalue. In other words, $\rho=\tr[P]^{-1}P$ for some projection $P$. Jordan's Lemma then lets us decompose pairs of such states into smaller pairs of states of at most rank $1$, which are pure states. This motivates us to consider a generalization of such pairs: pairs of classical-quantum (cq) states with pure components. These are pairs $(\rho,\sigma)$ such that there exists a classical register given by fixed orthonormal flags $\ket{i}$, for $i=1,\ldots,n$, $n\geq1$, such that 
\begin{equation}
    \rho=\sum_{i=1}^n p_i \ketbrar{i}\otimes\ketbrar{\alpha_i},\quad\sigma=\sum_{i=1}^n q_i \ketbrar{i}\otimes\ketbrar{\beta_i}, 
\end{equation}
where $(p_1,\ldots,p_n)$ is a probability distribution and $\ket{\alpha_i}$ is a normalized vector, and similarly for the $q_i$ and $\ket{\beta_i}$. We show that in this setting the divergences characterizing large-sample and catalytic convertibility are a subset of the $\alpha$-$z$ relative entropies $D_{\alpha,z}$ from \eqref{eq:introductionalphazdivergence}. This result establishes the first operational interpretation of these relative entropies, where the $\alpha$ and $z$ parameters are truly independent from each other. 

The axioms of additivity under tensor products \eqref{eq:axiomproduct} and the additivity under direct sums as in \eqref{eq:thirdaxiom1} and \eqref{eq:thirdaxiom2} involve two operations $\otimes$ and $\oplus$ on (possibly non-normalized) states that induce operations of ``multiplication'' $\boxtimes$ and ``addition'' $\boxplus$ on tuples of states in a straightforward way: 
\begin{align}
    &\left(\rho^{(1)},\ldots,\rho^{(d)}\right)\boxtimes\left(\sigma^{(1)},\ldots,\sigma^{(d)}\right):=\left(\rho^{(1)}\otimes\sigma^{(1)},\ldots,\rho^{(d)}\otimes\sigma^{(d)}\right),\\
    &\left(\rho^{(1)},\ldots,\rho^{(d)}\right)\boxplus\left(\sigma^{(1)},\ldots,\sigma^{(d)}\right):=\left(\rho^{(1)}\oplus\sigma^{(1)},\ldots,\rho^{(d)}\oplus\sigma^{(d)}\right). 
\end{align}
One can check that $S$ together with the operations $\boxplus$, $\boxtimes$ defined above has the structure of a semiring. Additionally, a preorder $\succeq$ can be defined on $S$ by setting 
\begin{equation}
    \left(\rho^{(1)},\ldots,\rho^{(d)}\right)\succeq\left(\sigma^{(1)},\ldots,\sigma^{(d)}\right) 
\end{equation}
if and only if the first tuple of states can be converted to the second by a channel $\Cha$. 

In the recent works \cite{fritz2023,fritz2022} by Tobias Fritz, a real-algebraic theory of preordered semirings was developed in the context of large-sample and catalytic majorization/convertibility.  The central results of these papers, called the Vergleichsstellens\"atze, give conditions for large-sample or catalytic majorization between two elements in the semiring in terms of the ordering of certain monotone homomorphisms (essentially the functions $Q$ as in \eqref{eq:thirdaxiom2}) and monotone derivations on the semiring, which correspond to the divergences. A particular version of such a Vergleichsstellensatz gives rise to the informal theorem discussed above. These techniques not only apply directly to the preordered semirings of tuples of states we are interested in in this thesis, but are applicable to much more general semirings, as well. 

Fritz' methods can be seen as an extension of the theory of asymptotic spectra previously developed by Volker Strassen. In \cite{Strassen1969}, Strassen introduced the first subcubic algorithm for matrix multiplication, by reducing the problem to $2\times2$-matrices and showing that these require only $7<2^3$ operations. Since multiplying two matrices is bilinear, the operation can be represented by a $3$-dimensional tensor, whose asymptotic rank equals the computational complexity of matrix multiplication. This led Strassen to develop a theory of preordered semirings of tensors in general \cite{strassen1986,strassen1987,strassen1988,strassen1991}. The theory was subsequently further generalized to preordered semirings dealing with other objects than just tensors. Strassen's theory and related works have found applications in probability and information theory, and computer science. Another extension of this theory was presented in \cite{Vrana2022}. Besides being applied in the author's works \cite{verhagen2025,verhagen2025alphaz} on which this thesis is based, and the previous matrix majorization study \cite{farooq2024}, these techniques have also been useful recently in other settings in (quantum) information theory \cite{bunth2021asymptotic,perry2022semiring,bunth2023,bunth2024,Rubboli2026,Haapasalo2025}. 

This thesis is organized as follows. In Chapter \ref{cha:VSS}, we introduce the necessary mathematical background on the Vergleichsstellens\"atze from \cite{fritz2022}, and use this to prove an expression for the optimal conversion rate. In Chapters \ref{cha:matrix} and \ref{cha:flat}, we apply the techniques from Chapter \ref{cha:VSS} to the setting of matrix majorization, and subsequently to pairs of pure-component cq-states. This is where we derive sufficient and (almost) necessary conditions for large-sample and catalytic convertibility, and show how the set of relevant divergences is minimal in certain settings. Chapter \ref{cha:matrix} also contains a brief discussion on how our results apply to Gibbs-preserving convertibility of states. Finally, we summarize some open problems in Chapter \ref{cha:open}.




\chapter{Mathematical Techniques}\label{cha:VSS}

\section{A Vergleichsstellensatz}
To derive our results, we need some mathematical machinery that comes in a real-algebraic form and involves preordered semirings in particular. The minimal background required to use these techniques will be presented in this section. All definitions and results presented here are from \cite{fritz2022,fritz2023}, where one can find further details. 

Throughout this work, we follow the convention $\N=\{1,2,\ldots\}$ and denote \sloppy{${[n]:=\{1,\ldots,n\}}$} for all $n\in\N$. We will only study \emph{commutative} algebraic structures. A tuple $(S,+,\cdot,0,1,\preceq)$ is called a {\it preordered semiring} if $(S,+,0)$ and $(S,1,\cdot)$ are commutative semigroups where the multiplication distributes over the addition and $\preceq$ is a preorder (a binary relation which is reflexive and transitive) satisfying
\begin{equation}
x\preceq y\ \Rightarrow\ \left\{\begin{array}{rcl}
x+a&\preceq&y+a,\\
xa&\preceq&ya,
\end{array}\right.
\end{equation}
for all $a\in S$. As above, we usually omit the multiplication dot between elements in the semiring. We denote by $\sim$ the equivalence relation generated by $\preceq$, i.e. $x\sim y$ if and only if there are $z_1,\ldots,z_n\in S$ such that
\begin{equation}
x\preceq z_1\succeq z_2\preceq\cdots\succeq z_n\preceq y.
\end{equation}
The preordered semiring $S$ is of {\it polynomial growth} if it has a {\it power universal} $u\in S$, i.e., for non-zero $x,y\in S$, 
\begin{equation}\label{eq:puproperty}
x\preceq y\quad\Rightarrow\quad\exists k\in\N:\ y\preceq xu^k.
\end{equation}
A preordered semiring $S$ is a {\it preordered semidomain} if
\begin{equation}
\begin{array}{rcl}
xy=0&\Rightarrow&x=0\ {\rm or}\ y=0,\\
0\preceq x\preceq 0&\Rightarrow&x=0.
\end{array}
\end{equation}
Moreover, $S$ is {\it zerosumfree} if $x+y=0$ implies $x=0=y$.
	
Given preordered semirings $S$ and $T$, we say that a map $\Phi:S\to T$ is a {\it monotone homomorphism} if
\begin{itemize}
\item $\Phi(x+y)=\Phi(x)+\Phi(y)$ for all $x,y\in S$ (additivity),
\item $\Phi(xy)=\Phi(x)\Phi(y)$ for all $x,y\in S$ (multiplicativity),
\item $x\preceq y$ $\Rightarrow$ $\Phi(x)\preceq\Phi(y)$ (monotonicity), and
\item $\Phi(0)=0$ and $\Phi(1)=1$.
\end{itemize}
We say that such a monotone homomorphism is {\it degenerate} if $x\preceq y$ implies $\Phi(x)=\Phi(y)$. Otherwise $\Phi$ is {\it non-degenerate}. For our results we need monotone homomorphisms with values in special semirings. These are the following:
\begin{itemize}
\item $\R_+$: The half-line $[0,\infty)$ equipped with the natural addition, multiplication, and total order.
\item $\R_+^{\rm op}$: The same as above but with the reversed order. Together $\R_+$ and $\R_+^{\rm op}$ are often called {\it temperate reals}.
\item $\T\R_+$: The half-line $[0,\infty)$ equipped with the natural multiplication, total order, and the tropical sum $x+y=\max\{x,y\}$.
\item $\T\R_+^{\rm op}$: The same as above but with the reversed order. Together $\T\R_+$ and $\T\R_+^{\rm op}$ are often called {\it tropical reals}.
\end{itemize}
Suppose that $S$ is a preordered semiring and that $\Phi:S\to\R_+$ is a monotone homomorphism. We say that an additive map $\Delta:S\to\R$ is a {\it derivation at $\Phi$} or a {\it $\Phi$-derivation} if it satisfies the {\it Leibniz rule}
\begin{equation}
\Delta(xy)=\Delta(x)\Phi(y)+\Phi(x)\Delta(y)
\end{equation}
for all $x,y\in S$. We are mainly only interested in derivations at degenerate homomorphisms that are also monotone, i.e. satisfy 
\begin{equation}
    x\preceq y\quad\Longrightarrow\quad\Delta(x)\leq\Delta(y). 
\end{equation}

Several results collectively called the ``Vergleichsstellens\"{a}tze'' have been derived in \cite{fritz2022}. Of all of them, we will need the following version:
	
\begin{theorem}[Based on Theorem 8.6 in \cite{fritz2022}]\label{thm:Fritz2022}
Let $S$ be a zerosumfree preordered semidomain with a power universal element $u$. Assume that for some $d\in\N$ there is a surjective homomorphism (called the \emph{norm}) $\|\cdot\|:S\to\R_{>0}^d\cup\{(0,\ldots,0)\}$\footnote{Note that $\R_{>0}^d\cup\{(0,\ldots,0)\}$ with component-wise addition and multiplication forms a semiring.} with trivial kernel and such that
\begin{equation}\label{eq:surjectivehomomorphismproperties}
a\succeq b\ \Rightarrow\ \|a\|=\|b\|\quad {\rm and} \quad \|a\|=\|b\|\ \Rightarrow\ a\sim b.
\end{equation}
Denote the component homomorphisms of $\|\cdot\|$ by $\|\cdot\|_{(j)}$, $j=1,\ldots,d$. Let $x,y\in S\setminus\{0\}$ with $\|x\|=\|y\|$. If 
\begin{enumerate}[(i)]
\item for every $\mb K\in\{\R_+,\R_+^{\rm op},\T\R_+,\T\R_+^{\rm op}\}$ and every non-degenerate monotone homomorphism $\Phi : S \to \mb K$ with trivial kernel, we have $\Phi(x) > \Phi(y)$, and
\item $\Delta(x) > \Delta(y)$ for every monotone $\|\cdot\|_{(j)}$-derivation\footnote{\label{ftn:interchangeability}In fact, this inequality only has to be verified for every derivation up to interchangeability: two derivations of some degenerate monotone homomorphism are called \emph{interchangeable} if their difference takes a (possibly different) constant value on every equivalence class of the relation $\sim$.} $\Delta : S \to \R$ with $\Delta(u) = 1$ for all component indices $j = 1,\ldots,d$,
\end{enumerate}
then
\begin{enumerate}
\item[(a)] there exists a nonzero $a\in S$ such that $ax\succeq ay$, and 
\item[(b)] if additionally $x$ is power universal, then $x^n \succeq y^n$ for all sufficiently large $n\in\N$. 
\end{enumerate}
Conversely, if either of these properties (a) or (b) holds for any $n$ or $a$, then the inequalities (i) and (ii) hold non-strictly.
\end{theorem}

As discussed in the Introduction, large-sample ordering as in (b) in the above theorem implies catalytic ordering as in (a), with a catalyst of the form $a=\sum_{\ell=0}^{n-1} x^\ell y^{n-1-\ell}$ for a sufficiently large $n\in\N$. This was shown in \cite{Duan2005} when $x,y$ are probability vectors, but the proof extends directly to the more abstract setting considered here. The converse implication typically does not hold; Theorem 3 of \cite{Feng_et_al_2006} provides a recipe for deriving counterexamples in a particular setting.

\section{Optimal Conversion Rate}\label{sec:optimalconversionrate}

Here we prove an expression for the optimal rate at which one can ``convert'' some $x\in S$ into another $y\in S$, where $x,y$ satisfy $x,y\succeq1$. This result can then be directly applied to the setting of conversion between tuples of quantum or classical states considered in this thesis. We call $r\geq0$ an \emph{achievable conversion rate} if $r\leq\liminf_{n\to\infty}m_n/n$ where $(m_n)_n$ is a sequence of natural numbers such that $x^n\succeq y^{m_n}$ for $n\in\N$ large enough. The \emph{optimal conversion rate} $r(x\to y)$ is defined as the supremum of all the achievable conversion rates. 

When the semiring $S$ describes tuples of states, it is more natural from an information-theoretic perspective to consider the \emph{divergences} associated with the non-degenerate monotone homomorphisms and monotone derivations of $S$, which we now define explicitly. Fix a power universal $u\in S$. For any non-degenerate homomorphism $\Phi:S\to\mb K$, $\mb K\in\{\R_+,\R_+^{\rm op},\T\R_+,\T\R_+^{\rm op}\}$, with trivial kernel, the associated divergence is
\begin{equation}
D(x)=\frac{\log{\Phi(x)}}{\log{\Phi(u)}},\qquad\forall x\in S\setminus\{0\}.
\end{equation}
For any $k\in\{1,\ldots,d\}$ and any derivation $\Delta$ at $\|\cdot\|_{(k)}$ normalized according to $\Delta(u)=1$, the corresponding divergence is
\begin{equation}
D(x)=\frac{\Delta(x)}{\|x\|_{(k)}},\qquad\forall x\in S\setminus\{0\}.
\end{equation}
Denote by $\mc D(S)$ the set of all the divergences defined in this way. We note that, whenever $x\succeq y$, then $D(x)\geq D(y)$ for all $D\in\mc D(S)$, i.e. the data-processing inequality holds. Moreover, $D(xy)=D(x)+D(y)$ for any $x,y\in S\setminus\{0\}$ and $D\in \mc D(S)$, i.e. additivity holds. 

According to Proposition 8.5 of \cite{fritz2022}, the set $\mc D(S)$ is a compact Hausdorff space with respect to the coarsest topology where the comparison maps
\begin{equation}
\mc D(S)\ni D\mapsto D(x)-D(y)\in\R,\qquad\forall x,y\in S\setminus\{0\},\ x\sim y,
\end{equation}
are continuous. Furthermore, note that the conditions (i) and (ii) of Theorem \ref{thm:Fritz2022} are now equivalent to the more concise condition that $D(x)>D(y)$ for all $D\in\mc D(S)$. In the setting of semirings describing tuples of states, i.e. $x\succeq1$, we have for any $D\in\mc D(S)$, using the data-processing inequality and additivity, that $D(x)\geq D(1)=0$. We now state and prove the following result on the optimal conversion rate, which is an immediate consequence of Theorem \ref{thm:Fritz2022}. 

\begin{proposition}\label{prop:optimalconversionrate}
    Consider a preordered semiring $S$ as in Theorem \ref{thm:Fritz2022}, with associated divergences $\mc D(S)$. Let $x,y\in S$, where $x$ is power universal, satisfy $x,y\succeq1$ and $D(x)>0$ for all $D\in\mc D(S)$. The optimal rate of conversion is given by 
    \begin{equation}\label{eq:optimalrategeneral}
        r(x\to y)=\min_{D\in\mc D(S)}\frac{D(x)}{D(y)}. 
    \end{equation}
\end{proposition}

Note that the quotient $D(x)/D(y)$ above may be infinite for some $D\in\mc D(S)$. In case it is infinite for all $D\in\mc D(S)$, we consider the minimum on the right-hand side above to be infinite. In all other cases, according to the compactness result referred to previously, the minimum of these quotients exists and is finite. Indeed, it follows from $x\succeq1$ that $x\sim1$. Hence, the function $\mc D(S)\ni D\mapsto D(x)=D(x)-D(1)$ is continuous. Similarly, $D\mapsto D(y)$ is continuous. Since $D(x)>0$, the function $D\mapsto D(y)/D(x)$ is continuous. Using the compactness of $\mc D(S)$ and considering the inverse of this function it now follows that the minimum above exists. One can also show that the condition $D(x)>0$ for all $D\in\mc D(S)$ appearing in the Proposition follows from the fact that $x$ is power universal. 

\begin{proof}
    We first show that 
    \begin{equation}\label{eq:upperboundrate}
        r(x\to y)\leq\min_{D\in\mc D(S)}\frac{D(x)}{D(y)}. 
    \end{equation}
    The proof for this is completely standard and does not use the fact that the divergences in $\mc D(S)$ characterize large-sample ordering. Let $D\in\mc D(S)$, and $r$ an achievable rate, i.e. there exists a sequence of natural numbers $(m_n)_{n=1}^\infty$ such that $x^n\succeq y^{m_n}$ for $n$ large enough and $r\leq\liminf_{n\rightarrow\infty}m_n/n$. By monotonicity of $D$, we have $D(x^n)\geq D(y^{m_n})$ for $n$ large enough. Then by additivity this implies that $m_n/n\leq D(x)/D(y)$ for $n$ large enough, where we note that $D(x)/D(y)$ may be infinite. Therefore, 
    \begin{equation}
        r\leq\liminf_{n\rightarrow\infty}\frac{m_n}{n}\leq\frac{D(x)}{D(y)}. 
    \end{equation}
    By taking the supremum over all achievable rates $r$ and then the minimum over all $D\in\mc D(S)$, we arrive at \eqref{eq:upperboundrate}. 

    Next, using Theorem \ref{thm:Fritz2022}, we prove that any $r<\min_{D\in\mc D(S)}D(x)/D(y)$ is an achievable rate. Together with \eqref{eq:upperboundrate}, this implies \eqref{eq:optimalrategeneral}. Consider any $r<\min_{D\in\mc D(S)}D(x)/D(y)$ and let $p,q\in\N$ be such that 
    \begin{equation}
        r\leq\frac{p}{q}<\min_{D\in\mc D(S)}\frac{D(x)}{D(y)}. 
    \end{equation}
    From this, it follows that $qD(x)>pD(y)$ for any $D\in\mc D(S)$. Then, by additivity, 
    \begin{equation}
        D(x^q)=qD(x)>pD(y)=D(y^p). 
    \end{equation}
    Note that since $x$ is power universal and $x\succeq1$, $x^q$ is also power universal. It now follows from Theorem \ref{thm:Fritz2022} that $x^{nq}\succeq y^{np}$ for all $n$ larger than or equal to some $n_0\in\N$. Any $k\in\N$ can be written as $k=n_kq+l_k$ with unique $n_k\in\N$ and $l_k\in\{0,\ldots,q-1\}$. This allows us to define $m_k:=n_kp$. Then, using that $x\succeq1$, we have for $k\geq n_0q$, 
    \begin{equation}
        x^k=x^{n_kq+l_k}=x^{n_kq}x^{l_k}\succeq x^{n_kq}\succeq y^{n_kp}=y^{m_k}. 
    \end{equation}
    Also, 
    \begin{equation}
        \liminf_{k\rightarrow\infty}\frac{m_k}{k}=\liminf_{k\rightarrow\infty}\frac{p}{q+l_k/n_k}=\frac{p}{q}\geq r. 
    \end{equation}
    This establishes that $r$ is an achievable rate. 
\end{proof}

\chapter{Matrix Majorization}\label{cha:matrix}

This chapter is based on the author's work \cite{verhagen2025}, with an additional result on the minimality of the set of monotones characterizing large-sample and catalytic majorization in some settings (Section \ref{Msec:minimality}), and additional discussion on defining a matrix majorization semiring based on the supports of a matrix' columns (Section \ref{Msec:pusupports}). 

\section{Introduction}

{\it Statistical experiments} are a mathematical description of the distributions of outcomes from an experiment on some physical system. They are described by a family $P=\{p^\tj\}_{\tj\in\Theta}$ of probability measures or vectors $p^\tj$ on the space of outcomes of the experiment, where the index set $\Theta$ represents the different states the system can be in. To avoid redundancy, it is important to characterize when one experiment is more informative than another, i.e.,\ we need to be able to compare experiments. Due to results by Blackwell \cite{Blackwell51,Blackwell53}, the information order of experiments can be characterized in an intuitive way: An experiment $P=\{p^\tj\}_{\tj\in\Theta}$ is more informative than an experiment $Q=\{q^\tj\}_{\tj\in \Theta}$, or $P$ {\it majorizes} $Q$, if and only if there is a stochastic map $T$ such that $q^\tj=Tp^\tj$ for all $\tj\in\Theta$. In this case, we write $P\rgeq Q$.

Majorization of one experiment $P=\{p^{\tj}\}_{\tj\in\Theta}$ over another experiment $Q=\{q^\tj\}_{\tj\in\Theta}$ can also be aided by having many copies of the experiments or by catalysis. The former situation means that
\begin{equation}
\left\{\big(p^\tj\big)^{\otimes n}\right\}_{\tj\in\Theta}\rgeq\left\{\big(q^\tj\big)^{\otimes n}\right\}_{\tj\in\Theta}
\end{equation}
for sufficiently large $n\in\N$. Note that $p^{\otimes n}$ denotes the $n$-th Kronecker product power of $p$. We call this {\it majorization in large samples}. The latter situation means that there exists a third experiment $R=\{r^\tj\}_{\tj\in\Theta}$ (a catalyst) such that
\begin{equation}
\big\{p^\tj\otimes r^\tj\big\}_{\tj\in\Theta}\rgeq\big\{q^\tj\otimes r^\tj\big\}_{\tj\in\Theta}.
\end{equation}
This is called {\it catalytic majorization}. Large-sample majorization implies catalytic majorization \cite{Duan2005}, but the reverse implication is false in general \cite{Feng_et_al_2006}. However, the exact difference between large-sample and catalytic majorization is still unclear. In this work, we are particularly interested in large-sample and catalytic majorization of finite experiments \cite{Heyer1982} where $\Theta=\{1,\ldots,d\}$ is a finite set and the probability measures of all experiments are finite probability vectors. Majorization in this case is also called {\it matrix majorization} \cite{dahl99}. When we identify a finite experiment $P=\big(p^{(1)},\ldots,p^{(d)}\big)$ with a matrix with columns $p^{(k)}$, majorization $P\rgeq Q$ is equivalent to the existence of a stochastic matrix $T$ such that $Q=TP$. Note that a stochastic matrix is a matrix with non-negative entries where all the columns are probability vectors.

In the case $\Theta=\{1,2\}$, majorization in large samples (under certain regularity conditions) is characterized by inequalities involving the R\'{e}nyi divergences \cite{Mu_et_al_2020}. In \cite{farooq2024}, this result was extended to the general case $\Theta=\{1,\ldots,d\}$ for any $d\in\N$ in the case of finite experiments. According to \cite{farooq2024}, matrix majorization in large samples is characterized by multivariate generalizations of the R\'{e}nyi divergences: For any $\underline{\alpha}=(\alpha_1,\ldots,\alpha_d)\in\R^d$ such that $\max_{1\leq k\leq d}\alpha_k\neq 1$, we define the quantity $D_{\underline{\alpha}}(P)$,
\begin{equation}\label{Meq:multivarD}
D_{\underline{\alpha}}(P)=\frac{1}{\max_{1\leq k\leq d}\alpha_k-1}\log\sum_{i=1}^n \big(p^{(1)}_i\big)^{\alpha_1}\cdots\big(p^{(d)}_i\big)^{\alpha_d}, 
\end{equation}
where $P=\big(p^{(1)},\ldots,p^{(d)}\big)$ and $p^{(k)}=\big(p^{(k)}_1,\ldots,p^{(k)}_n\big)$ for $k=1,\ldots,d$. The relevant quantities are those $\underline{\alpha}\in\R^d$ with 
\begin{equation}\label{Meq:A+}
\underline{\alpha}\in A_+:=\{(\alpha_1,\ldots,\alpha_d)\in\R^d\,|\,\alpha_1+\cdots+\alpha_d=1,\ \alpha_k\geq0,\ k=1,\ldots,d\}
\end{equation}
or
\begin{equation}\label{Meq:A-}
\underline{\alpha}\in A_-:=\bigcup_{k=1}^d \{(\alpha_1,\ldots,\alpha_d)\in\R^d\,|\,\alpha_1+\cdots+\alpha_d=1,\ \alpha_k>0,\ \alpha_\ell\leq0,\ \ell\neq k\}
\end{equation}
and not equal to one of the basis vectors $e_k$, $k=1,\ldots,d$. Also, one has to consider certain pointwise limits of $D_{\underline{\alpha}}$ in the $\underline{\alpha}$-parameter space. Roughly speaking, if $P$ attains a higher value than $Q$ in all these quantities, $P$ majorizes $Q$ in large samples. For illustration of the quantities in \eqref{Meq:multivarD} and their pointwise limits in the case $d=3$, see Figure \ref{Mfig:Restricted}.

\begin{figure}
\begin{center}
\begin{overpic}[scale=0.45,unit=1mm]{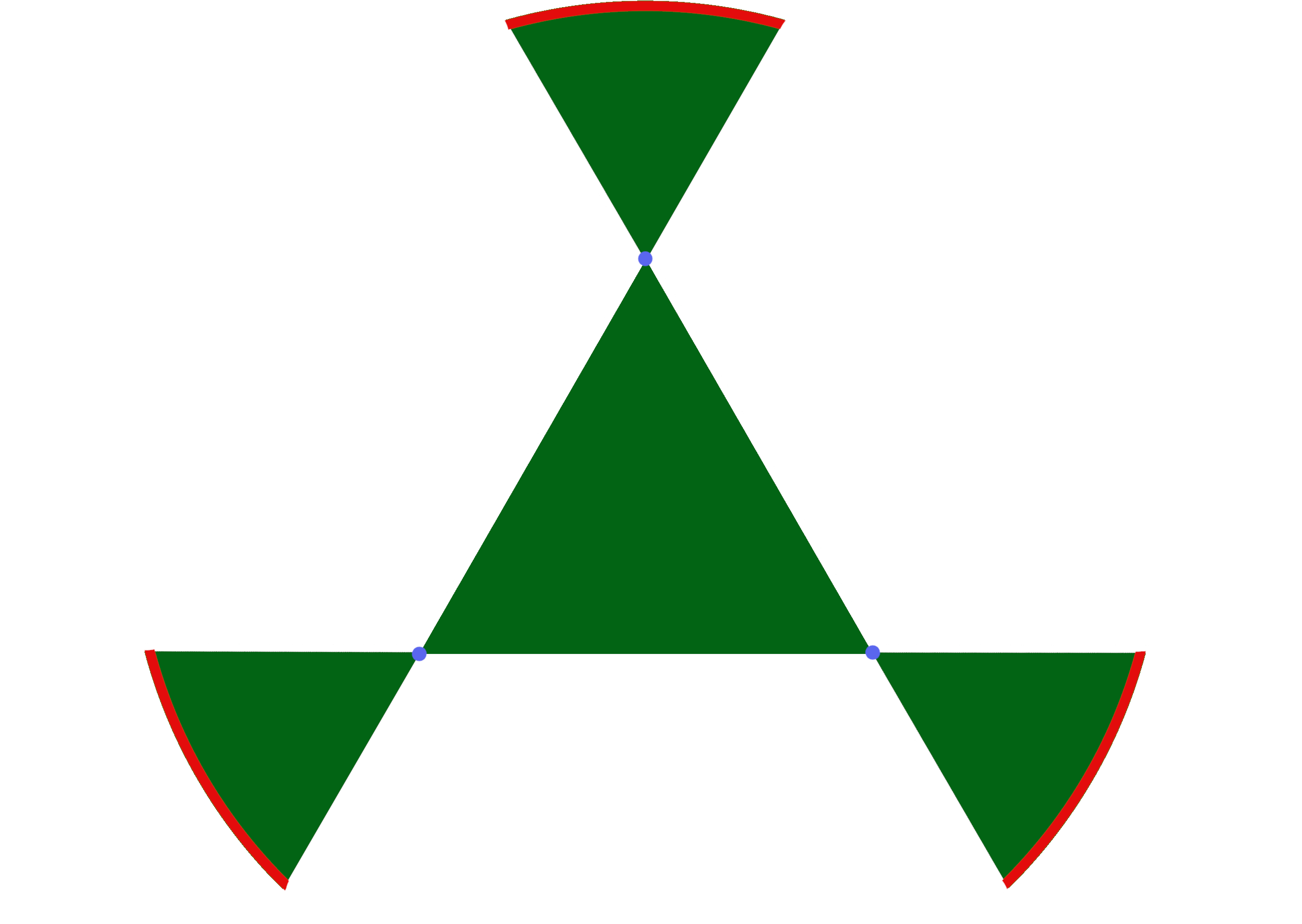}
\put(45,30){\begin{Huge}
${\color{white} A_+}$
\end{Huge}}
\put(17,12){\begin{Large}
${\color{white} A_-}$
\end{Large}}
\put(77,12){\begin{Large}
${\color{white} A_-}$
\end{Large}}
\put(47,62){\begin{Large}
${\color{white} A_-}$
\end{Large}}
\put(31,37){\begin{Large}
\darker{$D_{\underline{\alpha}}$}
\end{Large}}
\put(25,7){\begin{Large}
\darker{$D_{\underline{\alpha}}$}
\end{Large}}
\put(68,7){\begin{Large}
\darker{$D_{\underline{\alpha}}$}
\end{Large}}
\put(37,60){\begin{Large}
\darker{$D_{\underline{\alpha}}$}
\end{Large}}
\put(25,25){\begin{large}
\sininen{$\Delta^{(1)}_{\underline{\gamma}}$}
\end{large}}
\put(68,25){\begin{large}
\sininen{$\Delta^{(2)}_{\underline{\gamma}}$}
\end{large}}
\put(52,52){\begin{large}
\sininen{$\Delta^{(3)}_{\underline{\gamma}}$}
\end{large}}
\put(5,8){\begin{Large}
\punainen{$D^\T_{\underline{\beta}}$}
\end{Large}}
\put(86,8){\begin{Large}
\punainen{$D^\T_{\underline{\beta}}$}
\end{Large}}
\put(62,70){\begin{Large}
\punainen{$D^\T_{\underline{\beta}}$}
\end{Large}}
\end{overpic}
\caption{\label{Mfig:Restricted} The functions of \eqref{Meq:multivarD} together with their pointwise limits depicted in the case $d=3$. The figure is contained in the affine plane of $\R^3$ of those points whose coordinates sum up to 1. The maps of \eqref{Meq:multivarD} are associated with the areas $A_+$ and $A_-$ and are depicted in green. Note that, e.g.,\ in the lower-left part of $A_-$, we have $\alpha_1>1$ and $\alpha_2,\alpha_3\leq0$. The red arcs represent points at infinity and they correspond to the functions introduced later in \eqref{Meq:DbetaTropic}. The blue points correspond to the standard basis vectors $e_1=(1,0,0)$, $e_2=(0,1,0)$, and $e_3=(0,0,1)$. The pointwise limits of the maps in \eqref{Meq:multivarD} as $\underline{\alpha}\to e_1$ are $\Delta^{(1)}_{\underline{\gamma}}(P):=\gamma_2 D_1\big(p^{(1)}\big\|p^{(2)}\big)+\gamma_3 D_1\big(p^{(1)}\big\|p^{(3)}\big)$ with $\gamma_2,\gamma_3\geq0$ which depend on the angle of approach. Similar limits exist for $\underline{\alpha}\to e_2$ or $e_3$. Here $D_1$ is the Kullback--Leibler divergence. Interestingly, as shown in \cite{Csiszar2003}, we have $\min_{p^{(1)}}\Delta^{(1)}_{\underline{\gamma}}(P)\propto D_{\frac{\gamma_2}{\gamma_2+\gamma_3}}(p^{(2)}\|p^{(3)})$. In \cite{mosonyi2024geometric}, generalizations to larger $d$ were derived and involve the quantities $D_{\underline{\alpha}}$ in \ref{Meq:multivarD}. The situation depicted in this figure generalizes to the general $d\in\N$ case in a straightforward manner: the parameter area $A_+$ corresponds to the $(d-1)$-dimensional probability simplex and $A_-$ extends from the extreme points of the probability simplex as cones reaching to infinity.}
\end{center}
\end{figure}

Both in \cite{Mu_et_al_2020} and \cite{farooq2024} it was assumed that, within each statistical experiment, all the probability measures or vectors share the same support. In the finite case this means that, for an experiment $P=\big(p^{(1)},\ldots,p^{(d)}\big)$ consisting of probability vectors $p^{(k)}$, we have
\begin{equation}
{\rm supp}\,p^{(1)}=\cdots={\rm supp}\,p^{(d)}, 
\end{equation}
where ${\rm supp}\,p$ is the set of entries where $p$ is non-zero (the support of $p$). In this work, we consider more general support conditions. We first study the case where there are no support conditions except for the baseline
\begin{equation}\label{Meq:baselinecond}
\bigcap_{k=1}^d {\rm supp}\,p^{(k)}\neq\emptyset
\end{equation}
within each experiment $P=\big(p^{(1)},\ldots,p^{(d)}\big)$. We call this the \emph{minimal restrictions} case. Note that one could consider settings with even less restrictions by dropping the condition \eqref{Meq:baselinecond} and requiring instead that, for instance, the supports of each pair of probability vectors $p^{(k)}$, $p^{(k^\prime)}$ must have a non-empty intersection. However, this would restrict the divergences that characterize large-sample and catalytic majorization to just the bivariate R\'enyi divergences between each pair of probability vectors. Since we want to study their multivariate generalizations of the type \eqref{Meq:multivarD}, we impose slightly stronger conditions in order for these divergences to make an appearance. We will return to this issue in Remark \ref{Mrmk:minimalrestrictions}. 

The minimal restrictions case is analyzed in Section \ref{Msec:findings}, which culminates in the main results Theorem \ref{Mthm:majorization} for exact large-sample or catalytic majorization, and Theorem \ref{Mthm:genasymptocatalytic} for the asymptotic case where we consider approximate large-sample or catalytic majorization with vanishing error. 

In the second case, we study the case where the support of the final column dominates the others, i.e.,\ in addition to \eqref{Meq:baselinecond}, we require
\begin{equation}
{\rm supp}\,p^{(k)}\subseteq{\rm supp}\,p^{(d)}
\end{equation}
for $k=1,\ldots,d-1$. We call this the \emph{dominating column} case and discuss it in Section \ref{Msec:Asymmetric}. The main results here are Theorem \ref{Mthm:majorizationdominatingcolumn} for exact large-sample or catalytic majorization, and Theorem \ref{Mthm:asymptocatalytic} for asymptotic large-sample or catalytic majorization for $d=2$. 


In both the minimal restrictions and dominating column cases mentioned above, matrix majorization in large samples is characterized by inequalities involving subsets of the functions $D_{\underline{\alpha}}$. In the first case, only those $D_{\underline{\alpha}}$ where $\underline{\alpha}\in A_+$ appear in the sufficient conditions for majorization in large samples. However, on the boundary of $A_+$ some additional quantities surface. Simply put, in each lower-dimensional face of the simplex $A_+$, a set of additional quantities have to be considered. Note that the faces of $A_+$ are of the form
\begin{align}
A_C:=\{(\alpha_1,\ldots,\alpha_d)\in A_+\,|\,\alpha_k=0,\ k\in \{1,\ldots,d\}\setminus C\}
\end{align}
for some $C\subseteq\{1,\ldots,d\}$. With $C\subsetneq\{1,\ldots,d\}$, the new quantities arising from $A_C$ are essentially the original quantities when we restrict the vectors of the experiment $P=\big(p^{(1)},\ldots,p^{(d)}\big)$ onto $\bigcap_{k\in C}{\rm supp}\,p^{(k)}$. The occurrence of new quantities can be seen as a consequence of the non-commutativity of the limits $x\downarrow0$ and $\alpha\rightarrow0$: 
\begin{equation}
\lim_{x\downarrow0}x^\alpha=0\quad\text{for fixed }\alpha>0,\quad\quad\quad\lim_{\alpha\rightarrow0}x^\alpha=1\quad\text{for fixed }x>0.
\end{equation}
To be exact, the new quantities arise when ${\rm supp}\,\underline{\alpha}\subsetneq\{1,\ldots,d\}$, in which case we consider all ${\rm supp}\,\underline{\alpha}\subseteq C\subseteq\{1,\ldots,d\}$ and the new quantities are given by essentially the same formula as \eqref{Meq:multivarD} with the exception that the $i$-sum runs over only the set $\bigcap_{c\in C}{\rm supp}\,p^{(c)}$. In the case of one dominating column, the set of relevant $\underline{\alpha}$ is extended as the one-parameter families
\begin{equation}
D_{\underline{\alpha}(\lambda)},\quad \underline{\alpha}(\lambda)=(0,\ldots,0,\underbrace{1+\lambda}_{k{\rm th\ entry}},0,\ldots,0,-\lambda),\qquad\lambda\geq0,\quad k=1,\ldots,d-1,
\end{equation}
including the points at infinity where $\lambda\to\infty$ also have to be considered.

This work complements earlier work particularly in \cite{farooq2024} by filling in the missing cases that could not be dealt with in \cite{farooq2024} because of the support assumptions imposed there. The simplest complete results we obtain deals with dichotomies. By dichotomy, we mean a pair $\big(p^{(1)},p^{(2)}\big)$ of probability vectors where the second dominates the first, i.e.,\ ${\rm supp}\,p^{(1)}\subseteq{\rm supp}\,p^{(2)}$. Given dichotomies $\big(p^{(1)},p^{(2)}\big)$ and $\big(q^{(1)},q^{(2)}\big)$ with the innocuous assumption that $p^{(1)}\neq p^{(2)}$, we have that $\big(p^{(1)},p^{(2)}\big)$ majorizes $\big(q^{(1)},q^{(2)}\big)$ in large samples approximately with vanishing error if and only if
\begin{equation}
\label{Meq:dichotomiesconditionsintroduction}
D_\alpha\big(p^{(1)}\big\|p^{(2)}\big)\geq D_\alpha\big(q^{(1)}\big\|q^{(2)}\big),\qquad D_\alpha\big(p^{(2)}\big\|p^{(1)}\big)\geq D_\alpha\big(q^{(2)}\big\|q^{(1)}\big)
\end{equation}
for all $\alpha\geq 1/2$ where $D_\alpha(\cdot\|\cdot)$ are the R\'{e}nyi divergences (see \eqref{Meq:RenyiDivergences} for a definition of the R\'enyi divergences). The same conditions \eqref{Meq:dichotomiesconditionsintroduction} apply in the case of approximate catalytic majorization with vanishing error. See Corollary \ref{Mcor:asymptocatalytic} for an exact formulation of this result. We have similar results for general finite statistical experiments for $d\in\N$ (dichotomies being an example of the $d=2$ case) involving the multivariate generalizations $D_{\underline{\alpha}}$ of the R\'{e}nyi divergences in addition to the new quantities arising in the case where ${\rm supp}\,\underline{\alpha}\subsetneq\{1,\ldots,d\}$ as described earlier. 

The simplest case where both the multivariate divergences and these new quantities come into play is when we consider majorization of tuples of three probability vectors $P=(p^{(1)},p^{(2)},p^{(3)})$. In the following, we will illustrate the conditions for \emph{exact} large-sample or catalytic majorization of $P$ over another tuple $Q=(q^{(1)},q^{(2)},q^{(3)})$ in the setting of one dominating column, say the last, i.e. ${\rm supp}\,p^{(1)},{\rm supp}\,p^{(2)}\subseteq{\rm supp}\,p^{(3)}$. This is the case $d=3$ in Theorem \ref{Mthm:majorizationdominatingcolumn}. The majorization conditions can be formulated in terms of the divergences $D_{\underline{\alpha}}$ and the R\'enyi divergences $D_\alpha$. Figure \ref{Mfig:DominatingColumn} gives a depiction of the parameter vectors $\underline{\alpha}$ and parameters $\alpha$ that are applicable. The precise conditions for large-sample and catalytic majorization are: 

\begin{align}
    &D_{\underline{\alpha}}(P)>D_{\underline{\alpha}}(Q)&\forall\underline{\alpha}\in{\rm int}\,A_+,\label{Meq:d3condition1}\\
    &D_\alpha\big(p^{(1)}\|p^{(3)}\big)>D_\alpha\big(q^{(1)}\|q^{(3)}\big)&\forall\alpha\in[0,\infty],\label{Meq:d3condition2}\\
    &D_\alpha\big(p^{(1)}|_{{\rm supp}\,p^{(2)}}\|p^{(3)}|_{{\rm supp}\,p^{(2)}}\big)&\forall\alpha\in[0,1),\label{Meq:d3condition3}\\
    &\qquad\qquad>D_\alpha\big(q^{(1)}|_{{\rm supp}\,q^{(2)}}\|q^{(3)}|_{{\rm supp}\,q^{(2)}}\big)&\nonumber\\
    &D_\alpha\big(p^{(2)}\|p^{(3)}\big)>D_\alpha\big(q^{(2)}\|q^{(3)}\big)&\forall\alpha\in[0,\infty],\label{Meq:d3condition4}\\
    &D_\alpha\big(p^{(2)}|_{{\rm supp}\,p^{(1)}}\|p^{(3)}|_{{\rm supp}\,p^{(1)}}\big)&\forall\alpha\in[0,1),\label{Meq:d3condition5}\\
    &\qquad\qquad>D_\alpha\big(q^{(2)}|_{{\rm supp}\,q^{(1)}}\|q^{(3)}|_{{\rm supp}\,q^{(1)}}\big)&\nonumber\\
    &D_\alpha\big(p^{(1)}\|p^{(2)}\big)>D_\alpha\big(q^{(1)}\|q^{(2)}\big)&\forall\alpha\in[0,1/2],\label{Meq:d3condition6}\\
    &D_\alpha\big(p^{(2)}\|p^{(1)}\big)>D_\alpha\big(q^{(2)}\|q^{(1)}\big)&\forall\alpha\in[0,1/2].\label{Meq:d3condition7}
\end{align}

For the first condition \eqref{Meq:d3condition1} we consider $D_{\underline{\alpha}}$ for all parameter vectors $\underline{\alpha}$ in ${\rm int}\,A_+$, the interior of $A_+$ (defined in \eqref{Meq:A+}), i.e. all $\underline{\alpha}=(\alpha_1,\alpha_2,\alpha_3)\in\R^3$ such that $\alpha_1+\alpha_2+\alpha_3=1$ and $0<\alpha_k<1$ for $k=1,2,3$. Condition \eqref{Meq:d3condition2} is in terms of the R\'enyi divergence for $\alpha\in[0,\infty]$ between the first and third vectors. This coincides with the half-line parameterized by $(\alpha,0,1-\alpha)$ in Figure \ref{Mfig:DominatingColumn} starting from the top vertex of the triangle (corresponding to $\alpha=0$) and passing through the lower left vertex (corresponding to $\alpha=1$, i.e. the Kullback--Leibler divergence). In the limit $\alpha\rightarrow\infty$ this gives the max-divergence $D_\infty$. In condition \eqref{Meq:d3condition3}, we denote by $r|_{{\rm supp}\,s}$ the restriction of the vector $r$ to the support of $s$, i.e. $\big(r|_{{\rm supp}\,s}\big)_i=r_i$ if $i\in{\rm supp}\,s$ and $r|_{{\rm supp}\,s}=0$ otherwise. These divergences can be seen as point-wise limits of $D_{(\alpha_1,\alpha_2,\alpha_3)}$ when $\alpha_2\rightarrow0$, similarly as the R\'enyi divergence $D_0$ is the point-wise limit of $D_\alpha$ when $\alpha\rightarrow0$. Conditions \eqref{Meq:d3condition4} and \eqref{Meq:d3condition5} are similar to the preceding two, and correspond to the other half-line in Figure \ref{Mfig:DominatingColumn}. Finally, \eqref{Meq:d3condition6} and \eqref{Meq:d3condition7} correspond to the lower line segment of the triangle parameterized by $(\alpha,1-\alpha,0)$ for $\alpha\in[0,1]$. 

\begin{figure}
\vspace{5mm}
\begin{center}
\begin{overpic}[scale=0.4,unit=1mm]{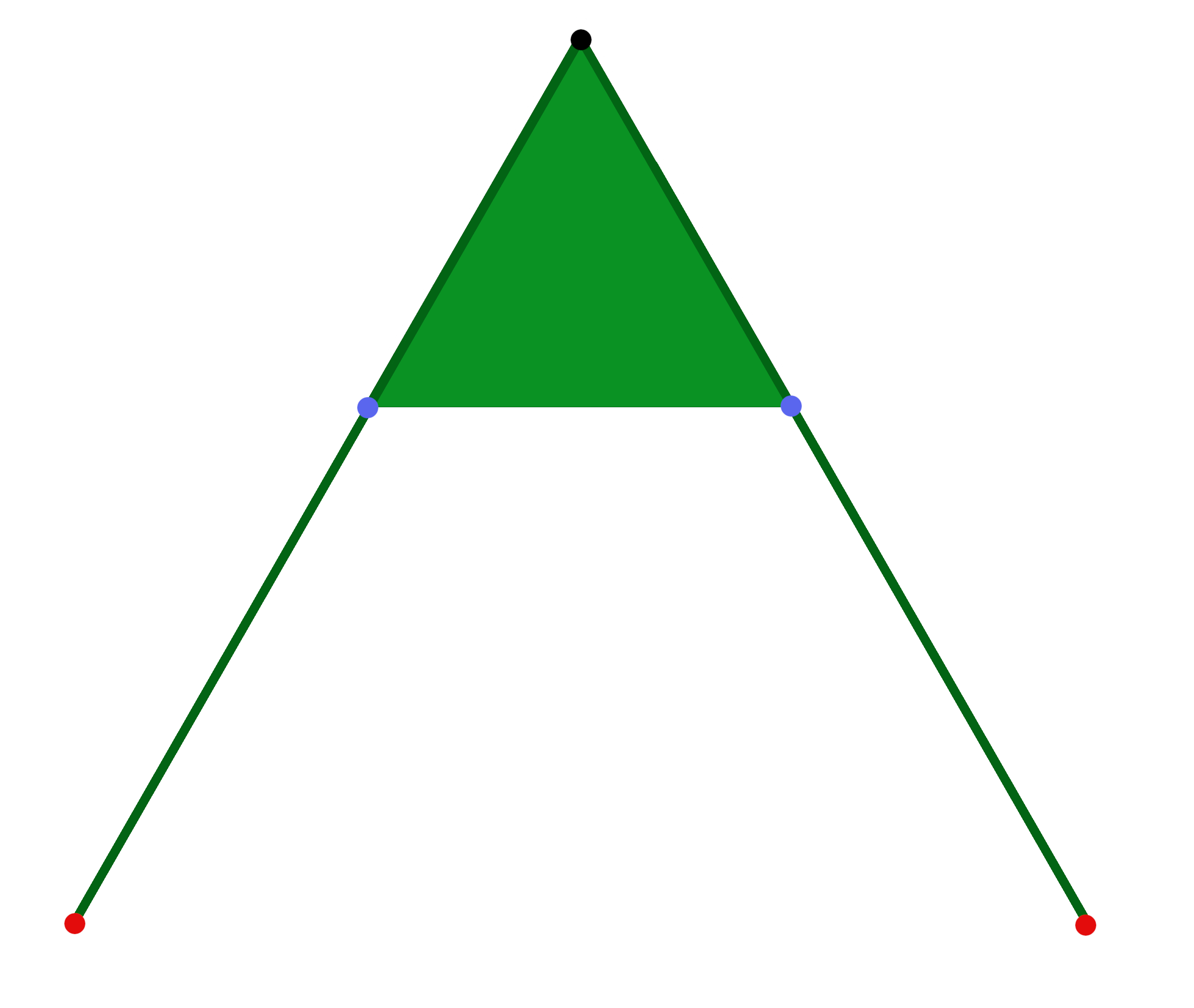}
\put(42,60){\begin{Huge}
$A_+$
\end{Huge}}
\put(-3,86){
\lighter{$D_0\big(p^{(1)}|_{{\rm supp}\,p^{(2)}}\|p^{(3)}|_{{\rm supp}\,p^{(2)}}\big),$} \darker{$D_0\big(p^{(2)}\big\|p^{(3)}\big)$, $D_0\big(p^{(1)}\big\|p^{(3)}\big)$}
}
\put(29,67){
\textbf{(a1)}
}
\put(59,67){
\textbf{(b1)}
}
\put(8,30){
\textbf{(a2)}
}
\put(80,30){
\textbf{(b2)}
}
\put(2,53){
\sininen{$D_1\big(p^{(1)}\big\|p^{(3)}\big)$}
}
\put(2,47){
\darker{$D_0\big(p^{(2)}\big\|p^{(1)}\big)$}
}
\put(70,53){
\sininen{$D_1\big(p^{(2)}\big\|p^{(3)}\big)$}
}
\put(70,47){
\darker{$D_0\big(p^{(1)}\big\|p^{(2)}\big)$}
}
\put(-12,1){
\punainen{$D_\infty\big(p^{(1)}\big\|p^{(3)}\big)$}
}
\put(78,1){
\punainen{$D_\infty\big(p^{(2)}\big\|p^{(3)}\big)$}
}
\put(0,67){
\lighter{$D_{\underline{\alpha},\{1,2,3\}}(P)$}
}
\thicklines
\put(24,68){\vector(1,-0.7){17}}
\put(35,46){
\textbf{(c1)}
}
\put(53,46){
\textbf{(c2)}
}
\end{overpic}
\caption{\label{Mfig:DominatingColumn} A subset of the set of parameter vectors that are depicted in Figure \ref{Mfig:Restricted}. These represent the family of divergences that appear in the conditions for exact large-sample and catalytic majorization for matrices $P=(p^{(1)},p^{(2)},p^{(3)})$ where the support of the last column dominates the support of the other two. The interior of the triangle $A_+$ corresponds to the trivariate divergences $D_{\underline{\alpha}}(P)$ given by \eqref{Meq:multivarD}. The two half lines beginning at $e_3$, going through $e_1$, respectively $e_2$, and extending to infinity represent the R\'enyi divergences $D_\alpha\big(p^{(1)}\|p^{(3)}\big)$ (\textbf{(a1), (a2)}), respectively $D_\alpha\big(p^{(2)}\|p^{(3)}\big)$ (\textbf{(b1), (b2)}), for $\alpha\in[0,\infty]$. The side of the triangle between $e_1$ and $e_2$ is split in half at $\underline{\alpha}=(1/2,1/2,0)$: the left half \textbf{(c1)} represents $D_\alpha\big(p^{(2)}\|p^{(1)}\big)$ for $\alpha\in[0,1/2]$ and the right half \textbf{(c2)} represents $D_\alpha\big(p^{(1)}\|p^{(2}\big)$ for $\alpha\in[0,1/2]$. Additionally, on the two sides of the triangle between $e_3$ and $e_1$, respectively $e_2$, we have the R\'enyi divergences between two columns when restricted to the support of the remaining column: $D_\alpha\big(p^{(1)}|_{{\rm supp}\,p^{(2)}}\|p^{(3)}|_{{\rm supp}\,p^{(2)}}\big)$ \textbf{(a1)}, respectively $D_\alpha\big(p^{(2)}|_{{\rm supp}\,p^{(1)}}\|p^{(3)}|_{{\rm supp}\,p^{(1)}}\big)$ \textbf{(b1)}. 
}
\end{center}
\end{figure}

At first glance, Theorems \ref{Mthm:majorization} and \ref{Mthm:majorizationdominatingcolumn} for the minimal restrictions and dominating column case, respectively, are very similar and the former seems to apply to a larger set of matrices than the latter. On top of that, the former theorem requires less inequalities to be satisfied. Therefore, one may wonder why we would need the latter theorem at all. The important thing to note here is that the strict inequalities in the conditions that the matrices $P$ and $Q$ have to satisfy force $P$ to be of a very specific form, it has to be \emph{power universal} (more on this later). In the case of Theorem \ref{Mthm:majorization}, these conditions require the support of $P=(p^{(1)},\ldots,p^{(d)})$ to satisfy the following: 
\begin{equation}
    {\rm supp}\,p^{(k)}\nsubseteq{\rm supp}\,p^{(k^\prime)}\text{ for all }k,k^\prime\in\{1,\ldots,d\}\text{ with }k\neq k^\prime. 
\end{equation}
However, for Theorem \ref{Mthm:majorizationdominatingcolumn}, $P$ has to satisfy 
\begin{equation}
    {\rm supp}\,p^{(k)}\subsetneq{\rm supp}\,p^{(d)}\text{ for all }k\in\{1,\ldots,d-1\}, 
\end{equation}
and 
\begin{equation}
    {\rm supp}\,p^{(k)}\nsubseteq{\rm supp}\,p^{(k^\prime)}\text{ for all }k,k^\prime\in\{1,\ldots,d-1\}\text{ with }k\neq k^\prime. 
\end{equation}
Therefore, both theorems apply to two disjoint types of matrices, depending on their support structure. In principle, one could derive similar theorems for matrices with different kinds of support structures, where the family of divergences that appear in the conditions of the theorem are different for each type of support structure. We will discuss this point in more detail in the Conclusion and outlook section. To understand why the type of matrices in Theorem \ref{Mthm:majorization} require less inequalities to be satisfied than those in Theorem \ref{Mthm:majorizationdominatingcolumn}, one can consider the catalytic case: both theorems guarantee the existence of a catalyst that is taken from the same set of matrices that $P$ and $Q$ are from. Hence, when this set is larger (which it is for Theorem \ref{Mthm:majorization} compared to Theorem \ref{Mthm:majorizationdominatingcolumn}), it intuitively makes sense that we require less conditions on $P$ and $Q$. 


Theorems \ref{Mthm:majorization} and \ref{Mthm:majorizationdominatingcolumn} also give an expression for the maximal achievable rate at which a tuple of probability distributions $P=(p^{(1)},\ldots,p^{(d)})$ can be transformed into another tuple $Q=(q^{(1)},\ldots,q^{(d)})$ using the same stochastic map for each distribution. This follows directly from Proposition \ref{prop:optimalconversionrate}. Denoting by $\mathcal{D}$ set of divergences associated with the monotone homomorphisms and derivations appearing in the conditions for large-sample and catalytic majorization, the maximal achievable rate can be expressed as 
\begin{equation}
r(P\to Q)=\min_{D\in\mathcal{D}}\,\frac{D(P)}{D(Q)}. 
\end{equation}

One can wonder whether the set of divergences that characterize large-sample and catalytic majorization is larger than necessary. In principle, there might be a smaller subset that, when properly ordered when evaluated on two matrices, guarantees large-sample or catalytic majorization. For the minimal restrictions case, we show that the set of divergences found is indeed minimal, in the following sense: one can define a natural topology on the set of divergences, and when any non-empty open subset is removed from this set, the remaining divergences no longer guarantee large-sample or catalytic majorization. This result is stated precisely in Theorem \ref{Mthm:minimal}. For the dominating column case, we expect a similar result to hold, but showing this analogously to the minimal restrictions case is not possible. We will discuss this issue later in Chapter \ref{cha:open}. 

One immediate application of the results obtained in this work is in the field of quantum thermodynamics. In thermodynamical equilibrium, the system is in a stable or thermal state $\gamma_\beta$ which is typically the Gibbs state of inverse temperature $\beta$ (or a KMS-state in general). A quantum channel $\Phi$ (a completely positive trace-preserving affine map on the states) is {\it Gibbs-preserving} if $\Phi(\gamma_\beta)=\gamma_\beta$. We denote $\rho\rgeq_\beta\sigma$ for quantum states $\rho$ and $\sigma$ if there is a Gibbs-preserving channel $\Phi$ such that $\Phi(\rho)=\sigma$. When we require the states $\rho$ and $\sigma$ to be diagonalizable in the eigenbasis of the thermal state $\gamma_\beta$, $\rho\rgeq_\beta\sigma$ means that
\begin{equation}
\big(\lambda^\rho,\lambda^\beta\big)\rgeq\big(\lambda^\sigma,\lambda^\beta\big), 
\end{equation}
where $\lambda^\rho$ and $\lambda^\sigma$ are the vectors of eigenvalues of $\rho$ and $\sigma$ where the eigenvalues appear in the order fixed by the thermal eigenbasis and $\lambda^\beta:=\lambda^{\gamma_\beta}$. This thermal majorization was characterized in the asymptotic catalytic regime in \cite{Brandao2014} using the results on catalytic vector majorization derived by Klimesh in \cite{klimesh2007inequalities}. To apply the vector majorization results, particular embedding methods were used in \cite{Brandao2014} approximating the thermal spectrum $\lambda^\beta$. These approximations introduce some convergence issues which were not completely taken into account in \cite{Brandao2014}, as was also pointed out in \cite{GourTomamichel2021}. Our method bypasses these issues and finally resolves the problem of catalytic thermal majorization also in the case when the states are not of full support. We may also now investigate multiple thermal majorization $(\rho_1,\ldots\rho_m)\rgeq_\beta(\sigma_1,\ldots,\sigma_m)$ i.e.,\ the existence of a Gibbs-preserving channel $\Phi$ such that $\Phi(\rho_i)=\sigma_i$ for $i=1,\ldots,m$.

This chapter is organized as follows. In Section \ref{Msec:background}, we set up basic notation and definitions, and review previous results in majorization in its various forms. In Section \ref{Msec:findings}, we define the preordered semiring for matrix majorization with minimal support conditions, identify the monotone homomorphisms and derivations relevant for large-sample and catalytic majorization, and show that the set of these monotones is minimal, in a certain sense. We partially build on earlier work in \cite{farooq2024} but see that lifting the assumption of mutual support of the columns in each matrix requires a dramatically different approach. In Section \ref{Msec:Asymmetric}, we study the case where the final column supports the others and in Sections \ref{Msec:2column} and \ref{Msec:thermal} we concentrate on the two-column case whose results we apply to the problem of thermal majorization briefly presented above. Section \ref{Msec:rates} presents results on the optimal conversion rate, and Section \ref{Msec:pusupports} touches briefly on the possibility of defining semirings with other conditions on the supports than considered in this chapter. 

\section{Background}\label{Msec:background}

\subsection{Basic definitions}
	
We follow the convention $\N=\{1,2,\ldots\}$. When $d\in\N$, we define $[d]:=\{1,\ldots,d\}$. For any $n\in\N$, we define the 1-norm in $\R^n$ as
\begin{equation}
\|(x_1,\ldots,x_n)\|_1=|x_1|+\cdots+|x_n|.
\end{equation}
Instead of full $\R$, we usually concentrate on the half-line $\R_+:=[0,\infty)$ or $\R_{>0}:=(0,\infty)$.
	
For any $n\in\N$, we say that $p=(p_1,\ldots,p_n)\in\R_+^n$ is a probability vector when $p_1+\cdots+p_n=1$, i.e.,\ $\|p\|_1=1$. Let us denote the set of all $n$-component probability vectors by $\mc P_n$. This set is convex, i.e.,\ when $p,q\in\mc P_n$ and $t\in[0,1]$, then $tp+(1-t)q\in\mc P_n$. Fix $m,n\in\N$. We say that $T:\mc P_m\to\mc P_n$ is a stochastic map if it is affine, i.e.,\ $T\big(tp+(1-t)q\big)=tTp+(1-t)Tq$ for all $p,q\in\mc P_m$ and $t\in[0,1]$. Equivalently, we may view $T$ as a linear map $T:\R^m\to\R^n$ such that $T\R_+^m\subseteq \R_+^n$ and $\|Tx\|_1=\|x\|_1$ for all $x\in\R^m$. This means that we may view $T$ as an $(n\times m)$-matrix $T=(T_{i,j})_{i,j}$ with entries in $\R_+$ such that $T_{1,j}+\cdots+T_{n,j}=1$ for all $j=1,\ldots,m$. We say that $T$ is bistochastic if also $T_{i,1}+\cdots+T_{i,m}=1$ for all $i=1,\ldots,n$.
	
\subsection{Majorization}
	
Majorization is a word commonly used when an object (often a probability measure or vector or a collection of them) can be reached from another object using a map of a relevant type (often a stochastic map). The arguably simplest notion of majorization deals with finite probability vectors. This is why we call this type of majorization {\it vector majorization}. In this mode of majorization, we say that a probability vector $p=(p_1,\ldots,p_n)$ majorizes another probability vector $q=(q_1,\ldots,q_n)$ if
\begin{equation}
\sum_{i=1}^k p_i^\downarrow\geq\sum_{i=1}^k q_i^\downarrow,\qquad k=1,\ldots,n-1,
\end{equation}
where $p^\downarrow=(p_1^\downarrow,\ldots,p_n^\downarrow)$ and $q^\downarrow=(q_1^\downarrow,\ldots,q_n^\downarrow)$ are $p$ and $q$ with their entries ordered in a non-increasing way. We write $p\rgeq q$ in this situation. This is equivalent with the existence of a bistochastic map $T$ such that $Tp=q$. A stochastic map $T$ is also bistochastic if and only if for the uniform probability vector $u=(1/n,\ldots,1/n)$, we have $Tu=u$. We subsequently generalize this notion of majorization in a straightforward manner.
	
\begin{definition}\label{Mdef:matrixmaj}
Fix $d\in\N$ and consider collections $P=\left(p^{(1)},\ldots,p^{(d)}\right)$ and $Q=\left(q^{(1)},\ldots,q^{(d)}\right)$ of probability vectors. We say that $P$ majorizes $Q$ and denote $P\rgeq Q$ if there is a stochastic map $T$ such that $Tp^{(k)}=q^{(k)}$ for $k=1,\ldots,d$. We call this form of majorization {\it matrix majorization}.
\end{definition}
	
We immediately notice that the preceding vector majorization is an example of this matrix majorization by noting that a probability vector $p=(p_1,\ldots,p_n)$ majorizes another probability vector $q=(q_1,\ldots,q_n)$ if $(p,u)\rgeq(q,u)$ where $u=(1/n,\ldots,1/n)$. Just like matrix majorization, also vector majorization can be studied in the large-sample and catalytic regimes.

Also the special case $d=2$ of matrix majorization is of particular interest. This is called {\it relative majorization}. Relative majorization is often restricted to the case of {\it dichotomies}, i.e.,\ pairs $(p,q)$ of probability measures or vectors where $q$ dominates $p$ (${\rm supp}\,p\subseteq{\rm supp}\,q$), or, symbolically, $p\ll q$.
	
\subsection{Earlier results}

The field of vector majorization is probably the best known mode of majorization. In \cite{klimesh2007inequalities}, necessary and sufficient conditions for catalytic vector majorization were derived. These conditions can be written using the {\it R\'{e}nyi divergences} $D_\alpha(\cdot\|\cdot)$ \cite{Renyi61}, $0\leq\alpha\leq\infty$, defined on any pair $(p,q)$ of probability measures with supports ${\rm supp}\,p=:I$ and ${\rm supp}\,q=:J$,
\begin{equation}\label{Meq:RenyiDivergences}
D_\alpha(p\|q)=\left\{\begin{array}{ll}
-\log{\sum_{i\in I\cap J}q_i}&\text{if}\ \alpha=0,\ {\rm and}\ I\cap J\neq\emptyset\\
\frac{1}{\alpha-1}\log{\sum_{i\in I\cap J}p_i^\alpha q_i^{1-\alpha}}&\text{if}\ \alpha\in(0,1),\ {\rm and}\ I\cap J\neq\emptyset\\
\sum_{i\in I}p_i\log{\frac{p_i}{q_i}}&\text{if}\ \alpha=1,\ {\rm and}\ I\subseteq J\\
\frac{1}{\alpha-1}\log{\sum_{i\in I}p_i^\alpha q_i^{1-\alpha}}&\text{if}\ \alpha\in(1,\infty),\ {\rm and}\ I\subseteq J\\
\max_{i\in I}\log{\frac{p_i}{q_i}}&\text{if}\ \alpha=\infty\ {\rm and}\ I\subseteq J\\
\infty&{\rm otherwise}
\end{array}\right.,
\end{equation}
in the following way: For probability vectors $p$ and $q$ such that $p\neq q$, there is a catalyst probability vector $r$ such that $p\otimes r\rgeq q\otimes r$ if and only if, denoting the uniform distribution on ${\rm supp}\,p\cup{\rm supp}\,q$ by $u$,
\begin{equation}\label{Meq:KlimeshCond}
D_\alpha(p\|u)>D_\alpha(q\|u),\quad D_\alpha(u\|p)>D_\alpha(u\|q),\qquad\forall \alpha\in[1/2,\infty).
\end{equation}

Conditions for vector majorization in large samples were derived in \cite{Jensen2019} which were further extended in \cite{farooq2024}. The latter extended result can be stated as follows: given probability vectors $p$ and $q$, there is $n\in\N$ such that $p^{\otimes n}\rgeq q^{\otimes n}$ if conditions \eqref{Meq:KlimeshCond} hold including the case $\alpha=\infty$.

In \cite{AubrunNechita2007}, asymptotic large-sample and catalytic vector majorization were studied. It was found that, given again probability vectors $p$ and $q$, the following conditions are equivalent: 
\begin{itemize}
\item[(i)] For all $\varepsilon>0$, there is $n_\varepsilon\in\N$ and a probability vector $q_\varepsilon$ such that $\|q-q_\varepsilon\|_1\leq\varepsilon$ and $p^{\otimes n}\rgeq q_\varepsilon^{\otimes n}$ for all $n\geq n_\varepsilon$.
\item[(ii)] For all $\varepsilon>0$, there are probability vectors $q_\varepsilon$ and $r_\varepsilon$ such that $\|q-q_\varepsilon\|_1\leq\varepsilon$ and $p\otimes r_\varepsilon\rgeq q_\varepsilon\otimes r_\varepsilon$.
\item[(iii)] $\|p\|_\alpha\geq \|q\|_\alpha$ for all $\alpha\geq1$, where
\begin{equation}
\|(p_1,\ldots,p_n)\|_\alpha:=\left(p_1^\alpha+\cdots+p_n^\alpha\right)^{1/\alpha}.
\end{equation}
\end{itemize}
Also this result was derived using real-algebraic methods in \cite{farooq2024}.

The problem of relative majorization in large samples was studied in \cite{Mu_et_al_2020}. This work studied dichotomies $(p,q)$ of probability measures over a standard Borel space $(X,\mc A)$ where $p\sim q$, i.e.,\ $p\ll q$ and $q\ll p$. In this case, the relevant R\'{e}nyi divergences $D_\alpha$, $\alpha\in[1/2,\infty]$ are given by
\begin{equation}
D_\alpha(p\|q)=\left\{\begin{array}{ll}
\frac{1}{\alpha-1}\int_X \frac{dp}{dq}(x)^\alpha\,dq(x)&\text{if}\ \alpha\in[1/2,\infty)\setminus\{1\}\\
\int_X \log{\frac{dp}{dq}(x)}\,dq(x)&\text{if}\ \alpha=1\\
q\text{-}\esssup_{x\in X}\log{\frac{dp}{dq}(x)}&\text{if}\ \alpha=\infty
\end{array}\right..
\end{equation}
It was found in \cite{Mu_et_al_2020} that a dichotomy $\big(p^{(1)},p^{(2)}\big)$, where $p^{(1)}\sim p^{(2)}$, majorizes $\big(q^{(1)},q^{(2)}\big)$, where $q^{(1)}\sim q^{(2)}$, if for all $\alpha\in[1/2,\infty]$, 
\begin{equation}
D_\alpha\big(p^{(1)}\big\|p^{(2)}\big)>D_\alpha\big(q^{(1)}\big\|q^{(2)}\big),\  D_\alpha\big(p^{(2)}\big\|p^{(1)}\big)>D_\alpha\big(q^{(2)}\big\|q^{(1)}\big).
\end{equation}

Large-sample and catalytic matrix majorization in the case of matrices (finite experiments) $P=\big(p^{(1)},\ldots,p^{(d)}\big)$ where the supports ${\rm supp}\,p^{(k)}$, $k=1,\ldots,d$, coincide was characterized in \cite{farooq2024}. For these results, let us define $A_+$ and $A_-$ as in \eqref{Meq:A+} and \eqref{Meq:A-}, and
\begin{align}
&B:=\{(\beta_1,\ldots,\beta_d)\in\R^d\,|\,\beta_1+\cdots+\beta_d=0\},\label{Meq:B}\\ 
&B_-:=\bigcup_{k=1}^d\{(\beta_1,\ldots,\beta_d)\in B\,|\,\beta_k>0,\ \beta_\ell\leq 0,\ \ell\neq k\}.\label{Meq:B-}
\end{align}
We also define the natural basis vectors $e_k:=(\delta_{\ell,k})_{\ell=1}^d$. We define the maps $D_{\underline{\alpha}}$ and $D_{\underline{\beta}}$ through
\begin{align}
D_{\underline{\alpha}}(P)&=\frac{1}{\max_{1\leq k\leq d}\alpha_k-1}\log{\sum_{i=1}^n \prod_{k=1}^d \big(p^{(k)}_i\big)^{\alpha_k}},\label{Meq:DalphaTemperate}\\
D_{\underline{\beta}}^\T(P)&=\frac{1}{\max_{1\leq k\leq d}\beta_k}\log{\max_{1\leq i\leq n}\prod_{i=1}^n \big(p^{(k)}_i\big)^{\beta_k}}\label{Meq:DbetaTropic},
\end{align}
for all for $P=\big(p^{(1)},\ldots,p^{(d)}\big)$ where $p^{(k)}=\big(p^{(k)}_1,\ldots,p^{(k)}_n\big)$, $k=1,\ldots,d$. The areas $A_+$ and $A_-$ together with the functions $D_{\underline{\alpha}}$ and $D^\T_{\underline{\beta}}$ are depicted in Figure \ref{Mfig:Restricted}. It follows that, if
\begin{equation}
\begin{array}{rcll}
D_{\underline{\alpha}}(P)&>&D_{\underline{\alpha}}(Q),&\forall\underline{\alpha}\in A_+\cup A_-\setminus\{e_1,\ldots,e_d\},\\
D_{\underline{\beta}}^\T(P)&>&D_{\underline{\beta}}^\T(Q),&\forall\underline{\beta}\in B_-,\\
D_1\big(p^{(k)}\big\|p^{(\ell)}\big)&>&D_1\big(q^{(k)}\big\|q^{(\ell)}\big),&\forall k,\ell\in\{1,\ldots,d\},\ k\neq\ell,
\end{array}
\end{equation}
then $P=\big(p^{(1)},\ldots,p^{(d)}\big)$ majorizes $Q=\big(q^{(1)},\ldots,q^{(d)}\big)$ in large samples. We also recover the earlier result of \cite{Mu_et_al_2020} for dichotomies of probability vectors as a special case. Also asymptotic large-sample and catalytic majorization can be characterized with similar non-strict inequalities involving $D_{\underline{\alpha}}$ with $\underline{\alpha}\in A_+\cup A_-\setminus\{e_1,\ldots,e_d\}$.

The above result crucially assumes that the columns of each matrix share a common support. In this work, we remove this condition and open the possibility for different support restrictions. This requires modifying the methodology used in \cite{farooq2024}. 

\section{The minimal restrictions case}\label{Msec:findings}

In order to derive conditions for large-sample majorization with varying support conditions for the columns of the matrices, we apply Theorem \ref{thm:Fritz2022}. To this end, we need to define semirings preordered by the extension of the matrix majorization ordering $\rgeq$ defined in Definition \ref{Mdef:matrixmaj} and find their monotone homomorphisms and derivations into each of $\R_+$, $\R_+^{\rm op}$, $\T\R_+$, and $\T\R_+^{\rm op}$. In this section, we look at the case where the only condition for the columns of the matrices is that their supports have a non-empty intersection (or they are all zero columns). 
	
\subsection{Minimal restrictions semiring}

We now define the semiring to which we will apply the Vergleichsstellensatz in Theorem \ref{thm:Fritz2022} in order to obtain large-sample majorization results for matrices whose columns need not be mutually dominating. Fix an integer $d\geq2$, which is the number of probability distributions in a tuple. For any integer $n\geq1$, write
\begin{align}
\V_n^d&:=\Bigg\{\,\left(p^{(1)},\ldots,p^{(d)}\right)\in\left(\R_+^n\right)^d\\
&\qquad\qquad\qquad\Bigg|\;\bigcap_{k=1}^d{\rm supp}\,p^{(k)}\neq\emptyset\textrm{ or }p^{(k)}=(0,\ldots,0)\;\forall k\in[d]\,\Bigg\},\\
\V^d&:=\bigcup_{n=1}^\infty \V_n^d. 
\end{align}
For any $P=(p^{(1)},\ldots,p^{(d)})\in \V_n^d$ we write $p^{(k)}=(p_1^{(k)},\ldots,p_n^{(k)})$ for $k\in[d]$ and view $P$ as an $n\times d$-matrix. 

We define two binary operations $\boxplus,\boxtimes$ on matrices in $\V^d$. Let $P,Q\in\V^d$. Then 
\begin{equation}
P\boxplus Q:=(p^{(1)}\oplus q^{(1)},\ldots,p^{(d)}\oplus q^{(d)}), 
\end{equation}
where $\oplus$ denotes the concatenation of two vectors. One can view this as stacking the matrices $P$ and $Q$: 
\begin{equation}
P\boxplus Q=\left(\frac{P}{Q}\right). 
\end{equation}
Also, we define 
\begin{equation}
P\boxtimes Q:=(p^{(1)}\otimes q^{(1)},\ldots,p^{(d)}\otimes q^{(d)}), 
\end{equation}
where $\otimes$ is the Kronecker product of two vectors. 

Next, we define a preorder $\rgeq$ on $\V^d$ as follows: for $P,Q\in\V^d$, we declare $P\rgeq Q$ if there exists a stochastic matrix $T$ such that $TP=Q$. 

The set $\V^d$ together with the operations $\boxplus$, $\boxtimes$ and the preorder $\rgeq$ do not yet form a preordered semiring. For instance, $P\boxplus Q$ and $Q\boxplus P$ have the same rows, but possibly in a different order. We now define an equivalence relation $\approx$ on $\V^d$ such that it becomes a preordered semiring together with the induced operations and preorder. Write $P\approx Q$ when one matrix can be obtained from the other by permuting the rows and adding rows of zeros $(0,\ldots,0)$. Denote by $[P]$ the equivalence class of $P\in\V^d$. The operations $\boxplus$ and $\boxtimes$ induce well-defined operations on $\V^d/\approx$ that are also commutative, associative, and distributive. The preorder $\rgeq$ on $\V^d$ induces a well-defined preorder on $\V^d/\approx$ which we still denote by $\rgeq$. Defining $0:=[(0,\ldots,0)]$, $1:=[(1,\ldots,1)]$, one checks that $S^d:=(\V^d/\approx,\boxplus,\boxtimes,0,1,\rgeq)$ is a preordered semiring, which we call the \emph{minimal restrictions matrix majorization semiring of length} $d$. 

In the rest of this chapter, we will mostly denote any element $[P]\in\V^d/\approx$ by just one representative matrix $P\in\V^d$, and do calculations with this matrix while freely changing the ordering of the rows and adding or removing rows of zeros. 

Let $P=(p^{(1)},\ldots,p^{(d)}),Q=(q^{(1)},\ldots,q^{(d)})\in\V_n^d$. Since a stochastic map preserves the $1$-norm of vectors, one sees that $P\rleq Q$ implies $\|p^{(k)}\|_1=\|q^{(k)}\|_1$ for all $k\in[d]$. Next, assume $\|p^{(k)}\|_1=\|q^{(k)}\|_1$ for all $k\in[d]$, and let $T$ be the stochastic map that takes the $1$-norm of every column. Then $TP=(\|p^{(1)}\|_1,\ldots,\|p^{(d)}\|_1)=TQ$, hence $P\sim Q$. We conclude that the conditions on $\|\cdot\|:S^d\rightarrow\R_{>0}^d\cup\{0\}$ in Theorem \ref{thm:Fritz2022} are met when we choose $\|P\|:=(\|p^{(1)}\|_1,\ldots,\|p^{(d)}\|_1)$ for $P\in\V^d$, which is clearly a surjective homomorphism. Later, in subsection \ref{Msubsec:poweruniversalnorestr}, we will show that $S^d$ is of polynomial growth by characterizing the power universals of $S^d$.
	
\subsection{Monotone homomorphisms}
\label{Msubsec:monotonehomomorphisms}
	
In this section, we characterize all monotone homomorphisms $\Phi:S^d\to\mathbb{K}$, where $\mathbb{K}\in\{\R_+,\R_+^{\rm op},\T\R_+,\T\R_+^{\rm op}\}$. It will turn out that only for $\mathbb{K}=\R_+^{\rm op}$ there exist non-degenerate $\Phi$. 

Define the one-row function $\fii:\R_{>0}^d\cup\{(0,\ldots,0)\}\to\mathbb{K}$, where $\fii(x_1,\ldots,x_d):=\Phi(x_1\,\cdots\,x_d)$. Note that $\fii$ is only defined on those rows of a matrix $P\in S^d$ that contain only positive entries, and the row containing only zeros. Exactly as in the proof of Proposition 13 in \cite{farooq2024}, we can show that 
\begin{equation}
\label{Meq:singlerowfunction}
\fii(x_1,\ldots,x_d)=x_1^{\alpha_1}\cdots x_d^{\alpha_d}, 
\end{equation}
where
\begin{equation}
\underline{\alpha}=(\alpha_1,\ldots,\alpha_d)\in\left\{\begin{array}{ll}
A_-&\text{if}\quad\mathbb{K}=\R_+\\
A_+&\text{if}\quad\mathbb{K}=\R_+^{\rm op}\\
B_-&\text{if}\quad\mathbb{K}=\T\R_+\\
B&\text{if}\quad\mathbb{K}=\T\R_+^{\rm op}
\end{array}\right..
\end{equation}
Here, the sets $A_+,A_-,B,B_-$ are as in \eqref{Meq:A+}, \eqref{Meq:A-}, \eqref{Meq:B}, \eqref{Meq:B-}, respectively. Note that any $P\in S^d$, $P\neq0$, has at least one row with all positive entries, since the columns' supports have a non-empty intersection. We will say that $\Phi$ \emph{is associated with} $\underline{\alpha}$. 

In \cite{farooq2024}, where only matrices $P$ with common support were considered, it was possible to break down $P$ into a $\boxplus$-sum of its rows which can all be viewed as elements of the semiring because each row contains only positive entries, or contains only zeros. In that case, denoting the rows of a matrix $P$ by $p_i$, we may write, for any monotone homomorphism $\Phi$ with codomain $\mathbb{K}$,
\begin{equation}
\Phi(P)=\left\{\begin{array}{ll}
\sum_i \fii(p_i)&\text{if}\quad\mathbb{K}\in\{\R_+,\R_+^{\rm op}\}\\
\max_i \fii(p_i)&\text{if}\quad\mathbb{K}\in\{\T\R_+,\T\R_+^{\rm op}\}
\end{array}\right..
\end{equation}
Thus, it is enough to characterize this one-row function $\fii$ in the common support case. In the semiring we are considering now, this no longer works because a row with both vanishing and positive entries cannot be seen as an element of our semiring; recall that the columns of any matrix in our semiring must have a non-empty intersection of their supports. This is why we have to consider the value of a homomorphism at matrices with a full-support row and different types of rows with vanishing entries. However, it turns out that we may characterize the monotone homomorphisms considering its values at only two-row matrices where the first row has full support and the second has some vanishing entries. 


We first prove a useful result for the monotone homomorphisms. For any $B\subseteq[d]$ we denote by $(\;B\;)$ the row of length $d$ whose $k$-th entry is 1 if $k\in B$ and otherwise the entry is 0. For any monotone homomorphism $\Phi:S^d\to\mathbb{K}$ for $\mathbb{K}\in\{\R_+,\R_+^{\rm op},\T\R_+,\T\R_+^{\rm op}\}$ denote
\begin{equation}
\Phi[B]:=\Phi\left(\begin{array}{ccc}
1&\cdots&1\\
&B&
\end{array}\right). 
\end{equation}

\begin{lemma}
\label{Mlem:1or2}
Let $B\subseteq[d]$, and $\Phi:S^d\to\mathbb{K}$ a monotone homomorphism for some $\mathbb{K}\in\{\R_+,\R_+^{\rm op},\T\R_+,\T\R_+^{\rm op}\}$. Then the following holds: 
\begin{enumerate}[(i)]
\item We have 
\begin{equation}
\label{Meq:1or2}
\Phi[B]=\left\{\begin{array}{ll}
1\text{ or }2&\text{if}\quad\mathbb{K}\in\{\R_+,\R_+^{\rm op}\}\\
1&\text{if}\quad\mathbb{K}\in\{\T\R_+,\T\R_+^{\rm op}\}
\end{array}\right.. 
\end{equation}
\item Suppose that $\mathbb{K}\in\{\R_+,\R_+^{\rm op}\}$ and that $\Phi$ is associated with $\underline{\alpha}=(\alpha_1,\ldots,\alpha_d)\in A_+\cup A_-$. If $\Phi[B]$ equals $2$, then $\alpha_k=0$ for all $k\notin B$. 
\end{enumerate}		
\end{lemma}
	
\begin{proof}
To prove (i), notice 
\begin{align}
&\left(\begin{array}{ccc}
1&\cdots&1\\
&B&
\end{array}\right)^{\boxtimes2}\boxplus\left(\begin{array}{ccc}
1&\cdots&1\\
1&\cdots&1
\end{array}\right)&\label{Meq:basicmatrixequation1}\\
&\qquad\qquad=\left(\begin{array}{ccc}
1&\cdots&1\\
&B&
\end{array}\right)\boxplus\left(\begin{array}{ccc}
1&\cdots&1\\
&B&
\end{array}\right)\boxplus\left(\begin{array}{ccc}
1&\cdots&1\\
&B&
\end{array}\right).&\label{Meq:basicmatrixequation2}
\end{align}
Using additivity and multiplicativity of $\Phi$, and the fact that we know that $\Phi$ evaluated on single rows with only positive entries is given by \eqref{Meq:singlerowfunction}, the above equality, for $\mathbb{K}\in\{\R_+,\R_+^{\rm op}\}$, yields the equation 
\begin{equation}
\left(\Phi[B]\right)^2+2=3\Phi[B], 
\end{equation}
and for $\mathbb{K}\in\{\T\R_+,\T\R_+^{\rm op}\}$, yields the equation 
\begin{equation}
\max\left\{\left(\Phi[B]\right)^2,1\right\}=\Phi[B]. 
\end{equation}
The solutions to the first equation are $1$ and $2$, and to the second the unique solution is $1$, proving \eqref{Meq:1or2}. 
		
To prove (ii), define the numbers $x_k$ coinciding with $2$ when $k\in B$ and otherwise $x_k=1$ and $y_k$ coinciding with $2$ when $k\in B$ and otherwise $y_k=0$. Denoting by $\substack{\rgeq\\\rleq}$ majorization in both directions simultaneously, we have
\begin{align}
&\left(\begin{array}{ccc}
1&\cdots&1\\
&B&
\end{array}\right)\boxplus\left(\begin{array}{ccc}
1&\cdots&1\\
&B&
\end{array}\right)\boxplus (x_1\,\cdots\,x_d)&\\
&\qquad\qquad\substack{\rgeq\\\rleq}\left(\begin{array}{ccc}
x_1&\cdots&x_d\\
y_1&\cdots&y_d
\end{array}\right)\boxplus\left(\begin{array}{ccc}
1&\cdots&1\\
1&\cdots&1
\end{array}\right),&
\end{align}
where the LHS majorizes the RHS because there is a stochastic matrix that adds the two rows $(\;B\;)$ to obtain the row $(y_1\,\cdots\,y_d)$, and the RHS majorizes the LHS because there is a stochastic matrix that splits the row $(y_1\,\cdots\,y_d)$ into two copies of the row $(\;B\;)$. Next, notice that the RHS is equal to
\begin{equation}
\label{Meq:1or2b}
    (x_1\,\cdots\,x_d)\boxtimes\left(\begin{array}{ccc}
        1&\cdots&1\\
        &B&
    \end{array}\right)\boxplus\left(\begin{array}{ccc}
        1&\cdots&1\\
        1&\cdots&1
    \end{array}\right). 
\end{equation}
This yields the equation 
\begin{equation}
2+2+2^{\sum_{k\in B}\alpha_k}=2^{\sum_{k\in B}\alpha_k}\cdot2+2, 
\end{equation}
where we used the monotonicity, additivity and multiplicativity of $\Phi$, and the single-row function given by \eqref{Meq:singlerowfunction} above. Therefore, $\sum_{k\in B}\alpha_k=1$. This implies $\alpha_k=0$ for all $k\notin B$, where for the case $\mathbb{K}=\R_+$ we use the fact that $\underline{\alpha}\in A_-$ and therefore only one of the coordinates $\alpha_k$ is positive. 
\end{proof}

Consider a monotone homomorphism $\Phi:S^d\to\mathbb{K}$ for $\mathbb{K}\in\{\R_+,\R_+^{\rm op},\T\R_+,\T\R_+^{\rm op}\}$. We call $B\subseteq[d]$ \emph{$\Phi$-trivial} if $\Phi[B]=1$ and \emph{$\Phi$-nontrivial} if $\Phi[B]=2$. Note that if $\mathbb{K}\in\{\T\R_+,\T\R_+^{\rm op}\}$, all sets $B$ are $\Phi$-trivial. Furthermore, when $\mathbb{K}\in\{\R_+,\R_+^{\rm op}\}$, the empty set is $\Phi$-trivial and $[d]$ is $\Phi$-nontrivial. 

\begin{lemma}
\label{Mclm:nontrivialB}
Fix a monotone homomorphism $\Phi:S^d\to\mathbb{K}$ for $\mathbb{K}\in\{\R_+,\R_+^{\rm op}\}$. Let $B_1,B_2\subseteq[d]$. The following holds: 
\begin{enumerate}[(i)]
\item If $B_1$ and $B_2$ are $\Phi$-nontrivial, then $B_1\cap B_2$ is also $\Phi$-nontrivial. 
\item If $B_1\subseteq B_2$ and $B_1$ is $\Phi$-nontrivial, then $B_2$ is also $\Phi$-nontrivial. 
\end{enumerate}
\end{lemma}

Note that this Lemma implies that when $\mathbb{K}\in\{\R_+,\R_+^{\rm op}\}$, the set of all $\Phi$-nontrivial sets is a filter of the power set of $[d]$. 

\begin{proof}
For the proof of (i), assume that $B_1$ and $B_2$ are $\Phi$-nontrivial. Notice 
\begin{equation}
\left(\begin{array}{ccc}
1&\cdots&1\\
1&\cdots&1
\end{array}\right)\boxplus\left(\begin{array}{ccc}
1&\cdots&1\\
&B_1&
\end{array}\right)\boxtimes\left(\begin{array}{ccc}
1&\cdots&1\\
&B_2&
\end{array}\right)=\left(\begin{array}{ccc}
1&\cdots&1\\
1&\cdots&1\\
1&\cdots&1\\
&B_1&\\
&B_2&\\
&B_1\cap B_2&
\end{array}\right). 
\end{equation}
This yields 
\begin{equation}
2+\Phi[B_1]\Phi[B_2]=\Phi[B_1]+\Phi[B_2]+\Phi[B_1\cap B_2]. 
\end{equation}
Since $\Phi[B_1]=2=\Phi[B_2]$, we conclude that $\Phi[B_1\cap B_2]=2$, hence $B_1\cap B_2$ is $\Phi$-nontrivial. 
	
Next, for the proof of (ii), assume that $B_1$ is $\Phi$-nontrivial. Notice 
\begin{equation}
\left(\begin{array}{ccc}
1&\cdots&1\\
1&\cdots&1
\end{array}\right)\boxplus\left(\begin{array}{ccc}
1&\cdots&1\\
&B_1&
\end{array}\right)\boxtimes\left(\begin{array}{ccc}
1&\cdots&1\\
&B_2&
\end{array}\right)=\left(\begin{array}{ccc}
1&\cdots&1\\
1&\cdots&1\\
1&\cdots&1\\
&B_1&\\
&B_2&\\
&B_1&
\end{array}\right). 
\end{equation}
This yields 
\begin{equation}
2+\Phi[B_1]\Phi[B_2]=2\Phi[B_1]+\Phi[B_2]. 
\end{equation}
Since $\Phi[B_1]=2$, we conclude that $\Phi[B_2]=2$, hence $B_2$ is $\Phi$-nontrivial.
\end{proof}

For any homomorphism $\Phi:S^d\to\mathbb{K}\in\{\R_+,\R_+^{\rm op}\}$, we define the {\it character} $C\subseteq[d]$ as the intersection of all $\Phi$-nontrivial subsets of $[d]$. $C$ is well-defined since $[d]$ itself is always $\Phi$-nontrivial. According to Lemma \ref{Mclm:nontrivialB}, $C$ is $\Phi$-nontrivial and, by construction, it is the smallest $\Phi$-nontrivial subset of $[d]$. Moreover, $C$ is non-empty, since the empty set is always $\Phi$-trivial. Additionally, observe that $B\subseteq[d]$ is $\Phi$-nontrivial if and only if $C\subseteq B$. 

Let $\underline{\alpha}=(\alpha_1,\ldots,\alpha_d)\in\R_+^d$ with $\alpha_1+\ldots+\alpha_d=1$, and $C=\{c_1,\ldots,c_\ell\}\subseteq[d]$, $C\neq\emptyset$, such that $C$ contains (at least) all columns $c$ with $\alpha_c\neq0$. Denote by $C^\prime:=\{c^\prime_1,\ldots,c^\prime_{\ell^\prime}\}\subseteq C$ the subset of columns $c\in C$ that satisfy $\alpha_c\neq0$. Define $\varphi_{\underline{\alpha},C}:\R_+^{\ell^\prime}\to\R$ by 
\begin{equation}
\varphi_{\underline{\alpha},C}(x_{c^\prime_1},\ldots,x_{c^\prime_{\ell^\prime}}):=x_{c^\prime_1}^{\alpha_{c^\prime_1}}\cdots x_{c^\prime_{\ell^\prime}}^{\alpha_{c^\prime_{\ell^\prime}}}. 
\end{equation}
One can show that $\varphi_{\underline{\alpha},C}$ is concave on $\R_{>0}^{\ell^\prime}$ (see the proof of Proposition 13 in \cite{farooq2024}). Then, by continuity of $\varphi_{\underline{\alpha},C}$ and the fact that all exponents $\alpha_{c^\prime}>0$ for $c^\prime\in C^\prime$, using that $\lim_{x\downarrow0}x^\alpha=0$ for $\alpha>0$, one concludes that $\varphi_{\underline{\alpha},C}$ is concave on its entire domain $\R_+^{\ell^\prime}$. In the following, we will abuse notation by writing $\varphi_{\underline{\alpha},C}(x)$ for $x\in\R_+^d$ instead of $x\in\R_+^{\ell^\prime}$, where it is understood that we ignore the coordinates not in $C^\prime$. 
	
We now define the homomorphism $\Phi_{\underline{\alpha},C}:S^d\to\R_+^{\rm op}$ by 
\begin{equation}\label{Meq:homomorphismsum}
\Phi_{\underline{\alpha},C}(P):=\sum_{i\in I(P,C)}\varphi_{\underline{\alpha},C}\left(p_i^{(1)},\ldots,p_i^{(d)}\right),\quad\quad P=(p^{(1)},\ldots,p^{(d)})\in \V_{n}^d, 
\end{equation}
where 
\begin{equation}\label{Meq:Iset}
I(P,C):=\left\{\,i\in[n]\,\middle|\,p_i^{(c)}>0\quad\forall c\in C\,\right\}=\bigcap_{c\in C}{\rm supp}\,p^{(c)}.
\end{equation}
One can easily check that these are indeed homomorphisms and we will later see that they are also monotone. By considering the rows $I(P,C)$ of $P$ that $\Phi_{\underline{\alpha},C}$ sums over, one sees that $C$ is the smallest $\Phi_{\underline{\alpha},C}$-nontrivial subset of $[d]$, and therefore coincides with the character of the homomorphism $\Phi_{\underline{\alpha},C}$ as defined above. The degenerate ones among these homomorphisms are exactly all $\Phi_{\underline{\alpha},C}=:\Phi^{(k)}$ with $\alpha_k=1$ for some $k$ and $\alpha_{k^\prime}=0$ for all $k^\prime\neq k$, and $C=\{k\}$; indeed, these give the same value for $P$ and $Q$ whenever $P\sim Q$. Denote all the non-degenerate homomorphisms $\Phi_{\underline{\alpha},C}$ by $\mathcal{H}^{\rm n.d.}(S^d,\R_+^{\rm op})$, i.e., denoting by ${\rm supp}\,\underline{\alpha}$ for all $\underline{\alpha}=(\alpha_1,\ldots,\alpha_d)\in A_+$ (where $A_+\subset\R^d$ is given in \eqref{Meq:A+}) the subset of those $k\in[d]$ such that $\alpha_k>0$,
\begin{equation}\label{Meq:temperateopposite}
\mathcal{H}^{\rm n.d.}(S^d,\R_+^{\rm op})=\left\{\Phi_{\underline{\alpha},C}\,\middle|\,\underline{\alpha}\in A_+,\ {\rm supp}\,\underline{\alpha}\subseteq C\subseteq[d]\right\}\setminus\big\{\Phi^{(k)}\big\}_{k=1}^d.
\end{equation}
	
Note that if $C$ consists of all $c\in[d]$ that satisfy $\alpha_c>0$, and no other $c\in[d]$, then $\Phi_{\underline{\alpha},C}$ essentially sums over all rows of any matrix in $S^d$. Hence a homomorphism with this property is continuous in its argument. From now on, we simply call elements of $\mathcal{H}^{\rm n.d.}(S^d,\R_+^{\rm op})$ with this property continuous, i.e. 
\begin{equation}\label{Meq:continuity}
\Phi_{\underline{\alpha},C}\ {\rm continuous}\quad \iff \quad C={\rm supp}\,\underline{\alpha}.
\end{equation}
A homomorphism not satisfying this property is always discontinuous. For example in the case $d=2$ and for the homomorphism $\Phi_{(0,1),\{1,2\}}$, this can be seen by considering 
\begin{equation}
P_\varepsilon:=\left(\begin{array}{ccc}
1-\varepsilon&1/2\\
\varepsilon&1/2
\end{array}\right), 
\end{equation}
where $\varepsilon\geq0$, and noting that there is a gap between $\Phi_{(0,1),\{1,2\}}(P_0)=1/2$ and $\Phi_{(0,1),\{1,2\}}(P_\varepsilon)=1$ for $\varepsilon>0$, although $\lim_{\varepsilon\rightarrow0}P_\varepsilon=P_0$. 
	
We first show that all homomorphisms in $\mathcal{H}^{\rm n.d.}(S^d,\R_+^{\rm op})$ are monotone, and then prove that these are in fact all non-degenerate monotone homomorphisms. 
	
\begin{proposition}
All homomorphisms in $\mathcal{H}^{\rm n.d.}(S^d,\R_+^{\rm op})$ are monotone. 
\end{proposition}
	
\begin{proof}
Let $\underline{\alpha}\in\R_+^d$ such that $\alpha_1+\ldots+\alpha_d=1$, and $C\in\mathcal{P}([d])\setminus\emptyset$ such that $C$ contains all columns $c$ with $\alpha_c\neq0$. Consider $P\in\V_n^d$ with rows $p_1,\ldots,p_n\in\R_+^d$ and a stochastic $m\times n$-matrix $T$, and write $Q:=TP\in\V_m^d$ with rows $q_1,\ldots,q_m\in\R_+^d$, and $T_i:=\sum_{j=1}^n T_{i,j}>0$. Using the concavity of $\varphi_{\underline{\alpha},C}$ in the first inequality, we find 
\begin{align}
\Phi_{\underline{\alpha},C}(Q)&=\sum_{i\in I(Q,C)}\varphi_{\underline{\alpha},C}(q_i)\\
&=\sum_{i\in I(Q,C)}\varphi_{\underline{\alpha},C}\left(\sum_{j=1}^n T_{i,j}p_j\right)=\sum_{i\in I(Q,C)}T_i\varphi_{\underline{\alpha},C}\left(\sum_{j=1}^n \frac{T_{i,j}}{T_i}p_j\right)\\
&\geq\sum_{i\in I(Q,C)}T_i\sum_{j=1}^n \frac{T_{i,j}}{T_i}\varphi_{\underline{\alpha},C}(p_j)=\sum_{j=1}^n \varphi_{\underline{\alpha},C}(p_j)\sum_{i\in I(Q,C)}T_{i,j}\\
&=\sum_{j\in I(P,C)} \varphi_{\underline{\alpha},C}(p_j)\underbrace{\sum_{i\in I(Q,C)}T_{i,j}}_{=1}+\sum_{j\notin I(P,C)} \varphi_{\underline{\alpha},C}(p_j)\sum_{i\in I(Q,C)}T_{i,j}\\
&\geq\sum_{j\in I(P,C)} \varphi_{\underline{\alpha},C}(p_j)=\Phi_{\underline{\alpha},C}(P). 
\end{align}
To see why $\sum_{i\in I(Q,C)}T_{i,j}=1$ for all $j\in I(P,C)$, consider any $j\in I(P,C)$ and $i\notin I(Q,C)$. Then there exists a column $k\in C$ such that $\sum_{j^\prime=1}^n T_{i,j^\prime}(p_{j^\prime})_k=0$, hence $T_{i,j}(p_j)_k=0$. Since $j\in I(P,C)$ implies that $(p_j)_k>0$, we conclude that $T_{i,j}=0$. This means that $\sum_{i\in I(Q,C)}T_{i,j}=\sum_{i=1}^m T_{i,j}=1$ for all $j\in I(P,C)$. 
\end{proof}
	
\begin{proposition}\label{Mpro:monhommin}
The non-degenerate monotone homomorphisms $S^d\rightarrow\R_+^{\rm op}$ are exactly those in $\mathcal{H}^{\rm n.d.}(S^d,\R_+^{\rm op})$. For $\mathbb{K}$ equal to $\R_+$, $\T\R_+$ or $\T\R_+^{\rm op}$, there are no non-degenerate monotone homomorphisms $S^d\rightarrow\mathbb{K}$. 
\end{proposition}
	
\begin{proof}
Let us first consider the case $\mathbb{K}\in\{\R_+,\R_+^{\rm op}\}$ and fix a monotone homomorphism $\Phi:S^d\to\mathbb{K}$. Let $\underline{\alpha}=(\alpha_1,\ldots,\alpha_d)\in A_+\cup A_-$ be the parameter vector which $\Phi$ is associated with and $C\subseteq[d]$ be the character of $\Phi$. 
For any $x\in\R_+^d$, denote by $\tilde{x}\in\R_{>0}^d$ the vector obtained from $x$ by replacing all vanishing entries with $1$. We now extend the domain $\R_{>0}^d\cup\{(0,\ldots,0)\}$ of the map $\fii$ defined at the beginning of subsection \ref{Msubsec:monotonehomomorphisms} to $\R_+^d$ through
\begin{equation}
\fii(x):=\Phi\big(\tilde{x}^T\boxplus x^T\big)-\fii(\tilde{x})
\end{equation}
for all $x\in\R_+^d$. Note that for $x=(x_1,\ldots,x_d)\in\R_+^d$, the vector $\tilde{x}$ has no vanishing entries, and hence 
\begin{equation}
\fii(\tilde{x})=\prod_{k\in{\rm supp}\,x}x_k^{\alpha_k}. 
\end{equation}
We see that $\fii$ coincides with its earlier definition on vectors in $\R_{>0}^d\cup\{(0,\ldots,0)\}$. Now we have, for all $x=(x_1,\ldots,x_k)\in\R_+^d$,
\begin{align}
\fii(x)&=\Phi\big(\tilde{x}^T\boxplus x^T\big)-\fii(\tilde{x})=\Phi\left(\tilde{x}^T\boxtimes\left(\begin{array}{ccc}
1&\cdots&1\\
&{\rm supp}\,x&
\end{array}\right)\right)-\fii(\tilde{x})\label{Meq:Phihajotus1}\\
&=\fii(\tilde{x})\left(\Phi\left(\begin{array}{ccc}
1&\cdots&1\\
&{\rm supp}\,x&
\end{array}\right)-1\right)=\fii(\tilde{x})\big(\Phi[{\rm supp}\,x]-1\big)\label{Meq:Phihajotus2}\\
&=\left\{\begin{array}{ll}
\fii(\tilde{x})&C\subseteq {\rm supp}\,x\\
0&{\rm otherwise}
\end{array}\right.,\label{Meq:Phihajotus3}
\end{align}
where the last equality is due to the fact that $\Phi[B]=2$ if and only if $C\subseteq B$ and $\Phi[B]=1$ otherwise.

Let $P=\big(p^{(1)},\ldots,p^{(d)}\big)\in\mc V^d$, $p^{(k)}=\big(p^{(k)}_1,\ldots,p^{(k)}_n\big)$ for $k=1,\ldots,d$. Define $p_i:=\big(p^{(1)}_i,\ldots,p^{(d)}_i\big)$ for $i=1,\ldots,n$. We may write
\begin{align}
\Phi(P)&=\Phi\left(\bigboxplus_{i=1}^n p_i^T\right)=\Phi\left(\bigboxplus_{i=1}^n \big(\tilde{p}_i^T\boxplus p_i^T\big)\right)-\sum_{i=1}^n\fii(\tilde{p}_i)\label{Meq:PhisumsCrowsa}\\
&=\sum_{i=1}^n\left(\Phi\big(\tilde{p}_i^T\boxplus p_i^T\big)-\fii(\tilde{p}_i)\right)=\sum_{i=1}^n \fii(p_i)\label{Meq:PhisumsCrowsb}\\
&=\sum_{i:\,C\subseteq{\rm supp}\,p_i}\prod_{k\in{\rm supp}\,p_i}\big(p^{(k)}_i\big)^{\alpha_k}.\label{Meq:PhisumsCrowsc}
\end{align}
It follows from Lemma \ref{Mlem:1or2}(ii) that if $k\notin C$, then $\alpha_k=0$. 

Thus, the form of $\Phi$ is the same as those in $\mathcal{H}^{\rm n.d.}(S^d,\R^{\rm op}_+)$ but with $\underline{\alpha}$ not necessarily in $A_+$. We already know that those within $\mathcal{H}^{\rm n.d.}(S^d,\R_+^{\rm op})$ are monotone in the case $\mathbb{K}=\R_+^{\rm op}$. We know from \cite{farooq2024} that, if $\underline{\alpha}\in A_-\setminus\{e_1,\ldots,e_d\}$, then $\Phi$ is monotone in the case $\mathbb{K}=\R_+$ when restricted on matrices whose columns have common support. Thus, in the case $\mathbb{K}=\R_+^{\rm op}$, the nondegenerate monotone homomorphisms are exactly those in $\mathcal{H}^{\rm n.d.}(S^d,\R^{\rm op})$. We next show that there are actually no nondegenerate monotone homomorphisms in the case $\mathbb{K}=\R_+$.

Assume $\mathbb{K}=\R_+$. It follows from \eqref{Meq:PhisumsCrowsa}, \eqref{Meq:PhisumsCrowsb} and \eqref{Meq:PhisumsCrowsc} that 
\begin{equation}\label{Meq:Phisinglex}
\Phi\left(\begin{array}{ccccccc}
1&\cdots&1&1&1&\cdots&1\\
0&\cdots&0&\smash{\underbrace{x}_{k\text{-th entry}}}&0&\cdots&0
\end{array}\right) 
\vspace{0.5cm}
\end{equation}
equals either $1+x$ or $1$ for all $x\geq0$ depending on whether $C=\{k\}$, respectively $C\neq\{k\}$. Furthermore, in the former case, \eqref{Meq:Phisinglex} equals $1$ for all other $k^\prime\neq k$, and since $k^\prime\notin C$, $\alpha_{k^\prime}=0$ for all $k^\prime\neq k$, and hence $\alpha_k=1$. 


Let $P=(p^{(1)},\ldots,p^{(d)})\in\mc V_n^d$. Notice the following inequalities, which follow from monotonicity, and for the second inequality by inserting a row of $1$'s for every non-zero entry in $P$ and putting every entry on its own row: 
\begin{align}
&\Phi(\|p^{(1)}\|_1\;\cdots\;\|p^{(d)}\|_1)\leq\Phi(P)&\label{Meq:Phiis01}\\
&\qquad\qquad\leq\sum_{i=1}^n\sum_{k=1}^d\left(\Phi\left(\begin{array}{ccccccc}
1&\cdots&1&1&1&\cdots&1\\
0&\cdots&0&\smash{\underbrace{p_i^{(k)}}_{k\text{-th entry}}}&0&\cdots&0 
\end{array}\right)-1\right).&\label{Meq:Phiis02}
\end{align}
Assume \eqref{Meq:Phisinglex} equals $1+x$ for some $k$. Then it equals $1$ for all $k^\prime\neq k$, and the RHS equals $\|p^{(k)}\|_1$. Also, since in this case $\alpha_k=1$ and $\alpha_{k^\prime}=0$ for all $k^\prime\neq k$, the LHS is equal to $\|p^{(k)}\|_1$. Hence $\Phi(P)=\|p^{(k)}\|_1$ for all $P$, meaning that $\Phi$ is degenerate. Assume \eqref{Meq:Phisinglex} equals $1$ for all $k$. Then it follows from the second inequality in \eqref{Meq:Phiis01}, \eqref{Meq:Phiis02} that $\Phi(P)=0$ for all $P$, which is not possible. In conclusion, no non-degenerate monotone homomorphism $S^d\to\R_+$ exists. 
		
Next, we consider the cases $\mathbb{K}$ equal to $\T\R_+$ or $\T\R_+^{\rm op}$. Let $\Phi:S^d\to\mathbb{K}$, where $\mathbb{K}\in\{\T\R_+,\T\R_+^{\rm op}\}$, be a non-degenerate monotone homomorphism, and $\underline{\beta}=(\beta_1,\ldots,\beta_d)\in B$ the parameter vector which $\Phi$ is associated with.

Let $\ell\in[d]$ and define $B_\ell:=\{\ell\}$ and $x^{(\ell)}:=(0,\ldots,x,\ldots,0)$ for some $x>0$. Then it follows from Lemma \ref{Mlem:1or2}(i) that 
\begin{align}
\Phi\big(\big(\widetilde{x^{(\ell)}}\big)^T\boxplus \big(x^{(\ell)}\big)^T\big)&=\Phi\left(\big(\widetilde{x^{(\ell)}}\big)^T\boxtimes\left(\begin{array}{ccc}
1&\cdots&1\\
&B_\ell&
\end{array}\right)\right)\\
&=\Phi\big((\widetilde{x^{(\ell)}}\big)^T)\Phi[B_\ell]=\Phi\big((\widetilde{x^{(\ell)}}\big)^T). 
\end{align}

Consider now the case $\mathbb{K}=\T\R_+$. Let $P\in\V_n^d$, $P\neq0$. By inserting rows of $1$'s and placing each non-zero entry in $P$ on its own row, we have 
\begin{align}
&P\boxplus\bigboxplus_{i,k\,:\,p_i^{(k)}>0}(1\,\cdots\,1\,p_i^{(k)}\,1\,\cdots\,1)&\\
&\qquad\preceq\bigboxplus_{i,k\,:\,p_i^{(k)}>0}\left(\begin{array}{ccccccc}
1&\cdots&1&p_i^{(k)}&1&\cdots&1\\[0.1cm]
0&\cdots&0&\smash{\underbrace{p_i^{(k)}}_{k\text{-th entry}}}&0&\cdots&0 
\end{array}\right).&
\vspace{0.5cm}
\end{align}
Assume all $\beta_k=0$.	Then from the majorization above and the fact that we know the value of $\Phi$ on single rows \eqref{Meq:singlerowfunction} it follows that 
\begin{equation}
\Phi(P)\leq\max\{\Phi(P),1\}\leq1. 
\end{equation}
Since also $1=\Phi(\|p^{(1)}\|_1\;\cdots\;\|p^{(d)}\|_1)\leq\Phi(P)$, we find that $\Phi(P)=1$. Hence $\Phi$ is degenerate, which leads to a contradiction. Thus there exists $k$ such that $\beta_k>0$. Note that 
\begin{align}
&\left(\begin{array}{cc}
1&\frac{1}{2}\\[0.1cm]
0&\frac{1}{2}
\end{array}\right)\left(\begin{array}{ccccccc}
\frac{3}{2}&\cdots&\frac{3}{2}&1&\frac{3}{2}&\cdots&\frac{3}{2}\\[0.1cm]
0&\cdots&0&\smash{\underbrace{1}_{k\text{-th entry}}}&0&\cdots&0
\end{array}\right)&\\\nonumber\\
&\qquad\qquad=\left(\begin{array}{ccccccc}
\frac{3}{2}&\cdots&\frac{3}{2}&\frac{3}{2}&\frac{3}{2}&\cdots&\frac{3}{2}\\[0.1cm]
0&\cdots&0&\smash{\underbrace{\textstyle\frac{1}{2}}_{k\text{-th entry}}}&0&\cdots&0
\end{array}\right).&
\end{align}
Using this and the fact that $\sum_{k^\prime\neq k}\beta_{k^\prime}<0$ it follows that
\begin{align}\label{Meq:tropicalhomomorphismcontradiction}
1&>\left(\frac{3}{2}\right)^{\sum_{k^\prime\neq k}\beta_{k^\prime}}\\
&=\Phi\bigg(\left(\frac{3}{2}\;\cdots\;\frac{3}{2}\;\smash{\underbrace{1}_{k\text{-th entry}}}\;\frac{3}{2}\;\cdots\;\frac{3}{2}\right)\\\nonumber\\
&\qquad\qquad\boxtimes\left(\begin{array}{ccccccc}
1&\cdots&1&1&1&\cdots&1\\[0.1cm]
0&\cdots&0&\smash{\underbrace{1}_{k\text{-th entry}}}&0&\cdots&0
\end{array}\right)\bigg)\\\nonumber\\
&=\Phi\left(\begin{array}{ccccccc}
\frac{3}{2}&\cdots&\frac{3}{2}&1&\frac{3}{2}&\cdots&\frac{3}{2}\\[0.1cm]
0&\cdots&0&\smash{\underbrace{1}_{k\text{-th entry}}}&0&\cdots&0
\end{array}\right)\\\nonumber\\
&\geq\Phi\left(\begin{array}{ccccccc}
\frac{3}{2}&\cdots&\frac{3}{2}&\frac{3}{2}&\frac{3}{2}&\cdots&\frac{3}{2}\\[0.1cm]
0&\cdots&0&\smash{\underbrace{\textstyle\frac{1}{2}}_{k\text{-th entry}}}&0&\cdots&0
\end{array}\right)\\\nonumber\\
&=\Phi\left(\left(\frac{3}{2}\;\cdots\;\frac{3}{2}\right)\boxtimes\left(\begin{array}{ccccccc}
1&\cdots&1&1&1&\cdots&1\\[0.1cm]
0&\cdots&0&\smash{\underbrace{\textstyle\frac{1}{3}}_{k\text{-th entry}}}&0&\cdots&0
\end{array}\right)\right)\\\nonumber\\
&\geq\Phi\left(\frac{3}{2}\;\cdots\;\frac{3}{2}\right)\Phi\left(1\;\cdots\;1\;\smash{\underbrace{\textstyle\frac{4}{3}}_{k\text{-th entry}}}\;1\;\cdots\;1\right)=\left(\frac{4}{3}\right)^{\beta_k}>1, \\\nonumber
\end{align}
which leads to a contradiction. We conclude that no non-degenerate monotone homomorphism $S^d\to\T\R_+$ exists. 

Finally, consider the case $\mathbb{K}=\T\R_+^{\rm op}$. Let $P\in\V_n^d$, $P\neq0$. $P$ has at least one row not containing any $0$. Assume WLOG it is the first row. Then 
\begin{equation}
P\preceq\left(\begin{array}{ccc}
p_1^{(1)}/2&\cdots&p_1^{(d)}/2\\
p_1^{(1)}/2&\cdots&p_1^{(d)}/2\\
p_2^{(1)}&\cdots&p_2^{(d)}\\
\vdots&\vdots&\vdots\\
p_n^{(1)}&\cdots&p_n^{(d)}
\end{array}\right), 
\end{equation}
hence
\begin{equation}
\Phi(P)\geq\max\left\{\Phi\left(p_1^{(1)}/2\;\cdots\;p_1^{(d)}/2\right),\Phi\left(\begin{array}{ccc}
p_1^{(1)}/2&\cdots&p_1^{(d)}/2\\
p_2^{(1)}&\cdots&p_2^{(d)}\\
\vdots&\vdots&\vdots\\
p_n^{(1)}&\cdots&p_n^{(d)}
\end{array}\right)\right\}. 
\end{equation}	
Assume all $\beta_k=0$. Then it follows from the above and \eqref{Meq:singlerowfunction} that $\Phi(P)\geq1$. Since also $1=\Phi(\|p^{(1)}\|_1\;\cdots\;\|p^{(d)}\|_1)\geq\Phi(P)$, we find that $\Phi(P)=1$. Hence $\Phi$ is degenerate, which leads to a contradiction. Thus there exists $k$ such that $\beta_k<0$. Then we can perform the same calculations as in \eqref{Meq:tropicalhomomorphismcontradiction} but now with $\beta_k<0$ and all the inequalities reversed. This leads to a contradiction. Hence there exists no non-degenerate monotone homomorphism $S^d\to\T\R_+^{\rm op}$.
\end{proof}
	
\subsection{Monotone derivations}
	
In this section, we identify the monotone derivations at the degenerate monotone homomorphisms on $S^d$, up to interchangeability (see footnote \ref{ftn:interchangeability} on page \pageref{ftn:interchangeability}): if the difference between two $\Phi$-derivations factors through $S^d/\sim$, then they are called interchangeable and we only need to consider one of these derivations in the Vergleichsstellensatz of Theorem \ref{thm:Fritz2022}. 

Recall that the degenerate monotone homomorphisms are exactly the $d$ $1$-norms of the columns $P\mapsto \|p^{(k)}\|_1$ for every $k\in[d]$. It will turn out that up to interchangeability, the only monotone derivation $\Delta$ is $\Delta(P)=0$ for all $P\in S^d$. 

\begin{proposition}
\label{Mpro:derivations}
There is a unique monotone derivation $\Delta$, up to interchangeability, at any degenerate monotone homomorphism $S^d\rightarrow\R_+$, and satisfies $\Delta(P)=0$ for all $P\in S^d$. 
\end{proposition}
 
Before proving this Proposition, we will use interchangeability to show that we only need to consider a certain type of derivations, namely those derivations $\Delta$ satisfying the vanishing property $\Delta(x,\ldots,x)=0$ for all $x\geq0$, as stated by the Lemma below. For the derivations of the semiring of matrices with common support studied in \cite{farooq2024}, the vanishing property was only shown for $x\in\Q$, $x\geq0$. However, one could show that it does in fact hold for all $x\geq0$ analogously to the proof of the Lemma below. 

\begin{lemma}
\label{Mlem:interchangeability}
Up to interchangeability, we may assume that any monotone derivation $\Delta$ at a degenerate monotone homomorphism satisfies 
\begin{equation}
\Delta(x\,\cdots\,x)=0 
\end{equation}
for any $x\geq0$. 
\end{lemma}
	
\begin{proof}
Let $\Delta$ be a monotone derivation at the degenerate monotone homomorphism $\Phi^{(k)}:P\mapsto \|p^{(k)}\|_1$ for some $k\in[d]$. Define the function $\delta:\R_+\to\R$ through $\delta(x)=-\Delta(x\,\cdots\,x)$. Consider $x,y\in\R_+$. Note that
\begin{equation}
\left(\begin{array}{ccc}
x&\cdots&x\\
y&\cdots&y
\end{array}\right)\rgeq \big(x+y\,\cdots\,x+y\big)\rgeq \left(\begin{array}{ccc}
x&\cdots&x\\
y&\cdots&y
\end{array}\right).
\end{equation}
Thus,
\begin{align}
\delta(x+y)&=-\Delta\left(\begin{array}{ccc}
x&\cdots&x\\
y&\cdots&y
\end{array}
\right)\\
&=-\Delta(x\,\cdots\,x)-\Delta(y\,\cdots\,y)=\delta(x)+\delta(y).
\end{align}
Moreover,
\begin{align}
\delta(xy)&=-\Delta\big((x\,\cdots\,x)\boxtimes(y\,\cdots\,y)\big)=-\Delta(x\,\cdots\,x)y-x\Delta(y\,\cdots\,y)\\
&=\delta(x)y+x\delta(y).
\end{align}

Let us define 
\begin{equation}
\Delta'(P):=\Delta(P)+\delta\big(\|p^{(k)}\|_1\big), 
\end{equation}
for $P\in S^d$. Since $\Delta$ is a monotone derivation at $\Phi^{(k)}$, so is $\Delta'$ because of the above calculations for $\delta$ and since $P\succeq Q$ implies $\|p^{(k)}\|_1=\|q^{(k)}\|_1$ for all $P,Q\in S^d$. Also, $\Delta$ and $\Delta'$ are clearly interchangeable. We now have, for all $x\geq0$,
\begin{equation}
\Delta'(x\,\cdots\,x)=\Delta(x\,\cdots\,x)+\delta(x)=\Delta(x\,\cdots\,x)-\Delta(x\,\cdots\,x)=0.
\end{equation}
\end{proof}

\begin{remark}
According to Lemma \ref{Mlem:interchangeability}, up to interchangeability, any monotone derivation $\Delta:S^d\to\R$ has the vanishing property $\Delta(p,\ldots,p)=0$ for any $p\in\R_+^n$, $n\in\N$. Indeed, denoting $p=(p_1,\ldots,p_n)$, we have
\begin{equation}
\Delta(p,\ldots,p)=\sum_{i=1}^n \Delta(p_i\,\cdots\,p_i)=0.
\end{equation}
This same result also applies to the results of \cite{farooq2024} where the assumption that $\|p\|_1\in\Q_+$ was made to guarantee this vanishing property.
\end{remark}

Recall that the degenerate monotone homomorphisms are all $\Phi^{(k)}$ for $k\in[d]$. Let $\Delta:S^d\to\mathbb{\R}$ be a monotone $\Phi^{(k)}$-derivation for some $k\in[d]$. For any $B\subseteq[d]$ we write 
\begin{equation}
    \Delta[B]:=\Delta\left(\begin{array}{ccc}
            1&\cdots&1\\
            &B&
        \end{array}\right). 
\end{equation}
The following Lemma will be used in the proof of Proposition \ref{Mpro:derivations}. 
\begin{lemma}
    \label{Mlem:derivationszero}
    Let $B\subseteq[d]$ and $\Delta:S^d\to\mathbb{\R}$ be a monotone $\Phi^{(k)}$-derivation for some $k\in[d]$. Then $\Delta[B]=0$. 
\end{lemma}

\begin{proof}
Recall the equality of matrices given in \eqref{Meq:basicmatrixequation1}, \eqref{Meq:basicmatrixequation2}. Using the Leibniz rule and additivity of $\Delta$, and Lemma \ref{Mlem:interchangeability}, we find 
    \begin{equation}
        2\Phi^{(k)}[B]\Delta[B]=3\Delta[B], 
    \end{equation}
    where $\Phi^{(k)}[B]=2$, respectively $1$, depending on whether $k\in B$, respectively $k\notin B$. In both cases, the only solution is $\Delta[B]=0$. 
\end{proof}

We now prove Proposition \ref{Mpro:derivations}. 
	
\begin{proof}[Proof of Proposition \ref{Mpro:derivations}]
Fix $k\in[d]$ and assume $\Delta:S^d\to\R$ is a monotone $\Phi^{(k)}$-derivation. For $\ell=1,\ldots,d$ define
\begin{equation}
    C_\ell:=\left(\begin{array}{ccccccc}
        1&\cdots&1&1&1&\cdots&1\\[0.1cm]
        0&\cdots&0&\smash{\underbrace{1}_{\ell\text{-th entry}}}&0&\cdots&0
    \end{array}\right). 
    \vspace{0.5cm}
\end{equation}
According to Lemma \ref{Mlem:derivationszero}, $\Delta(C_\ell)=0$ for $\ell=1,\ldots,d$. Let us consider $P=\big(p^{(1)},\ldots,p^{(d)}\big)\in\mc V^d$ with $p^{(\ell)}=\big(p^{(\ell)}_1,\ldots,p^{(\ell)}_n\big)$. Define the matrix
\begin{equation}
P':=\bigboxplus_{i=1}^n\bigboxplus_{\ell=1}^d \big(p^{(\ell)}_i\,\cdots\,p^{(\ell)}_i\big)\boxtimes C_\ell.
\end{equation}
One easily sees that
\begin{equation}
P'\rgeq P'':=P\boxplus \bigboxplus_{i=1}^n\bigboxplus_{\ell=1}^d \big(p^{(\ell)}_i\,\cdots\,p^{(\ell)}_i\big).
\end{equation}
Since $\Delta(x\,\cdots\,x)=0$ for all $x\geq0$, we notice that $\Delta(P'')=\Delta(P)$ and, hence
\begin{align}
\Delta(P)&=\Delta(P'')\leq\Delta(P')=\sum_{i=1}^n\sum_{\ell=1}^d\Delta\left(\big(p^{(\ell)}_i\,\cdots\,p^{(\ell)}_i\big)\boxtimes C_\ell\right)\label{Meq:Deltais0a}\\
&=\sum_{i=1}^n\sum_{\ell=1}^d\Big(\underbrace{\Delta\big(p^{(\ell)}_i\,\cdots\,p^{(\ell)}_i\big)}_{=0}\Phi^{(k)}(C_\ell)+\Phi^{(k)}\big(p^{(\ell)}_i\,\cdots\,p^{(\ell)}_i\big)\underbrace{\Delta(C_\ell)}_{=0}\Big)\label{Meq:Deltais0b}\\
&=0.\label{Meq:Deltais0c}
\end{align}
Thus, $\Delta$ is nonpositive.
  
We finish the proof by showing that $\Delta$ is also nonnegative. Define $M:=\max_{1\leq\ell\leq d}\|p^{(\ell)}\|_1$, $m_\ell:=M-\|p^{(\ell)}\|_1$ for $\ell=1,\ldots,d$, and $m:=m_1+\cdots+m_d$. Consider
\begin{equation}
\tilde{P}:=P\boxplus\bigboxplus_{\ell=1}^d\big(m_\ell\,\cdots\,m_\ell\big)\boxtimes C_\ell, 
\end{equation}
which has columns that have the same $1$-norm $M+m$ as one easily verifies. By monotonicity, we find 
\begin{align}
0&=\Delta(M+m\,\cdots\,M+m)\leq \Delta(\tilde{P})\\
&=\Delta(P)+\sum_{\ell=1}^d\underbrace{\Delta\Big(\big(m_\ell\,\cdots\,m_\ell\big)\boxtimes C_\ell\Big)}_{=0}=\Delta(P),
\end{align}
where the terms in the sum vanish according to a calculation analogous to \eqref{Meq:Deltais0a}, \eqref{Meq:Deltais0b} and \eqref{Meq:Deltais0c}.
\end{proof}

\begin{figure}
\begin{center}
\begin{overpic}[scale=0.43,unit=1mm]{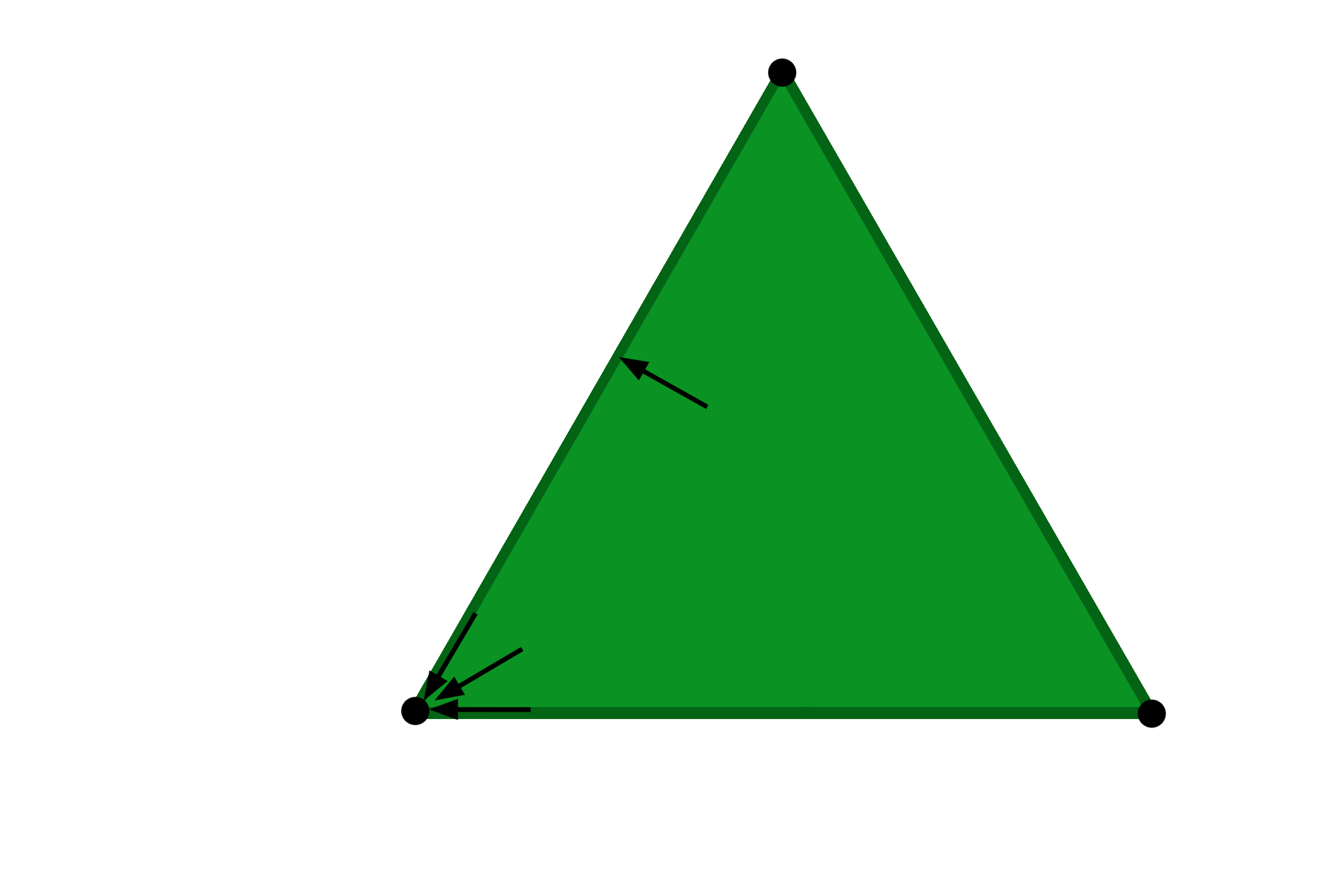}
\put(60,30){\begin{Huge}
$A_+$
\end{Huge}}
\put(32,11){\begin{Large}
$e_1$
\end{Large}}
\put(65,65){\begin{Large}
$e_3$
\end{Large}}
\put(95,13){\begin{Large}
$e_2$
\end{Large}}
\put(75,47){\begin{Large}
${\color{lighter} \Phi_{\underline{\alpha},\{1,2,3\}}}$
\end{Large}}
\put(23,46){
${\color{lighter} \Phi_{(\alpha,0,1-\alpha),\{1,2,3\}}}$
}
\put(12,13){
${\color{lighter} \Phi_{e_1,\{1,2,3\}}}$
}
\put(20,42){
${\color{darker} \Phi_{(\alpha,0,1-\alpha),\{1,3\}}}$
}
\put(17,18){
${\color{darker} \Phi_{e_1,\{1,3\}}}$
}
\put(27,7){
${\color{darker} \Phi_{e_1,\{1,2\}}}$
}
\put(0,6){
$\Phi_{e_1,\{1\}}=\Phi^{(1)}$
}
\end{overpic}
\caption{\label{Mfig:NoRestrictions} Illustration of the monotone homomorphisms of $S^3$, where the figure concentrates on the line segment from $e_1$ to $e_3$ and the vertex $e_1$: The monotone homomorphisms associated with $\underline{\alpha}$ in the interior of $A_+$ are $\Phi_{\underline{\alpha},\{1,2,3\}}$, depicted in lighter green. The smaller faces of $A_+$, i.e.,\ the three line segments and the singletons $\{e_1\}$, $\{e_2\}$, and $\{e_3\}$ inherit similar homomorphisms still with $C=\{1,2,3\}$, also depicted in lighter green. These homomorphisms can be seen as pointwise limits from the interior of $A_+$, when $\alpha_2\to0$. Additionally, the extreme point $e_1$ inherits the homomorphisms $\Phi_{e_1,\{1,3\}}$ and $\Phi_{e_1,\{1,2\}}$, depicted in darker green, as limits from the two line segments it connects, and similarly for $e_2$ and $e_3$. We also have the homomorphisms $\Phi_{(\alpha,0,1-\alpha),\{1,3\}}$, $\alpha\in[0,1]$, on the line segment between $e_1$ and $e_3$, depicted in darker green, and similarly for the other two line segments. At $e_1$, the degenerate homomorphism with $C=\{1\}$, i.e.,\ $\Phi^{(1)}$, is also present, depicted in black, and similarly for $e_2$ and $e_3$. There are no non-zero derivations at the three extreme points.}
\end{center}
\end{figure}
		
	\subsection{Power universal elements}\label{Msubsec:poweruniversalnorestr}
	
	We determine which matrices in $\V^d$ with columns that sum to $1$ are power universal. Recall that for large-sample or catalytic majorization of a matrix $P$ over $Q$, it is a requirement in the Vergleichsstellensatz (Theorem \ref{thm:Fritz2022}) that $P$ is power universal. 
	\begin{proposition}\label{Mlem:poweruniversalconditions}
		Let $U=(u^{(1)},\ldots,u^{(d)})\in \V^d$ with $u^{(k)}$ for $k\in[d]$ probability distributions. Then the following are equivalent: 
		\begin{enumerate}[(i)]
			\item $U$ is power universal. 
			\item There exists $m\in\N$ such that for all $k\in[d]$, $U^{\boxtimes m}$ has a row $(0\,\cdots\,0\underbrace{x_k}_{k\text{-th entry}}0\,\cdots\,0)$ for some $x_k>0$. 
			\item It holds that ${\rm supp}\,u^{(k)}\nsubseteq{\rm supp}\,u^{(k^\prime)}$ for all $k,k^\prime\in[d]$, $k\neq k^\prime$. 
			\item Any nondegenerate homomorphism $\Phi\in\mathcal{H}^{\rm n.d.}(S^d,\R_+^{\rm op})$ satisfies $\Phi(U)<1$. 
		\end{enumerate}
	\end{proposition}
	
	\begin{proof}
		We show the implications (iii) $\Rightarrow$ (ii) $\Rightarrow$ (i) $\Rightarrow$ (iv) $\Rightarrow$ (iii). 
		
		(iii) $\Rightarrow$ (ii): Let $k\in[d]$. The matrix $U^{\boxtimes(d-1)}$ contains a row which is the element-wise product of the $d-1$ rows $(u^{(1)}_{i_{k^\prime}}\;\cdots\;u^{(d)}_{i_{k^\prime}})$ with $u^{(k)}_{i_{k^\prime}}>0,u^{(k^\prime)}_{i_{k^\prime}}=0$ for $k^\prime\in[d]\setminus\{k\}$, which exist according to (iii). This row in $U^{\boxtimes(d-1)}$ is clearly of the form $(0\,\cdots\,0\underbrace{x_k}_{k\text{-th entry}}0\,\cdots\,0)$ for some $x_k>0$. We conclude that (ii) holds for $m=d-1$. 
		
		(ii) $\Rightarrow$ (i): Assume $U\in \V_n^d$ satisfies (ii), and let $P,Q\in \V_n^d$, $Q\preceq P$. We may assume WLOG that $P,Q$ have columns of unit $1$-norm. By (ii) there exists $m_0\in\N$ such that $\tilde{U}:=U^{\boxtimes m_0}$ has a row $(0\,\cdots\,0\underbrace{x_k}_{k\text{-th entry}}0\,\cdots\,0)$ with $x_k>0$ for every $k\in[d]$. We then have 
		\begin{align}
			\tilde{U}&=\left(\begin{array}{ccc}
				\tilde{u}^{(1)}_0&\cdots&\tilde{u}^{(d)}_0\\
				1-\|\tilde{u}^{(1)}_0\|_1&\cdots&0\\
				0&\ddots&0\\
				0&\cdots&1-\|\tilde{u}^{(d)}_0\|_1
			\end{array}\right)\\
            &\succeq\left(\begin{array}{ccc}
				\|\tilde{u}^{(1)}_0\|_1&\cdots&\|\tilde{u}^{(d)}_0\|_1\\
				1-\|\tilde{u}^{(1)}_0\|_1&\cdots&0\\
				0&\ddots&0\\
				0&\cdots&1-\|\tilde{u}^{(d)}_0\|_1
			\end{array}\right), 
		\end{align}
		for some vectors $\tilde{u}_0^{(k)}\in\R_+^l$, $k\in[d]$, where $l\in\N$, and such that $1-\|\tilde{u}_0^{(k)}\|_1>0$ for all $k\in[d]$. Then for any $m\geq1$ 
		\begin{align}
			\tilde{U}^{\boxtimes m}&\succeq\left(\begin{array}{ccc}
				\|\tilde{u}^{(1)}_0\|_1&\cdots&\|\tilde{u}^{(d)}_0\|_1\\
				1-\|\tilde{u}^{(1)}_0\|_1&\cdots&0\\
				0&\ddots&0\\
				0&\cdots&1-\|\tilde{u}^{(d)}_0\|_1
			\end{array}\right)^{\boxtimes m}\\
			&\succeq\left(\begin{array}{ccc}
				\|\tilde{u}^{(1)}_0\|_1^m&\cdots&\|\tilde{u}^{(d)}_0\|_1^m\\
				1-\|\tilde{u}^{(1)}_0\|_1^m&\cdots&0\\
				0&\ddots&0\\
				0&\cdots&1-\|\tilde{u}^{(d)}_0\|_1^m
			\end{array}\right)=:\tilde{U}_m. 
		\end{align}
		Denote the stochastic matrix that maps $\tilde{U}^{\boxtimes m}$ to $\tilde{U}_m$ as $T_m$. Write
		\begin{equation}
			P=\left(\begin{array}{ccc}
				p_0^{(1)}&\cdots&p_0^{(d)}\\
				p_1^{(1)}&\cdots&p_1^{(d)}
			\end{array}\right), 
		\end{equation}
		where for $k\in[d]$ the vectors $p_0^{(k)}\in\R_{>0}^{n_0}$ form the part of common support of $P$ and $p_1^{(k)}\in\R_+^{n_1}$, and $n_0+n_1=n$. Also, define $v_0:=(1/n_0,\ldots,1/n_0)\in\R_+^{n_0}$. Next, consider 
		\begin{equation}
			S_m:=\left(\begin{array}{cccc}
				v_0&\frac{p_0^{(1)}}{1-\|\tilde{u}_0^{(1)}\|_1^m}-\frac{\|\tilde{u}_0^{(1)}\|_1^m}{1-\|\tilde{u}_0^{(1)}\|_1^m}v_0&\cdots&\frac{p_0^{(d)}}{1-\|\tilde{u}_0^{(d)}\|_1^m}-\frac{\|\tilde{u}_0^{(d)}\|_1^m}{1-\|\tilde{u}_0^{(d)}\|_1^m}v_0\\
				0&\frac{1-\|p_0^{(1)}\|_1}{1-\|\tilde{u}_0^{(1)}\|_1^m}&\cdots&0\\
				\vdots&0&\ddots&0\\
				0&0&\cdots&\frac{1-\|p_0^{(d)}\|_1}{1-\|\tilde{u}_0^{(d)}\|_1^m}
			\end{array}\right), 
		\end{equation}
		which has columns that sum to $1$ and non-negative entries for sufficiently large $m$ because $\|\tilde{u}_0^{(k)}\|_1<1$ and $p_0^{(k)}$ has positive entries for all $k\in[d]$. Then 
		\begin{align}
			P&\preceq\left(\begin{array}{ccc}
				p_0^{(1)}&\cdots&p_0^{(d)}\\
				1-\|p_0^{(1)}\|_1&\cdots&0\\
				0&\ddots&0\\
				0&\cdots&1-\|p_0^{(d)}\|_1
			\end{array}\right)=S_m\tilde{U}_m=S_mT_m\tilde{U}^{\boxtimes m}\\
			&\preceq \tilde{U}^{\boxtimes m}=(1\,\cdots\,1)\boxtimes U^{\boxtimes(m_0m)}\preceq Q\boxtimes U^{\boxtimes(m_0m)}, 
		\end{align}
		which is what we aimed to prove. 
		
		(i) $\Rightarrow$ (iv): Let $\Phi_{\underline{\alpha},C}\in\mathcal{H}^{\rm n.d.}(S^d,\R_+^{\rm op})$, and $P,Q\in \V^d$ such that $Q\preceq P$. Since $U$ is power universal, there exists $m\in\N$ such that $P\preceq Q\boxtimes U^{\boxtimes m}$. Using monotonicity and multiplicativity: 
		\begin{equation}
			\Phi_{\underline{\alpha},C}(P)\geq \Phi_{\underline{\alpha},C}(Q\boxtimes U^{\boxtimes m})=\Phi_{\underline{\alpha},C}(Q)\Phi_{\underline{\alpha},C}(U)^m\geq \Phi_{\underline{\alpha},C}(P)\Phi_{\underline{\alpha},C}(U)^m. 
		\end{equation}
		If $\Phi_{\underline{\alpha},C}(U)=1$, then $\Phi_{\underline{\alpha},C}(P)=\Phi_{\underline{\alpha},C}(Q)$. This implies that $\Phi_{\underline{\alpha},C}$ is degenerate, leading to a contradiction. Therefore, $\Phi_{\underline{\alpha},C}(U)<1$. 
		
		(iv) $\Rightarrow$ (iii): Let $k,k^\prime\in[d]$, $k\neq k^\prime$. Since $\Phi_{e_k,\{k,k^\prime\}}(U)<1$ and the columns in $U$ sum to $1$, $I(U,\{k,k^\prime\})$ cannot contain all rows $i$ in $U$ where $u^{(k)}_i>0$. Hence there exists a row $i$ such that $u^{(k)}_i>0$, but $u^{(k^\prime)}_i=0$, proving the claim. 
	\end{proof}
	
	\begin{remark}
		Note that it follows from the proof of (iv) $\Rightarrow$ (iii) above that we only have to require that monotones from a certain subset of $\mathcal{H}^{\rm n.d.}(S^d,\R_+^{\rm op})$ are strictly less than $1$, namely all those of the type $\Phi_{e_k,\{k,k^\prime\}}$, where $k,k^\prime\in[d]$, $k\neq k^\prime$. 		
	\end{remark}
	
	\subsection{Results: applying the Vergleichsstellensatz}
    \label{Msubsec:applyingvssnorestrictions}
	
	In the previous sections we have characterized all monotone homomorphisms, the monotone derivations at degenerate monotone homomorphisms, and the power universal elements in the matrix majorization semiring $S^d$. Now, we will apply the Vergleichsstellensatz in the form of Theorem \ref{thm:Fritz2022} to $S^d$ to obtain results on the large-sample and catalytic majorization of tuples of probability distributions. We obtain sufficient (and almost necessary) conditions for {\it exact} majorization in large samples and catalytically in Theorem \ref{Mthm:majorization} below. Furthermore, we derive necessary and sufficient conditions for majorization in large samples and catalytically where the majorization only happens {\it approximately} with asymptotically vanishing error in Theorem \ref{Mthm:genasymptocatalytic}.
	
	The conditions, in terms of the monotone homomorphisms, for a matrix $P\in \V^d$  to majorize another $Q\in \V^d$ in the large-sample or catalytic setting, can be expressed as 
	\begin{equation}\label{Meq:majorizationconditions}
		\Phi_{\underline{\alpha},C}(P)<\Phi_{\underline{\alpha},C}(Q)\quad\quad\forall\; \Phi_{\underline{\alpha},C}\in\mathcal{H}^{\rm n.d.}(S^d,\R_+^{\rm op}). 
	\end{equation}
    These homomorphisms $\Phi_{\underline{\alpha},C}$ are defined in \eqref{Meq:homomorphismsum} and if we adopt the convention $0^0=1$, we can rewrite these more conveniently as follows: 
    \begin{equation}
        \label{Meq:PhiR+op}
        \Phi_{\underline{\alpha},C}(P)=\sum_{i\in I(P,C)}\big(p^{(1)}_i\big)^{\alpha_1}\cdots\big(p^{(d)}_i\big)^{\alpha_d}, 
    \end{equation}
    where $\underline{\alpha}\in A_+$ (defined in \eqref{Meq:A+}), $C\subseteq[d]$ such that ${\rm supp}\,\underline{\alpha}\subseteq C$, and the sum runs only over the rows in $I(P,C)$ (defined in \eqref{Meq:Iset}). Furthermore, we exclude the degenerate homomorphisms $\Phi^{(k)}(P)=\|p^{(k)}\|_1$, which equals the sum of the entries in column $k$, for $k\in[d]$. 
	
	We note that if $P,Q\in\V^d$ have columns that sum to $1$ and satisfy \eqref{Meq:majorizationconditions}, then $P$ is automatically power universal. To see this, observe that for any $\Phi_{\underline{\alpha},C}\in\mathcal{H}^{\rm n.d.}(S^d,\R_+^{\rm op})$ we have $\Phi_{\underline{\alpha},C}(P)<\Phi_{\underline{\alpha},C}(Q)\leq1$, hence (iv) in Proposition \ref{Mlem:poweruniversalconditions} holds for $U=P$, which means $P$ is power universal. 
	
	Putting everything together, Theorem \ref{thm:Fritz2022} yields the following result. 
	\begin{theorem}\label{Mthm:majorization}
		Let $P=(p^{(1)},\ldots,p^{(d)})$, $Q=(q^{(1)},\ldots,q^{(d)})$ be tuples of probability vectors such that $\bigcap_{k=1}^d{\rm supp}\,p^{(k)}\neq\emptyset$ and $\bigcap_{k=1}^d{\rm supp}\,q^{(k)}\neq\emptyset$. If the inequalities in \eqref{Meq:majorizationconditions} hold, then we have the following: 
		\begin{enumerate}[(a)]
			\item There exist a stochastic map $T$ and a tuple of probability vectors $(r^{(1)},\ldots,r^{(d)})$, satisfying $\bigcap_{k=1}^d{\rm supp}\,r^{(k)}\neq\emptyset$, such that
			\begin{equation}
				T(r^{(k)}\otimes p^{(k)})=r^{(k)}\otimes q^{(k)}\quad\text{for }k\in[d], 
			\end{equation}
			and we may choose, for $n$ sufficiently large, 
			\begin{equation}
                \label{Meq:rform}
				r^{(k)}=\frac{1}{n+1}\bigoplus_{\ell=0}^n\left(q^{(k)}\right)^{\otimes \ell}\otimes\left(p^{(k)}\right)^{\otimes(n-\ell)}\quad\text{ for }k\in[d]. 
			\end{equation}
			\item For $n$ sufficiently large there exists a stochastic map $T_n$ such that 
			\begin{equation}
				T_n(p^{(k)})^{\otimes n}=(q^{(k)})^{\otimes n}\quad\text{for }k\in[d]. 
			\end{equation}
		\end{enumerate}
		Conversely, if either one of these holds, then the inequalities in \eqref{Meq:majorizationconditions} hold non-strictly. 
	\end{theorem}
	
	Just as in the case with equal support in \cite{farooq2024}, we have an asymptotic version of majorization. Note that now, as opposed to the exact version of majorization in the Theorem above, we have to additionally require that $P$ is power universal, which is also necessary for the approximate version in the equal-support case in \cite{farooq2024}.  
	\begin{theorem}\label{Mthm:genasymptocatalytic}
		Let $P=(p^{(1)},\ldots,p^{(d)})$, $Q=(q^{(1)},\ldots,q^{(d)})$ be tuples of probability vectors such that $\bigcap_{k=1}^d{\rm supp}\,p^{(k)}\neq\emptyset$ and $\bigcap_{k=1}^d{\rm supp}\,q^{(k)}\neq\emptyset$, and $P$ is power universal. The following are equivalent: 
		\begin{enumerate}[(i)]
			\item For all $\Phi_{\underline{\alpha},C}\in\mathcal{H}^{\rm n.d.}(S^d,\R_+^{\rm op})$ that are continuous (see \eqref{Meq:continuity}) it holds that 
            \begin{equation}\label{Meq:genasymptocatalyticcondition}
                \Phi_{\underline{\alpha},C}(P)\leq \Phi_{\underline{\alpha},C}(Q). 
            \end{equation}
			\item For every $\varepsilon>0$ there exist two tuples of probability vectors $(q_\varepsilon^{(1)},\ldots,q_\varepsilon^{(d)})$ and $(r_\varepsilon^{(1)},\ldots,r_\varepsilon^{(d)})$, satisfying $\bigcap_{k=1}^d{\rm supp}\,q_\varepsilon^{(k)}\neq\emptyset$ and $\bigcap_{k=1}^d{\rm supp}\,r_\varepsilon^{(k)}\neq\emptyset$, and a stochastic map $T_\varepsilon$ such that
			\begin{equation}
				\|q^{(k)}-q_\varepsilon^{(k)}\|_1\leq\varepsilon\quad\text{for }k\in[d], 
			\end{equation}
			and 
			\begin{equation}
				T_\varepsilon(r_\varepsilon^{(k)}\otimes p^{(k)})=r_\varepsilon^{(k)}\otimes q_\varepsilon^{(k)}\quad\text{for }k\in[d]. 
			\end{equation}
			Also, we may choose, for $n$ sufficiently large
			\begin{equation}
				r_\varepsilon^{(k)}=\frac{1}{n+1}\bigoplus_{\ell=0}^n\left(q_\varepsilon^{(k)}\right)^{\otimes\ell}\otimes\left(p^{(k)}\right)^{\otimes(n-\ell)}\quad\text{ for }k\in[d]. 
			\end{equation}
			\item For every $\varepsilon>0$ there exist a tuple of probability vectors $(q_\varepsilon^{(1)},\ldots,q_\varepsilon^{(d)})$, satisfying\\$\bigcap_{k=1}^d{\rm supp}\,q_\varepsilon^{(k)}\neq\emptyset$, such that
			\begin{equation}
				\|q^{(k)}-q_\varepsilon^{(k)}\|_1\leq\varepsilon\quad\text{for }k\in[d], 
			\end{equation}
			and for $n$ sufficiently large a stochastic map $T_{\varepsilon,n}$ such that 
			\begin{equation}
				T_{\varepsilon,n}(p^{(k)})^{\otimes n}=(q_\varepsilon^{(k)})^{\otimes n}\quad\text{for }k\in[d]. 
			\end{equation}
		\end{enumerate}
	\end{theorem}
	\begin{proof}
		(i) $\Rightarrow$ (ii) and (iii): Fix any $k\in[d]$. Define $W:=(w,\ldots,w)$, where $w:=(1/n,\ldots,1/n)\in\R^n$, and 
		\begin{equation}
			Q_\varepsilon:=\left(1-\frac{\varepsilon}{2}\right)Q+\frac{\varepsilon}{2}W, 
		\end{equation}
		where $+$ means entry-wise addition of two matrices. Write $Q_\varepsilon=(q_\varepsilon^{(1)},\ldots,q_\varepsilon^{(d)})$. Let $\Phi_{\underline{\alpha},C}\in\mathcal{H}^{\rm n.d.}(S^d,\R_+^{\rm op})$ be continuous. Then 
		\begin{align}
			\Phi_{\underline{\alpha},C}&(Q_\varepsilon)=\Phi_{\underline{\alpha},C}\left(\left(1-\frac{\varepsilon}{2}\right)Q+\frac{\varepsilon}{2}W\right)\geq \Phi_{\underline{\alpha},C}\left(\left(1-\frac{\varepsilon}{2}\right)Q\boxplus\frac{\varepsilon}{2}W\right)\\
			&=\left(1-\frac{\varepsilon}{2}\right)\Phi_{\underline{\alpha},C}(Q)+\frac{\varepsilon}{2}\Phi_{\underline{\alpha},C}(W)=\Phi_{\underline{\alpha},C}(Q)+\frac{\varepsilon}{2}[1-\Phi_{\underline{\alpha},C}(Q)]\\
            &\geq \Phi_{\underline{\alpha},C}(Q)\geq \Phi_{\underline{\alpha},C}(P), 
		\end{align}
		where in the first inequality we used monotonicity, in the equality that follows we used additivity and the fact that $\Phi_{\underline{\alpha},C}(cA)=c\Phi_{\underline{\alpha},C}(A)$ for any $c>0$ and $A\in \V^d$, in the next equality we used that $\Phi_{\underline{\alpha},C}(W)=1$ since all the columns in $W$ are the same and sum to $1$, in the penultimate inequality we used that $\Phi_{\underline{\alpha},C}(Q)\leq1$, and in the final inequality we used the assumption in (i). If $\Phi_{\underline{\alpha},C}(Q)=1$, then the final inequality is strict since $\Phi_{\underline{\alpha},C}(P)<1$ by Lemma \ref{Mlem:poweruniversalconditions}. If $\Phi_{\underline{\alpha},C}(Q)<1$, then the penultimate inequality is strict. Thus, in both cases $\Phi_{\underline{\alpha},C}(P)<\Phi_{\underline{\alpha},C}(Q_\varepsilon)$.
		
		Let $\Phi_{\underline{\alpha},C}\in\mathcal{H}^{\rm n.d.}(S^d,\R_+^{\rm op})$ be discontinuous. Since $Q_\varepsilon$ has no $0$ entries, $\Phi_{\underline{\alpha},C}(Q_\varepsilon)=\Phi_{\underline{\alpha},D}(Q_\varepsilon)$, where $D\subseteq C$ consists of all $c\in C$ that satisfy $\alpha_c>0$. If $\Phi_{\underline{\alpha},D}$ is one of the degenerate homomorphisms (i.e. the $1$-norm of one of the columns), then 
        \begin{equation}
			\Phi_{\underline{\alpha},C}(P)<1=\Phi_{\underline{\alpha},D}(Q_\varepsilon)=\Phi_{\underline{\alpha},C}(Q_\varepsilon), 
		\end{equation}
        where the inequality follows from Lemma \ref{Mlem:poweruniversalconditions}. If $\Phi_{\underline{\alpha},D}$ is nondegenerate, then we find 
		\begin{equation}
			\Phi_{\underline{\alpha},C}(P)\leq \Phi_{\underline{\alpha},D}(P)<\Phi_{\underline{\alpha},D}(Q_\varepsilon)=\Phi_{\underline{\alpha},C}(Q_\varepsilon), 
		\end{equation}
		where the first inequality is due to the fact that $\Phi_{\underline{\alpha},C}$ sums over all or a subset of the rows that $\Phi_{\underline{\alpha},D}$ sums over, and the second inequality is due to the fact that $\Phi_{\underline{\alpha},D}$ is continuous and nondegenerate. 
		
		We conclude that $\Phi_{\underline{\alpha},C}(P)<\Phi_{\underline{\alpha},C}(Q_\varepsilon)$ for all $\Phi_{\underline{\alpha},C}\in\mathcal{H}^{\rm n.d.}(S^d,\R_+^{\rm op})$, hence Theorem \ref{Mthm:majorization} yields (ii) and (iii). Note that one easily sees that $\bigcap_{k=1}^d{\rm supp}\,q_\varepsilon^{(k)}\neq\emptyset$ and the fact that $\bigcap_{k=1}^d{\rm supp}\,r_\varepsilon^{(k)}\neq\emptyset$ in (ii) follows from the possible choice of the $r_\varepsilon^{(k)}$ in \eqref{Meq:rform}. 
		
		(ii) or (iii) $\Rightarrow$ (i): The fact that $\bigcap_{k=1}^d{\rm supp}\,q_\varepsilon^{(k)}\neq\emptyset$, respectively $\bigcap_{k=1}^d{\rm supp}\,r_\varepsilon^{(k)}\neq\emptyset$, implies that $(q_\varepsilon^{(1)},\ldots,q_\varepsilon^{(d)})\in S^d$, respectively $(r_\varepsilon^{(1)},\ldots,r_\varepsilon^{(d)})\in S^d$. Then, for any continuous $\Phi\in\mathcal{H}^{\rm n.d.}(S^d,\R_+^{\rm op})$ we find for any $\varepsilon>0$, using monotonicity and multiplicativity, that $\Phi(P)\leq \Phi(Q_\varepsilon)$, where $Q_\varepsilon:=(q_\varepsilon^{(1)},\ldots,q_\varepsilon^{(d)})$. For any $\varepsilon>0$, write $Q_\varepsilon=Q_\varepsilon^0\boxplus Q_\varepsilon^1$, where $Q_\varepsilon^0$ consists of the first $n_Q$ rows of $Q_\varepsilon$ and $n_Q$ is the number of rows in $Q$ (if $Q_\varepsilon$ has less than $n_Q$ rows, then we add rows of zeros to form $Q_\varepsilon^0$). Note that $\lim_{\varepsilon\rightarrow0}Q_\varepsilon^0=Q$ and the sum of every column in $Q_\varepsilon^1$ is at most $\varepsilon$. Then 
        \begin{equation}
            \Phi(Q_\varepsilon)+1=\Phi(Q_\varepsilon\boxplus(1\,\cdots\,1))=\Phi(Q_\varepsilon^0)+\underbrace{\Phi(Q_\varepsilon^1\boxplus(1\,\cdots\,1))}_{\leq1+\varepsilon}, 
        \end{equation}
        where in the inequality we used monotonicity and the fact that the rows of $Q_\varepsilon^1\boxplus(1\,\cdots\,1)$ sum to at most $1+\varepsilon$. Hence $\Phi(Q_\varepsilon)\leq\Phi(Q_\varepsilon^0)+\varepsilon$. This implies that $\Phi(P)\leq\Phi(Q_\varepsilon^0)+\varepsilon$. Taking the limit $\varepsilon\rightarrow0$ in this inequality, since $\Phi$ is a continuous function when restricted to matrices with a fixed number of rows, we deduce that $\Phi(P)\leq\Phi(Q)$, completing the proof.     
	\end{proof}

    \begin{remark}
    \label{Mrmk:minimalrestrictions}
        As mentioned in the Introduction, we could drop the requirement \eqref{Meq:baselinecond}, namely that the intersection of the supports of the columns in a matrix is non-empty, and require instead that the pairwise intersection of supports is non-empty. This would yield a semiring with even less restrictions than the one studied in this section. However, in that case, the homomorphisms and derivations appearing in the conditions of the Vergleichsstellensatz (Theorem \ref{thm:Fritz2022}) will turn out to be only the bivariate ones, namely those associated with the R\'enyi divergences between each pair of columns of the matrix. To get an idea why this is so, consider the $d=3$ case. A matrix of the form 
        \begin{equation}
            P:=\left(\begin{array}{ccc}
            *&*&0\\
            *&0&*\\
            0&*&*
            \end{array}\right), 
        \end{equation}
        where the asterisks represent positive entries, has pairwise non-empty intersections of the supports, but does not have non-empty intersection over all the supports. Any homomorphism $\Phi_{\underline{\alpha},C}$ with a character $C$ containing more than two columns will evaluate to zero on the matrix above, since the set of rows $I(P,C)$ defined in \eqref{Meq:Iset} is empty. Hence $\Phi_{\underline{\alpha},C}$ has a non-trivial kernel, and such homomorphisms are excluded from the set of homomorphisms appearing in the conditions of Theorem \ref{thm:Fritz2022}. Therefore, the only non-degenerate monotone homomorphisms with trivial kernel associated with this semiring are $\Phi_{\underline{\alpha},\{k,k'\}}$ with $k,k'\in[d]$, $k\neq k'$, $\underline{\alpha}\in A_+$ having at most two non-zero components $\alpha_k,\alpha_{k'}\geq0$, and $\alpha_{l}=0$ for all $l\neq k,k'$. These homomorphisms correspond to the R\'enyi divergences between any pair of columns, with parameter $\alpha\in[0,1)$; see also Figure \ref{Mfig:emptytriangle} for a depiction.

        Additionally, the conditions for a matrix to be power universal are even stricter in this semiring than in $S^d$. Namely, a matrix having any common support among all columns is now no longer power universal, hence we can no longer use this semiring to find majorization conditions for such a matrix.  
    \end{remark}

\begin{figure}
\begin{center}
\begin{overpic}[scale=0.43,unit=1mm]{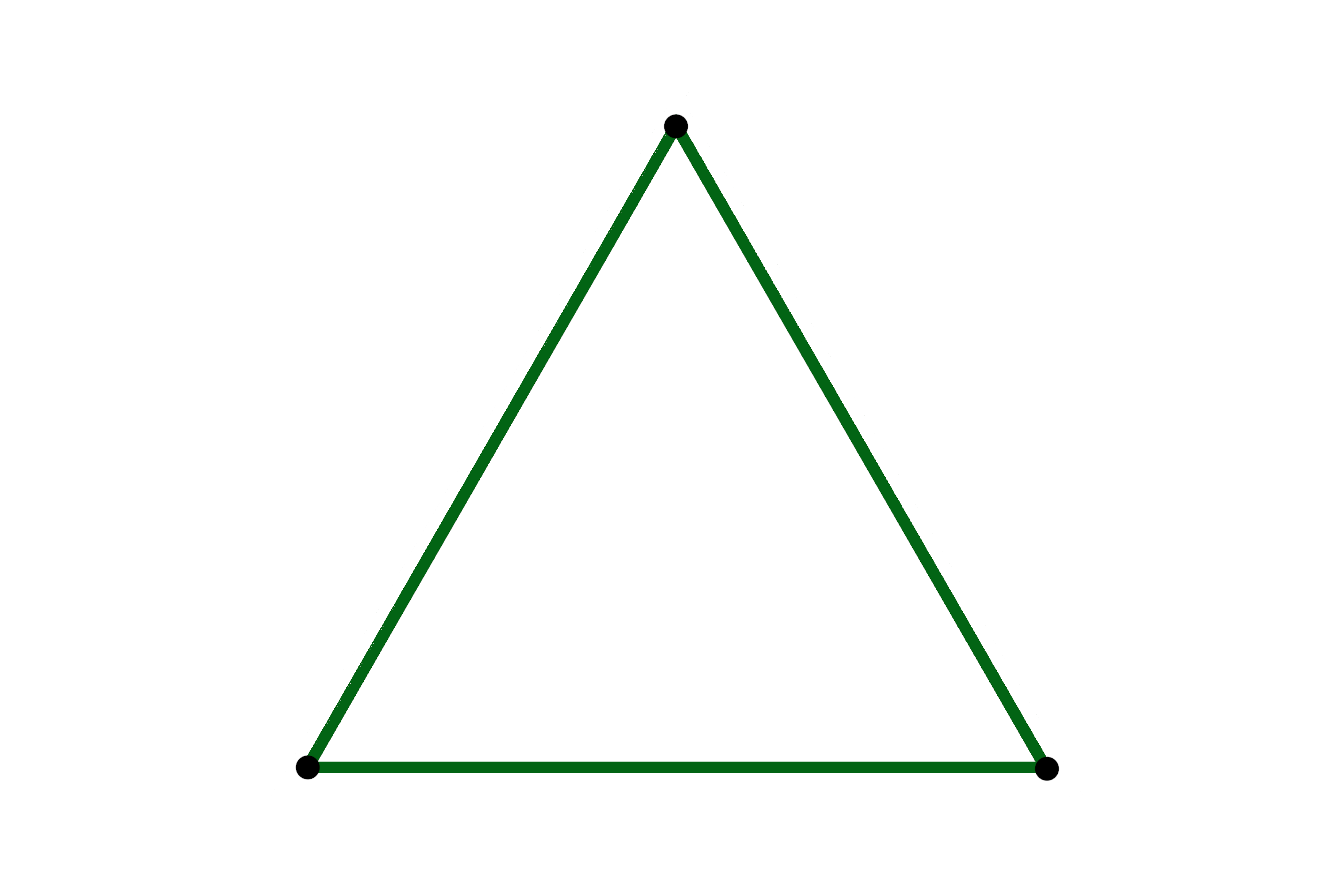}
\put(12,5){\begin{large}
$\Phi^{(1)}$
\end{large}}
\put(88,5){\begin{large}
$\Phi^{(2)}$
\end{large}}
\put(50,66){\begin{large}
$\Phi^{(3)}$
\end{large}}
\put(41,13){\begin{large}
\darker{$\Phi_{\alpha,1-\alpha,0,\{1,2\}}$}
\end{large}}
\put(70,35){\begin{large}
\darker{$\Phi_{0,\beta,1-\beta,\{2,3\}}$}
\end{large}}
\put(11,35){\begin{large}
\darker{$\Phi_{\gamma,0,1-\gamma,\{1,3\}}$}
\end{large}}
\put(21,6){
\darker{$\Phi_{e_1,\{1,2\}}$}
}
\put(10,13){
\darker{$\Phi_{e_1,\{1,3\}}$}
}
\put(74,6){
\darker{$\Phi_{e_2,\{1,2\}}$}
}
\put(83,13){
\darker{$\Phi_{e_2,\{2,3\}}$}
}
\put(38,62){
\darker{$\Phi_{e_3,\{1,3\}}$}
}
\put(55,62){
\darker{$\Phi_{e_3,\{2,3\}}$}
}
\end{overpic}
\caption{\label{Mfig:emptytriangle} A depiction of the monotone homomorphisms of the semiring discussed in Remark \ref{Mrmk:minimalrestrictions} in the case $d=3$. Only bivariate R\'{e}nyi divergences involving only two columns of the matrices are present. On the line segment connecting $e_1$ and $e_2$, e.g.,\ the homomorphisms correspond to $D_\alpha\big(p^{(1)}\big\|p^{(2)}\big)$ and $D_\alpha\big(p^{(2)}\big\|p^{(1)}\big)$ for $0\leq\alpha\leq 1/2$. Note that at the corner points we have again two nondegenerate (bivariate) homomorphisms in addition to the degenerate $\Phi^{(1)}$, $\Phi^{(2)}$, and $\Phi^{(3)}$.}
\end{center}
\end{figure}

\subsection{Minimality}\label{Msec:minimality}
First, we define a natural topology on $\mathcal{H}^{\rm n.d.}(S^d,\R_+^{\rm op})$ to make precise what is meant with an open subset. Let $C\subseteq[d]$ be a character with $|C|\geq2$. Define 
\begin{equation}
    A_C:=\left\{\,\underline{\alpha}\in A_+\,\middle|\,{\rm supp}\,\underline{\alpha}\subseteq C\,\right\}, 
\end{equation}
which is a compact subset of $\R^d$. Next, define the set of homomorphisms corresponding to the character $C$ and $\underline{\alpha}$ in $A_C$, i.e. 
\begin{equation}
    T_C:=\left\{\,\Phi_{\underline{\alpha},C}\,\middle|\,\underline{\alpha}\in A_C,\,\right\}. 
\end{equation}
Via the surjective map $A_C\rightarrow T_C$, $\underline{\alpha}\mapsto\Phi_{\underline{\alpha},C}$, $T_C$ inherits the compact and Hausdorff topology from $A_C$. The disjoint union of all such $T_C$ consists of precisely all non-degenerate monotone homomorphisms: 
\begin{equation}
    \mathcal{H}^{\rm n.d.}(S^d,\R_+^{\rm op})=\bigsqcup_{C\subseteq[d]\,:\,|C|\geq2} T_C. 
\end{equation}
Now, $\mathcal{H}^{\rm n.d.}(S^d,\R_+^{\rm op})$ can be equipped with the disjoint-union topology, making it a compact and Hausdorff space. 

Our result on the minimality of $\mathcal{H}^{\rm n.d.}(S^d,\R_+^{\rm op})$ in relation to characterizing large-sample and catalytic majorization can now be stated as follows. The proof uses an idea from the proof of Corollary 2.7 in \cite{strassen1988}. 
\begin{theorem}\label{Mthm:minimal}
     Let $O$ be a non-empty open subset of $\mathcal{H}^{\rm n.d.}(S^d,\R_+^{\rm op})$. There exist tuples of probability vectors $P=(p^{(1)},\ldots,p^{(d)})$, $Q=(q^{(1)},\ldots,q^{(d)})$ such that $\bigcap_{k=1}^d{\rm supp}\,p^{(k)}\neq\emptyset$ and $\bigcap_{k=1}^d{\rm supp}\,q^{(k)}\neq\emptyset$, $P$ is power universal, and such that \eqref{Meq:majorizationconditions} holds for all $\Phi_{\underline{\alpha},C}\in\mathcal{H}^{\rm n.d.}(S^d,\R_+^{\rm op})\setminus O$, but both (a) and (b) of Theorem \ref{Mthm:majorization} do not hold. Additionally, in Theorem \ref{Mthm:genasymptocatalytic}, \eqref{Meq:genasymptocatalyticcondition} holds for all continuous $\Phi_{\underline{\alpha},C}\in\mathcal{H}^{\rm n.d.}(S^d,\R_+^{\rm op})\setminus O$ \footnote{Note that any non-empty open $O\subseteq\mathcal{H}^{\rm n.d.}(S^d,\R_+^{\rm op})$ contains some continuous homomorphisms.}, but both (ii) and (iii) do not hold. 
\end{theorem}

\begin{proof}
For a fixed matrix $P\in S^d$, define the evaluation function 
\begin{equation}
    f_P:\mathcal{H}^{\rm n.d.}(S^d,\R_+^{\rm op})\rightarrow\R 
\end{equation}
by 
\begin{equation}
    f_P(\Phi_{\underline{\alpha},C}):=\Phi_{\underline{\alpha},C}(P), 
\end{equation}
for any $\Phi_{\underline{\alpha},C}\in\mathcal{H}^{\rm n.d.}(S^d,\R_+^{\rm op})$. For a fixed $C\subseteq[d]$ with $|C|\geq2$, one sees that $\alpha\mapsto\Phi_{\underline{\alpha},C}(P)$ is continuous on $A_C$, and therefore $f_P$ is continuous on $T_C$. Thus, by the properties of the discrete union topology, $f_P$ is continuous on $\mathcal{H}^{\rm n.d.}(S^d,\R_+^{\rm op})$. In other words, $f_P$ is in $C(\mathcal{H}^{\rm n.d.}(S^d,\R_+^{\rm op}),\R)$, the algebra of all continuous real-valued functions on $\mathcal{H}^{\rm n.d.}(S^d,\R_+^{\rm op})$. 

Next, let $\mc A$ be the $\R$-algebra generated by multiplication of the functions $f_P$, $P\in S^d$, by real numbers, and finite repetitions of point-wise addition and multiplication. Then $\mc A$ is a subalgebra of $C(\mathcal{H}^{\rm n.d.}(S^d,\R_+^{\rm op}),\R)$. Clearly, $f_{(1,\ldots,1)}(\Phi_{\underline{\alpha},C})=1$ for all $\Phi_{\underline{\alpha},C}\in\mathcal{H}^{\rm n.d.}(S^d,\R_+^{\rm op})$, hence $\mc A$ contains a non-zero constant function. One also sees that for any two distinct $\Phi_{\underline{\alpha},C},\Phi_{\underline{\alpha'},C'}\in\mathcal{H}^{\rm n.d.}(S^d,\R_+^{\rm op})$, there exists a $P\in S^d$ such that $f_P(\Phi_{\underline{\alpha},C})\neq f_P(\Phi_{\underline{\alpha'},C'})$, hence $\mc A$ separates points. Therefore, by the Stone-Weierstrass theorem, $\mc A$ is dense in $C(\mathcal{H}^{\rm n.d.}(S^d,\R_+^{\rm op}),\R)$.

Since $O$ is open, it contains a $\Phi_{\underline{\alpha_0},C_0}\in O$ that is continuous, i.e. satisfying ${\rm supp}\,\underline{\alpha_0}=C_0$ (see \eqref{Meq:continuity}). Clearly, there exists a continuous function $f:\mathcal{H}^{\rm n.d.}(S^d,\R_+^{\rm op})\rightarrow\R$ such that $f(\Phi_{\underline{\alpha},C})=1$ for all $\Phi_{\underline{\alpha},C}$ outside $O$, and $f(\Phi_{\underline{\alpha_0},C_0})=-1$. Since $\mc A$ is dense in $C(\mathcal{H}^{\rm n.d.}(S^d,\R_+^{\rm op}),\R)$, $f$ can be approximated by functions in $\mc A$ with arbitrary precision, from which it follows that there exist $r_1,\ldots,r_t\in\R$ and $P_1,\ldots,P_t\in S^d$ such that 
\begin{equation}\label{Meq:SWapproximation}
    r_1\Phi_{\underline{\alpha},C}(P_1)+\ldots+r_t\Phi_{\underline{\alpha},C}(P_t)\left\{\begin{array}{ll}>0&\text{if }\ \Phi_{\underline{\alpha},C}\in \mathcal{H}^{\rm n.d.}(S^d,\R_+^{\rm op})\setminus O\\<0&\text{if }\ \Phi_{\underline{\alpha},C}=\Phi_{\underline{\alpha_0},C_0}\end{array}\right..
\end{equation}
Here we have used the multiplicativity of the homomorphisms, which implies that $\mc A$ is the span of the set of all $f_P$, $P\in S^d$, with real coefficients. 
Define the matrices 
\begin{equation}
    S:=\bigboxplus_{i\,:\,r_i>0}r_iP_i,\quad R:=\bigboxplus_{i\,:\,r_i<0}-r_iP_i, 
\end{equation}
which are both contained in $S^d$. Then, by rearranging \eqref{Meq:SWapproximation}, if $\Phi_{\underline{\alpha},C}\in\mathcal{H}^{\rm n.d.}(S^d,\R_+^{\rm op})\setminus O$, 
\begin{equation}\label{Meq:counterexample1}
    \Phi_{\underline{\alpha},C}(S)>\Phi_{\underline{\alpha},C}(R),  
\end{equation}
and 
\begin{equation}\label{Meq:counterexample2}
    \Phi_{\underline{\alpha_0},C_0}(S)<\Phi_{\underline{\alpha_0},C_0}(R). 
\end{equation}
Note that it follows in particular that both matrices $R,S$ are non-zero. 

We may add a diagonal component to $R,S$ via 
\begin{equation}
    \tilde R:=R\boxplus{\rm diag}\{a_1,\ldots,a_d\},\quad\tilde S:=S\boxplus{\rm diag}\{b_1,\ldots,b_d\}, 
\end{equation}
with $a_1,\ldots,a_d>0$ and $b_1,\ldots,b_d\geq0$, to ensure that each column in both $\tilde R$ and $\tilde S$ sums to the same value $t$. Define the normalized matrices $P:=\tilde R/t$ and $Q:=\tilde S/t$, in $S$. Since $a_1,\ldots,a_d>0$, $P$ satisfies condition (ii) of Proposition \ref{Mlem:poweruniversalconditions}, hence $P$ is power universal. For any character $C\subseteq[d]$ with $|C|\geq2$, we have 
\begin{equation}
    \Phi_{\underline{\alpha},C}(P)=t^{-1}\Phi_{\underline{\alpha},C}(\tilde R)=t^{-1}\Phi_{\underline{\alpha},C}(R), 
\end{equation}
and similarly for $Q$. Hence, 
\begin{equation}
    \Phi_{\underline{\alpha},C}(Q)>\Phi_{\underline{\alpha},C}(P) 
\end{equation}
for all $\Phi_{\underline{\alpha},C}\in\mathcal{H}^{\rm n.d.}(S^d,\R_+^{\rm op})\setminus O$ and the reverse inequality when $\Phi_{\underline{\alpha},C}=\Phi_{\underline{\alpha_0},C_0}$. The statement of the theorem now follows from the fact that (a) or (b) in Theorem \ref{Mthm:majorization} imply the non-strict version of \ref{Meq:majorizationconditions} for the case of exact majorization, and from the implications (ii) $\Rightarrow$ (i), and (iii) $\Rightarrow$ (i) in Theorem \ref{Mthm:genasymptocatalytic} for the case of asymptotic majorization. 
\end{proof}
	
	\section{One dominating column case}\label{Msec:Asymmetric}
	
	In this section we will consider tuples of probability vectors where the support of one vector contains the support of all other vectors. Such tuples have applications in quantum thermodynamics, as we will see later. Similarly as for $S^d$, we will use the Vergleichsstellensatz given by Theorem \ref{thm:Fritz2022} to derive conditions for large-sample and catalytic majorization. 
	
	We define the \emph{dominating column matrix majorization semiring} $S^d_{\rm d.c.}$ as the subset of all (equivalence classes of) matrices in $S^d$ where the support of the last column contains all others: 
	\begin{equation}
		S^d_{\rm d.c.}:=\left\{\,\left[\left(p^{(1)},\ldots,p^{(d)}\right)\right]\in S^d\;\middle|\;{\rm supp}\,p^{(k)}\subseteq {\rm supp}\,p^{(d)}\quad\forall k\in[d-1]\,\right\}. 
	\end{equation}
 The set of individual matrices (i.e.,\ lifting the projection into equivalence classes) generating $S^d_{\rm d.c.}$ is denoted by $\mc V^d_{\rm d.c.}$. The subset of matrices in $\mc V^d_{\rm d.c.}$ with columns of length $n\in\N$ is denoted by $(\mc V_{\rm d.c.})^d_n$. In subsection \ref{Msubsec:poweruniversaldc}, we identify the power universals of $S^d_{\rm d.c.}$ thus showing that $S^d_{\rm d.c.}$ is of polynomial growth.
		
\subsection{Monotone homomorphisms}

Since $S^d_{\rm d.c.}$ is a subset of $S^d$, all monotone homomorphisms of $S^d$, when restricted to $S^d_{\rm d.c.}$, are also monotone homomorphisms of $S^d_{\rm d.c.}$. Thus, especially functions in $\mathcal{H}^{\rm n.d.}(S^d_{\rm d.c.},\R_+^{\rm op})$, which is defined in exactly the same way as in \eqref{Meq:temperateopposite}, are non-degenerate monotone homomorphisms $S^d_{\rm d.c.}\to\R_+^{\rm op}$. Note that in this case some homomorphisms that are different on $S^d$ coincide on $S^d_{\rm d.c.}$, since column $d$ always has a non-zero entry and therefore $\Phi_{\underline{\alpha},C}=\Phi_{\underline{\alpha},C\cup\{d\}}$ for any $\underline{\alpha}\in A_+$ and ${\rm supp}\,\underline{\alpha}\subseteq C\subseteq[d]$. This also means that the homomorphisms $\Phi_{e_k,\{k,d\}}$ for $k\in[d-1]$, which are nondegenerate on $S^d$, are degenerate on $S^d_{\rm d.c.}$. 

Let $\alpha>1$ and $c\in[d-1]$. Define $\varphi_{\alpha,c}:\R_+\times\R_{>0}\to\R$ by 
\begin{equation}
    \varphi_{\alpha,c}(x_c,x_d):=x_c^\alpha x_d^{1-\alpha}. 
\end{equation}
It follows from the proof of Proposition 13 in \cite{farooq2024} that $\varphi_{\alpha,c}$ is convex on $\R_{>0}^2$. Using the fact that $\alpha>0$ and the continuity of $\varphi_{\alpha,c}$, we find that $\varphi_{\alpha,c}$ is convex on its entire domain $\R_+\times\R_{>0}$. In the following, we will abuse notation by writing $\varphi_{\alpha,c}(x)$ for $x$ in a higher dimensional space $\R^d$ instead of $x\in \R_+\times\R_{>0}$, where it is understood that we ignore the coordinates not equal to $c$ or $d$. 

We now define the homomorphism $\Phi_{\alpha,c}:S^d\to\R_+$ (where we assume WLOG that $p^{(d)}_i>0$ for all $i=1,\ldots,n$) by 
\begin{equation}
\label{Meq:homdcR+}
    \Phi_{\alpha,c}(P):=\sum_{i=1}^n\varphi_{\alpha,c}\left(p_i^{(c)},p_i^{(d)}\right),\quad\quad P=(p^{(1)},\ldots,p^{(d)})\in (\V_{\rm d.c.})_{n}^d. 
\end{equation}
We denote all these homomorphisms $\Phi_{\alpha,c}$ that are non-degenerate by $\mathcal{H}^{\rm n.d.}(S^d_{\rm d.c.},\R_+)$, i.e.,
\begin{equation}
\mathcal{H}^{\rm n.d.}(S^d_{\rm d.c.},\R_+)=\big\{\Phi_{\alpha,c}\,\big|\,1<\alpha<\infty,\ c=1,\ldots,d-1\big\}.
\end{equation}
	
	\begin{proposition}
		All homomorphisms in $\mathcal{H}^{\rm n.d.}(S^d_{\rm d.c.},\R_+)$ are monotone. 
	\end{proposition}
	
	\begin{proof}
		Let $\alpha>1$ and $c\in[d-1]$. Consider $P\in S^d_{\rm d.c.}$ with rows $p_1,\ldots,p_n\in\R_+^d$ and a stochastic $m\times n$-matrix $T$, and write $Q:=TP\in S^d_{\rm d.c.}$ with rows $q_1,\ldots,q_m\in\R_+^d$, and $T_i:=\sum_{j=1}^n T_{i,j}>0$. Using the convexity of $\varphi_{\alpha,c}$ and the fact that $\left(p_j\right)_d>0$ for all $j\in[n]$, we find 
		\begin{align}
			\Phi_{\alpha,c}(Q)&=\sum_{i=1}^m\varphi_{\alpha,c}(q_i)=\sum_{i=1}^m\varphi_{\alpha,c}\left(\sum_{j=1}^n T_{i,j}p_j\right)\\
            &=\sum_{i=1}^m T_i\varphi_{\alpha,c}\left(\sum_{j=1}^n \frac{T_{i,j}}{T_i}p_j\right)\leq\sum_{i=1}^m T_i\sum_{j=1}^n \frac{T_{i,j}}{T_i}\varphi_{\alpha,c}(p_j)\\
            &=\sum_{j=1}^n \varphi_{\alpha,c}(p_j)\underbrace{\sum_{i=1}^m T_{i,j}}_{=1}=\Phi_{\alpha,c}(P), 
		\end{align}
		which completes the proof. 		
	\end{proof}

Let us also define, for all $P=\big(p^{(1)},\ldots,p^{(d)}\big)\in\mc V^d_{\rm d.c.}$ with $p^{(k)}=\big(p^{(k)}_1,\ldots,p^{(k)}_n\big)$ for $k=1,\ldots,d$ (where we assume WLOG that $p^{(d)}_i>0$ for all $i=1,\ldots,n$) and $c=1,\ldots,d-1$,
\begin{equation}
\Phi^\T_c(P):=\max_{1\leq i\leq n}\frac{p^{(c)}_i}{p^{(d)}_i}.
\end{equation}
We also define
\begin{equation}
\mathcal{H}^{\rm n.d.}(S^d_{\rm d.c.},\T\R_+):=\big\{\Phi^\T_{\beta,c}:=\big(\Phi^\T_c(\cdot)\big)^\beta\,\big|\,\beta>0\big\}.
\end{equation}
The functions in $\mathcal{H}^{\rm n.d.}(S^d_{\rm d.c.},\T\R_+\}$ are easily seen to be homomorphisms $S^d_{\rm d.c.}\to\T\R_+$. The conditions for large-sample or catalytic majorization of $P$ over $Q$ require in particular that $\Phi^\T_{\beta,c}(P)>\Phi^\T_{\beta,c}(Q)$ for all $\Phi^\T_{\beta,c}\in\mathcal{H}^{\rm n.d.}(S^d_{\rm d.c.},\T\R_+)$. Notice that this is equivalent to $\Phi^\T_c(P)>\Phi^\T_c(Q)$ for all $c=1,\ldots,d-1$, and that $\Phi^\T_c(P)=D_\infty\big(p^{(c)}\|p^{(d)}\big)$, where $D_\infty$ is the max-divergence. We next show that the homomorphisms in $\mathcal{H}^{\rm n.d.}(S^d_{\rm d.c.},\T\R_+)$ are monotone.

\begin{proposition}
All homomorphisms in $\mathcal{H}^{\rm n.d.}(S^d_{\rm d.c.},\T\R_+)$ are monotone.
\end{proposition}

\begin{proof}
Let $\beta>0$ and $c\in[d-1]$. Note that, denoting the rows of $P\in\mc V^d_{\rm d.c.}$ by $p_1,\ldots,p_n$, we have
\begin{equation}
\Phi^\T_{\beta,c}(P)=\sup_{1\leq j\leq n}\fii^\T_{\beta,c}(p_j), 
\end{equation}
where $\fii^\T_{\beta,c}(x_1,\ldots,x_d)=(x_c/x_d)^\beta$ for all $x_1,\ldots,x_{d-1}\geq0$ and $x_d>0$. The sublevel sets
\begin{equation}
\left\{(x_1,\ldots,x_d)\in\R_+^{d-1}\times\R_{>0}\,\middle|\,\left(\frac{x_c}{x_d}\right)^\beta\leq a\right\},\qquad a>0,
\end{equation}
are immediately seen to be convex, implying that $\fii^\T_{c,\beta}$ is quasi-convex, i.e.,\ for all $x,y\in\R_+^{d-1}\times\R_{>0}$ and $t\in(0,1)$,
\begin{equation}
\fii^\T_{c,\beta}\big(tx+(1-t)y\big)\leq\max\{\fii^\T_{c,\beta}(x),\fii^\T_{c,\beta}(y)\}.
\end{equation}

Let $P,Q\in\mc V^d_{\rm d.c.}$, $P$ with rows $p_1,\ldots,p_n$ and $Q$ with rows $q_1,\ldots,q_m$. Assume that $P\rgeq Q$ and let $T=(T_{i,j})_{i,j}$ be a stochastic matrix such that $Q=TP$. We again denote $T_i=\sum_{j=1}^n T_{i,j}$ which we may assume to be positive WLOG. Noting that $\fii^\T_{c,\beta}(cx)=\fii^\T_{c,\beta}(x)$ for all $c>0$ and $x\in\R_+^{d-1}\times\R_{>0}$ and using the quasi-convexity of $\fii^\T_{c,\beta}$, we now have
\begin{align}
\Phi^\T_{c,\beta}(Q)&=\max_{1\leq i\leq m}\fii^\T_{c,\beta}(q_i)=\max_{1\leq i\leq m}\fii^\T_{c,\beta}\left(\sum_{j=1}^n T_{i,j}p_j\right)\\
&=\max_{1\leq i\leq m}\fii^\T_{c,\beta}\left(\sum_{j=1}^n\frac{T_{i,j}}{T_i}p_j\right)\leq\max_{1\leq i\leq m}\max_{1\leq j\leq n}\fii^\T_{c,\beta}(p_j)\\
&=\max_{1\leq j\leq n}\fii^\T_{c,\beta}(p_j)=\Phi^\T_{c,\beta}(P),
\end{align}
proving the claim.
\end{proof}
	
	\begin{proposition}
		The non-degenerate monotone homomorphisms in $S_{\rm d.c.}^d$ are exactly those in the union of $\mathcal{H}^{\rm n.d.}(S^d_{\rm d.c.},\R_+)$, $\mathcal{H}^{\rm n.d.}(S^d_{\rm d.c.},\R_+^{\rm op})$, and $\mathcal{H}^{\rm n.d.}(S^d_{\rm d.c.},\T\R_+)$. There are no nondegenerate monotone homomorphisms $S^d_{\rm d.c.}\to\T\R_+^{\rm op}$.
	\end{proposition}
	
	\begin{proof}
Throughout this proof, we consider an arbitrary matrix $P=\big(p^{(1)},\ldots,p^{(d)}\big)\in\mc V^d_{\rm d.c.}$ with rows $p^{(k)}=\big(p^{(k)}_1,\ldots,p^{(k)}_n\big)$ for $k=1,\ldots,d$. We also denote the rows of $P$ by $p_1,\ldots,p_n$.

We first consider the case of a monotone homomorphism $\Phi:S^d_{\rm d.c.}\to\mathbb{K}\in\{\R_+,\R_+^{\rm op}\}$. As in the minimally restricted case $S^d$, also now $\Phi$ is associated with $\underline{\alpha}\in A_+\cup A_-\setminus\{e_1,\ldots,e_d\}$. Lemmas \ref{Mlem:1or2} and \ref{Mclm:nontrivialB} can be proven in exactly the same way (except now we only consider $B\subsetneq[d]$ such that $d\in B$) and we will use these results in the following. Hence, the notion of $\Phi$-nontrivial sets can again be defined and $\Phi$ also has a character $C\subseteq[d]$. Through the same proof as the one for Proposition \ref{Mpro:monhommin}, we know that
\begin{equation}
\Phi(P)=\sum_{i:\,C\subseteq{\rm supp}\,p_i}\prod_{k\in{\rm supp}\,p_i}\big(p_i^{(k)}\big)^{\alpha_k}
\end{equation}
and that ${\rm supp}\,\underline{\alpha}=\{k\,|\,\alpha_k\neq0\}\subseteq C$. As for the case $\mathbb{K}=\R_+^{\rm op}$, we already know that $\underline{\alpha}\in A_+$, so this case is dealt with and the set of nondegenerate monotone homomorphisms coincides with $\mathcal{H}^{\rm n.d.}(S^d_{\rm d.c.},\R_+^{\rm op})$. Thus we may concentrate on the case $\mathbb{K}=\R_+$.

Assume that $\mathbb{K}=\R_+$. From the results of \cite{farooq2024}, we know that $\underline{\alpha}\in A_-$. Let us first concentrate on the subset of those $\underline{\alpha}=(\alpha_1,\ldots,\alpha_d)$ where $\alpha_1>1$ and $\alpha_2,\ldots,\alpha_d\leq 0$. We show that $\alpha_2=\cdots=\alpha_{d-1}=0$. Consider
\begin{equation}
P_t:=\left(\begin{array}{ccccc}
1/2&1-t&1/2&\cdots&1/2\\
1/2&t&1/2&\cdots&1/2
\end{array}\right),\qquad 0\leq t\leq\frac{1}{2}.
\end{equation}
Assume $\alpha_2\neq0$, hence $\alpha_2<0$. Then, by Lemma \ref{Mlem:1or2}(ii), $2\in C$. Whenever $0\leq s\leq t\leq 1/2$, we have $P_s\rgeq P_t$. Thus, $\Phi(P_\varepsilon)\leq\Phi(P_0)$, i.e.,
\begin{equation}
\label{Meq:onlytwocolumns}
2^{-\alpha_1-\alpha_3-\cdots-\alpha_d}\left[(1-\varepsilon)^{\alpha_2}+\varepsilon^{\alpha_2}\right]\leq2^{-\alpha_1-\alpha_3-\cdots-\alpha_d}
\end{equation}
for all $\varepsilon\in(0,1/2]$. However, the left-hand side above diverges as $\varepsilon\to 0$, so $\alpha_2=0$. Similarly, $\alpha_3=\cdots=\alpha_{d-1}=0$. The same logic does not work for $\alpha_d$ because
\begin{equation}
\left(\begin{array}{cccc}
1/2&\cdots&1/2&1\\
1/2&\cdots&1/2&0
\end{array}\right)\notin\mc V^d_{\rm d.c.}, 
\end{equation}
hence we are not able to consider matrices similar to the $P_t$ for $0\leq t\leq1/2$ as above, but now with the last column equal to $(1-t,t)$. 

We know that $1,d\in C$ and we now show that there are no other columns in $C$. Assume that $C$ contains another column $k\in C\setminus\{1,d\}$. WLOG we assume that $k=2$. Then, \eqref{Meq:onlytwocolumns} holds with $\alpha_1+\alpha_d=1$ and $\alpha_3=\cdots=\alpha_{d-1}=0$, yielding $1\leq1/2$, a contradiction. Hence $C=\{1,d\}$. Therefore, in the part of $A_-$ where $\alpha_1>1$, the only non-degenerate homomorphisms $S^d_{\rm d.c.}\to\R_+$ are all $\Phi_{\alpha,1}$ for $\alpha_d>1$. In the same way, we can deal with the parts of $A_-$ where $\alpha_c>1$, $c=2,\ldots,d-1$ to show that the only homomorphisms are all $\Phi_{\alpha,c}$, for $\alpha_c>1$, $c=2,\ldots,d-1$. 

It remains to be shown what the homomorphisms are for the part of $A_-$ where $\alpha_d>1$. Assume that $\alpha_d\geq1$. 
Considering matrices
\begin{equation}
Q_t:=\left(\begin{array}{cccc}
1-t&1/2&\cdots&1/2\\
t&1/2&\cdots&1/2
\end{array}\right),\qquad 0\leq t\leq\frac{1}{2},
\end{equation}
we find that $Q_s\rgeq Q_t$ for $0\leq s\leq t\leq 1/2$, so $\Phi(Q_\varepsilon)\leq\Phi(Q_0)$ for all $\varepsilon\in(0,1/2]$. However, for $\alpha_1<0$, this again leads to a contradiction, so $\alpha_1=0$. Similarly, $\alpha_2=\cdots=\alpha_{d-1}=0$. Thus, $\alpha_d=1$ since $\alpha_1+\cdots+\alpha_d=1$. However, this case falls under the set $\mathcal{H}^{\rm n.d.}(S^d_{\rm d.c.},\R_+^{\rm op})$ which is already dealt with. Thus, the set of nondegenerate monotone homomorphisms $S^d_{\rm d.c.}\to\R_+$ is exactly $\mathcal{H}^{\rm n.d.}(S^d_{\rm d.c.},\R_+)$.

We now consider the case $\mathbb{K}=\T\R_+$. From the results of \cite{farooq2024}, we know that there is $\underline{\beta}=(\beta_1,\ldots,\beta_d)\in B_-$ such that, when the columns of $P$ have full support,
\begin{equation}
\Phi(P)=\max_{1\leq i\leq n}\prod_{k=1}^d \big(p^{(k)}_i\big)^{\beta_k}.
\end{equation}
Proceeding similarly as above for $\mathbb{K}=\R_+$ and considering the matrices $P_t$ and $Q_t$, one can again show that $\beta_d<0$ and that there is $c\in\{1,\ldots,d-1\}$ such that $\beta_c>0$ and $\beta_k=0$ for $k\in[d-1]\setminus\{c\}$. Thus, $\beta:=\beta_c=-\beta_d$.

  

Note that we may assume $P$ to have columns that sum up to $1$ because, in the general case, 
\begin{equation}
\Phi(P)=\Phi\left(\|p^{(1)}\|_1\,\cdots\,\|p^{(d)}\|_1\right)\Phi(P^\prime) 
\end{equation}
for a matrix $P'$ with columns that sum to $1$, and $\Phi\left(\|p^{(1)}\|_1\,\cdots\,\|p^{(d)}\|_1\right)$ is known since $\Phi$ is evaluated on a matrix with common support. Thus, we assume for now that $\|p^{(k)}\|=1$ for $k=1,\ldots,d$. By monotonity, this implies that $\Phi(P)\geq1$. 
  
        Let $p_1,\ldots,p_n\in\R_+^d$ denote the rows of $P$. Let $\delta>0$, which we will fix later. If a row $p_i$ contains a zero, but $(p_i)_k>0$, then we denote by $\widetilde{p_i}$ the row $p_i$ with every $0$ replaced by a $1$. If a row $p_i$ contains a zero, and $(p_i)_k=0$ is one of them, then we denote by $\widetilde{p_i}^\delta$ the row $p_i$ with $\delta$ on position $k$ and every other $0$ replaced by a $1$. Then, notice that 
		\begin{align}
			&P\boxplus\bigboxplus_{k\in{\rm supp}\,p_i\neq[d]}\widetilde{p_i}\boxplus\bigboxplus_{k\notin{\rm supp}\,p_i\neq[d]}\widetilde{p_i}^\delta&\label{Meq:PhiPdecomposedtropical1}\\
            &\qquad\qquad=\bigboxplus_{{\rm supp}\,p_i=[d]}p_i\boxplus\bigboxplus_{k\in{\rm supp}\,p_i\neq[d]}\left(\begin{array}{c}
				\widetilde{p_i}\\
				p_i
			\end{array}\right)\boxplus\bigboxplus_{k\notin{\rm supp}\,p_i\neq[d]}\left(\begin{array}{c}
				\widetilde{p_i}^\delta\\
				p_i
			\end{array}\right).&\label{Meq:PhiPdecomposedtropical2}
		\end{align}

Let us note that Lemma \ref{Mlem:1or2} still holds (with the obvious modification that $d\in B$), especially $\Phi[B]=1$ whenever $d\in B\subseteq[d]$. Using this, we have for the terms in the penultimate $\boxplus$-sum on the RHS of \eqref{Meq:PhiPdecomposedtropical1}, \eqref{Meq:PhiPdecomposedtropical2} (i.e.,\ when $i$ is such that $k\in{\rm supp}\,p_i$ or, simply, $i\in{\rm supp}\,p^{(k)}$ and ${\rm supp}\,p_i\neq [d]$)
\begin{align}
\Phi\big(\widetilde{p_i}\boxplus p_i\big)&=\Phi\left(\widetilde{p_i}\boxtimes\left(
\begin{array}{ccc}
1&\cdots&1\\
&{\rm supp}\,p_i&
\end{array}\right)\right)\\
&=\Phi(\widetilde{p_i})\Phi[{\rm supp}\,p_i]=\Phi(\widetilde{p_i})=\left(\frac{p^{(k)}_i}{p^{(d)}_i}\right)^\beta.
\end{align}
Similarly, for the terms in the final $\boxplus$-sum on the RHS of \eqref{Meq:PhiPdecomposedtropical1}, \eqref{Meq:PhiPdecomposedtropical2}, we have
\begin{equation}
\Phi(\widetilde{p_i}^\delta\boxplus p_i)=\left(\frac{\delta}{p^{(d)}_i}\right)^\beta.
\end{equation}
Then \eqref{Meq:PhiPdecomposedtropical1}, \eqref{Meq:PhiPdecomposedtropical2} implies 
		\begin{align}
			\label{Meq:Phitropicaldominatingcolumnalmost}
   \max&\left(\{\Phi(P)\}\cup\left\{\left(\frac{p^{(k)}_i}{p^{(d)}_i}\right)^\beta\,\middle|\,i:\ k\in{\rm supp}\,p_i\neq[d]\right\}\right.\\
   &\qquad\qquad\qquad\cup\left.\left\{\left(\frac{\delta}{p^{(d)}_i}\right)^\beta\,\middle|\,i:\ k\notin{\rm supp}\,p_i\neq[d]\right\}\right)\\
   &=\max\left(\left\{\left(\frac{p^{(k)}_i}{p^{(d)}_i}\right)^\beta\,\middle|\,i:\ k\in{\rm supp}\,p_i\right\}\right.\\
   &\qquad\qquad\qquad\qquad\cup\left.\left\{\left(\frac{\delta}{p^{(d)}_i}\right)^\beta\,\middle|\,i:\ k\notin{\rm supp}\,p_i\right\}\right).
		\end{align}
        By making $\delta$ suitably small, we can make sure that both maximizations above do not have an element of the form $\delta^\beta\big(p^{(d)}_i\big)^{-\beta}$, for some row $i$, as largest element. 
		
		For any $\varepsilon>0$, we have 
		\begin{equation}
			P\boxplus\left(\begin{array}{ccc}
				\varepsilon&\cdots&\varepsilon\\
				\vdots&\ddots&\vdots\\
				\varepsilon&\cdots&\varepsilon
			\end{array}\right)\rgeq P+\left(\begin{array}{ccc}
				\varepsilon&\cdots&\varepsilon\\
				\vdots&\ddots&\vdots\\
				\varepsilon&\cdots&\varepsilon
			\end{array}\right), 
		\end{equation}
		where $+$ denotes the entry-wise addition of matrices. This implies 
		\begin{equation}
			\Phi(P)\geq\max_{i\in[n]}\left(\frac{p^{(k)}_i+\varepsilon}{p^{(d)}_i+\varepsilon}\right)^\beta, 
		\end{equation}
        where we used the assumption that $\Phi(P)\geq1$. By continuity of the functions $\varepsilon\mapsto\frac{p^{(k)}_i+\varepsilon}{p^{(d)}_i+\varepsilon}$, $i\in[n]$, there exists $\varepsilon_{\rm max}$ such that the above maximum is obtained at the same $i_0$ for all $\varepsilon\leq\varepsilon_{\rm max}$. Then, taking $\varepsilon\rightarrow0$ on both sides of the inequality above, we deduce 
		\begin{equation}
			\Phi(P)\geq\left(\frac{p^{(k)}_{i_0}}{p^{(d)}_{i_0}}\right)^\beta\quad\forall\,i\in[n]. 
		\end{equation}
		This together with \eqref{Meq:Phitropicaldominatingcolumnalmost} yields 
        \begin{equation}
            \Phi(P)=\max_{i\in[n]}\left(\frac{p^{(k)}_i}{p^{(d)}_i}\right)^\beta. 
        \end{equation}
		
		Finally, we consider the case $\mathbb{K}=\T\R_+^{\rm op}$. Exactly as in the proof of Proposition 13 in \cite{farooq2024} (see Equations (89), (90)), we can show that $\Phi(P)$ only depends on the supports of the columns of $P$. More precisely, if $P,Q\in S_{\rm d.c.}^d$ satisfy ${\rm supp}\,p^{(k)}={\rm supp}\,q^{(k)}$ for all $k\in[d]$, then $\Phi(P)=\Phi(Q)$. It then follows that for any $Q\in S_{\rm d.c.}^d$, we have $\Phi(Q)\leq\Phi(\|q^{(1)}\|_1\,\cdots\,\|q^{(d)}\|_1)=\Phi(1\,\cdots\,1)=1$, since $Q\rgeq(\|q^{(1)}\|_1\,\cdots\,\|q^{(d)}\|_1)$ and the latter one-row matrix has the same support structure as $(1\,\cdots\,1)$. $P$ has at least one row not containing any $0$. Assume WLOG it is the first row. Then $P$ majorizes and is majorized by 
		\begin{equation}
			\left(\begin{array}{ccc}
				p_1^{(1)}/2&\cdots&p_1^{(d)}/2\\
				p_1^{(1)}/2&\cdots&p_1^{(d)}/2\\
				p_2^{(1)}&\cdots&p_2^{(d)}\\
				\vdots&\vdots&\vdots\\
				p_n^{(1)}&\cdots&p_n^{(d)}
			\end{array}\right), 
		\end{equation}
		hence 
        \begin{equation}
            \Phi(P)=\max\left(\underbrace{\Phi\left(p_1^{(1)}/2\;\cdots\;p_1^{(d)}/2\right)}_{=1},\underbrace{\Phi\left(\begin{array}{ccc}
				p_1^{(1)}/2&\cdots&p_1^{(d)}/2\\
				p_2^{(1)}&\cdots&p_2^{(d)}\\
				\vdots&\vdots&\vdots\\
				p_n^{(1)}&\cdots&p_n^{(d)}
			\end{array}\right)}_{\leq1}\right)=1. 
        \end{equation}
        Therefore, there are no non-degenerate monotone homomorphisms for $\mathbb{K}=\T\R_+^{\rm op}$. 
	\end{proof}

\begin{remark}
Note that, compared to the minimally restricted semiring $S^d$, some elements in $\mathcal{H}^{\rm n.d.}(S^d_{\rm d.c.},\R_+^{\rm op})$ coincide. Since the final column dominates all the other columns in all matrices in $\mc V^d_{\rm d.c.}$, it immediately follows that, whenever $\underline{\alpha}=(\alpha_1,\ldots,\alpha_{d-1},0)\in A_+$ and ${\rm supp}\,\underline{\alpha}\subseteq C\subseteq[d]$, we have $\Phi_{\underline{\alpha},C}=\Phi_{\underline{\alpha},C\cup\{d\}}$. See also Figure \ref{Mfig:DominatingColumn} for a depiction of the monotone homomorphisms in the case $d=3$. The figure also shows the derivations which we study in the next subsection.
\end{remark}

	\subsection{Monotone derivations}
	

    Recall that the degenerate monotone homomorphisms are all $\Phi^{(k)}$ for $k\in[d]$. Note that we only need to consider monotone derivations of these homomorphisms up to interchangeability. Similarly as for the minimal restrictions semiring, interchangeability allows us to show a result analogous to Lemma \ref{Mlem:interchangeability}, which we will use to characterize all monotone derivations, up to interchangeability. 
	
	\begin{proposition}
		\label{Mpro:derivationsdominatingcolumn}
		For each $k\in[d-1]$, any monotone derivation $\Delta$ of the degenerate homomorphism $\Phi^{(k)}$, up to interchangeability, satisfies $\Delta(P)=\gamma D_1(p^{(k)}\|p^{(d)})$ for all $P\in S^d_{\rm d.c}$, for some $\gamma\geq0$. The only monotone derivation $\Delta$ of the degenerate homomorphism $\Phi^{(d)}$, up to interchangeability, vanishes, i.e.,\ $\Delta(P)=0$ for all $P\in S^d_{\rm d.c}$. 
	\end{proposition}

    Let $\Delta:S^d_{\rm d.c}\to\mathbb{\R}$ be a monotone $\Phi^{(k)}$-derivation for some $k\in[d]$. For any $B\subsetneq[d]$, $d\in B$, we write 
    \begin{equation}
        \Delta[B]:=\Delta\left(\begin{array}{ccc}
                1&\cdots&1\\
                &B&
            \end{array}\right). 
    \end{equation}
Before going on to the proof, note that, whenever $B\subseteq[d]$ is such that $d\in B$, then $\Delta[B]=0$. The proof is essentially the same as that of Lemma \ref{Mlem:derivationszero} and will not be repeated here.
    

\begin{proof}[Proof of Proposition \ref{Mpro:derivationsdominatingcolumn}]
Let $\Delta:S_{\rm d.c.}^d\to\mathbb{\R}$ be a monotone $\Phi^{(k)}$-derivation for some $k\in[d]$. For any $x\in\R_+^d$ with $d\in{\rm supp}\,x$ we find using the above note and the Leibniz rule for $\Delta$,
\begin{align}
\Delta\big(\tilde{x}^T\boxplus x^T\big)&=\Delta\left(\tilde{x}^T\boxtimes\left(\begin{array}{ccc}
1&\cdots&1\\
&{\rm supp}\,x&
\end{array}\right)\right)\\
&=\Delta(\tilde{x}^T)\Phi^{(k)}[{\rm supp}\,x]+\Phi^{(k)}(\tilde{x}^T)\underbrace{\Delta[{\rm supp}\,x]}_{=0}\\
&=\left\{\begin{array}{ll}
2\Delta(\tilde{x}^T)&{\rm if}\ k\in{\rm supp}\,x\\
\Delta(\tilde{x}^T)&{\rm otherwise}
\end{array}\right.. 
\end{align}
Through a similar calculation as in \eqref{Meq:Phihajotus1}, \eqref{Meq:Phihajotus2} and \eqref{Meq:Phihajotus3}, we now find that
\begin{equation}
\Delta(P)=\sum_{i\in{\rm supp}\,p^{(k)}}\Delta(\widetilde{p_i}) 
\end{equation}
for any $P=(p^{(1)},\ldots,p^{(d)})\in S_{\rm d.c.}^d$ with rows $p_i$. Again, $\widetilde{p_i}$ is the row $p_i$ where the vanishing entries are replaced with 1. It follows that $\Delta(P)$ can be seen as $\Delta$ evaluated on a matrix with non-vanishing entries, more specifically the matrix resulting from taking only the rows in $P$ corresponding to ${\rm supp}\,p^{(k)}$ and replacing all vanishing entries with $1$. Hence, we can use the results in the proof of Proposition 14 in \cite{farooq2024} for the equal-support case. Combining Equations (111) and (117) in \cite{farooq2024}, we find 
\begin{equation}
\Delta(P)=\sum_{\ell\neq k}\gamma_\ell\sum_{i\in{\rm supp}\,p^{(k)}}p_i^{(k)}\log\left(\frac{p_i^{(k)}}{\widetilde{p_i^{(\ell)}}}\right) 
\end{equation}
for some constants $\gamma_\ell\geq0$, $\ell\in[d]\setminus\{k\}$. 
        
        Next, we show that $\gamma_\ell=0$ if $\ell\neq d$. Let $\ell\in[d-1]\setminus\{k\}$. Consider the following matrices in $S_{\rm d.c.}^d$ 
        \begin{align}
            P:=\left(\begin{array}{ccccccc}
                3/4&\cdots&3/4&1&3/4&\cdots&3/4\\
                1/4&\cdots&1/4&0&1/4&\cdots&1/4
            \end{array}\right),\\
            Q:=\left(\begin{array}{ccccccc}
                7/8&\cdots&7/8&1&7/8&\cdots&7/8\\
                1/8&\cdots&1/8&0&1/8&\cdots&1/8
            \end{array}\right),
        \end{align}
        where the column containing a $0$ is the $\ell$-th column. Notice that $P\rgeq Q$. A simple calculation shows that 
        \begin{equation}
            \Delta(P)=\gamma_\ell\left(\frac{3}{4}\log\frac{3}{4}+\frac{1}{4}\log\frac{1}{4}\right)\leq\gamma_\ell\left(\frac{7}{8}\log\frac{7}{8}+\frac{1}{8}\log\frac{1}{8}\right)=\Delta(Q), 
        \end{equation}
        with strict inequality if $\gamma_\ell>0$. This would imply that $\Delta$ is not monotone. Hence, we find that $\gamma_\ell=0$ for all $\ell\in[d-1]\setminus\{k\}$. Therefore, if $k\neq d$, 
        \begin{equation}
            \Delta(P)=\gamma_d\sum_{i\in{\rm supp}\,p^{(k)}}p_i^{(k)}\log\left(\frac{p_i^{(k)}}{\widetilde{p_i^{(d)}}}\right)=\gamma_d D_1\left(p^{(k)}\|p^{(d)}\right), 
        \end{equation}
        where we used that ${\rm supp}\,p^{(k)}\subseteq {\rm supp}\,p^{(d)}$ to identify the sum as the Kullback--Leibler divergence $D_1$. Note that the monotonicity of $\Delta$ follows from the data-processing inequality satisfied by $D_1$. For the case $k=d$, we have $\gamma_\ell=0$ for all $\ell\neq d$, hence $\Delta(P)=0$ for every $P\in S_{\rm d.c.}^d$. 		
	\end{proof}

\subsection{Power universal elements}\label{Msubsec:poweruniversaldc}

The following Proposition characterizes the power universal elements in $S_{\rm d.c.}^d$. 

\begin{proposition}\label{Mpro:poweruniversaldcgeneral}
    Let $U=(u^{(1)},\ldots,u^{(d)})\in S_{\rm d.c.}^d$ with all $u^{(k)}$ for $k\in[d]$, probability distributions. Then the following are equivalent: 
    \begin{enumerate}[(i)]
        \item $U$ is power universal. 
        \item It holds that ${\rm supp}\,u^{(k)}\subsetneq{\rm supp}\,u^{(d)}$ for all $k\in[d-1]$, and ${\rm supp}\,u^{(k)}\nsubseteq{\rm supp}\,u^{(k^\prime)}$ for all $k,k^\prime\in[d-1]$, $k\neq k^\prime$. 
        \item Any nondegenerate homomorphism
        \begin{equation}
            \Phi\in\mathcal{H}^{\rm n.d.}(S^d_{\rm d.c.},\R_+)\cup\mathcal{H}^{\rm n.d.}(S^d_{\rm d.c.},\R_+^{\rm op})\cup\mathcal{H}^{\rm n.d.}(S^d_{\rm d.c.},\T\R_+)
        \end{equation}
        satisfies $\Phi(U)\neq1$. 
    \end{enumerate}
\end{proposition}

\begin{proof}
    (iii) $\Rightarrow$ (ii): Note that for a monotone homomorphism $\Phi:S^d_{\rm d.c.}\to\mathbb{\R^{\rm op}}$, $\Phi(U)\neq1$ implies $\Phi(U)<1$. Since $\Phi_{e_d,\{k,d\}}(U)<1$ for $k\in[d-1]$, there is a row with a vanishing entry at $k$ and a positive value at $d$. Similarly, since $\Phi_{e_k,\{k,k^\prime,d\}}(U)<1$ for $k,k^\prime\in[d-1]$, $k\neq k^\prime$, there is a row with a vanishing entry at $k^\prime$ and a positive value at $k$. Hence, the conditions in (ii) are satisfied.  
    
    (ii) $\Rightarrow$ (i): From (ii), it follows that for each $k=2,\ldots,d-1$, $U$ has a row with positive entries at positions $1$ and $d$, and $0$ at position $k$. Then $U^{\boxtimes (d-2)}$ will contain the $\boxtimes$-product of these $d-2$ rows, which is a row with positive entries at positions $1$ and $d$, and $0$ everywhere else. Similarly, for the other columns $k^\prime=2,\ldots,d-1$, $U^{\boxtimes (d-2)}$ will contain a row with positive entries at positions $k^\prime$ and $d$ and $0$ everywhere else. Note that $U^{\boxtimes (d-1)}=U^{\boxtimes (d-2)}\boxtimes U$ will also contain rows of these specific types. From (ii) it also follows that for each $k=1,\ldots,d-1$, $U$ contains a row with a positive entry at position $d$ and $0$ at position $k$. Then $U^{\boxtimes (d-1)}$ will contain the $\boxtimes$-product of these $d-1$ rows, which is a row with a positive entry at position $d$ and $0$ everywhere else. One sees that it is possible to sum the rows in $U^{\boxtimes (d-1)}$ in such a way that 
    \begin{equation}
        \label{Meq:almostidentity}
        U^{\boxtimes (d-1)}\rgeq\left(\begin{array}{cccc}
            a_1&\cdots&a_{d-1}&b_0\\
            1-a_1&\cdots&0&b_1\\
            0&\ddots&\vdots&\vdots\\
            0&\cdots&1-a_{d-1}&b_{d-1}\\
            0&\cdots&0&1-b			
        \end{array}\right):=\tilde{U}, 
    \end{equation}
    where $b:=\sum_{i=0}^{d-1}b_i$, and $a_k,b_i,b\in(0,1)$ for $k=1,\ldots,d-1$ and $i=0,\ldots,d-1$. 
    Define the stochastic map (matrix) $T$ mapping probability vectors on $\{0,1,\ldots,d\}^2$ into probability vectors on $\{0,1.\ldots,d\}$ through
    \begin{equation}
    (Tp)_m=p_{m,m}+\sum_{j=0}^{m-1}\big(p_{j,m}+p_{m,j}\big),\qquad m=0,\ldots,d,
    \end{equation}
    for all $p=(p_{j,k})_{i,j=0}^d$ where empty sums $\sum_{j=0}^{-1}(\cdots)$ vanish. Also define matrices $U_\ell$ through $U_1=\tilde{U}$ and $U_{\ell+1}=T(\tilde{U}\boxtimes U_\ell)$ for all $\ell=1,2,\ldots$. One easily sees that $\tilde{U}^{\boxtimes\ell}\rgeq U_\ell$ and a straightforward calculation combined with an induction argument shows that, for all $\ell\in\N$, there are $b^{(\ell)}_i$, $i=0,\ldots,d-1$ such that $b^{(\ell)}_0+\cdots+b^{(\ell)}_{d-1}=b^\ell$ and
    \begin{equation}\label{Meq:majorizedbyU}
        U_\ell= \left(\begin{array}{cccc}
            a_1^\ell&\cdots&a_{d-1}^\ell&b_0^{(\ell)}\\
            1-a_1^\ell&\cdots&0&b_1^{(\ell)}\\
            0&\ddots&\vdots&\vdots\\
            0&\cdots&1-a_{d-1}^\ell&b_{d-1}^{(\ell)}\\
            0&\cdots&0&1-b^\ell		
        \end{array}\right).
    \end{equation}
    Hence, by choosing $\ell$ sufficiently large, we can make $a_k^\ell$ for $k=1,\ldots,d-1$ and $b_i^{(\ell)}$ for $i=0,\ldots, d-1$ arbitrarily small. 
    
Let $P=\big(p^{(1)},\ldots,p^{(d)}\big)\in \mc V_{\rm d.c.}^d$ with columns of unit $1$-norm, where $p^{(k)}=\big(p^{(k)}_1,\ldots,p^{(k)}_n\big)$. We may assume WLOG that $p^{(k)}_1>0$ for $k=1,\ldots,d$. Let us write $p_1:=\min_{1\leq k\leq d-1}p^{(k)}_1>0$, $t^{(0)}:=(1,0,\ldots,0)$ (with $n$ entries), and 
\begin{equation}
t^{(k)}_v:=\frac{1}{1-v}\big(p^{(k)}-vt^{(0)}\big)
\end{equation}
for $v\in(0,p_1)$ and $k=1,\ldots,d-1$. For any $v=(v_1,\ldots,v_{d-1})\in(0,p_1)^{d-1}$ and $w=(w_0,\ldots,w_{d-1})\in\R_{>0}^d$ such that $w_0+\cdots+w_{d-1}<1$, let us also define
\begin{equation}
t^{(d)}_{v,w}:=\frac{1}{1-w_0-\cdots-w_{d-1}}\big(p^{(d)}-w_0 t^{(0)}-w_1 t^{(1)}_{v_1}-\cdots-w_{d-1}t^{(d-1)}_{v_{d-1}}\big).
\end{equation}
A sufficient condition for $t^{(d)}_{v,w}$ to be a probability vector is
\begin{align}
p^{(d)}_1-w_0-\sum_{k=1}^{d-1}w_k\frac{p^{(k)}_1}{1-p_1}&\geq0,\\
p^{(d)}_i-\sum_{k=1}^{d-1}w_k\frac{p^{(k)}_i}{1-p_1}&\geq0,\qquad i=2,\ldots,n,
\end{align}
independent of $v_1,\ldots,v_{d-1}$. Thus, there is $w=(w_0,\ldots,w_{d-1})\in\R_{>0}^{d}$ such that $t^{(d)}_{v,w'}$ is a probability vector for all $v\in(0,p_1)^{d-1}$ and $w'\in\R_{>0}^d$ whose entries are upper-bounded by those of $w$. Suppose then that $\ell\in\N$ is large enough so that $a_k^\ell<p_1$ for $k=1,\ldots,d-1$ and $b^\ell<w_i$ (so that $b^{(\ell)}_i\leq b^\ell<w_i$) for $i=0,\ldots,d-1$ where $a_k$, $b_i$, and $b^{(\ell)}_i$ are as in \eqref{Meq:almostidentity} and \eqref{Meq:majorizedbyU}. Setting $v:=\big(a_1^\ell,\ldots,a_{d-1}^\ell\big)$ and $w':=\big(b_0^{(\ell)},\ldots,b_{d-1}^{(\ell)}\big)$ and $t^{(k)}:=t^{(k)}_{v}$ for $k=1,\ldots,d-1$ and $t^{(d)}:=t^{(d)}_{v,w'}$, we may set up the stochastic matrix $T=\big(t^{(0)},\ldots,t^{(d)}\big)$. It follows that we may write $TU_\ell=P$, i.e.,\ $U^{\boxtimes\ell(d-1)}\rgeq\tilde{U}^{\boxtimes\ell}\rgeq U_\ell\rgeq P$.

We now have shown that, for all $P\in\mc V^d_{\rm d.c.}$ with columns of unit $1$-norm, there is $\ell\in\N$ such that $U^{\boxtimes\ell}\rgeq P$. Let us assume next that $P,Q\in\mc V_{\rm d.c.}^d$ are such that $P\rgeq Q$. WLOG we may assume that the columns of $P$ and $Q$ have unit 1-norm. Thus, there is $\ell\in\N$ such that
\begin{equation}
Q\boxtimes U^{\boxtimes\ell}\rgeq(1\,\cdots\,1)\boxtimes U^{\boxtimes\ell}=U^{\boxtimes\ell}\rgeq P,
\end{equation}
showing that $U$ is a power universal.

(i) $\Rightarrow$ (iii): Let $P\in\mc V^d_{\rm d.c.}$ be such that all the columns of $P$ have unit 1-norm and $\Phi(P)\neq 1$ for all the nondegenerate monotone homomorphisms of $S^d_{\rm d.c.}$ into $\mb K\in\{\R_+,\R_+^{\rm op},\T\R_+,\T\R_+^{\rm op}\}$. We may choose, e.g.,
\begin{equation}
P=\left(\begin{array}{cccc}
1/2&\cdots&1/2&1/d\\
1/2&\cdots&0&1/d\\
\vdots&\ddots&\vdots&\vdots\\
0&\cdots&1/2&1/d
\end{array}\right), 
\end{equation}
for which $\Phi(P)>1$ when $\Phi\in\mc H^{\rm n.d.}(S^d_{\rm d.c.},\R_+)\cup\mc H^{\rm d.c.}(S^d_{\rm d.c.},\T\R_+)$ and $\Phi(P)<1$ otherwise. According to (i), there is $\ell\in\N$ such that $U^{\boxtimes\ell}\rgeq P$, so that $\Phi(U)^\ell\geq\Phi(P)>1$ for $\Phi\in\mc H^{\rm n.d.}(S^d_{\rm d.c.},\R_+)\cup\mc H^{\rm d.c.}(S^d_{\rm d.c.},\T\R_+)$ and $\Phi(U)^\ell\leq\Phi(P)<1$ otherwise. All in all, $\Phi(U)\neq1$ for all the monotone homomorphisms.

\end{proof}

\begin{remark}
    Note that for the implication (iii) $\Rightarrow$ (ii) above, we do not need all non-degenerate monotone homomorphisms to be unequal to $1$. It is sufficient that this holds for the homomorphisms $\Phi_{e_d,\{k,d\}}$ for $k\in[d-1]$, and $\Phi_{e_k,\{k,k^\prime,d\}}$ for $k,k^\prime\in[d-1]$, $k\neq k^\prime$. 
\end{remark}

\begin{remark}
Note that, according to item (ii) of Proposition \ref{Mpro:poweruniversaldcgeneral} and item (iii) of Proposition \ref{Mlem:poweruniversalconditions}, $U=\big(u^{(1)},\ldots,u^{(d)}\big)$ is a power universal of $S^d_{\rm d.c.}$ if and only if $\big(u^{(1)},\ldots,u^{(d-1)}\big)$ is a power universal of $S^{d-1}$ (n.b.,\ {\it not} of $S^{d-1}_{\rm d.c.}$) and ${\rm supp}\,u^{(k)}\subsetneq{\rm supp}\,u^{(d)}$ for $k=1,\ldots,d-1$. 

For $d=2$, $U=\big(u^{(1)},u^{(2)}\big)$ is a power universal of $S^d_{\rm d.c.}$ if and only if ${\rm supp}\,u^{(1)}\subsetneq{\rm supp}\,u^{(2)}$. For $d=3$, the general form of a power universal, up to the $\approx$-equivalence, is
\begin{equation}
U=\left(\begin{array}{ccc}
*&*&*\\
*&0&*\\
0&*&*
\end{array}\right), 
\end{equation}
where the asterisks represent positive entries (or vectors with positive entries).
\end{remark}

\subsection{Results: applying the Vergleichsstellensatz}

Applying the Vergleichsstellensatz in the form of Theorem \ref{thm:Fritz2022}, we can use the results from the previous sections to find the conditions for large-sample and catalytic majorization of a matrix $P$ over $Q$ in the semiring $S^d_{\rm d.c.}$. These conditions on the monotone homomorphisms and derivations can be expressed as 
\begin{align}
    \label{Meq:majorizationconditionsdominatingcolumn1}
      &\Phi_{\underline{\alpha},C}(P)<\Phi_{\underline{\alpha},C}(Q)&\forall\; \Phi_{\underline{\alpha},C}\in\mathcal{H}^{\rm n.d.}(S^d_{\rm d.c.},\R_+^{\rm op}),\\
    \label{Meq:majorizationconditionsdominatingcolumn2}
    &\Phi_{\alpha,c}(P)>\Phi_{\alpha,c}(Q)&\forall\; \Phi_{\alpha,c}\in\mathcal{H}^{\rm n.d.}(S^d_{\rm d.c.},\R_+),\\
    \label{Meq:majorizationconditionsdominatingcolumn3}
    &D_1(p^{(k)}\|p^{(d)})>D_1(q^{(k)}\|q^{(d)})&\forall k=1,\ldots,d-1,\\
    \label{Meq:majorizationconditionsdominatingcolumn4}
    &D_\infty(p^{(k)}\|p^{(d)})>D_\infty(q^{(k)}\|q^{(d)})&\forall k=1,\ldots,d-1. 
\end{align}
Recall that $D_1$ is the Kullback--Leibler divergence and $D_\infty$ the max-divergence between two columns. The homomorphisms $\Phi_{\underline{\alpha},C}\in\mathcal{H}^{\rm n.d.}(S^d_{\rm d.c.},\R_+^{\rm op})$ are exactly as in \eqref{Meq:PhiR+op}, where we again exclude all degenerate homomorphisms $\Phi^{(k)}$, $k\in[d]$ (i.e. those with either $\underline{\alpha}=e_k$ and $C=\{k\}$ for some $k\in[d]$, or $\underline{\alpha}=e_k$ and $C=\{k,d\}$ for some $k\in[d-1]$). The homomorphisms $\Phi_{\alpha,c}\in\mathcal{H}^{\rm n.d.}(S^d_{\rm d.c.},\R_+)$ were defined in \eqref{Meq:homdcR+} and can be conveniently written as 
\begin{equation}
    \Phi_{\alpha,c}(P)=\sum_{i=1}^n\big(p^{(c)}\big)^\alpha\big(p^{(d)}\big)^{1-\alpha}, 
\end{equation}
where $c\in[d-1]$ and $\alpha>1$. 

Note that just like the analogous result for the semiring $S^d$, the strict inequalities in \eqref{Meq:majorizationconditionsdominatingcolumn1} and \eqref{Meq:majorizationconditionsdominatingcolumn2} already imply that $P$ is power universal by Proposition \ref{Mpro:poweruniversaldcgeneral}, hence we do not need to require this explicitly in the conditions of the following Theorem. 

\begin{theorem}\label{Mthm:majorizationdominatingcolumn}
		Let $P=(p^{(1)},\ldots,p^{(d)})$, $Q=(q^{(1)},\ldots,q^{(d)})$ be tuples of probability vectors such that $\bigcap_{k=1}^d{\rm supp}\,p^{(k)}\neq\emptyset$ and ${\rm supp}\,p^{(k)}\subseteq{\rm supp}\,p^{(d)}$ for $k\in[d-1]$, and similarly for $Q$. If the inequalities in \eqref{Meq:majorizationconditionsdominatingcolumn1}, \eqref{Meq:majorizationconditionsdominatingcolumn2}, \eqref{Meq:majorizationconditionsdominatingcolumn3} and \eqref{Meq:majorizationconditionsdominatingcolumn4} hold, then we have the following: 
		\begin{enumerate}[(a)]
			\item There exist a stochastic map $T$ and a tuple of probability vectors $(r^{(1)},\ldots,r^{(d)})$, satisfying $\bigcap_{k=1}^d{\rm supp}\,r^{(k)}\neq\emptyset$ and ${\rm supp}\,r^{(k)}\subseteq{\rm supp}\,r^{(d)}$ for $k\in[d-1]$, such that
			\begin{equation}
				T(r^{(k)}\otimes p^{(k)})=r^{(k)}\otimes q^{(k)}\quad\text{for }k\in[d], 
			\end{equation}
			and we may choose, for $n$ sufficiently large, 
			\begin{equation}
				r^{(k)}=\frac{1}{n+1}\bigoplus_{\ell=0}^n\left(q^{(k)}\right)^{\otimes \ell}\otimes\left(p^{(k)}\right)^{\otimes(n-\ell)}\quad\text{ for }k\in[d]. 
			\end{equation}
			\item For $n$ sufficiently large there exists a stochastic map $T_n$ such that 
			\begin{equation}
				T_n(p^{(k)})^{\otimes n}=(q^{(k)})^{\otimes n}\quad\text{for }k\in[d]. 
			\end{equation}
		\end{enumerate}
		Conversely, if either one of these holds, then the inequalities \eqref{Meq:majorizationconditionsdominatingcolumn1}, \eqref{Meq:majorizationconditionsdominatingcolumn2}, \eqref{Meq:majorizationconditionsdominatingcolumn3} and \eqref{Meq:majorizationconditionsdominatingcolumn4} hold non-strictly. 
	\end{theorem}

    \begin{remark}
        We may derive an asymptotic version of this result, analogous to Theorem \ref{Mthm:genasymptocatalytic}. However, we would have to require the matrix $P$ to be power universal, which according to the characterization in (ii) of Proposition \ref{Mpro:poweruniversaldcgeneral} is quite restrictive. Nevertheless, for the case $d=2$, the situation is manageable, since there are essentially three interesting cases: either both columns have the same support and are different, or both columns have part of the support outside of the support of the other column, or one column has support strictly contained in the other. Majorization results (both for the exact and asymptotic setting) for the first case were derived in \cite{farooq2024}, and the second case is covered by the results in subsection \ref{Msubsec:applyingvssnorestrictions}. The third case is covered by Theorem \ref{Mthm:majorizationdominatingcolumn} (for the exact setting) above and by Theorem \ref{Mthm:asymptocatalytic} (for the asymptotic setting) in the next section, where we will derive an analogue of Theorem \ref{Mthm:genasymptocatalytic} for the case $d=2$. 
        
        Note that for $d>2$ there are many possibilities for the relationships of the supports of the columns, and each would require a different semiring to be studied. We will discuss this issue further in Section \ref{Msec:pusupports}. 
    \end{remark}

\subsection{The two-column case: dichotomies}\label{Msec:2column}

Here we specialize to the case $d=2$. In this context, any pair $\big(p^{(1)},p^{(2)}\big)$, where $p^{(1)}$ and $p^{(2)}$ are probability vectors (i.e.,\ $\big\|p^{(1)}\big\|_1=1=\big\|p^{(2)}\big\|_1$) and satisfy ${\rm supp}\,p^{(1)}\subseteq{\rm supp}\,p^{(2)}$, is called a \emph{dichotomy}. This special case will turn out to be particularly interesting in the theory of thermal majorization in the next section. We derive sufficient conditions for exact large-sample and catalytic majorization in this setting and sufficient and necessary conditions for approximate large-sample and catalytic majorization with asymptotically vanishing error in Corollary \ref{Mcor:dichotomyexact} and Theorem \ref{Mthm:asymptocatalytic}. We describe the conditions in the form of inequalities involving the traditional R\'{e}nyi divergences. We also complete the picture of asymptotic majorization in large samples and catalytically in Corollary \ref{Mcor:asymptocatalytic} to obtain a full characterization without having to make any support restrictions.

For all dichotomies $\big(p^{(1)},p^{(2)}\big)$, we have 
\begin{equation}
\frac{1}{\alpha-1}\log{\Phi_{(\alpha,1-\alpha),\{1,2\}}\big(p^{(1)}\|p^{(2)}\big)}=D_\alpha\big(p^{(1)}\big\|p^{(2)}\big), 
\end{equation}
when $\alpha\in[0,1)$, corresponding to the homomorphisms in $\mathcal{H}^{\rm n.d.}(S^d_{\rm d.c.},\R_+^{\rm op})$, and 
\begin{equation}
\frac{1}{\alpha-1}\log{\Phi_{\alpha,1}\big(p^{(1)}\|p^{(2)}\big)}=D_\alpha\big(p^{(1)}\big\|p^{(2)}\big), 
\end{equation}
when $\alpha\in(1,\infty)$, corresponding to the homomorphisms in $\mathcal{H}^{\rm n.d.}(S^d_{\rm d.c.},\R_+)$. 
We have the two degenerate homomorphisms $\Phi^{(1)}=\Phi_{(1,0),\{1\}}=\Phi_{(1,0),\{1,2\}}$ and $\Phi^{(2)}=\Phi_{(0,1),\{2\}}$. Up to a positive multiplier (and interchangeability), we have only one monotone derivation at $\Phi^{(1)}$, which is simply the Kullback--Leibler divergence, and none at $\Phi^{(2)}$. We also have essentially only one nondegenerate tropical homomorphism $\Phi^\T_1$ and
\begin{equation}
\log{\Phi^\T_1\big(p^{(1)},p^{(2)}\big)}=D_\infty\big(p^{(1)}\big\|p^{(2)}\big)
\end{equation}
for all dichotomies $\big(p^{(1)},p^{(2)}\big)$. Thus, we may unify the nondegenerate monotone homomorphisms and derivations of $S^2_{\rm d.c.}$ within the family $D_\alpha$, $0\leq\alpha\leq\infty$, of R\'{e}nyi divergences; see Figure \ref{Mfig:2d} for a depiction. The monotonicity of the homomorphisms and derivation is now encapsulated by the data processing inequality
\begin{equation}
D_\alpha\big(p^{(1)}\big\|p^{(2)}\big)\geq D_\alpha\big(Tp^{(1)}\big\|Tp^{(2)}\big),\qquad\forall\alpha\in[0,\infty],
\end{equation}
for all dichotomies $\big(p^{(1)},p^{(2)}\big)$ and stochastic matrices $T$.


We reformulate Theorem \ref{Mthm:majorizationdominatingcolumn} specified to the case of dichotomies and in terms of inequalities of R\'enyi divergences. 

\begin{corollary}\label{Mcor:dichotomyexact}
Let $\big(p^{(1)},p^{(2)}\big)$ and $\big(q^{(1)},q^{(2)}\big)$ be dichotomies. If
\begin{equation}\label{Meq:RenyiOrdered}
D_\alpha\big(p^{(1)}\big\|p^{(2)}\big)>D_\alpha\big(q^{(1)}\big\|q^{(2)}\big),\qquad\forall\alpha\in[0,\infty],
\end{equation}
then we have the following:
\begin{itemize}
\item[(a)] There is a stochastic matrix $T$ and a dichotomy $\big(r^{(1)},r^{(2)}\big)$ such that
\begin{equation}
T\big(r^{(k)}\otimes p^{(k)}\big)=r^{(k)}\otimes q^{(k)},\qquad k=1,2.
\end{equation}
There is $n\in\N$ such that we may choose
\begin{equation}\label{Meq:BinomialCatalyst}
r^{(k)}=\frac{1}{n+1}\bigoplus_{\ell=0}^n \big(q^{(k)}\big)^{\otimes\ell}\otimes\big(p^{(k)}\big)^{\otimes(n-\ell)},\quad k=1,2.
\end{equation}
\item[(b)] For any $n\in\N$ sufficiently large, there is a stochastic matrix $T_n$ such that
\begin{equation}
T_n\big(p^{(k)}\big)^{\otimes n}=\big(q^{(k)}\big)^{\otimes n},\qquad k=1,2.
\end{equation}
\end{itemize}
Conversely, if item (a) or (b) above holds, then \eqref{Meq:RenyiOrdered} holds with non-strict inequalities.
\end{corollary}

\begin{figure}
\begin{center}
\resizebox{\linewidth}{!}{
\begin{overpic}[scale=0.65,unit=1mm]{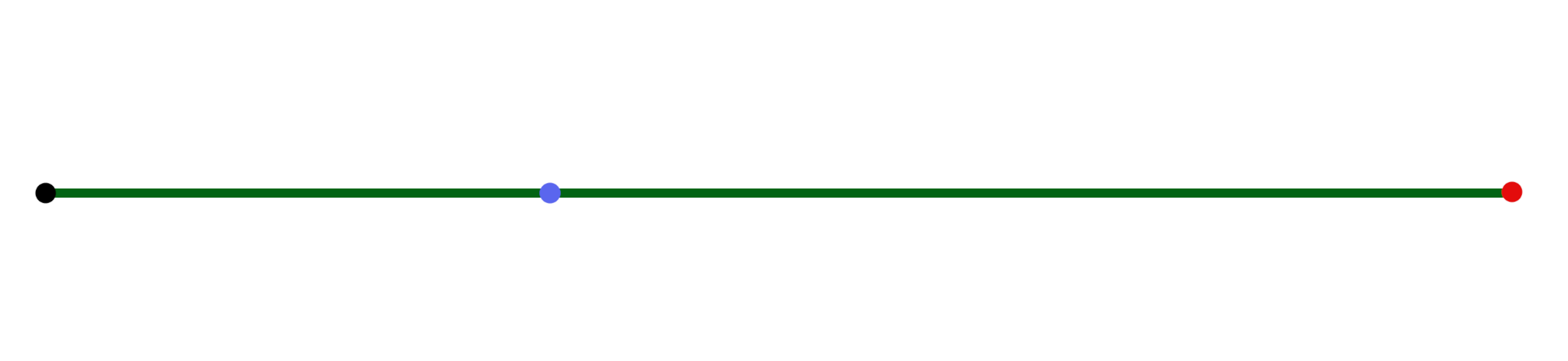}
\put(0,9){\begin{large}
$\Phi^{(2)}$
\end{large}}
\put(49,9){\begin{large}
$\Phi^{(1)}$
\end{large}}
\put(0,20){\begin{Large}
\darker{$D_0$}
\end{Large}}
\put(0,28){
\darker{$\Phi_{e_2,\{1,2\}}$}
}
\put(18,9){\begin{large}
\darker{$0<\alpha<1$}
\end{large}}
\put(18,20){\begin{Large}
\darker{$D_\alpha$}
\end{Large}}
\put(18,28){
\darker{$\Phi_{(\alpha,1-\alpha),\{1,2\}}$}
}
\put(90,9){\begin{large}
\darker{$1<\alpha<\infty$}
\end{large}}
\put(90,20){\begin{Large}
\darker{$D_\alpha$}
\end{Large}}
\put(90,28){
\darker{$\Phi_{(\alpha,1-\alpha),\{1,2\}}$}
}
\put(49,20){\begin{Large}
\sininen{$D_1$}
\end{Large}}
\put(135,9){\begin{large}
\punainen{$\alpha=\infty$}
\end{large}}
\put(140,20){\begin{Large}
\punainen{$D_\infty$}
\end{Large}}
\put(140,28){
\punainen{$\Phi^\T_1$}
}
\end{overpic}
}
\hspace{3cm}
\caption{\label{Mfig:2d} The monotone homomorphisms and derivations depicted for $S^2_{\rm d.c.}$: The line is the compactification of the half line $[0,\infty]$. Nondegenerate temperate homomorphisms are in green in the area $\alpha\in[0,1)\cup(1,\infty)$. These correspond to the R\'{e}nyi divergences $D_\alpha$. At points $\alpha=0$ and $\alpha=1$ we also have the degenerate homomorphisms $\Phi^{(2)}$ and $\Phi^{(1)}$, the latter of which is associated with the non-zero derivation $D_1$, the Kullback--Leibler divergence. At the infinity (in red) we have the tropical homomorphism corresponding to the max-divergence $D_\infty$.}
\end{center}
\end{figure}

We may also derive an asymptotic version of the above result using largely the same methods as in the proof of Theorem \ref{Mthm:genasymptocatalytic}. However, since the second probability vector dominates the first one in a dichotomy, we may remove the vanishing error completely on the second probability vector as stated at the end of the following result.

\begin{minipage}{\linewidth}
\begin{theorem}\label{Mthm:asymptocatalytic}
Consider dichotomies $\big(p^{(1)},p^{(2)}\big)$ and $\big(q^{(1)},q^{(2)}\big)$ 
and assume ${\rm supp}\,p^{(1)}\subsetneq{\rm supp}\,p^{(2)}$. The following conditions are equivalent:
\begin{itemize}
\item[(i)] For all $\alpha\in(0,\infty]$,
\begin{equation}\label{Meq:0alphainfty}
D_\alpha\big(p^{(1)}\big\|p^{(2)}\big)\geq D_\alpha\big(q^{(1)}\big\|q^{(2)}\big).
\end{equation}
\item[(ii)] For all $\varepsilon>0$, there is a stochastic matrix $T_\varepsilon$ and dichotomies $\big(q_\varepsilon^{(1)},q_\varepsilon^{(2)}\big)$ and $\big(r_\varepsilon^{(1)},r_\varepsilon^{(2)}\big)$ such that
\begin{equation}
\big\|q^{(k)}-q^{(k)}_\varepsilon\big\|_1\leq\varepsilon,\qquad k=1,2,
\end{equation}
and
\begin{equation}
T_\varepsilon\big(r_\varepsilon^{(k)}\otimes p^{(k)}\big)=r_\varepsilon^{(k)}\otimes q_\varepsilon^{(k)},\qquad k=1,2.
\end{equation}
Again, we may choose $r_\varepsilon^{(k)}$ as in \eqref{Meq:BinomialCatalyst} (with $q_\varepsilon^{(k)}$ substituted for $q^{(k)}$) for sufficiently large $n\in\N$.
\item[(iii)] For all $\varepsilon>0$, there is a dichotomy $\big(q_\varepsilon^{(1)},q_\varepsilon^{(2)}\big)$ 
such that
\begin{equation}
\big\|q^{(k)}-q_\varepsilon^{(k)}\big\|_1\leq\varepsilon,\qquad k=1,2,
\end{equation}
and for $n\in\N$ sufficiently large there is a stochastic matrix $T_{\varepsilon,n}$ such that
\begin{equation}
T_{\varepsilon,n}\big(p^{(k)}\big)^{\otimes n}=\big(q_\varepsilon^{(k)}\big)^{\otimes n},\qquad k=1,2.
\end{equation}
\end{itemize}
In statements (ii) and (iii) above, we may assume that $q_\varepsilon^{(2)}=q^{(2)}$ for all $\varepsilon>0$.
\end{theorem}
\end{minipage}

\begin{proof}
Let us assume condition (i). Let us choose $w\in\mc P_n$ which dominates $q^{(2)}$ (and, hence, automatically $q^{(1)}$). Naturally, we may choose $w=q^{(2)}$; this choice yields the final claim of the statement of the theorem, as one clearly subsequently sees. Let us denote, for brevity, $\big(p^{(1)},p^{(2)}\big)=:P$, $\big(q^{(1)},q^{(2)}\big)=:Q$, and $W:=(w,w)$ and define, for all $\varepsilon\in(0,2]$,
\begin{equation}
q^{(k)}_\varepsilon:=\left(1-\frac{\varepsilon}{2}\right)q^{(k)}+\frac{\varepsilon}{2}w,\qquad k=1,2,
\end{equation}
and
\begin{equation}
Q_\varepsilon:=\big(q^{(1)}_\varepsilon,q^{(2)}_\varepsilon\big)=\left(1-\frac{\varepsilon}{2}\right)Q+\frac{\varepsilon}{2}W.
\end{equation}
Clearly, $Q_\varepsilon$ is a dichotomy with mutually supporting columns for all $\varepsilon\in(0,2]$. Moreover, $\big\|q^{(k)}-q^{(k)}_\varepsilon\big\|_1\leq\varepsilon$, as one easily verifies. We may show that $D_\alpha\big(p^{(1)}\big\|p^{(2)}\big)>D_\alpha\big(q^{(1)}_\varepsilon\big\|q^{(2)}_\varepsilon\big)$ for all $\alpha\in[0,1)$ in the same way as in the proof of Theorem \ref{Mthm:genasymptocatalytic}.

Using the joint convexity of the Kullback--Leibler divergence, we have, for all $\varepsilon\in(0,2]$,
\begin{align}
D_1\big(q^{(1)}_\varepsilon\big\|q^{(2)}_\varepsilon\big)&\leq \left(1-\frac{\varepsilon}{2}\right)D_1\big(q^{(1)}\big\|q^{(2)}\big)+\frac{\varepsilon}{2}\underbrace{D_1(w\|w)}_{=0}\\
&\leq D_1\big(q^{(1)}\big\|q^{(2)}\big)\leq D_1\big(p^{(1)}\big\|p^{(2)}\big).
\end{align}
If $D_1\big(q^{(1)}\big\|q^{(2)}\big)=0$, then $D_1\big(q^{(1)}\big\|q^{(2)}\big)=0<D_1\big(p^{(1)}\big\|p^{(2)}\big)$ since $p^{(1)}\neq p^{(2)}$. If, on the other hand, $D_1\big(q^{(1)}\big\|q^{(2)}\big)>0$, then the penultimate inequality in the above calculation is strict. Thus, we have
\begin{equation}
D_1\big(q^{(1)}_\varepsilon\big\|q^{(2)}_\varepsilon\big)<D_1\big(p^{(1)}\big\|p^{(2)}\big),\qquad\forall\varepsilon\in(0,2].
\end{equation}

Let now $\alpha\in(1,\infty)$ and denote $\Phi_\alpha:=\Phi_{(\alpha,1-\alpha),\{1,2\}}$. Similarly as in the proof of Theorem \ref{Mthm:genasymptocatalytic}, we may evaluate, for all $\varepsilon\in(0,2]$,
\begin{align}
D_\alpha\big(q^{(1)}_\varepsilon\big\|q^{(2)}_\varepsilon\big)&\leq\frac{1}{\alpha-1}\log{\bigg(\Phi_\alpha(Q)-\frac{\varepsilon}{2}\underbrace{\big(\Phi_\alpha(Q)-1\big)}_{\geq0}\bigg)}\\
&\leq\frac{1}{\alpha-1}\log{\Phi_\alpha(Q)}=D_\alpha\big(q^{(1)}\big\|q^{(2)}\big)\leq D_\alpha\big(p^{(1)}\big\|p^{(2)}\big).
\end{align}
If $\Phi_\alpha(Q)=1$, then $\Phi_\alpha(Q)=1<\Phi_\alpha(P)$ and the final inequality above is strict. If $\Phi_\alpha(Q)>1$, then the penultimate inequality above is strict. Thus,
\begin{equation}
D_\alpha\big(q^{(1)}_\varepsilon\big\|q^{(2)}_\varepsilon\big)<D_\alpha\big(p^{(1)}\big\|p^{(2)}\big),\qquad\forall\varepsilon\in(0,2].
\end{equation}

Let us then look at the tropical homomorphism $\Phi^\T_1$. Using monotonicity, we have, for all $\varepsilon\in(0,2]$,
\begin{align}
\Phi^\T_1(Q_\varepsilon)&=\Phi^\T_1\left(\left(1-\frac{\varepsilon}{2}\right)Q+\frac{\varepsilon}{2}W\right)\leq\Phi^\T_1\left(\left(1-\frac{\varepsilon}{2}\right)Q\boxplus\frac{\varepsilon}{2}W\right)\\
&=\max\big\{\Phi^\T_1(Q),1\big\}=\Phi^\T_1(Q)\leq\Phi^\T_1(P).
\end{align}
Especially, $\Phi^\T_1(Q_\varepsilon)\leq\Phi^\T_1(Q)$, i.e.,
\begin{equation}
\max_{1\leq i\leq n}\frac{\left(1-\frac{\varepsilon}{2}\right)q^{(1)}_i+\frac{\varepsilon}{2}w_i}{\left(1-\frac{\varepsilon}{2}\right)q^{(2)}_i+\frac{\varepsilon}{2}w_i}\leq\max_{1\leq i\leq n}\frac{q^{(1)}_i}{q^{(2)}_i}
\end{equation}
for all $\varepsilon\in(0,2]$. When $0<\varepsilon\leq\varepsilon'\leq 2$, then $Q_\varepsilon\rgeq Q_{\varepsilon'}$, implying that the LHS above is non-increasing in $\varepsilon$ and coincides with the RHS when $\varepsilon\to0$. Thus, if we have equality with some $\varepsilon_0\in(0,2]$, then we have equality for all $\varepsilon\in[0,\varepsilon_0]$. Let us assume this. Since the LHS is a pointwise maximum of a finite set of real-analytic functions, the maximum is reached with the same index $i$ for all $\varepsilon\in[0,\varepsilon_0]$. The real-analyticity of the LHS now implies that the LHS is constant for all $\varepsilon\in[0,2]$, i.e.,\ $\Phi^\T_1(Q)=1=\Phi^\T_1(W)$. Thus, if $\Phi^\T_1(Q_\varepsilon)=\Phi^\T_1(Q)$ for some $\varepsilon\in(0,2]$, this is the case for all $\varepsilon\in[0,2]$ and $\Phi^\T_1(Q)=1<\Phi^\T_1(P)$. On the other hand, if $\Phi^\T_1(Q)>1$ then $\Phi^\T(Q_\varepsilon)<\Phi^\T(Q)$, so that $\Phi^\T_1(Q_\varepsilon)<\Phi^\T_1(Q)\leq\Phi^\T_1(P)$. All in all,
\begin{equation}
D_\infty\big(q^{(1)}_\varepsilon\big\|q^{(2)}_\varepsilon\big)<D_\infty\big(p^{(1)}\big\|p^{(2)}\big),\qquad\forall\varepsilon\in(0,2].
\end{equation}
Thus, the conditions of Corollary \ref{Mcor:dichotomyexact} are satisfied for $Q_\varepsilon$ and $P$, and conditions (ii) and (iii) of the claim follow.

Let us then assume that (ii) or (iii) holds. The proof that $D_\alpha\big(p^{(1)}\big\|p^{(2)}\big)\geq D_\alpha\big(q^{(1)}\big\|q^{(2)}\big)$ for $\alpha\in(0,1)$ follows from the same argument as in the end of the proof of Theorem \ref{Mthm:genasymptocatalytic}. To show the inequality for $\alpha\in(1,\infty)$, we use almost the same argument: consider any $D_\alpha$ with $\alpha\in(1,\infty)$, which corresponds to the monotone homomorphism $\Phi_{\alpha,1}$. Then we use similar notation and reasoning as in the proof of Theorem \ref{Mthm:genasymptocatalytic}, except that $\Phi=\Phi_{\alpha,1}$ and the direction of the inequalities is reversed, and we note that 
\begin{equation}
    \Phi_{\alpha,1}(Q_\varepsilon^1\boxplus(1\,\cdots\,1))\geq(1+\varepsilon)^{1-\alpha}, 
\end{equation}
where we used monotonicity of $\Phi_{\alpha,1}$ and that the sum of each column in $Q_\varepsilon^1\boxplus(1\,\cdots\,1)$ is between $1$ and $1+\varepsilon$. Finally, the inequalities $D_\alpha\big(p^{(1)}\big\|p^{(2)}\big)\geq D_\alpha\big(q^{(1)}\big\|q^{(2)}\big)$ for $\alpha\in\{1,\infty\}$ follow from the same inequalities for $\alpha\in(0,1)\cup(1,\infty)$ and the continuity of $D_\alpha$ in $\alpha$. 
\end{proof}

Combining the above result with Theorem 22 in \cite{farooq2024} gives immediately the following corollary.

\begin{corollary}\label{Mcor:asymptocatalytic}
Denote by $\mc P_n$ the set of probability vectors with $n$ components and consider probability vectors $p^{(2)},q^{(2)}\in\mc P_n$ of full support and some additional $p^{(1)},q^{(1)}\in\mc P_n$. Also assume that $p^{(1)}\neq p^{(2)}$. The following conditions are equivalent:
\begin{itemize}
\item[(i)] For all $\alpha\geq 1/2$,
\begin{equation}\label{Meq:jointcond}
D_\alpha\big(p^{(1)}\big\|p^{(2)}\big)\geq D_\alpha\big(q^{(1)}\big\|q^{(2)}\big),\qquad D_\alpha\big(p^{(2)}\big\|p^{(1)}\big)\geq D_\alpha\big(q^{(2)}\big\|q^{(1)}\big).
\end{equation}
\item[(ii)] For all $\varepsilon>0$, there is a stochastic matrix $T_\varepsilon$ and dichotomies $\big(q_\varepsilon^{(1)},q_\varepsilon^{(2)}\big)$ and $\big(r_\varepsilon^{(1)},r_\varepsilon^{(2)}\big)$ such that
\begin{equation}
\big\|q^{(k)}-q^{(k)}_\varepsilon\big\|_1\leq\varepsilon,\qquad k=1,2,
\end{equation}
and
\begin{equation}
T_\varepsilon\big(r_\varepsilon^{(k)}\otimes p^{(k)}\big)=r_\varepsilon^{(k)}\otimes q_\varepsilon^{(k)},\qquad k=1,2.
\end{equation}
Again, we may choose $r_\varepsilon^{(k)}$ as in \eqref{Meq:BinomialCatalyst} (with $q_\varepsilon^{(k)}$ substituted for $q^{(k)}$) for sufficiently large $n\in\N$.
\item[(iii)] For all $\varepsilon>0$, there is a dichotomy $\big(q_\varepsilon^{(1)},q_\varepsilon^{(2)}\big)$ with probability vectors in $\mc P_n$
such that
\begin{equation}
\big\|q^{(k)}-q_\varepsilon^{(k)}\big\|_1\leq\varepsilon,\qquad k=1,2,
\end{equation}
and for $n\in\N$ sufficiently large there is a stochastic matrix $T_{\varepsilon,n}$ such that
\begin{equation}
T_{\varepsilon,n}\big(p^{(k)}\big)^{\otimes n}=\big(q_\varepsilon^{(k)}\big)^{\otimes n},\qquad k=1,2.
\end{equation}
\end{itemize}
In statements (ii) and (iii), we may assume that $q_\varepsilon^{(2)}=q^{(2)}$ for all $\varepsilon>0$.
\end{corollary}

\begin{proof}
Let us first assume that ${\rm supp}\,p^{(1)}={\rm supp}\,p^{(2)}=\{1,\ldots,n\}$ and ${\rm supp}\,q^{(1)}={\rm supp}\,q^{(2)}=\{1,\ldots,n\}$. In this case, the desired result is given by Theorem 22 in \cite{farooq2024}. Strictly speaking, only the equivalence of statements (i) and (ii) is shown there, but adding (iii) can be done in the same way that we have already employed earlier.

Let us now assume that ${\rm supp}\,p^{(1)}\subsetneq{\rm supp}\,p^{(2)}$ or ${\rm supp}\,q^{(1)}\subsetneq{\rm supp}\,q^{(2)}$. Let us first show that the latter case is out of the question if the first one does not hold. So, assume that ${\rm supp}\,p^{(1)}={\rm supp}\,p^{(2)}$ while ${\rm supp}\,q^{(1)}\subsetneq{\rm supp}\,q{(2)}$. If we assume (i), we obtain from the second set of inequalities that
\begin{equation}
1=\sum_{i\in{\rm supp}\,p^{(1)}}p^{(2)}_i\leq\sum_{i\in{\rm supp}\,q^{(1)}}q^{(2)}_i<1
\end{equation}
as $\alpha\to 1$, a contradiction. Thus, we only have to consider the case where ${\rm supp}\,p^{(1)}\subsetneq{\rm supp}\,p^{(2)}$. In this case, the second set of the inequalities in \eqref{Meq:jointcond} are irrelevant when $\alpha\geq1$. Indeed, $D_\alpha\big(p^{(2)}\big\|p^{(1)}\big)=\infty$ when $\alpha\geq1$. Thus we may ignore this set of inequalities. For $0<\alpha\leq1/2$, we have $D_\alpha\big(p^{(1)}\big\|p^{(2)}\big)=\frac{1-\alpha}{\alpha}D_{1-\alpha}\big(p^{(2)}\big\|p^{(1)}\big)$. This means that, if (i) holds, then $D_\alpha\big(p^{(1)}\big\|p^{(2)}\big)\geq D_\alpha\big(q^{(1)}\big\|q^{(2)}\big)$ for all $\alpha>0$. Naturally this also holds for $\alpha=0$ as a pointwise limit of the above. Finally, condition (i) is equivalent with \eqref{Meq:0alphainfty}, and we are in the situation of Theorem \ref{Mthm:asymptocatalytic} providing the proof for the remainder of the claim.
\end{proof}

\section{Some applications}

In this section, we discuss some applications and corollaries of the results we have derived in this work. The first immediate result deals with the optimal rate of transforming a given statistical experiment into another. Our second field of applications is quantum thermodynamics and the characterization of asymptotic catalytic thermal majorization of quantum states commuting with the thermal (Gibbs) state. Together with the results of \cite{farooq2024}, we give a coherent proof for the statements made in \cite{Brandao2014}.

\subsection{Achievable rates}\label{Msec:rates}
We now apply Proposition \ref{prop:optimalconversionrate} on the optimal conversion rate to the setting of Section \ref{Msec:2column} dealing with dichotomies. For the semiring of dichotomies with a dominating column $S_{\rm d.c.}^2$, the set of divergences $\mc D(S_{\rm d.c.}^2)$ defined as in Section \ref{sec:optimalconversionrate} consists of the familiar R\'{e}nyi divergences as defined in \eqref{Meq:RenyiDivergences}. For this, we choose the power universal 
\begin{equation}
u=\left(\begin{array}{cc}
1&1/2\\
0&1/2
\end{array}\right),
\end{equation}
in the definition of the divergences. 

Suppose now that $P=\big(p^{(1)},p^{(2)}\big)$ and $Q=\big(q^{(1)},q^{(2)}\big)$ are dichotomies such that ${\rm supp}\,p^{(1)}\subsetneq{\rm supp}\,p^{(2)}$ (i.e.,\ $P$ is a power universal of the $d=2$ dominating column semiring). The rate \eqref{eq:optimalrategeneral} now takes the form
\begin{equation}\label{Meq:dominatingcolumndichotomyrate}
r(P\to Q)=\min_{\alpha\in[0,\infty]}\frac{D_\alpha\big(p^{(1)}\big\|p^{(2)}\big)}{D_\alpha\big(q^{(1)}\big\|q^{(2)}\big)}.
\end{equation}
Naturally, similar results also hold in the higher $d$ case, but these look quite complex. Let us see however what the optimal rate is in the case of dichotomies with full support studied as a special case in \cite{farooq2024}: Given dichotomies $P=\big(p^{(1)},p^{(2)}\big)$ and $Q=\big(q^{(1)},q^{(2)}\big)$ where ${\rm supp}\,p^{(1)}={\rm supp}\,p^{(2)}$, ${\rm supp}\,q^{(1)}={\rm supp}\,q^{(2)}$, and $p^{(1)}\neq p^{(2)}$, we have
\begin{equation}
r(P\to Q)=\min\left\{\frac{D_\alpha\big(p^{(1)}\big\|p^{(2)}\big)}{D_\alpha\big(q^{(1)}\big\|q^{(2)}\big)},\frac{D_\alpha\big(p^{(2)}\big\|p^{(1)}\big)}{D_\alpha\big(q^{(2)}\big\|q^{(1)}\big)}\,\middle|\,\alpha\in[1/2,\infty]\right\}.
\end{equation}
We note that in this case the minimization has to be performed on a larger set, where we take into account the different parametrizations in $\alpha$ between this expression and in \eqref{Meq:dominatingcolumndichotomyrate}. In the $d=2$ case of the minimal restrictions semiring discussed in Section \ref{Msec:findings}, we notice similarly that the minimization task is smaller: If $P=\big(p^{(1)},p^{(2)}\big)$ and $Q=\big(q^{(1)},q^{(2)}\big)$ are pairs of probability distributions such that ${\rm supp}\,p^{(1)}\cap{\rm supp}\,p^{(2)}\neq\emptyset\neq{\rm supp}\,q^{(1)}\cap{\rm supp}\,q^{(2)}$, ${\rm supp}\,p^{(1)}\nsubseteq{\rm supp}\,p^{(2)}$, and ${\rm supp}\,p^{(2)}\nsubseteq{\rm supp}\,p^{(1)}$ (i.e.,\ $P$ is a power universal of the minimal restrictions semiring),
\begin{equation}
r(P\to Q)=\min\left\{\frac{D_\alpha\big(p^{(1)}\big\|p^{(2)}\big)}{D_\alpha\big(q^{(1)}\big\|q^{(2)}\big)},\frac{D_\alpha\big(p^{(2)}\big\|p^{(1)}\big)}{D_\alpha\big(q^{(2)}\big\|q^{(1)}\big)}\,\middle|\,\alpha\in[0,1/2]\right\}.
\end{equation}

\subsection{Thermal majorization}\label{Msec:thermal}

Let us consider a finite-dimensional quantum thermodynamic system. We assume that the system's steady state is described by the Gibbs state $\gamma_\beta$ at inverse temperature $\beta>0$, i.e.,
\begin{equation}
\gamma_\beta=\frac{1}{Z}e^{-\beta H}, 
\end{equation}
where $Z=\tr{e^{-\beta H}}$ and $H$ is the Hamiltonian of the system. Let $\{h_i\}_{i=0}^{d-1}$ be the orthonormal basis of the system's Hilbert space that diagonalizes $H$, i.e.,\ $H=\sum_{i=0}^{d-1}E_i|h_i\>\<h_i|$ where $E_i$ are the eigenvalues of $H$. This means
\begin{equation}
\gamma_\beta=\sum_{i=0}^{d-1} \lambda^\beta_i |h_i\>\<h_i|, 
\end{equation}
where $\lambda^\beta_i=e^{-\beta E_i}/Z$. We define the probability vector $\lambda^\beta:=(\lambda^\beta_0,\ldots,\lambda^\beta_{d-1})$. From now on, we assume that $E_i>0$ for $i=0,\ldots,d-1$; we are free to assume this as the zero energy can be shifted if necessary.

We are interested now in the question when a state of this thermodynamic system can be transformed into another with a Gibbs-preserving channel. Recall that a quantum channel (completely positive trace-preserving linear map) $\Phi$ whose input and output systems coincide with this thermodynamic system is Gibbs-preserving if $\Phi(\gamma_\beta)=\gamma_\beta$. When $\Phi(\rho)=\sigma$ where $\Phi$ is a Gibbs-preserving channel, we say that $\rho$ {\it thermally majorizes} $\sigma$ and denote $\rho\rgeq_\beta\sigma$. Thermal majorization of $\rho$ over $\sigma$ in large samples means that there is $n\in\N$ such that $\rho^{\otimes n}\rgeq_\beta\sigma^{\otimes n}$, i.e.,\ there is a channel $\Phi_n$ on $n$ copies of the system such that $\Phi_n(\rho^{\otimes n})=\sigma^{\otimes n}$ and $\Phi_n(\gamma_\beta^{\otimes n})=\gamma_\beta^{\otimes n}$. Catalytic thermal majorization means that there are states $\tau^{(1)},\tau^{(2)}$ on possibly another system (e.g.,\ a heat bath), where we may interpret $\tau^{(2)}$ as the Gibbs state of the other system, and a channel $\Phi$ such that $\Phi(\rho\otimes\tau^{(1)})=\sigma\otimes\tau^{(1)}$ and $\Phi(\gamma_\beta\otimes\tau^{(2)})=\gamma_\beta\otimes\tau^{(2)}$. We are particularly interested in asymptotic catalytic thermal majorization: for any $\varepsilon>0$, there are catalysts $\tau_\varepsilon^{(1)},\tau_\varepsilon^{(2)}$ on some system (which may depend on $\varepsilon$), a state $\sigma_\varepsilon$ on our system of interest such that $\|\sigma-\sigma_\varepsilon\|_1<\varepsilon$, and a channel $\Phi_\varepsilon$ such that $\Phi_\varepsilon(\rho\otimes\tau_\varepsilon^{(1)})=\sigma_\varepsilon\otimes\tau_\varepsilon^{(1)}$ and $\Phi_\varepsilon(\gamma_\beta\otimes\tau_\varepsilon^{(2)})=\gamma_\beta\otimes\tau_\varepsilon^{(2)}$.

We concentrate on the case where the states commute with the Gibbs state. If the energy spectrum is nondegenerate, i.e.,\ there are no repeating eigenvalues $E_i$, then $[\rho,\gamma_\beta]=0$ if and only if $\rho$ diagonalizes in the energy eigenbasis, i.e.,
\begin{equation}\label{Meq:diagonalizable}
\rho=\sum_{i=0}^{d-1}\lambda_i |h_i\>\<h_i|.
\end{equation}
Even if the energy spectrum is degenerate, we will focus on states which are diagonalizable in the energy eigenbasis, i.e.,\ those $\rho$ as in \eqref{Meq:diagonalizable} with $\lambda_i\geq0$, $\lambda_0+\cdots+\lambda_{d-1}=1$. To explicitly refer to the state, we will associate the above $\{h_i\}_i$-diagonal state $\rho$ with the probability vector $\lambda^\rho=(\lambda^\rho_0,\ldots,\lambda^\rho_{d-1})$ where
\begin{equation}
\lambda^\rho_i=\lambda_i=\<h_i|\rho\,h_i\>,\qquad i=0,\ldots,d-1.
\end{equation}


If $\rho$ and $\sigma$ are $\{h_i\}_i$-diagonal, then one can easily see that $\Phi(\rho)=\sigma$ for a channel $\Phi$ if and only if $T\lambda^\rho=\lambda^\sigma$ where $T=(T_{i,j})_{i,j=0}^{d-1}$ is a stochastic matrix given by $T_{i,j}=\<h_i|\Phi(|h_j\>\<h_j|)h_i\>$. This means that, for $\{h_i\}_i$-diagonal states $\rho$ and $\sigma$, we have $\rho\rgeq_\beta\sigma$ if and only if there is a stochastic matrix $T$ such that $T\lambda^\rho=\lambda^\sigma$ and $T\lambda^\beta=\lambda^\beta$, i.e.,
\begin{equation}
\big(\lambda^\rho,\lambda^\beta\big)\rgeq (\lambda^\sigma,\lambda^\beta).
\end{equation}
Thus, we arrive at matrix majorization in the case $d=2$ where the second column dominates the first one since we assume that each energy $E_i$ in the energy spectrum occurs with non-vanishing probability $e^{-\beta E_i}/Z$. Moreover, we have that $\rho$ thermally majorizes $\sigma$ in the asymptotic catalytic regime if and only if, for all $\varepsilon>0$, there are catalytic probability vectors $r_\varepsilon^{(1)},r_\varepsilon^{(2)}$ (of undetermined finite length), $\lambda_\varepsilon\in\mc P_d$ such that $\|\lambda^\sigma-\lambda_\varepsilon\|_1<\varepsilon$ and a stochastic matrix $T_\varepsilon$ such that $T_\varepsilon(\lambda^\rho\otimes r_\varepsilon^{(1)})=\lambda_\varepsilon\otimes r_\varepsilon^{(1)}$ and $T_\varepsilon(\lambda^\beta\otimes r_\varepsilon^{(2)})=\lambda^\beta\otimes r_\varepsilon^{(2)}$. Thus, asymptotic catalytic thermal majorization can be recast as asymptotic catalytic matrix majorization which we have exhaustively characterized in Section \ref{Msec:2column}.

Consider now states $\rho$ and $\sigma$ both diagonalizable in the energy eigenbasis $\{h_i\}_i$. If $\rho$ and $\sigma$ are both of full rank (as is the Gibbs state $\gamma_\beta$), then, according to Corollary 24 of \cite{farooq2024}, $\rho$ asymptotically catalytically thermally majorizes $\sigma$ if and only if $D_\alpha(\lambda^\rho\|\lambda^\beta)\geq D_\alpha(\lambda^\sigma\|\lambda^\beta)$ and $D_\alpha(\lambda^\beta\|\lambda^\rho)\geq D_\alpha(\lambda^\beta\|\lambda^\sigma)$ for all $\alpha\geq 1/2$. Now we can also address the case of general rank. If $\rho$ is not of full rank, then, according to Proposition \ref{Mpro:poweruniversaldcgeneral}, $(\lambda^\rho,\lambda^\beta)$ is a power universal of $S^2_{\rm d.c.}$ and, according to Theorem \ref{Mthm:asymptocatalytic}, $\rho$ asymptotically catalytically thermally majorizes $\sigma$ if and only if only $D_\alpha(\lambda^\rho\|\lambda^\beta)\geq D_\alpha(\lambda^\sigma\|\lambda^\beta)$ for all $\alpha\geq0$. In this case, we have $D_\alpha(\lambda^\beta\|\lambda^\rho)=\infty$ for $\alpha\geq 1$. Because also
\begin{align}
D_\alpha(\lambda^\beta\|\lambda^\rho)&=\frac{\alpha}{1-\alpha}D_{1-\alpha}(\lambda^\rho\|\lambda^\beta)\geq \frac{\alpha}{1-\alpha}D_{1-\alpha}(\lambda^\sigma\|\lambda^\beta)=D_\alpha(\lambda^\beta\|\lambda^\sigma)
\end{align}
for all $\alpha\in (0,1)$, we see that the asymptotic catalytic thermal majorization of the state $\rho$ of not full rank over $\sigma$ is equivalent with exactly the same inequalities as in the full-rank case. The case where $\rho$ asymptotically catalytically thermally majorizes $\sigma$ while $\rho$ is of full support and $\sigma$ is not is impossible. We can see this as follows: Suppose that $\rho$ asymptotically catalytically thermally majorizes $\sigma$. Because $D_\alpha$ are monotones, we now have that $D_\alpha(\lambda^\rho\|\lambda^\beta)\geq D_\alpha(\lambda^\sigma\|\lambda^\beta)$ for all $\alpha\geq0$. As $\alpha\to0$, it follows that
\begin{equation}
\sum_{i\in{\rm supp}\,\lambda^\rho} e^{-\beta E_i}/Z\leq \sum_{i\in{\rm supp}\,\lambda^\sigma} e^{-\beta E_i}/Z.
\end{equation}
Since $\rho$ has full support, the left-hand side equals $1$, forcing $\sigma$ to also have full support.

Note that, whenever $\rho$ diagonalizes in the energy eigenbasis, then $D_\alpha(\lambda^\rho\|\lambda^\beta)=\tilde{D}_\alpha(\rho\|\gamma_\beta)$ where $\tilde{D}_\alpha$ is the `sandwiched' quantum R\'{e}nyi divergence, i.e. $\tilde{D}_\alpha(\rho\|\tau)$ equals 
\begin{equation}\label{Meq:sandwiched}
\left\{\begin{array}{ll}
-\log{\tr{\sigma\,{\rm supp}\,\rho}}&\text{if}\ \alpha=0\ {\rm and}\ {\rm supp}\,\rho\cap{\rm supp}\,\sigma\neq\{0\}\\
\frac{1}{\alpha-1}\log{\tr{\big(\sigma^{\frac{1-\alpha}{2\alpha}}\rho\sigma^{\frac{1-\alpha}{2\alpha}}\big)^\alpha}}&\text{if}\ \alpha\in(0,1) \ {\rm and}\ {\rm supp}\,\rho\cap{\rm supp}\,\sigma\neq\{0\}\\
\tr{\rho(\log{\rho}-\log{\sigma})}&\text{if}\ \alpha=1\ {\rm and}\ {\rm supp}\,\rho\subseteq{\rm supp}\,\sigma\\
\frac{1}{\alpha-1}\log{\tr{\big(\sigma^{\frac{1-\alpha}{2\alpha}}\rho\sigma^{\frac{1-\alpha}{2\alpha}}\big)^\alpha}}&\text{if}\ \alpha\in(1,\infty)\ {\rm and}\ {\rm supp}\,\rho\subseteq{\rm supp}\,\sigma\\
\log{\big\|\sigma^{-1/2}\rho\sigma^{-1/2}\big\|}&\text{if}\ \alpha=\infty\ {\rm and}\ {\rm supp}\,\rho\subseteq{\rm supp}\,\sigma\\
\infty&{\rm otherwise}
\end{array}\right.,
\end{equation}
where $\|\cdot\|$ is the operator norm and we use ${\rm supp}\,\rho$ for both the supporting subspace of $\rho$ and for the orthogonal projection onto this subspace. Let us define the free energies $F^\pm_{\alpha,\beta}$ through
\begin{equation}
F^+_{\alpha,\beta}(\rho):=\tilde{D}_\alpha(\rho\|\gamma_\beta),\quad F^-_{\alpha,\beta}(\rho):=\tilde{D}_\alpha(\gamma_\beta\|\rho)
\end{equation}
for all $\alpha\geq1/2$. Thus, the above considerations yield the following result.
\begin{corollary}\label{Mcor:thermal}
Suppose that the eigen-energies $E_i>0$ for all $i$, and $\rho$ and $\sigma$ diagonalize in the energy eigenbasis. Then $\rho$ asymptotically catalytically thermally majorizes $\sigma$ if and only if
\begin{equation}
F^\pm_{\alpha,\beta}(\rho)\geq F^\pm_{\alpha,\beta}(\sigma),\qquad\forall\alpha\geq\frac{1}{2}.
\end{equation}
\end{corollary}
Note also that, when $\rho$ and $\sigma$ are as in Corollary \ref{Mcor:thermal} above, then the optimal rate at which i.i.d.\ copies of $\rho$ can be transformed into $\sigma$ with Gibbs-preserving channels is
\begin{equation}
r\left(\rho\overset{\beta}{\to}\sigma\right)=\min\left\{\frac{F^+_{\alpha,\beta}(\rho)}{F^+_{\alpha,\beta}(\sigma)},\frac{F^-_{\alpha,\beta}(\rho)}{F^-_{\alpha,\beta}(\sigma)}\,\middle|\,\alpha\in[1/2,\infty]\right\},
\end{equation}
a direct consequence of the results of Section \ref{Msec:rates}. As mentioned in the Introduction, although this result was previously claimed in \cite{Brandao2014}, our present proof avoids some technical issues in the original proof. We would also like to highlight the fact that, instead of using the result on vector majorization proven in \cite{klimesh2007inequalities}, our semiring methods allow us to directly study the problem from the matrix majorization view without having to embed the problem in the vector majorization form as was done in \cite{Brandao2014}. 

Our results also allow us to study multiple thermal majorization: for $d$-tuples $\big(\rho^{(1)},\ldots,\rho^{(d)}\big)$ and $\big(\sigma^{(1)},\ldots,\sigma^{(d)}\big)$ of states, we denote
\begin{equation}
\big(\rho^{(1)},\ldots,\rho^{(d)}\big)\rgeq_\beta \big(\sigma^{(1)},\ldots,\sigma^{(d)}\big)
\end{equation}
if there is a Gibbs-preserving channel $\Phi$ (i.e.,\ $\Phi(\gamma_\beta)=\gamma_\beta$) such that $\Phi\big(\rho^{(k)}\big)=\sigma^{(k)}$ for $k=1,\ldots,d$. This majorization can naturally also be studied in the large-sample, catalytic, and asymptotic regimes, the latter one arguably being the most physically relevant case. If we assume that all the states $\rho^{(k)}$ and $\sigma^{(k)}$ diagonalize in the energy eigenbasis, we are naturally again back to matrix majorization with the dominating column $\lambda^\beta$. We point out though that, in order to derive necessary and sufficient conditions for asymptotic thermal majorization in this case, $\big(\lambda^{\rho^{(1)}},\ldots,\lambda^{\rho^{(d)}},\lambda^\beta\big)$ has to be a power universal of the semiring $S^{d+1}_{\rm d.c.}$. Conditions for this are quite strict, as we have seen in Proposition \ref{Mpro:poweruniversaldcgeneral}. Thus, there remains plenty of unanswered questions in multiple thermal majorization.

\section{More General Support Conditions and Power Universality}\label{Msec:pusupports}

In this section, we will briefly touch upon matrix majorization semirings of matrices that have different conditions on their columns' supports than those considered in this thesis and in \cite{farooq2024}. We show how one can define such semirings in general, and conjecture what form the power universals take and what the corresponding monotone homomorphisms and derivations are. 

For any matrix $P\in\V^d$ and a pair of columns $k,k'\in[d]$ we may impose one of the following conditions on the supports of these columns in $P$: 
\begin{itemize}
    \item[$\subeq$:] ${\rm supp}\,p^{(k)}\subseteq{\rm supp}\,p^{(k^\prime)}$; 
    \item[$\supeq$:] ${\rm supp}\,p^{(k)}\supseteq{\rm supp}\,p^{(k^\prime)}$; 
    \item[$\eq$:] ${\rm supp}\,p^{(k)}={\rm supp}\,p^{(k^\prime)}$; 
    \item[$\no$:] no condition on the supports. 
\end{itemize}
Define a \emph{collection of support conditions} $\Xi$ to be a choice $\Xi(k,k')\in\{\subeq,\supeq,\eq,\no\}$ for every pair of columns $k,k'\in[d]$, $k\neq k'$, of one of the above conditions. Naturally, we require that 
\begin{align}
    &\Xi(k_1,k_2)=\subeq\,\Leftrightarrow\,\Xi(k_2,k_1)=\supeq,&\\
    &\Xi(k_1,k_2)=\eq\,\Leftrightarrow\,\Xi(k_2,k_1)=\eq,\ \Xi(k_1,k_2)=\no\,\Leftrightarrow\,\Xi(k_2,k_1)=\no,&
\end{align}
for all distinct $k_1,k_2\in[d]$. Note that if ${\rm supp}\,p^{(k_1)}\subseteq{\rm supp}\,p^{(k_2)}$ and ${\rm supp}\,p^{(k_2)}\subseteq{\rm supp}\,p^{(k_3)}$ for $k_1,k_2,k_3\in[d]$, then we must also have ${\rm supp}\,p^{(k_1)}\subseteq{\rm supp}\,p^{(k_3)}$. Consequently, we also require that the conditions are compatible in this sense, i.e., for all distinct $k_1,k_2,k_3\in[d]$, if 
\begin{equation}
    \big(\Xi(k_1,k_2),\Xi(k_2,k_3)\big)\in\{(\subeq,\subeq),(\subeq,\eq),(\eq,\subeq)\}, 
\end{equation}
then $\Xi(k_1,k_3)=\subeq$, and the same with $\subeq$ replaced with $\supeq$, and also if $\Xi(k_1,k_2)=\eq=\Xi(k_2,k_3)$, then $\Xi(k_1,k_3)=\eq$. 

Denote by $\V^d_\Xi$ all matrices in $\V^d$ satisfying the support conditions set by $\Xi$, i.e., 
\begin{equation}
    \V^d_\Xi:=\big\{\,P\in \V^d\,\big|\,P\text{ satisfies condition }\Xi(k,k')\ \forall\,k,k'\in[d],k\neq k'\,\big\}. 
\end{equation}
One easily sees that $\V^d_\Xi$ is closed under the operations $\boxplus$ and $\boxtimes$. Also, it is closed under the application of any stochastic matrix. Indeed, let $P\in \V^d_\Xi$, $T$ a stochastic matrix, and set $Q:=TP$. Assume ${\rm supp}\,p^{(k)}\subseteq{\rm supp}\,p^{(k^\prime)}$ for a pair of columns $k,k'$. We show that also ${\rm supp}\,q^{(k)}\subseteq{\rm supp}\,q^{(k^\prime)}$, from which our claim follows. Let $i\in{\rm supp}\,q^{(k)}$. Then 
\begin{equation}
    q^{(k)}_i=\sum_jT_{i,j}p^{(k)}_j>0, 
\end{equation}
hence there exists $j$ such that $T_{i,j},p^{(k)}_j>0$. Then also $p^{(k')}_j>0$, whence $q^{(k')}_i>0$. Therefore, $i\in{\rm supp}\,q^{(k')}$. We then define the semiring $S^d_\Xi:=\V^d_\Xi/\!\approx$ similarly as before. 

If we make the choice $\Xi(k,k')=\eq$ for all distinct $k,k'\in[d]$, then the resulting semiring $S^d_\Xi$ will contain the smallest possible number of matrices. This semiring is the common support matrix majorization semiring from \cite{farooq2024}. The semiring obtained from the support conditions $\Xi(k,k')=\no$ for all distinct $k,k'\in[d]$, will contain the largest possible number of matrices, and is the minimal restrictions matrix majorization semiring from Section \ref{Msec:findings}. The semiring obtained from choosing $\Xi(k,k')=\subeq$ and $\Xi(k',k)=\supeq$ when $k\in[d-1]$, $k'=d$, and $\Xi(k,k')=\no$ otherwise, is the dominating column matrix majorization semiring from Section \ref{Msec:Asymmetric}. 

Extrapolating the conditions we found for a matrix to be power universal in Lemma 12 in \cite{farooq2024}, Proposition \ref{Mlem:poweruniversalconditions}, and Proposition \ref{Mpro:poweruniversaldcgeneral} (the case of common supports, minimal restrictions, and one dominating column, respectively), we conjecture that they generalize to the semirings $S^d_\Xi$ for any collection of support conditions $\Xi$, as follows. 

\begin{conjecture}\label{con:puconditiongeneral}
Let $\Xi$ be a collection of support conditions and $U\in S^d_\Xi$ with normalized columns. Then $U$ is power universal in $S^d_\Xi$ if and only if it satisfies the following for every pair of columns $k,k'\in[d]$, $k\neq k'$, according to the condition $\Xi(k,k')$: 

\vspace{1mm}
\begin{center}
    \begin{tabular}{|c|c|}
        \hline
        Condition $\Xi(k,k')$&Condition on $U\in S_\Xi^d$\\
        \hline
        \textup{$\subeq$}&${\rm supp}\,u^{(k)}\subsetneq{\rm supp}\,u^{(k^\prime)}$\\
        \textup{$\supeq$}&${\rm supp}\,u^{(k)}\supsetneq{\rm supp}\,u^{(k^\prime)}$\\
        \textup{$\eq$}&${\rm supp}\,u^{(k)}={\rm supp}\,u^{(k^\prime)}$ and $u^{(k)}\neq u^{(k^\prime)}$\\
        \textup{$\no$}&${\rm supp}\,u^{(k)}\nsubseteq{\rm supp}\,u^{(k^\prime)}$ and ${\rm supp}\,u^{(k^\prime)}\nsubseteq{\rm supp}\,u^{(k)}$\\
        \hline
    \end{tabular}
\end{center}
\end{conjecture}

Furthermore, we may reasonably conjecture that the homomorphisms and derivations corresponding to $S_\Xi^d$ are similar to the ones found for the matrix majorization semirings studied in this thesis. We expect that the homomorphisms are precisely those for parameters $\underline{\alpha}$ or $\underline{\beta}$ that are well-defined, i.e. those where there does not appear a factor $\left(p^{(k)}_i\right)^\alpha$ or $\left(p^{(k)}_i\right)^\beta$ with $\alpha,\beta<0$ where $p^{(k)}_i$ vanishes for at least one non-zero matrix $P$ in $S_\Xi^d$. Also, we expect that the derivations will be precisely those that can be reached as point-wise limits of the divergences corresponding to these homomorphisms. 

Consider a situation in which we have two matrices $P,Q\in \V^d$ and we want to know whether $P$ majorizes $Q$ in the large-sample or catalytic setting. Since $P$ should be power universal in order to apply our conjectured results on majorization, we should choose the support conditions $\Xi$ such that $P$ is power universal in $S_\Xi^d$. Consulting Conjecture \ref{con:puconditiongeneral}, if $P$ has no repeated columns, this is always possible for a unique choice of $\Xi$, as we will now argue. Note that we may assume that $P$ does not have two columns that are equal: if the same two columns are equal in $Q$ as well, then we can delete duplicate columns in $P$ and $Q$ and consider matrices with less columns $d$, and if the same two columns are not equal in $Q$, then majorization is clearly impossible. Therefore, any two columns $k,k'\in[d]$ of $P$ satisfy exactly one of the conditions on the right of the table in Conjecture \ref{con:puconditiongeneral}. Choosing the corresponding condition $\Xi(k,k')$ on the left of the table then uniquely defines $\Xi$ such that $P$ is power universal in $S_\Xi^d$. Since $S_\Xi^d$ is closed under the application of stochastic maps, we may also assume that $Q\in S_\Xi^d$: otherwise there is no possibility of $P$ majorizing $Q$ in the large-sample or catalytic setting. The homomorphisms and derivations corresponding to $S_\Xi^d$ (as conjectured in the previous paragraph) now give sufficient and (almost) necessary conditions for large-sample and catalytic majorization. 

\chapter{Pure-Component Classical-Quantum State Pairs}\label{cha:flat}

This chapter is based on the author's work \cite{verhagen2025alphaz}. We continue our study of monotones characterizing large-sample and catalytic majorization of tuples of states, and take a first step into the realm of quantum states. Instead of (classical) matrix majorization, we will consider majorization of pairs of so-called flat quantum states and certain generalizations of them. Our results lead to the first operational interpretation of the $\alpha$-$z$ relative entropies, which turn out to be the relevant monotones characterizing large-sample and catalytic majorization of these pairs of states. 

\section{Introduction}

As explained in Chapter \ref{cha:introduction}, studying majorization of tuples, or even pairs, of general quantum states in the setting of the preordered semiring techniques that we have been applying, is not straightforward. Studying pairs of cq-states with pure components, however, does turn out to be a tractable endeavor. Recall that $(\rho,\sigma)$ is a pair of pure-component cq-states if there exists a classical register given by fixed orthonormal flags $\ket{i}$, for $i=1,\ldots,n$, $n\geq1$, such that 
\begin{equation}\label{eq:pairform}
\rho=\sum_{i=1}^n p_i \ketbrar{i}\otimes\ketbrar{\alpha_i},\quad\sigma=\sum_{i=1}^n q_i \ketbrar{i}\otimes\ketbrar{\beta_i}, 
\end{equation}
where $(p_1,\ldots,p_n)$ is a probability distribution and $\ket{\alpha_i}$ is a normalized vector when $p_i>0$ and the zero vector when $p_i=0$, for $i=1,\ldots,n$, and similarly for the $q_i$ and $\ket{\beta_i}$. We say that $\rho$ and $\sigma$ of \eqref{eq:pairform} {\it have some overlap} if there is $i$ such that $\langle\alpha_i|\beta_i\rangle\neq 0$. We also say that $\rho$ and $\sigma$ are {\it non-parallel} if $|\braket{\alpha_i}{\beta_i}|<1$ for all $i$. For technical reasons, we will exclusively consider pairs of cq-states with pure components as in \eqref{eq:pairform} that have some overlap, and denote this collection of pairs by $\mc F$. 

In this setting, the set of relative entropies (corresponding to the monotone homomorphisms and derivations) giving sufficient and (almost) necessary conditions for exact or asymptotic, large-sample or catalytic majorization, are within the set of the $\alpha$-$z$ relative entropies $D_{\alpha,z}$. The $\alpha$-$z$ relative entropies for any pair of quantum states $(\rho,\sigma)$, introduced in \cite{Jaksic2012} and \cite{Audenaert_Datta_2015} (see also a new generalization in \cite{Hiai_Jencova_2024}), are defined as 
\begin{equation}\label{eq:alphazdivergence}
    D_{\alpha,z}(\rho\|\sigma):=\frac{1}{\alpha-1}\log\tr\left[\left(\sigma^\frac{1-\alpha}{2z}\rho^\frac{\alpha}{z}\sigma^\frac{1-\alpha}{2z}\right)^z\right], 
\end{equation}
as well as certain limits $\alpha\rightarrow1$, $z\rightarrow\infty$, of this expression, and with some restrictions on the supports of $\rho$, $\sigma$. The question of what the precise range of $(\alpha,z)$ is where the $D_{\alpha,z}$ satisfy the data-processing inequality, or equivalently where monotonicity of the corresponding homomorphisms and derivations holds, has been investigated in a number of papers \cite{Audenaert_Datta_2015,Hiai2013,Carlen_et_al_2018} and conclusively answered in \cite{Zhang2020}. 

The relative entropies giving conditions for exact or asymptotic, large-sample or catalytic majorization of pairs of states as in \eqref{eq:pairform} are essentially those $D_{\alpha,z}$ with parameters $\alpha$ and $z$ such that 
\begin{equation}\label{eq:parameterset}
    \alpha\in(0,1)\text{ and }z\geq\max\{\alpha,1-\alpha\}, 
\end{equation}
together with certain extensions to $\alpha=0,1$ and $z\rightarrow\infty$. Note that the $D_{\alpha,z}$ with $\alpha>1$ do not play a role here because they would diverge on the pairs of states we consider, since the component states are pure. To our knowledge, this is the first time an operational interpretation of the $\alpha$-$z$ relative entropies, where $\alpha$ and $z$ are truly independent of each other, has been demonstrated. Let us note however that the special cases of the so-called `sandwiched' quantum relative entropies \cite{Muller-Lennert_et_al_2013,Wilde2013StrongCF}, where $\alpha=z$, and the Petz-type relative entropies \cite{Petz_85,Petz_86}, where $z=1$, have their own fields of applications and operational interpretations.

Before presenting our results formally, we make some additional definitions. For a pair of states $(\rho,\sigma)\in\mc F$, and $\alpha\in[0,1]$, $z\geq\max\{\alpha,1-\alpha\}$, define     
\begin{equation}\label{eq:alphazexplicit}
    \hat{D}_{\alpha,z}(\rho\|\sigma):=-\frac{1}{z+1}\log\left(\sum_{i=1}^n(p_i)^\alpha(q_i)^{1-\alpha}|\braket{\alpha_i}{\beta_i}|^{2z}\right), 
\end{equation}
and 
\begin{equation}\label{eq:DT}
    \hat{D}^\T(\rho\|\sigma):=-\log\left(\max_{i=1,\ldots,n}|\braket{\alpha_i}{\beta_i}|^{2}\right). 
\end{equation}
In \eqref{eq:alphazexplicit} for $\alpha=0$ or $1$, we adopt the convention that $0^0:=0$. When $\alpha\in(0,1)$, the $\hat{D}_{\alpha,z}$ coincide with the $\alpha$-$z$ relative entropies $D_{\alpha,z}$ as in \eqref{eq:alphazdivergence}, when evaluated on pairs of states in $\mc F$, except for a different prefactor in front of the $\log$ for technical reasons. Since the $D_{\alpha,z}$ satisfy DPI for $\alpha\in(0,1)$, $z\geq\max\{\alpha,1-\alpha\}$, the same holds for $\hat{D}_{\alpha,z}$, where  we use the continuity of $\hat{D}_{\alpha,z}$ in $\alpha$ to establish the DPI when $\alpha=0,1$. $\hat{D}^\T$ is the point-wise limit $z\rightarrow\infty$ of $\hat{D}_{\alpha,z}$, for any $\alpha\in[0,1]$ as one may rather easily verify. As a consequence, also $\hat{D}^\T$ satisfies DPI. See Figure \ref{fig:alphazregion} for a depiction of the family of all these relative entropies. 

\begin{figure}
\begin{overpic}[scale=0.4,unit=1mm]{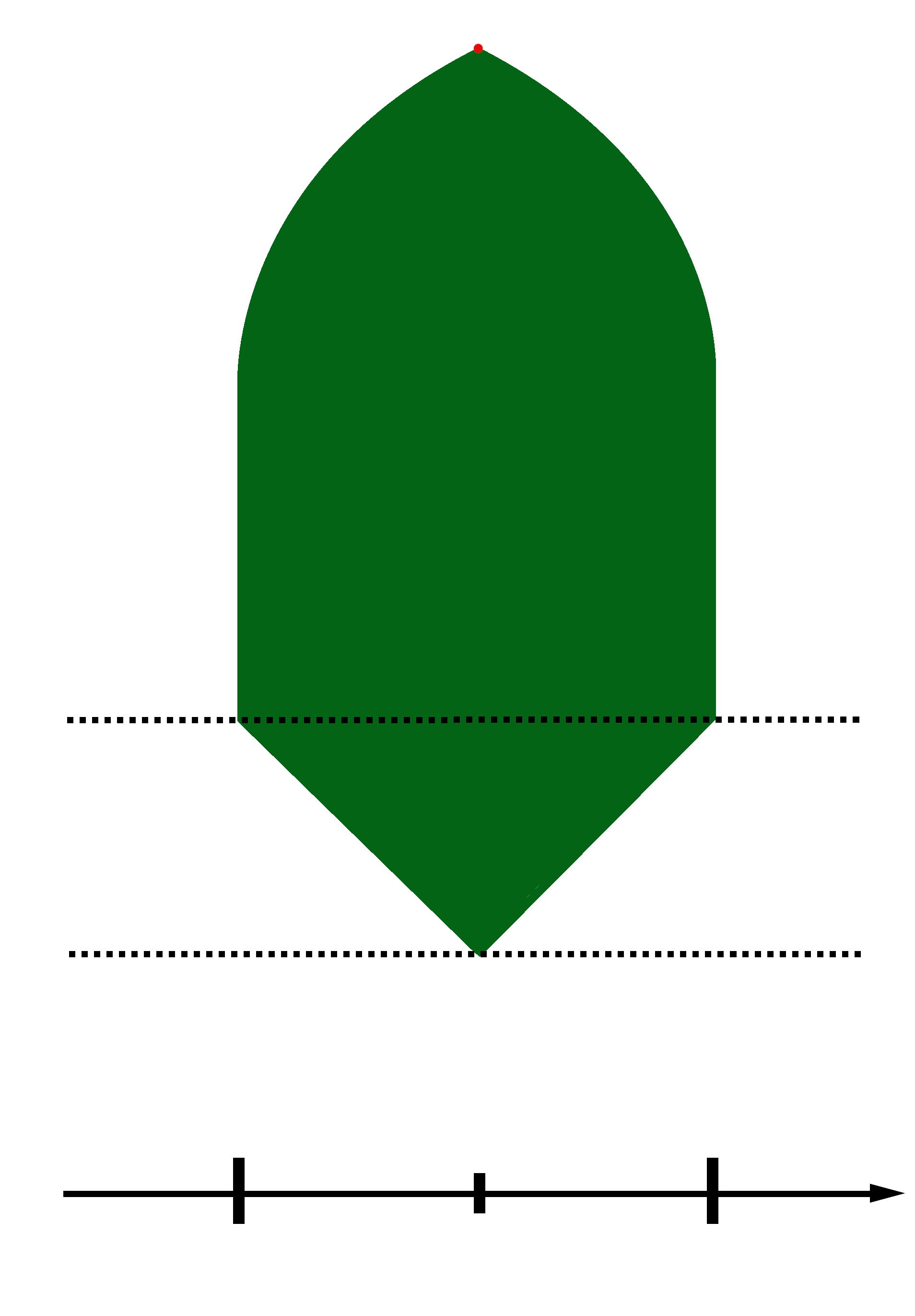}
\put(0,43){$z=1$}
\put(0,27){$z=1/2$}
\put(0,10){$z=0$}
\put(11,2){$\alpha=0$}
\put(28,2){$\alpha=1/2$}
\put(45,2){$\alpha=1$}
\put(22,57){\begin{Huge}
$\textcolor{white}{\hat{D}_{\alpha,z}}$
\end{Huge}}
\put(28.5,73){\begin{large}
\rotatebox{90}{$\textcolor{white}{\xrightarrow{z\rightarrow\infty}}$}
\end{large}}
\put(30,91){\begin{large}
\punainen{$\hat{D}^\T$}
\end{large}}
\end{overpic}
\caption{\label{fig:alphazregion} A depiction of the family of all relative entropies given by \eqref{eq:alphazexplicit} and \eqref{eq:DT}. For $\alpha\in(0,1)$ and $z\geq\max\{\alpha,1-\alpha\}$, $\hat{D}_{\alpha,z}$ equals the $\alpha$-$z$ relative entropies given by \eqref{eq:alphazdivergence}, up to a prefactor. With this different prefactor, the $\alpha$-$z$ divergences can be extended easily to the segments $\alpha=0$ or $\alpha=1$, and $z\geq1$. The additional relative entropy $\hat{D}^\T=\lim_{z\rightarrow\infty}\hat{D}_{\alpha,z}$, independent of $\alpha$, turns the set of all relative entropies into a compact space in the topology of point-wise convergence.}
\end{figure}
    
A result by Uhlmann \cite{Uhlmann1985} states that a pair of pure states $(\ketbrar{\alpha},\ketbrar{\beta})$ majorizes another pair of pure states $(\ketbrar{\alpha'},\ketbrar{\beta'})$ if and only if $|\braket{\alpha}{\beta}|\leq|\braket{\alpha'}{\beta'}|$. Note that in the expressions \eqref{eq:alphazexplicit} and \eqref{eq:DT}, the fidelity between pure states $|\braket{\alpha_i}{\beta_i}|$ also makes an appearance. 

Recall that a pair $(\rho,\sigma)\in\mc F$ is called non-parallel if 
\begin{equation}\label{eq:pu1}
    |\braket{\alpha_i}{\beta_i}|<1\quad\forall\;i\in\{1,\ldots,n\}. 
\end{equation}
We introduce the following additional property for a pair $(\rho,\sigma)\in\mc F$: 
\begin{equation}\label{eq:pu2}
    \exists\;i\in\{1,\ldots,n\}\ \text{s.t.}\ \braket{\alpha_i}{\beta_i}=0\ \text{and}\ \ket{\alpha_i},\ket{\beta_i}\neq0.  
\end{equation}

The following theorem gives sufficient and almost necessary conditions for \emph{exact} large-sample and catalytic majorization. This result is established in Section \ref{sec:exactproof}.  

\begin{theorem}\label{thm:LS}
    Let $(\rho,\sigma)$, $(\rho',\sigma')$ be two pairs of states in $\mc F$, where additionally $(\rho,\sigma)$ satisfies \eqref{eq:pu2}. If 
    \begin{align}
        \label{eq:LScondition1}&\hat{D}_{\alpha,z}(\rho\|\sigma)>\hat{D}_{\alpha,z}(\rho'\|\sigma')\ \text{when}\ \alpha\in[0,1],\ z\geq\max\{\alpha,1-\alpha\}, &\\
        \label{eq:LScondition2}&\hat{D}^\mathbb{T}(\rho\|\sigma)>\hat{D}^\mathbb{T}(\rho'\|\sigma'), 
    \end{align}            
    then $(\rho,\sigma)$ majorizes $(\rho',\sigma')$ both in the large-sample setting and in the catalytic setting, with a pair of catalyst states in $\mc F$. For the catalytic result, $(\rho,\sigma)$ does not need to satisfy \eqref{eq:pu2}. 
    
    Conversely, if $(\rho,\sigma)$ majorizes $(\rho',\sigma')$ in large samples or catalytically, then the above inequalities \eqref{eq:LScondition1} and \eqref{eq:LScondition2} hold non-strictly. 
\end{theorem}

As a consequence of this theorem, the following result gives necessary and sufficient conditions for \emph{asymptotic} large-sample and catalytic majorization. The proof can be found in Section \ref{sec:approximate}. 

\begin{theorem}\label{thm:approximateLS}
    Let $(\rho,\sigma)$, $(\rho',\sigma')$ be two pairs of states in $\mc F$, where $\rho$, $\sigma$ are non-parallel (i.e. satisfy \eqref{eq:pu1}) and $\rho'$, $\sigma'$ do not commute. Consider the following statements: 
    \begin{enumerate}[(i)]
        \item $D_{\alpha,z}(\rho\|\sigma)\geq D_{\alpha,z}(\rho'\|\sigma')$ when $\alpha\in(0,1)$ and $z>\max\{\alpha,1-\alpha\}$. 
        \item For each $\varepsilon>0$, there exist a state $\rho'_\varepsilon$ such that $F(\rho'_\varepsilon,\rho')\geq1-\varepsilon$, and a quantum channel $\mc{C}_{\varepsilon,n}$ for $n\geq n_\varepsilon$ large enough such that $\mc{C}_{\varepsilon,n}(\rho^{\otimes n})=(\rho'_\varepsilon)^{\otimes n}$ and $\mc{C}_{\varepsilon,n}(\sigma^{\otimes n})=(\sigma')^{\otimes n}$. 
        \item For each $\varepsilon>0$, there exist a state $\rho'_\varepsilon$ such that $F(\rho'_\varepsilon,\rho')\geq1-\varepsilon$, a quantum channel $\mc C_\varepsilon$, and a catalyzer pair $(\tau_\varepsilon$, $\omega_\varepsilon)\in\mc F$ such that 
        \begin{equation}
            \mc C_\varepsilon(\rho\otimes\tau_\varepsilon)=\rho'_\varepsilon\otimes\tau_\varepsilon\quad\text{and}\quad \mc C_\varepsilon(\sigma\otimes\omega_\varepsilon)=\sigma'\otimes\omega_\varepsilon. 
        \end{equation}
        Statements (i) and (iii) are equivalent. If $(\rho,\sigma)$ furthermore satisfies \eqref{eq:pu2}, then (i), (ii) and (iii) are equivalent. 
    \end{enumerate}        
\end{theorem}

Note that the state $\rho'_\varepsilon$ in statements (ii) and (iii) above is allowed to come from a Hilbert space that is larger than the one associated to $\rho'$. Additionally, $(\rho'_\varepsilon,\sigma')$ need not be a pair in $\mc F$. However, from the construction of $\rho'_\varepsilon$ in the proof of the implications (i) $\Longrightarrow$ (ii) and (i) $\Longrightarrow$ (iii), it follows that $\rho'_\varepsilon$ can be chosen such that $(\rho'_\varepsilon,\sigma')\in\mc F$, and $\rho'_\varepsilon$ and $\rho'$ have the same classical probability distribution associated to them, but with rotated pure-state components. 

It might be possible to relax the assumption \eqref{eq:pu2} in Theorems \ref{thm:LS} and \ref{thm:approximateLS}, but this requires the study of a slightly different semiring. We will discuss this issue later. Note that, for the results on large-sample majorization, Theorem \ref{thm:approximateLS} requires $(\rho,\sigma)$ to satisfy both \eqref{eq:pu1} and \eqref{eq:pu2}, while Theorem \ref{thm:LS} only explicitly requires \eqref{eq:pu2}. However, \eqref{eq:pu1} is automatically satisfied for $(\rho,\sigma)$ in Theorem \ref{thm:LS} since the strict inequality \eqref{eq:LScondition2} implies \eqref{eq:pu1}. 

One may wonder whether the set of relative entropies that give sufficient conditions for large-sample and catalytic majorization is bigger than necessary. The following theorem, similar to Theorem \ref{Mthm:minimal} for the minimal restrictions matrix majorization semiring, tells us that this is in fact not the case in the context of asymptotic majorization, since the set of relative entropies is minimal in the sense that we cannot remove any open subset from the range of parameters 
\begin{equation}\label{eq:range}
    R:=\big\{\,(\alpha,z)\in\R^2\,|\,\alpha\in(0,1),\ z>\max\{\alpha,1-\alpha\}\,\big\}
\end{equation}
appearing in condition (i) of Theorem \ref{thm:approximateLS}. The proof is given in Section \ref{sec:minimal}. 
\begin{theorem}\label{thm:minimal}
     Let $O$ be a non-empty open subset of $R$. There exist pairs $(\rho,\sigma)$, $(\rho',\sigma')$ in $\mc F$, where $\rho$, $\sigma$ are non-parallel and satisfy \eqref{eq:pu2}, and $\rho'$, $\sigma'$ do not commute, such that 
     \begin{equation}\label{eq:smallercondition}
         D_{\alpha,z}(\rho\|\sigma)\geq D_{\alpha,z}(\rho'\|\sigma')\ \text{when}\ (\alpha,z)\in R\setminus O, 
     \end{equation}
     but both (ii) and (iii) of Theorem \ref{thm:approximateLS} do not hold. 
\end{theorem}

The optimal conversion rate $r\big((\rho,\sigma)\to(\rho',\sigma')\big)$ between two pairs in $\mc F$ is given by the following result, which is a direct consequence of Proposition \ref{prop:optimalconversionrate}. 

\begin{theorem}\label{thm:rates}
    For pairs $(\rho,\sigma)$ and $(\rho',\sigma')$ of states in $\mc F$, where $\rho$, $\sigma$ are non-parallel and satisfy \eqref{eq:pu2}, the optimal conversion rate from $(\rho, \sigma)$ to $(\rho', \sigma')$ is 
    \begin{align}
    &r\big((\rho,\sigma)\to(\rho',\sigma')\big)=\inf\left\{\,\frac{D_{\alpha,z}(\rho\|\sigma)}{D_{\alpha,z}(\rho'\|\sigma')}\,\middle|\,\begin{array}{l}\alpha\in(0,1),\\z>\max\{\alpha,1-\alpha\}\end{array}\,\right\}&\label{eq:rate1}\\
    &\quad=\min\left\{\,\frac{\hat{D}_{\alpha,z}(\rho\|\sigma)}{\hat{D}_{\alpha,z}(\rho'\|\sigma')}\,\middle|\,\begin{array}{l}\alpha\in[0,1],\\z\geq\max\{\alpha,1-\alpha\}\end{array}\,\right\}\cup\left\{\frac{\hat{D}^\T(\rho\|\sigma)}{\hat{D}^\T(\rho'\|\sigma')}\right\}.&\label{eq:rate2} 
    \end{align}
\end{theorem}

The additional relative entropies $\hat{D}_{\alpha,z}$ for $\alpha=0,1$ and the limit point $\hat{D}^\T$ make the set of relative entropies appearing in \eqref{eq:rate2} compact in the topology of pointwise convergence \cite[Proposition 8.5]{fritz2022}. This is why in \eqref{eq:rate2} we are able to take the minimum, instead of taking the infimum when using the original definition of the $\alpha$-$z$ relative entropies $D_{\alpha,z}$ in \eqref{eq:rate1}. 

The framework of majorization we study in this work generalizes the classical notion of the comparison of statistical experiments in the context of matrix majorization for $d=2$ in the previous chapter, characterized by the existence of a stochastic map between pairs of finite probability distributions. This classical majorization can be equivalently characterized by the majorization of Lorenz curves \cite{Blackwell53}. The latter notion was generalized to the quantum case using quantum relative Lorenz curves in \cite{Buscemi2017}, where it was called \emph{quantum relative majorization}. Our notion of majorization characterized by the existence of a quantum channel between pairs of quantum states implies quantum relative majorization. In the case of pairs of qubit states the two notions coincide \cite{Alberti1980}, but in general, they are not equivalent, unlike in the classical case \cite{matsumoto2014}. 

This chapter is organized as follows: In Section \ref{sec:definingthesemirings} we formally define the semirings containing pairs as in \eqref{eq:pairform}. Then we characterize power universal elements in Section \ref{sec:pu} and identify the relevant monotone homomorphisms and derivations in Section \ref{sec:homomorphismsderivations}. In Section \ref{sec:results}, we use these technical results to derive our main theorems on large-sample or catalytic, and exact or asymptotic, majorization, and the minimality of the range of parameters $(\alpha,z)$.

\section{Defining the Semirings}\label{sec:definingthesemirings}
We consider the set $V$ consisting of all pairs of positive semi-definite matrices that are mutually block-diagonalizable into blocks of rank at most $1$. More precisely, $(A,B)\in V$ when $A,B$ are positive semi-definite operators on a finite-dimensional Hilbert space
\begin{equation}
    \Hil=\bigoplus_{i=1}^{n}\Hil_i, 
\end{equation}
where $\Hil_1,\ldots,\Hil_n$, are $n\geq1$ Hilbert spaces (these spaces may all vary and also their number $n$ may vary), such that $A,B$ admit a decomposition  
\begin{equation}\label{eq:Vpair}
    A=\bigoplus_{i=1}^{n}\ketbra{\alpha_i}{\alpha_i},\quad B=\bigoplus_{i=1}^{n}\ketbra{\beta_i}{\beta_i}, 
\end{equation}
where $\ket{\alpha_i},\ket{\beta_i}\in\Hil_i$ are vectors (possibly equal to the zero-vector). 

We define the set of \emph{minimal restrictions} pairs $V_{\rm m.r.}\subsetneq V$ to be those $(A,B)\in V$ such that either $A$ and $B$ are both the $0$-operator in their associated Hilbert space, or $A$ and $B$ have some overlap in the sense that 
\begin{equation}\label{eq:cscondition}
    \braket{\alpha_i}{\beta_i}\neq0\text{ for some }i\in[n]. 
\end{equation}
Also, we define the set of \emph{everywhere overlapping} pairs $V_{\rm e.o.}\subsetneq V_{\rm m.r.}$ to be those $(A,B)\in V_{\rm m.r.}$ such that, 
for any $i\in[n]$, we have either 
\begin{equation}
    \braket{\alpha_i}{\beta_i}\neq0 
\end{equation}
or 
\begin{equation}
    \ket{\alpha_i}=0=\ket{\beta_i}. 
\end{equation}
Note that any pair of flat states (in particular, pair of pure states) that are at least partially overlapping is in $V$. 

$V$ is closed under the $\boxplus$ and $\boxtimes$ operations, defined by 
\begin{equation}
    (A,B)\boxplus(C,D):=(A\oplus C,B\oplus D),\quad(A,B)\boxtimes(C,D):=(A\otimes C,B\otimes D), 
\end{equation}
as one easily sees. We say that $(A,B)\in V$ with associated Hilbert space $\Hil$ majorizes $(C,D)\in V$ with associated Hilbert space $\Hil'$ and write $(A,B)\succeq(C,D)$ if and only if there exists a quantum channel $\Cha:\Lin(\Hil)\rightarrow\Lin(\Hil')$ such that $\Cha(A)=C$ and $\Cha(B)=D$. 

We define the equivalence relation $\approx$ as follows: Let $(A,B),(A',B')\in V$ with associated Hilbert spaces $\Hil$, respectively $\Hil'$, and write $(A,B)\approx(A',B')$ if and only if they are equal up to a permutation of the component Hilbert spaces $\Hil_i$, adding/removing copies of $(0,0)$, or by transformations of the form 
\begin{equation}
    V_i\ketbrar{\alpha_i}V_i^*=W_i\ketbrar{\alpha'_i}W_i^*,\quad V_i\ketbrar{\beta_i}V_i^*=W_i\ketbrar{\beta'_i}W_i^*, 
\end{equation}
where the $V_i,W_i$ are isometries from $\Hil_i,\Hil'_i$ into another Hilbert space. 
For $(A,B)\in V$, we denote by $[(A,B)]$ the $\approx$-equivalence class of $(A,B)$. Define $S_{\rm m.r.}:=V_{\rm m.r.}/\!\approx$. The operations $\boxplus$ and $\boxtimes$ induce well-defined operations on $S_{\rm m.r.}$. These operations are now also commutative: for $\boxplus$, this follows from the freedom to permute the Hilbert spaces $\Hil_i$, and for $\boxtimes$, this can be seen by considering unitaries that swap two vectors in a tensor product. Also, the preorder $\succeq$ induces a preorder on $S_{\rm m.r.}$. In conclusion, $(S_{\rm m.r.},\boxplus,0:=[(0,0)],\boxtimes,1:=[(1,1)],\succeq)$ is a preordered semiring. We may also define the preordered semiring based on $S_{\rm e.o.}=V_{\rm e.o.}/\!\approx$ in the same way. In the sequel, we will often treat elements of $V_{\rm m.r.}$ or $V_{\rm e.o.}$ as elements of the semirings $S_{\rm m.r.}$ or $S_{\rm e.o.}$, i.e.,\ we usually ignore the $\approx$-equivalence classes, and simply write $(A,B)=(A',B')$ when $(A,B)\approx(A',B')$ for $(A,B),(A',B')\in V_{\rm m.r.}$ or $V_{\rm e.o.}$. We will also treat functions (such as the monotone homomorphisms) on the semirings as functions on $V_{\rm m.r.}$ or $V_{\rm e.o.}$; we simply assume that they are constant on each equivalence class. 

In the rest of this chapter, we will usually write a pair $(A,B)$ in $V_{\rm m.r.}$ or $V_{\rm e.o.}$ as follows: 
\begin{equation}\label{eq:Vpair2}
    A=\bigoplus_{i=1}^{n}p_i\ketbra{\alpha_i}{\alpha_i},\quad B=\bigoplus_{i=1}^{n}q_i\ketbra{\beta_i}{\beta_i}, 
\end{equation}
where, for $i=1,\ldots,n$, the vector $\ket{\alpha_i}$ is normalized when $p_i>0$ and equal to the zero vector when $p_i=0$, and similarly for the $\ket{\beta_i}$, $q_i$. Note that when $(A,B)$ is a pair of quantum states, meaning that $\tr[A]=1=\tr[B]$, then $p_1+\ldots+p_n=1=q_1+\ldots+q_n$, i.e. the $p_i$ and the $q_i$ are both finite probability distributions. In this case, the pair $(A,B)$ is  essentially the same as the pair $(\rho,\sigma)$ in \eqref{eq:pairform}, but without the classical flags $\ketbrar{i}$ and using the direct sum instead of the ordinary sum, since the component states $\ketbrar{\alpha_i}$ are now associated to different Hilbert spaces $\Hil_i$. One sees that such pairs are precisely those in $\mc F$, i.e. cq-states with pure components that have some overlap. 

Note that our current semiring $S_{\rm e.o.}$ contains the full-support matrix majorization semiring from \cite{farooq2024} with $d=2$. The semiring $S_{\rm m.r.}$ contains the minimal restrictions matrix majorization semiring for $d=2$ from Chapter \ref{cha:matrix}, Section \ref{Msec:findings}. More precisely, in both $S_{\rm e.o.}$ and $S_{\rm m.r.}$, the earlier studied matrix majorization semirings appear as subsemirings, whose elements correspond to those $(A,B)$ in $S_{\rm e.o.}$ or $S_{\rm m.r.}$ where $A$ and $B$ commute. This means that we already know something about the monotone homomorphisms of $S_{\rm e.o.}$ and $S_{\rm m.r.}$ on commuting pairs, and we will use this knowledge in our further analysis. 

We introduce some notation that will make the connection with the previously studied matrix majorization semirings for $d=2$ from Chapter \ref{cha:matrix} more clear. Let $(A,B)$ be a pair as in \eqref{eq:Vpair2}. If it has a component $(p_i\ketbrar{\alpha_i},q_i\ketbrar{\beta_i})$ such that $|\braket{\alpha_i}{\beta_i}|=1$ for some $i$, then we may regard this as a ``classical'' component and we simply denote it as $(p_i,q_i)$, which is allowed by the equivalence relation defined on the semirings. Also, when $\braket{\alpha_i}{\beta_i}=0$, but $\ket{\alpha_i},\ket{\beta_i}\neq0$, we denote it as $(p_i,0)\boxplus(0,q_i)$. Similarly, when $\ket{\alpha_i}=0$, respectively $\ket{\beta_i}=0$, we denote the component as $(0,q_i)$, respectively $(p_i,0)$. As an example, see \eqref{eq:preparation} later on. 


\section{Power Universal Elements}\label{sec:pu}
We will identify elements in the semirings $S_{\rm m.r.}$ and $S_{\rm e.o.}$ that are power universal. First, we present the following useful result due to Uhlmann \cite{Uhlmann1985}. 
\begin{proposition}[Uhlmann]\label{propo:Uhlmann}
    Let $\ket{\alpha_1},\ket{\beta_1}\in\Hil_1$ and $\ket{\alpha_2},\ket{\beta_2}\in\Hil_2$ be normalized vectors in some Hilbert spaces $\Hil_1,\Hil_2$. Then 
    \begin{equation}
        (\,\ketbrar{\alpha_1},\ketbrar{\beta_1}\,)\succeq(\,\ketbrar{\alpha_2},\ketbrar{\beta_2}\,)\quad\Longleftrightarrow\quad|\braket{\alpha_1}{\beta_1}|\leq|\braket{\alpha_2}{\beta_2}|. 
    \end{equation}
\end{proposition}

Using this proposition, we first prove a characterization for certain elements in $S_{\rm e.o.}$ that are power universal. 

\begin{proposition}\label{prop:puS}
A pair $(A,B)$ of (normalized) quantum states of the form \eqref{eq:Vpair2} such that $0<|\braket{\alpha_i}{\beta_i}|<1$ for all $i=1,\ldots,n$ is a power universal of $S_{\rm e.o.}$.
\end{proposition}

Note that the above condition for the power universals of $S_{\rm e.o.}$ essentially means that all the components of the two cq-states are, in a sense, purely quantum.

\begin{proof}
Let us first consider a simpler case where we show that a pair $(\ketbrar{\alpha},\ketbrar{\beta})$ of pure states such that $0<|\braket{\alpha}{\beta}|<1$ is a power universal of $S_{\rm e.o.}$. To this end, let us fix $(A',B')\in S_{\rm e.o.}$ such that $\tr[A']=1=\tr[B']$, where
\begin{equation}\label{eq:a'b'}
    A'=\bigoplus_{i=1}^{n'}p'_i\ketbrar{\alpha'_i},\quad B'=\bigoplus_{i=1}^{n'}q'_i\ketbrar{\beta'_i}.
\end{equation}
Since we are within $S_{\rm e.o.}$, we may assume that $p'_i,q'_i>0$ for all $i=1,\ldots,n'$. We next show that 
\begin{equation}
    (\ketbrar{\alpha},\ketbrar{\beta})^{\boxtimes m}\succeq (A',B')
\end{equation}    
for $m\in\N$ sufficiently large. 

Let us denote $F:=|\braket{\alpha}{\beta}|^2$. Let $m_0\in\N$ be such that $F^m p'_i<q'_i$ for all $m\geq m_0$ and $i=1,\ldots,n'$. Let $\ket{\alpha_m^\perp}$ be a unit vector in the 2-dimensional Hilbert space $\mc M_m$ spanned by $\ket{\alpha}^{\otimes m}$ and $\ket{\beta}^{\otimes m}$ which is perpendicular to $\ket{\alpha}^{\otimes m}$ for all $m\in\N$. It follows that we may define, for all $m\geq m_0$, the positive-operator-valued measure $E^m=(E^m_i)_{i=1}^{n'}$ on $\mc M_m$ through
\begin{equation}
E^m_i=p'_i\ketbrar{\alpha}^{\otimes m}+\frac{q'_i-F^mp'_i}{1-F^m}\ketbrar{\alpha_m^\perp}.
\end{equation}
In what follows, we write, e.g.,\ $|\alpha^{\otimes m}\rangle:=|\alpha\rangle^{\otimes m}$. Through a straightforward calculation, we find that
\begin{equation}
\left\langle\alpha^{\otimes m}\middle|E^m_i\middle|\alpha^{\otimes m}\right\rangle=p'_i,\qquad \left\langle\beta^{\otimes m}\middle|E^m_i\middle|\beta^{\otimes m}\right\rangle=q'_i.
\end{equation}
Let us define the unit vectors $\big|\tilde{\beta}_i^m\big\rangle=(q'_i)^{-1/2}(E^m_i)^{1/2}\ket{\beta}^{\otimes m}$. It is easily seen that
\begin{equation}
\big\langle\alpha^{\otimes m}\big|\tilde{\beta}_i^m\big\rangle
=\sqrt{p'_i/q'_i}\braket{\alpha}{\beta}^m\to 0
\end{equation}
as $m\to\infty$. Since $|\braket{\alpha'_i}{\beta'_i}|>0$ for $i=1,\ldots,n'$, there is $m_1\in\N$, $m_1\geq m_0$, such that, for any $m\geq m_1$ and $i=1,\ldots,n'$, $|\big\langle\alpha^{\otimes m}\big|\tilde{\beta}_i^m\big\rangle|\leq|\braket{\alpha'_i}{\beta'_i}|$. Then, according to Proposition \ref{propo:Uhlmann}, for any $m\geq m_1$ and $i=1,\ldots,n'$, there is a channel $\mc C^m_i$ such that $\mc C^m_i(\ketbrar{\alpha}^{\otimes m})=\ketbrar{\alpha'_i}$ and $\mc C^m_i\big(\big|\tilde{\beta}_i^m\big\rangle\big\langle\tilde{\beta}_i^m\big|\big)=\ketbrar{\beta'_i}$. Thus, when $m\geq m_1$, we may define the channel $\mc C_m$ through
\begin{equation}
\mc C_m(\tau)=\bigoplus_{i=1}^{n'}\mc C_i^m\big((E^m_i)^{1/2}\tau(E_i^m)^{1/2}\big),
\end{equation}
and it is seen through a simple calculation that
\begin{align}
\mc C_m(\ketbrar{\alpha}^{\otimes m})&=\bigoplus_{i=1}^{n'}p'_i\ketbrar{\alpha'_i}=A',\\
\mc C_m(\ketbrar{\beta}^{\otimes m})&=\bigoplus_{i=1}^{n'}q'_i\ketbrar{\beta'_i}=B'.
\end{align}

Let us next assume that the state pair $(A,B)$ is as in \eqref{eq:Vpair2} with $0<|\braket{\alpha_i}{\beta_i}|<1$ for $i=1,\ldots,n$. We are free to assume that $|\braket{\alpha_1}{\beta_1}|^2$ is the greatest of the fidelities $|\braket{\alpha_i}{\beta_i}|^2$. According to Proposition \ref{propo:Uhlmann}, we have channels $\mc D_i$ such that $\mc D_i(\ketbrar{\alpha_i})=\ketbrar{\alpha_1}$ and $\mc D_i(\ketbrar{\beta_i})=\ketbrar{\beta_1}$ for $i=1,\ldots,n$; for $i=1$ we may naturally choose the identity map $\mc D_1={\rm id}$. Let $P_i$ be the projection of the Hilbert space onto the $i$'th component in \eqref{eq:Vpair2} and define the channel $\mc D$ by $\mc D(\tau)=\sum_{i=1}^n\mc D_i(P_i\tau P_i)$. It follows that $\mc D(A)=\ketbrar{\alpha_1}$ and $\mc D(B)=\ketbrar{\beta_1}$, so that $(A,B)$ majorizes $(\ketbrar{\alpha_1},\ketbrar{\beta_1})$. Since, by assumption, $0<|\braket{\alpha_1}{\beta_1}|<1$, we know from what we have shown thus far that $(\ketbrar{\alpha_1},\ketbrar{\beta_1})^{\boxtimes m}\succeq (A',B')$ for any normalized element $(A',B')\in S_{\rm e.o.}$ for sufficiently large $m\in\N$. 
Thus, also for any normalized $(A',B')\in S_{\rm e.o.}$,
\begin{equation}\label{eq:apu2}
(A,B)^{\boxtimes m}\succeq(\ketbrar{\alpha_1},\ketbrar{\beta_1})^{\boxtimes m}\succeq (A',B')
\end{equation}
for $m\in\N$ large enough. 

We are finally able to show that the above $(A,B)$ is a power universal. For this, consider pairs $(A_1,B_1),(A_2,B_2)\in S_{\rm e.o.}$ such that $(A_1,B_1)\preceq(A_2,B_2)$. This clearly means especially that $\tr[A_1]=\tr[A_2]$ and $\tr[B_1]=\tr[B_2]$. Define $A':=\tr[A_2]^{-1}A_2$ and $B':=\tr[B_2]^{-1}B_2$. Since we have already proven that, for $m\in\N$ large enough, we have \eqref{eq:apu2}, we now have
\begin{align}
(A_1,B_1)\boxtimes(A,B)^{\boxtimes m}&=(A_1\otimes A^{\otimes m},B_1\otimes B^{\otimes m})\\
&\succeq\big(\tr[A_1]A^{\otimes m},\tr[B_1]B^{\otimes m}\big)\\
&=\big(\tr[A_2]A^{\otimes m},\tr[B_2]B^{\otimes m}\big)\\
&\succeq\big(\tr[A_2]A',\tr[B_2]B'\big)=(A_2,B_2)
\end{align}
for $m\in\N$ large enough. This shows that $(A,B)$ is a power universal.
\end{proof}

Using the above result, we may easily characterize certain power universals of $S_{\rm m.r.}$. These are essentially of the same type as the power universals in $S_{\rm e.o.}$, but with an additional orthogonal component. 

\begin{proposition}\label{prop:pumr}
A pair $(A,B)$ of (normalized) quantum states of the form \eqref{eq:Vpair2} such that $|\braket{\alpha_i}{\beta_i}|<1$ for all $i=1,\ldots,n$, and $\braket{\alpha_i}{\beta_i}=0$ for some $i$ with $\ket{\alpha_i},\ket{\beta_i}\neq0$, is a power universal of $S_{\rm m.r.}$.
\end{proposition}

\begin{proof}
Let $(A,B)$ be a pair as stated in the claim. By using the equivalence relation defined on $S_{\rm m.r.}$, we may assume that there are $\ell$ with $1\leq\ell<n$, unit vectors $\ket{\alpha_i}$ and $\ket{\beta_i}$ such that $0<|\braket{\alpha_i}{\beta_i}|<1$, $i=1,\ldots,\ell$, positive probabilities $p_1,\ldots,p_\ell$ and $q_1,\ldots,q_\ell$ with $P:=p_1+\cdots+p_\ell<1$ and $Q:=q_1+\cdots+q_\ell<1$, and non-zero vectors $p,q\in\R_+^{n-\ell}$, $p=(p_{\ell+1},\ldots,p_{n})$, $q=(q_{\ell+1},\ldots,q_n)$, with $\|p\|_1+P=1=\|q\|_1+Q$ such that
\begin{align}
A&=\left(\bigoplus_{i=1}^\ell p_i\ketbrar{\alpha_i}\right)\oplus p,\\
B&=\left(\bigoplus_{i=1}^\ell q_i\ketbrar{\beta_i}\right)\oplus q,
\end{align}
and additionally 
\begin{equation}
{\rm supp}\,p\cap{\rm supp}\,q=\emptyset,
\end{equation}
i.e.,\ $p$ and $q$ are disjoint.

It follows that, for each $m\in\N$, 
\begin{align}
A^{\otimes m}&=\left(\bigoplus_{i_1,\ldots,i_m=1}^\ell p_{i_1}\cdots p_{i_m}\ketbrar{\alpha_{i_1}}\otimes\cdots\otimes\ketbrar{\alpha_{i_m}}\right)\oplus A^{(m)},\\
B^{\otimes m}&=\left(\bigoplus_{i_1,\ldots,i_m=1}^\ell q_{i_1}\cdots q_{i_m}\ketbrar{\beta_{i_1}}\otimes\cdots\otimes\ketbrar{\beta_{i_m}}\right)\oplus B^{(m)}
\end{align}
for some operators $A^{(m)}$ and $B^{(m)}$ where $A^{(m)}$ is a tensor product of $p$ (seen as a diagonal operator) with some other operator and $B^{(m)}$ is a tensor product of $q$ with some other operator. Since $p$ and $q$ have disjoint supports, $A^{(m)}$ and $B^{(m)}$ commute, so we may identify them with some sequences $p^{(m)}$ and $q^{(m)}$ (classical blocks) of non-negative numbers with disjoint supports. We may assume that the block $i=1$ has the highest fidelity among $|\braket{\alpha_i}{\beta_i}|^2$ for $i=1,\ldots,\ell$.
Then, by Proposition \ref{propo:Uhlmann}, there are channels $\mc C_{i_1,\ldots,i_m}$ such that
\begin{align}
\mc C_{i_1,\ldots,i_m}\big(\ketbrar{\alpha_{i_1}}\otimes\cdots\otimes\ketbrar{\alpha_{i_m}}\big)&=\ketbrar{\alpha_1}^{\otimes m},\\
\mc C_{i_1,\ldots,i_m}\big(\ketbrar{\beta_{i_1}}\otimes\cdots\otimes\ketbrar{\beta_{i_m}}\big)&=\ketbrar{\beta_1}^{\otimes m}.
\end{align}
Let $P_{i_1,\ldots,i_m}$ be the projections onto the $\ell^m$ first blocks and $Q_j$, $j=\ell^m+1,\ldots,n^m$ be the projections onto the last (classical) blocks. Let us also denote by $\big({\rm supp}\,p^{(m)}\big)^c$ the complement of the support of $p^{(m)}$ among the $n^m-\ell^m$ final classical blocks. Naturally, this complement contains the support of $q^{(m)}$. Let us set up the quantum channel $\mc C_m$
\begin{align}
\mc C_m(R)=&\sum_{i_1,\ldots,i_m=1}^\ell \mc C_{i_1,\ldots,i_m}(P_{i_1,\ldots,i_m}RP_{i_1,\ldots,i_m})\\
&\oplus\,\left(\sum_{j\in{\rm supp}\,p^{(m)}}\tr[Q_jRQ_j]\right)\,\oplus\,\left(\sum_{k\in\big({\rm supp}\,p^{(m)}\big)^c}\tr[Q_k RQ_k]\right).
\end{align}
We now have 
\begin{align}
\mc C_m(A^{\otimes m})&=P^m\ketbrar{\alpha_1}^{\otimes m}\oplus(1-P^m)\oplus 0=:A_m,\\
\mc C_m(B^{\otimes m})&=Q^m\ketbrar{\beta_1}^{\otimes m}\oplus 0\oplus(1-Q^m)=:B_m,
\end{align}
so that $(A,B)^{\boxtimes m}\succeq (A_m,B_m)$.

Let $(A',B')$ be a pair of operators within $S_{\rm m.r.}$ which we intend to majorize with a high enough power of the input $(A,B)$. Let us assume that this pair has the form of \eqref{eq:a'b'}, except that now not all $p'_i$ and $q'_i$ need to be non-zero. However, there needs to be a block $i_0$ such that $p'_{i_0},q'_{i_0}>0$ and $\braket{\alpha'_{i_0}}{\beta'_{i_0}}\neq0$. By reordering, we may assume that $i_0=1$. Let us now upper bound this target pair $(A',B')$. We may write $A'=A_1\oplus A_2$ and $B'=B_1\oplus B_2$ where $(A_1,B_1)\in S_{\rm e.o.}$ and $A_2$ and $B_2$ have blocks that are completely disjoint. Let $s:=\tr[A_1]$ and $t:=\tr[B_1]$. We now claim that $(\tilde{A},\tilde{B})\succeq (A',B')$ where
\begin{align}
\tilde{A}&=s\ketbrar{\alpha}\oplus(1-s)\oplus 0,\\
\tilde{B}&=t\ketbrar{\beta}\oplus 0\oplus(1-t)
\end{align}
where $\ket{\alpha}$ and $\ket{\beta}$ are some non-orthogonal unit vectors. Indeed, from the proof Proposition \ref{prop:puS}, we know that (when $|\braket{\alpha}{\beta}|$ is sufficiently small) \sloppy{${(\ketbrar{\alpha},\ketbrar{\beta})\succeq (\tr[A_1]^{-1}A_1,\tr[B_1]^{-1}B_1)}$}, and from the remaining orthogonal part \sloppy{${((1-s)\oplus0,0\oplus(1-t))}$} of $(\tilde{A},\tilde{B})$ we can easily find a channel that creates $(A_2,B_2)$. We next show that there is $m\in\N$ such that $(A,B)^{\boxtimes m}\succeq (A',B')$ by showing the intermediary $(A_m,B_m)\succeq(\tilde{A},\tilde{B})$ for sufficiently large $m\in\N$.

Let us denote the 2-dimensional space spanned by $\ket{\alpha_1}^{\otimes m}$ and $\ket{\beta_1}^{\otimes m}$ by $\mc M_m$ and the projection of $\mc M_m\oplus\C\oplus\C$ onto $\mc M_m$ by $R_m$. The projection orthogonal to this we denote by $R_m^\perp$. Using again Proposition \ref{propo:Uhlmann}, we find that, for $m\in\N$ large enough, there is a channel $\mc D_m$ such that $\mc D_m(\ketbrar{\alpha_1}^{\otimes m})=\ketbrar{\alpha}$ and $\mc D_m(\ketbrar{\beta_1}^{\otimes m})=\ketbrar{\beta}$. Let us assume that $m\in\N$ is sufficiently large so that $P^m<s$ and $Q^m<t$. 
We may now define a channel $\mc E_m$ on the last two blocks of $(A_m,B_m)$ (which together can be regarded as a pair of qubits) such that
\begin{align}
\mc E_m(1\oplus 0)&=\frac{s-P^m}{1-P^m}\ketbrar{\alpha}\oplus\frac{1-s}{1-P^m}\oplus0,\\
\mc E_m(0\oplus 1)&=\frac{t-Q^m}{1-Q^m}\ketbrar{\beta}\oplus0\oplus\frac{1-t}{1-Q^m}.
\end{align}
Define the channel $\mc C_m$ through
\begin{equation}
\mc C_m(\tau)=\mc D_m(R_m\tau R_m)+\mc E_m(R_m^\perp\tau R_m^\perp).
\end{equation}
It now follows from a straightforward calculation that $\mc C_m(A_m)=\tilde{A}$ and $\mc C_m(B_m)=\tilde{B}$. Thus, we may deduce that $(A,B)^{\boxtimes m}\succeq(A',B')$ for $m\in\N$ large enough. The fact that $(A,B)$ is a power universal now follows through the same simple logic as in the end of the proof of Proposition \ref{prop:puS}.
\end{proof}
    
\section{Identifying the Monotone Homomorphisms and Derivations}\label{sec:homomorphismsderivations}
We will identify all non-degenerate monotone homomorphisms with trivial kernel and non-zero monotone derivations associated with the semiring $S_{\rm m.r.}$, see Proposition \ref{prop:monhom}. We will also derive partial results for $S_{\rm e.o.}$, summarized in Remark \ref{rmk:Seo}. According to Theorem \ref{thm:Fritz2022}, sufficient conditions for large-sample and catalytic majorization can be stated in terms of these homomorphisms and derivations. In the following calculations and proofs we will repeatedly use vectors
\begin{equation}\label{eq:anglevectors}
|\psi(\theta)\rangle:=\cos{\theta}\,|0\rangle+\sin{\theta}\,|1\rangle\in\C^2,\qquad\theta\in[0,2\pi),
\end{equation}
where $\{|0\rangle,|1\rangle\}$ is some fixed orthonormal basis of $\C^2$. The next proposition shows that, compared with the earlier classical results on matrix majorization, a new parameter $z$ appears due to the quantum nature of the pairs of states. 


\begin{proposition}\label{prop:z}
    Let $\Phi:S_{*}\rightarrow\K$ be a monotone homomorphism, where $S_{*}\in\{S_{\rm m.r.},S_{\rm e.o.}\}$ and $\K\in\{\R_+,\R_+^{\rm op},\T\R_+,\T\R_+^{\rm op}\}$. There exists $z\in\R$ such that 
    \begin{equation}
        \Phi(\ketbrar{\alpha},\ketbrar{\beta})=|\braket{\alpha}{\beta}|^{2z}
    \end{equation}
    for all normalized $\ket{\alpha},\ket{\beta}\in\Hil$ such that $\braket{\alpha}{\beta}\neq0$, and $\Hil$ is a Hilbert space. 
\end{proposition}

\begin{proof}
    Fix orthonormal vectors $\ket{0}$ and $\ket{1}$ that span the subspace $\mc M$ where $\ket{\alpha}$ and $\ket{\beta}$ reside. For any normalized $|\alpha^\perp\rangle\in\mc M$ perpendicular to $\ket{\alpha}$ we have 
    \begin{equation}
        |\braket{\alpha}{\beta}|^2+|\langle\alpha^\perp|\beta\rangle|^2=1. 
    \end{equation}
    Hence 
    \begin{equation}
        |\braket{\alpha}{\beta}|=\cos\theta,\quad|\langle\alpha^\perp|\beta\rangle|=\sin\theta,
    \end{equation}
    for some $\theta\in[0,\pi/2]$. By multiplying with an appropriate phase, we may choose $|\alpha^\perp\rangle$ such that 
    \begin{equation}
        \braket{\alpha}{\beta}\langle\beta|\alpha^\perp\rangle=\cos\theta\sin\theta. 
    \end{equation}
    Via the unitary operation defined by $\ket{\alpha}\mapsto\ket{0},|\alpha^\perp\rangle\mapsto\ket{1}$, the pair $(\ketbra{\alpha}{\alpha},\ketbra{\beta}{\beta})$ is equivalent (according to the equivalence relation $\approx$ defined in Section \ref{sec:definingthesemirings}) to the pair 
    \begin{equation}\label{eq:pairstandardform}
        (\ketbra{0}{0},\cos^2\theta\ketbra{0}{0}+\cos\theta\sin\theta\ketbra{0}{1}+\cos\theta\sin\theta\ketbra{1}{0}+\sin^2\theta\ketbra{1}{1}). 
    \end{equation}
    From this, we see that $\Phi(\ketbrar{\alpha},\ketbrar{\beta})$ only depends on the fidelity $\cos^2\theta=|\braket{\alpha}{\beta}|^2$. Therefore, we can define $\varphi:(0,1]\rightarrow\R$, 
    \begin{equation}
        \varphi(x):=\Phi(\ketbrar{\alpha},\ketbrar{\beta}), 
    \end{equation}
    where $\ket{\alpha},\ket{\beta}$ are any normalized vectors such that $|\braket{\alpha}{\beta}|^2=x$. Note that 
    \begin{align}
        &(\ketbrar{\alpha_1},\ketbrar{\beta_1})\boxtimes(\ketbrar{\alpha_2},\ketbrar{\beta_2})\\            &\quad\quad=((\ket{\alpha_1}\otimes\ket{\alpha_2})(\bra{\alpha_1}\otimes\bra{\alpha_2}),(\ket{\beta_1}\otimes\ket{\beta_2})(\bra{\beta_1}\otimes\bra{\beta_2})) 
    \end{align}
    and 
    \begin{equation}
        |(\bra{\alpha_1}\otimes\bra{\alpha_2})(\ket{\beta_1}\otimes\ket{\beta_2})|^2=|\braket{\alpha_1}{\beta_1}|^2|\braket{\alpha_2}{\beta_2}|^2. 
    \end{equation}
    Thus, by multiplicativity of $\Phi$, the function $\varphi$ is multiplicative as well, i.e. 
    \begin{equation}
        \varphi(xy)=\varphi(x)\varphi(y)\text{ for all }x,y\in(0,1]. 
    \end{equation}    
    Using Proposition \ref{propo:Uhlmann}, for any $\ket{\alpha},\ket{\beta}$ such that 
    \begin{equation}
        |\braket{\alpha}{\beta}|^2\geq\frac{1}{2}=|\braket{0}{\psi(\pi/4)}|^2, 
    \end{equation}
    where $|\psi(\pi/4)\rangle$ is defined as in \eqref{eq:anglevectors}, we have 
    \begin{equation}
        (1,1)\preceq(\ketbrar{\alpha},\ketbrar{\beta})\preceq(\ketbrar{0},\ketbrar{\psi(\pi/4)}). 
    \end{equation}
    Hence for all $x\in[1/2,1]$, we have $\varphi(x)\leq1$ if $\K=\R_+^{\rm op},\T\R_+^{\rm op}$, and $\varphi(x)\leq\varphi(1/2)$ if $\K=\R_+,\T\R_+$. Therefore, $\varphi$ is bounded from above on a set of positive measure, which establishes that $\varphi$ is a regular solution of the multiplicative version of Cauchy's Functional Equation \cite[Chapter 3, Proposition 6]{Aczel1989}: there exists $z\in\R$ such that $\varphi(x)=x^z$ or $\varphi$ vanishes everywhere. The latter case is impossible since $\varphi(1)=\Phi(1,1)=1$ by the properties of a homomorphism. 
\end{proof}

\begin{remark}\label{rmk:classical}
Let us summarize some results that were established in Chapter \ref{cha:matrix} on the possible monotone homomorphisms and derivations on commuting pairs $(A,B)$ in $S_{\rm m.r.}$, i.e.,\ a pair $(A,B)$ such as in \eqref{eq:Vpair2} where the component spaces $\Hil_i\simeq\C$. We will use these results extensively in the analysis of our current semiring $S_{\rm m.r.}$. 

The commuting pairs in $S_{\rm m.r.}$ correspond precisely to the elements in the semiring denoted in Chapter \ref{cha:matrix}, Section \ref{Msec:findings} as $S^2$, the minimal restrictions matrix majorization semiring of length $2$. Hence any non-degenerate monotone homomorphism or monotone derivation on $S_{\rm m.r.}$, when restricted to these commuting pairs, must correspond to a (possibly degenerate) monotone homomorphism or monotone derivation on $S^2$. 

We now list these homomorphisms and derivations. Let $\Phi:S_{\rm m.r.}\rightarrow\K$, $\K\in\{\R_+,\R_+^{\rm op},\T\R_+,\T\R_+^{\rm op}\}$, be a monotone homomorphism. When $\K=\R_+^{\rm op}$, $\Phi$ restricted to a commuting pair $(A,B)$ takes the form 
\begin{equation}\label{eq:Phitemperate}
\Phi(A,B)=\sum_{i=1}^n p_i^\alpha q_i^{1-\alpha}
\end{equation}
for some $\alpha\in(0,1)$, or 
\begin{equation}\label{eq:Phicharacter}
\Phi(A,B)=\sum_{i:p_i>0} q_i,\quad\text{or}\quad\Phi(A,B)=\sum_{i:q_i>0} p_i, 
\end{equation}
or $\Phi$ restricts to one of the two degenerate homomorphisms $\Phi(A,B)=\tr[A]$ or $\tr[B]$.\footnote{Note that the two cases in \eqref{eq:Phicharacter} correspond to the character $C=\{1,2\}$, i.e. only those $i$ for which $p_i>0$ and $q_i>0$ are summed over.} When $\K=\R_+$, $\Phi$ always restricts to one of these latter two degenerate cases. When $\mb K=\T\R_+$ or $\T\R_+^{\rm op}$, the only form $\Phi$ can take on a commuting pair is the highly degenerate $\Phi(A,B)=1$ when $(A,B)\neq 0$ and $\Phi(0,0)=0$. (However, it will turn out that the new parameter $z$ in Proposition \ref{prop:z} makes the extension of this homomorphism to $S_{\rm m.r.}$ \emph{non}-degenerate in the case $\K=\T\R_+^{\rm op}$.) Finally, the derivations corresponding to the two degenerate cases $\Phi(A,B)=\tr[A]$ or $\tr[B]$, up to interchangeability (see footnote \ref{ftn:interchangeability} on page \pageref{ftn:interchangeability}), must vanish on commuting pairs.
\end{remark}

Using the additivity and multiplicativity of any homomorphism $\Phi:S_{\rm m.r.}\rightarrow\K$ in the cases $\K=\R_+,\R_+^{\rm op}$, we can decompose $\Phi(A,B)$, where $(A,B)\in S_{\rm m.r.}$ as in \eqref{eq:Vpair2}, as follows: 
\begin{equation}\label{eq:decompositiontemperate}
    \Phi(A,B)=\sum_{i=1}^n\Phi(p_i,q_i)\Phi(\ketbrar{\alpha_i},\ketbrar{\beta_i}). 
\end{equation}
And similarly for the cases $\K=\T\R_+,\T\R_+^{\rm op}$: 
\begin{equation}\label{eq:decompositiontropical}
    \Phi(A,B)=\max_{i=1,\ldots,n}\Phi(p_i,q_i)\Phi(\ketbrar{\alpha_i},\ketbrar{\beta_i}). 
\end{equation}

Note that if $\ket{\alpha_i},\ket{\beta_i}$ are orthogonal or exactly one of them is the zero-vector, then $(\ketbrar{\alpha_i},\ketbrar{\beta_i})$ is not part of $S_{\rm m.r.}$. Similarly for $(p_i,q_i)$ if either $p_i$ or $q_i$ is zero. Hence, some terms in the decompositions \eqref{eq:decompositiontemperate} and \eqref{eq:decompositiontropical} are not well-defined. We will address this issue in the first part of the proof of Proposition \ref{prop:monhom}. 

By Proposition \ref{prop:z}, we know the value of all the ``quantum'' terms $\Phi(\ketbrar{\alpha_i},\ketbrar{\beta_i})$ in \eqref{eq:decompositiontemperate} and \eqref{eq:decompositiontropical}. Also, from the results on commuting pairs as summarized in Remark \ref{rmk:classical}, we know the value of all the ``classical'' terms $\Phi(p_i,q_i)$, as well. Therefore, the value of $\Phi(A,B)$ for any $(A,B)\in S_{\rm m.r.}$ follows. What remains is to determine exactly which of these homomorphisms are non-degenerate and monotone, which will yield the allowed range of the parameters $\alpha$ and $z$, and investigating the existence of non-zero monotone derivations. Our results are contained in the following proposition. 



\begin{proposition}\label{prop:monhom}
The non-degenerate monotone homomorphisms $\Phi:S_{\rm m.r.}\to\mb K$, where $\K\in\{\R_+,\R_+^{\rm op},\T\R_+,\T\R_+^{\rm op}\}$, with trivial kernel are exactly the following: 
\begin{itemize}
\item $\mb K=\R_+^{\rm op}$: $\Phi$ is of the form 
\begin{equation}\label{eq:temperate}
    \Phi_{\alpha,z}(A,B):=\sum_{i=1}^n(p_i)^\alpha(q_i)^{1-\alpha}|\braket{\alpha_i}{\beta_i}|^{2z}, 
\end{equation}    
for any $\alpha\in[0,1]$ and $z\geq\max\{\alpha,1-\alpha\}$. 
\item $\mb K=\T\R_+^{\rm op}$: $\Phi$ is of the form 
\begin{equation}\label{eq:tropicalz}
    \Phi_z^\T(A,B):=\max_{i=1,\ldots,n}|\braket{\alpha_i}{\beta_i}|^{2z}, 
\end{equation}    
for any $z>0$. 
\item $\mb K\in\{\R_+,\T\R_+\}$: There are no non-degenerate monotone homomorphisms with trivial kernel. 
\end{itemize}
Moreover, there are no non-zero monotone derivations.
\end{proposition}

\begin{remark}    
    Note that when $\alpha=0$ or $1$, it may occur that the expression \eqref{eq:temperate} contains a factor $0^0$ when $p_i=0$ or $q_i=0$ for some $i$. However, since in this case also the corresponding $\ket{\alpha_i}$ or $\ket{\beta_i}$ equals the zero vector, there is no ambiguity: the entire expression $(p_i)^\alpha(q_i)^{1-\alpha}|\braket{\alpha_i}{\beta_i}|^{2z}$ should be considered to be equal to $0$. 

    One checks that for pairs $(A,B)\in S_{\rm m.r.}$, we have 
    \begin{equation}
        \Phi_{\alpha,z}(A,B)=\tr\left[\left(B^{\frac{1-\alpha}{2z}}A^{\frac{\alpha}{z}}B^{\frac{1-\alpha}{2z}}\right)^z\right], 
    \end{equation}
    where the right-hand side is the expression appearing inside the logarithm in the definition of the $\alpha$-$z$ relative entropy in \eqref{eq:alphazdivergence}. 

    For two pairs $(A,B),(A',B')\in S_{\rm m.r.}$, note that $\Phi_z^\T(A,B)>\Phi_z^\T(A',B')$ for all $z>0$ if and only if this inequality holds for one such $z$. Since the conditions for majorization given by Theorem \ref{thm:Fritz2022} are stated in terms of such inequalities, it is sufficient for our purposes to consider only $\Phi_z^\T$ for the choice $z=1$. Hence, we define 
    \begin{equation}\label{eq:tropical}
        \Phi^\T(A,B):=\Phi_1^\T(A,B)=\max_{i=1,\ldots,n}|\braket{\alpha_i}{\beta_i}|^{2}=\|({\rm supp}\,A)({\rm supp}\,B)\|_\infty^2. 
    \end{equation}
    Here, the right-hand side gives an alternative expression for the homomorphism, where $\|\cdot\|_\infty$ is the operator norm and, e.g.,\ ${\rm supp}\,A$ is the support projection of the positive semidefinite $A$. 
\end{remark}

\begin{proof}[Proof of Proposition \ref{prop:monhom}]
Before proceeding to the characterization of all relevant monotone homomorphisms and derivations, we will first show how to interpret the decompositions \eqref{eq:decompositiontemperate} and \eqref{eq:decompositiontropical} when there are $i$ such that $\ket{\alpha_i},\ket{\beta_i}$ are orthogonal or exactly one of them is the zero-vector. 

Fix a monotone homomorphism $\Phi:S_{\rm m.r.}\rightarrow\K$, where $\K\in\{\R_+,\R_+^{\rm op},\allowbreak\T\R_+,\T\R_+^{\rm op}\}$, and let $(A,B)\in S_{\rm m.r.}$. Note that we may assume that there exists a representation as in \eqref{eq:Vpair2} such that there is no $j$ for which $\braket{\alpha_j}{\beta_j}=0$ and $\ket{\alpha_j},\ket{\beta_j}\neq0$. Indeed, if there is such $j$, then the pair $(p_j\ketbrar{\alpha_j},q_j\ketbrar{\beta_j})$ can be replaced by $(p_j\oplus0,0\oplus q_j)$ under the equivalence relation $\approx$. 

Due to \eqref{eq:cscondition}, there exists $k\in[n]$ such that $|\braket{\alpha_k}{\beta_k}|\neq0$. Assume there is a $j_1$ such that $\ket{\alpha_{j_1}}\neq0$ and $\ket{\beta_{j_1}}=0$, or a $j_2$ such that $\ket{\alpha_{j_2}}=0$ and $\ket{\beta_{j_2}}\neq0$. In the following, we will show that $\Phi(A,B)$ may be expressed as \eqref{eq:decompositiontemperate} or \eqref{eq:decompositiontropical}, where we set 
\begin{equation}\label{eq:define0trA}
    \Phi(p_{j_1}\ketbrar{\alpha_{j_1}},0)=p_{j_1}\Phi(\ketbrar{\alpha_k},\ketbrar{\beta_k})\text{ and }\Phi(0,q_{j_2}\ketbrar{\beta_{j_2}})=0
\end{equation}
if $\K=\R_+$ or $\R_+^{\rm op}$ and $\Phi$ restricts to $\tr[A']$ on commuting pairs $(A',B')$, 
\begin{equation}\label{eq:define0trB}
    \Phi(p_{j_1}\ketbrar{\alpha_{j_1}},0)=0\text{ and }\Phi(0,q_{j_2}\ketbrar{\beta_{j_2}})=q_{j_2}\Phi(\ketbrar{\alpha_k},\ketbrar{\beta_k})
\end{equation}
if $\K=\R_+$ or $\R_+^{\rm op}$ and $\Phi$ restricts to $\tr[B']$ on commuting pairs $(A',B')$, and 
\begin{equation}\label{eq:define0otherwise}
    \Phi(p_{j_1}\ketbrar{\alpha_{j_1}},0)=0\text{ and }\Phi(0,q_{j_2}\ketbrar{\beta_{j_2}})=0
\end{equation}
otherwise. Note that, although \eqref{eq:define0trA}, \eqref{eq:define0trB} seem to depend on the choice of $k$ due to the appearance of the term $\Phi(\ketbrar{\alpha_k},\ketbrar{\beta_k})$, we will show later that $z=0$ when $\Phi$ restricts to $\tr[A']$ or $\tr[B']$ on commuting pairs $(A',B')$, hence $\Phi(\ketbrar{\alpha_k},\ketbrar{\beta_k})=1$, independent of $k$. 

Assume there is a $j_1$ such that $\ket{\alpha_{j_1}}\neq0$ and $\ket{\beta_{j_1}}=0$; the analysis of the case of $j_2$ such that $\ket{\alpha_{j_2}}=0$ and $\ket{\beta_{j_2}}\neq0$ is analogous. Denote by $\substack{\succeq\\\preceq}$ majorization in both directions. There exist a quantum channel that splits block $k$, and another channel that merges two blocks, which implies 
\begin{align}
    \big(p_k\ketbrar{\alpha_k},q_k\ketbrar{\beta_k}\big)\substack{\succeq\\\preceq}&\big(p_k/2\ketbrar{\alpha_k},q_k/2\ketbrar{\beta_k}\big)\\
    &\boxplus\big(p_k/2\ketbrar{\alpha_k},q_k/2\ketbrar{\beta_k}\big). 
\end{align}
Denote by $(\tilde A,\tilde B)\in S_{\rm m.r.}$ the pair $(A,B)$ with blocks $j_1$ and $k$ removed. Using the majorization relation above, we have 
\begin{align}
    \big(A,B\big)\substack{\succeq\\\preceq}&\big(p_k/2\ketbrar{\alpha_k}\oplus \tilde A,q_k/2\ketbrar{\beta_k}\oplus\tilde B\big)\label{eq:split1}\\
    &\boxplus\big(p_k/2\ketbrar{\alpha_k}\oplus p_{j_1}\ketbrar{\alpha_{j_1}},q_k/2\ketbrar{\beta_k}\oplus0\big).\label{eq:split2}
\end{align}
Since $\ket{\beta_{j_1}}=0$, we may assume that $\ket{\alpha_{j_1}}=\ket{\alpha_k}$ under the equivalence relation $\approx$. Then, notice that 
\begin{align}
    &\big(p_k/2\ketbrar{\alpha_k}\oplus p_{j_1}\ketbrar{\alpha_{j_1}},q_k/2\ketbrar{\beta_k}\oplus0\big)\label{eq:specialblock1}\\
    &\quad\quad=\big(\ketbrar{\alpha_k},\ketbrar{\beta_k}\big)\boxtimes\big(p_k/2\oplus p_{j_1},q_k/2\oplus0\big).\label{eq:specialblock2}
\end{align}

Now assume $\K=\R_+$ or $\R_+^{\rm op}$. Assume further that the restriction of $\Phi$ on commuting pairs $(A',B')$ does not equal $\tr[A']$. Considering the remaining possible restrictions (see Remark \ref{rmk:classical}), it follows from \eqref{eq:specialblock1}, \eqref{eq:specialblock2} that 
\begin{align}
    &\Phi\big(p_k/2\ketbrar{\alpha_k}\oplus p_{j_1}\ketbrar{\alpha_{j_1}},q_k/2\ketbrar{\beta_k}\oplus0\big)\\
    &\quad\quad=\Phi\big(\ketbrar{\alpha_k},\ketbrar{\beta_k}\big)\Phi\big(p_k/2\oplus p_{j_1},q_k/2\oplus0\big)\\
    &\quad\quad=\Phi\big(\ketbrar{\alpha_k},\ketbrar{\beta_k}\big)\Phi\big(p_k/2,q_k/2\big)\\
    &\quad\quad=\Phi\big(p_k/2\ketbrar{\alpha_k},q_k/2\ketbrar{\beta_k}\big). 
\end{align}
Using this equality and \eqref{eq:split1}, \eqref{eq:split2}, via splitting and merging of blocks, we find 
\begin{align}
    \Phi(A,B)&=\Phi\big(p_k/2\ketbrar{\alpha_k}\oplus\tilde A,q_k/2\ketbrar{\beta_k}\oplus\tilde B\big)\\
    &\quad+\Phi\big(p_k/2\ketbrar{\alpha_k},q_k/2\ketbrar{\beta_k}\big)\\        
    &=\Phi\big(p_k\ketbrar{\alpha_k}\oplus\tilde A,q_k\ketbrar{\beta_k}\oplus\tilde B\big).
\end{align}
Note that $(p_k\ketbrar{\alpha_k}\oplus\tilde A,q_k\ketbrar{\beta_k}\oplus\tilde B)$ equals $(A,B)$ with the pair $(p_{j_1}\ketbrar{\alpha_{j_1}},0)$ removed. Hence, we may set $\Phi(p_{j_1}\ketbrar{\alpha_{j_1}},0)=0$. Assume now that $\Phi$ restricts to $\tr[A']$ on commuting pairs $(A',B')$. Then it follows from \eqref{eq:specialblock1}, \eqref{eq:specialblock2} that 
\begin{align}
    &\Phi\big(p_k/2\ketbrar{\alpha_k}\oplus p_{j_1}\ketbrar{\alpha_{j_1}},q_k/2\ketbrar{\beta_k}\oplus0\big)\\
    &\quad\quad=\Phi\big(\ketbrar{\alpha_k},\ketbrar{\beta_k}\big)(p_k/2+p_{j_1})\\
    &\quad\quad=\Phi\big(\ketbrar{\alpha_k},\ketbrar{\beta_k}\big)(\Phi\big(p_k/2,q_k/2\big)+p_{j_1})\\        &\quad\quad=\Phi\big(p_k/2\ketbrar{\alpha_k},q_k/2\ketbrar{\beta_k}\big)+p_{j_1}\Phi\big(\ketbrar{\alpha_k},\ketbrar{\beta_k}\big). 
\end{align}
Now it follows from \eqref{eq:split1}, \eqref{eq:split2} that 
\begin{equation}
    \Phi(A,B)=\Phi\big(p_k\ketbrar{\alpha_k}\oplus\tilde A,q_k\ketbrar{\beta_k}\oplus\tilde B\big)+p_{j_1}\Phi(\ketbrar{\alpha_k},\ketbrar{\beta_k}), 
\end{equation}
allowing us to set $\Phi(p_{j_1}\ketbrar{\alpha_{j_1}},0)=p_{j_1}\Phi(\ketbrar{\alpha_k},\ketbrar{\beta_k})$. Analogously one can show that, when there is $j_2$ such that $\ket{\alpha_{j_2}}=0$ and $\ket{\beta_{j_2}}\neq0$, and $\Phi$ does not restrict to $\tr[B']$ on commuting pairs $(A',B')$, we may set $\Phi(0,q_{j_2}\ketbrar{\beta_{j_2}})=0$, and otherwise $\Phi(0,q_{j_2}\ketbrar{\beta_{j_2}})=q_{j_2}\Phi(\ketbrar{\alpha_k},\ketbrar{\beta_k})$. We conclude that \eqref{eq:decompositiontemperate} with the conventions \eqref{eq:define0trA}, \eqref{eq:define0trB}, \eqref{eq:define0otherwise} are correct for $\K=\R_+,\R_+^{\rm op}$. 

Next, assume $\K=\T\R_+$ or $\T\R_+^{\rm op}$. Recall that $\Phi$ must restrict to $\Phi(A,B)=1$ for commuting pairs $(A,B)\neq 0$ and $\Phi(0,0)=0$. Assume there is a $j_1$ such that $\ket{\alpha_{j_1}}\neq0$ and $\ket{\beta_{j_1}}=0$. Using \eqref{eq:specialblock1}, \eqref{eq:specialblock2}, we now have 
\begin{equation}
    \Phi\big(p_k/2\ketbrar{\alpha_k}\oplus p_{j_1}\ketbrar{\alpha_{j_1}},q_k/2\ketbrar{\beta_k}\oplus0\big)=\Phi\big(\ketbrar{\alpha_k},\ketbrar{\beta_k}\big). 
\end{equation}    
Then, using similar arguments as above using \eqref{eq:split1}, \eqref{eq:split2}, and remembering that addition in $\T\R_+,\T\R_+^{\rm op}$ is defined by taking the maximum of two numbers, we may conclude that 
\begin{equation}
    \Phi(A,B)=\Phi\big(p_k\ketbrar{\alpha_k}\oplus\tilde A,q_k\ketbrar{\beta_k}\oplus\tilde B\big). 
\end{equation}
The case of $j_2$ such that $\ket{\alpha_{j_2}}=0$ and $\ket{\beta_{j_2}}\neq0$ can be handled analogously. We conclude that \eqref{eq:decompositiontropical} with the convention \eqref{eq:define0otherwise} is correct for $\K=\T\R_+,\T\R_+^{\rm op}$. 

In the next part of the proof, we will first exclude certain ranges that the $\alpha$ and $z$ parameters may be in by providing counterexamples to the monotonicity property of homomorphisms. Similarly, we will use monotonicity to show that there are no non-zero monotone derivations. Finally, we prove that the homomorphisms that remain are in fact monotone. 

Fix a monotone homomorphism $\Phi:S_{\rm m.r.}\rightarrow\K$, where $\K\in\{\R_+,\R_+^{\rm op},\allowbreak\T\R_+,\T\R_+^{\rm op}\}$. From the results in Propositions \ref{propo:Uhlmann} and \ref{prop:z} on majorization of pairs of pure states, and using the monotonicity of $\Phi$, it follows that the parameter $z$ must be non-positive when ${\mb K}=\R_+,\T\R_+$ and non-negative when ${\mb K}=\R_+^{\rm op},\T\R_+^{\rm op}$. 

Let us assume $\mb K=\R_+^{\rm op}$ first. In the following, we will consider case-by-case the different ways that $\Phi$ restricts to commuting pairs. 

The quantum channel arising from measuring a qubit in a fixed basis $\ket{0},\ket{1}$ yields the following majorization relation: 
\begin{equation}
\label{eq:alphazrelationship1}
    \left(\,\ketbrar{\psi(\pi/4)},\ketbrar{0}\,\right)\succeq(1/2,1)\boxplus(1/2,0), 
\end{equation}
where $|\psi(\theta)\rangle$ is defined as in \eqref{eq:anglevectors}. Consider the case that $\Phi$ is of the form \eqref{eq:Phitemperate} with $\alpha\in(0,1)$ when restricted to commuting pairs. Evaluating $\Phi$ on both sides of the inequality \eqref{eq:alphazrelationship1}, it follows by monotonicity that 
\begin{equation}
\label{eq:alphazrelationship2}
    \frac{1}{2^z}\leq\frac{1}{2^\alpha}. 
\end{equation}
By switching around the states in the two pairs in \eqref{eq:alphazrelationship1}, we have similarly
\begin{equation}
\label{eq:alphazrelationship3}
    \frac{1}{2^z}\leq\frac{1}{2^{1-\alpha}}. 
\end{equation}    
Hence \sloppy{${z\geq\max\{\alpha,1-\alpha\}}$}. It follows by the decomposition \eqref{eq:decompositiontemperate} with the convention \eqref{eq:define0otherwise} that $\Phi$ must be as in \eqref{eq:temperate} with $\sloppy{\alpha\in(0,1)}$ and ${z\geq\max\{\alpha,1-\alpha\}}$. 

Next, consider the case that $\Phi$ is one of \eqref{eq:Phicharacter} when restricted to a commuting pair. When it restricts to the second form in \eqref{eq:Phicharacter}, monotonicity and \eqref{eq:alphazrelationship1} imply that 
\begin{equation}
    \frac{1}{2^z}\leq\frac{1}{2}. 
\end{equation}
Hence, $z\geq1$. Similarly, when $\Phi$ restricts to the first form in \eqref{eq:Phicharacter}, by switching around the states in the two pairs in \eqref{eq:alphazrelationship1}, we find $z\geq1$. It follows for these latter two cases that $\Phi$ is equal to \eqref{eq:temperate} with $\alpha=0$ or $1$, and $z\geq1=\max\{\alpha,1-\alpha\}$. 

A measure-and-prepare channel establishes the following majorization relation: 
\begin{align}\label{eq:preparation}
&(1/2,1/2)\boxplus(1/2,0)\boxplus(0,1/2)\\
&\qquad\qquad\succeq(1/2,1/2)\boxplus(1/2\ketbrar{0},1/2\ketbrar{\psi(\pi/4)}). 
\end{align}
Consider the remaining case that $\Phi$ is equal to $\tr[A]$ or $\tr[B]$ on commuting pairs. By monotonicity, it follows from \eqref{eq:preparation} that 
\begin{equation}\label{eq:preparationinequality}
    1\leq\frac{1}{2}+\frac{1}{2}\frac{1}{2^z}.
\end{equation} 
Since $z\geq0$, this implies that $z=0$. Hence, using the convention \eqref{eq:define0trA} or \eqref{eq:define0trB}, $\Phi$ is one of the degenerate homomorphisms $\Phi(A,B)=\tr[A]$ or $\tr[B]$ for all $(A,B)\in S_{\rm m.r.}$. 

Assume now that ${\mb K}=\R_+$. Then $\Phi$ equals $\tr[A]$ or $\tr[B]$ on commuting pairs. Monotonicity and \eqref{eq:preparation} now yield \eqref{eq:preparationinequality} with the inequality reversed. Since $z\leq0$, it follows that $z=0$. Again, $\Phi$ is one of the degenerate cases $\Phi(A,B)=\tr[A]$ or $\tr[B]$ for all $(A,B)\in S_{\rm m.r.}$. 

Next, assume ${\mb K}=\T\R_+^{\rm op}$. Then $\Phi$ is the highly degenerate $\Phi(A,B)=1$ when $(A,B)\neq 0$ and $\Phi(0,0)=0$ on commuting pairs. Also, $z\geq0$. Using the decomposition \eqref{eq:decompositiontropical} with the convention \eqref{eq:define0otherwise}, we see that when $z=0$ we get the degenerate homomorphism that is equal to $1$ for all non-zero elements in $S_{\rm m.r.}$. When $z>0$, we find that $\Phi$ must be as in \eqref{eq:tropicalz}. 

Finally, assume ${\mb K}=\T\R_+$. Then $\Phi$ is again the highly degenerate $\Phi(A,B)=1$ when $(A,B)\neq 0$ and $\Phi(0,0)=0$ on commuting pairs. The majorization relation \eqref{eq:preparation} implies that 
\begin{equation}
\label{eq:zrelationship2}
    1\geq\max\left\{1,\frac{1}{2^z}\right\}, 
\end{equation}    
Since $z\leq0$, this implies $z=0$. Hence, using the decomposition \eqref{eq:decompositiontropical} with the convention \eqref{eq:define0otherwise}, we get the degenerate homomorphism that is equal to $1$ for all non-zero elements in $S_{\rm m.r.}$. 

We now turn to studying the monotone derivations on $S_{\rm m.r.}$ at the two degenerate homomorphisms $\Phi(A,B)=\tr[A]$ or $\tr[B]$ for all $(A,B)\in S_{\rm m.r.}$. Let $\Delta:S_{\rm m.r.}\rightarrow\R$ be a monotone derivation at $\Phi(A,B)=\tr[A]$, \sloppy{${(A,B)\in S_{\rm m.r.}}$}. Observe that for any $(A,B)\in S_{\rm m.r.}$, 
\begin{equation}
    (\tr[A],\tr[B])\boxplus(1,1)\preceq(A,B)\boxplus(1,1)\preceq(1,1)\boxplus(\tr[A],0)\boxplus(0,\tr[B]), 
\end{equation}
where the second majorization is achieved by using a measure-and-prepare channel. By interchangeability, we may assume that $\Delta$ vanishes on commuting pairs. Since the outer left and outer right sides of the above inequalities are commuting pairs, monotonicity and additivity of $\Delta$ yields 
\begin{equation}
    0=\Delta((A,B)\boxplus(1,1))=\Delta(A,B)+\underbrace{\Delta(1,1)}_{=0}. 
\end{equation}
Hence, $\Delta$ vanishes everywhere on $S_{\rm m.r.}$. The same conclusion holds for any derivation at the other degenerate monotone homomorphism $\Phi(A,B)=\tr[B]$, $(A,B)\in S_{\rm m.r.}$. 

Finally, we show that the maps in \eqref{eq:temperate} and \eqref{eq:tropicalz} are in fact monotone homomorphisms with trivial kernel. Additivity and multiplicativity are immediate from the decompositions \eqref{eq:decompositiontemperate}, \eqref{eq:decompositiontropical} together with the multiplicativity of the factors appearing in the decompositions. Also, they all have a trivial kernel, since for every non-zero $(A,B)\in S_{\rm m.r.}$ there is at least one $k$ such that $\braket{\alpha_k}{\beta_k}\neq0$. The fact that they are all monotone can be seen as follows. The homomorphisms $\Phi_{\alpha,z}$, $\Phi^\T$ give rise to the relative entropies $\hat{D}_{\alpha,z}$, $\hat{D}^\T$ defined in \eqref{eq:alphazexplicit}, \eqref{eq:DT}. Since DPI is known to be satisfied for the $\alpha$-$z$ relative entropies when $\alpha\in(0,1)$ and $z\geq\max\{\alpha,1-\alpha\}$ \cite{Zhang2020} (and hence also for our version with a different prefactor), the pointwise limits 
\begin{equation}
    \hat D^\T=\lim_{z\rightarrow\infty}\hat{D}_{\alpha,z}\ \text{(independent of}\ \alpha\text{)}, 
\end{equation}
and, for $z\geq1$, 
\begin{equation}
    \hat D_{0,z}=\lim_{\alpha\rightarrow0}\hat{D}_{\alpha,z},\quad\hat D_{1,z}=\lim_{\alpha\rightarrow1}\hat{D}_{\alpha,z} 
\end{equation}    
also satisfy DPI. Since these relative entropies satisfy DPI and the prefactors in front of $\log$ are negative, the homomorphisms that they are associated with (namely $\Phi_{\alpha,z}$ for $\alpha\in[0,1]$, $z\geq\max\{\alpha,1-\alpha\}$, and $\Phi^\T$) are monotone according to the opposite ordering of $\R_+^{\rm op}$, $\T\R_+^{\rm op}$. Since 
\begin{equation}
    \Phi_z^\T(A,B)=\left(\Phi^\T(A,B)\right)^z
\end{equation}
for all $(A,B)\in S_{\rm m.r.}$ and $z>0$, we find that $\Phi_z^\T$ for any $z>0$ is also monotone. 
\end{proof}

\begin{remark}\label{rmk:Seo}
For the semiring $S_{\rm e.o.}$ we do not have a complete picture: we know that all monotone homomorphisms in Proposition \ref{prop:monhom} associated with $S_{\rm m.r.}$ are also associated with $S_{\rm e.o.}$, but there might be more. Namely, it is still an open question whether the homomorphisms for $\K=\R_+$ or $\R_+^{\rm op}$ that have the same form as \eqref{eq:temperate}, but with parameters $\alpha=0$, $z<1$, or $\alpha=1$, $z<1$, are monotone or not. The same question of monotonicity can also be asked for the homomorphisms for $\K=\T\R_+$ that are similar to \eqref{eq:tropicalz}, but with $z<0$. Additionally, there might exist monotone derivations that do not vanish everywhere. In the following, we will show that we can rule out the existence of any $\Phi$ other than the aforementioned additional ones or the ones associated with $S_{\rm m.r.}$, using techniques similar to the proof of Proposition \ref{prop:monhom}. We suspect that these additional homomorphisms are in fact not monotone either, since they correspond to parameters $\alpha,z$ for which the $\alpha$-$z$ relative entropy is known not to satisfy DPI for general pairs of quantum states (with some conditions on the support). However, in principle, DPI might be satisfied when restricting to the specific pairs of states we consider in this work. 

Fix a monotone homomorphism $\Phi:S_{\rm e.o.}\rightarrow\K$, where $\K\in\{\R_+,\R_+^{\rm op},\allowbreak\T\R_+,\T\R_+^{\rm op}\}$. First, let us summarize results from \cite{farooq2024} on the monotone homomorphisms and derivations associated to the \emph{matrix majorization semiring of length} $2$, which consists precisely of all commuting pairs in $S_{\rm e.o.}$. When $\K=\R_+^{\rm op}$, $\Phi$ restricts to \eqref{eq:Phitemperate} for some $\alpha\in[0,1]$, on commuting pairs. When $\K=\R_+$, $\Phi$ restricts to \eqref{eq:Phitemperate} for some $\alpha\in(-\infty,0]\cup[1,\infty)$. Note that for $\K=\R_+^{\rm op}$ or $\R_+$ the cases $\alpha=0$, respectively $\alpha=1$, correspond to the degenerate homomorphisms $\tr[B]$, respectively $\tr[A]$, for a commuting pair $(A,B)$. When $\K=\T\R_+$, $\Phi$ on commuting pairs restricts to 
\begin{equation}\label{eq:Phitropicalcommuting}
\Phi(A,B)=\max_{1\leq i\leq n}\left(\frac{p_i}{q_i}\right)^\beta, 
\end{equation}
for some $\beta\in\R$. For $\beta=0$, this corresponds to the highly degenerate $\Phi(A,B)=1$ for $(A,B)\neq 0$ a commuting pair and $\Phi(0,0)=0$. When $\K=\T\R_+^{\rm op}$, $\Phi$ always restricts to this highly degenerate case on commuting pairs. The two degenerate homomorphisms $\tr[A],\tr[B]$ for commuting pairs $(A,B)$ both have non-zero monotone derivations associated with them (similar in form to the Kullback--Leibler divergence), but we will not use these here. 

Note that the pair on the right-hand side in \eqref{eq:alphazrelationship1} is not inside $S_{\rm e.o.}$. However, for the case $\K=\R_+^{\rm op}$ and $\alpha\in(0,1)$, we can still deduce that $\Phi$ is of the form \eqref{eq:temperate} with $z\geq\max\{\alpha,1-\alpha\}$. To see this, consider again measuring a qubit in a fixed basis $\ket{0},\ket{1}$: 
\begin{equation}
\label{eq:alphazrelationship1eo}
    \left(\,\ketbrar{\psi(\pi/4)},\ketbrar{\psi(\theta)}\,\right)\succeq(1/2,\cos^2\theta)\boxplus(1/2,\sin^2\theta), 
\end{equation}
where $\theta\in(0,\pi/2)$, which ensures that the pair on the right-hand side above is in $S_{\rm e.o.}$. Assume $\K=\R_+^{\rm op}$ and $\alpha\in(0,1)$. Using monotonicity of $\Phi$, \eqref{eq:alphazrelationship1eo} yields 
\begin{equation}
\label{eq:alphazrelationship2eo}
    \left(\cos(\theta-\pi/4)\right)^{2z}\leq\frac{1}{2^\alpha}\left((\cos\theta)^{2-2\alpha}+(\sin\theta)^{2-2\alpha}\right) 
\end{equation}
for all $\theta\in(0,\pi/2)$, and by switching around the states in the pairs in \eqref{eq:alphazrelationship1eo}, we have similarly 
\begin{equation}
\label{eq:alphazrelationship3eo}
    \left(\cos(\theta-\pi/4)\right)^{2z}\leq\frac{1}{2^{1-\alpha}}\left((\cos\theta)^{2\alpha}+(\sin\theta)^{2\alpha}\right). 
\end{equation}
Taking the limit $\theta\rightarrow\pi/2$ in the above two inequalities allows us to conclude that $\sloppy{z\geq\max\{\alpha,1-\alpha\}}$ for $\alpha\in(0,1)$. Note that for $\alpha=0$ or $1$, we are not able to use this method to conclude that $z\geq1=\max\{\alpha,1-\alpha\}$. 

Next, assume $\K=\R_+$ and $\alpha\in(-\infty,0)\cup(1,\infty)$. Note that \eqref{eq:alphazrelationship2eo} and \eqref{eq:alphazrelationship3eo} hold with the inequality reversed. Assume $\alpha>1$. When taking the limit $\theta\rightarrow\pi/2$, the left-hand side of \eqref{eq:alphazrelationship2eo} converges to $\frac{1}{2^z}$, but the right-hand side diverges to positive infinity, leading to a contradiction. When $\alpha<0$, a similar argument using \eqref{eq:alphazrelationship3eo} leads to a contradiction. This rules out the existence of any monotone homomorphism on $S_{\rm e.o.}$ associated with $\alpha$ outside $[0,1]$. 

Now, we assume $\K=\T\R_+$ and $\Phi$ is of the form \eqref{eq:Phitropicalcommuting} for some $\beta\neq0$ when restricted to commuting pairs. Now \eqref{eq:alphazrelationship1eo} implies 
\begin{equation}
\label{eq:TR+}
    \left(\cos(\theta-\pi/4)\right)^{2z}\geq\frac{1}{2^\beta}\max\left\{(\cos\theta)^{-2\beta},(\sin\theta)^{-2\beta}\right\} 
\end{equation}
for all $\theta\in(0,\pi/2)$. Assume $\beta>0$. Then for $\theta$ restricted to the smaller interval $[\pi/4,\pi/2)$, 
\begin{equation}
    \left(\cos(\theta-\pi/4)\right)^{2z}\geq\frac{1}{2^\beta}(\cos\theta)^{-2\beta}. 
\end{equation}
When taking the limit $\theta\rightarrow\pi/2$, the left-hand side converges to $\frac{1}{2^z}$, but the right-hand side diverges to positive infinity, leading to a contradiction. When $\beta<0$, an analogous argument where we switch around the states in the pairs in \eqref{eq:alphazrelationship1eo} leads to a contradiction. Note that we cannot rule out the case $\beta=0$ using this method. Hence there might exist monotone homomorphisms that are similar to \eqref{eq:tropicalz}, but with $z<0$. 

Finally, assume $\K=\T\R_+^{\rm op}$. Analogously as in the case of the semiring $S_{\rm m.r.}$, we can conclude that any non-degenerate $\Phi$ must be of the form \eqref{eq:tropicalz} for some $z>0$.  
\end{remark}

\begin{remark}
    Proposition \ref{prop:puS} tells us that any $(A,B)\in S_{\rm e.o.}$, where $A,B$ have trace $1$ and satisfy $|\braket{\alpha_i}{\beta_i}|<1$ for all $i=1,\ldots,n$, is power universal. We can now show that these are in fact all possible power universals in $S_{\rm e.o.}$, using the tropical monotone homomorphism $\Phi^\T$ defined in \eqref{eq:tropical}. 

    To see this, let $(A,B)\in S_{\rm e.o.}$ be power universal. Then there exists $m\in\N$ such that $(A,B)^{\boxtimes m}\succeq(\ketbrar{0},\ketbrar{\psi(\pi/4)})$. Hence, 
    \begin{equation}
        \left(\Phi^\T(A,B)\right)^m\leq\Phi^\T(\ketbrar{0},\ketbrar{\psi(\pi/4)})=1/2<1. 
    \end{equation}
    Therefore, $\Phi^\T(A,B)<1$, from which it follows that $|\braket{\alpha_i}{\beta_i}|<1$ for all $i=1,\ldots,n$. 
\end{remark}

\section{Results}\label{sec:results}
We next show how our main results Theorems \ref{thm:LS}, \ref{thm:approximateLS} and \ref{thm:minimal} follow from the analysis of the semiring $S_{\rm m.r.}$ in the previous section. 

\subsection{Exact Majorization}\label{sec:exactproof}

Our result Theorem \ref{thm:LS} is a direct application of the Vergleichsstellensatz in the form of Theorem \ref{thm:Fritz2022}. Note that the monotone homomorphisms associated with $S_{\rm m.r.}$ (see Proposition \ref{prop:monhom}) give rise to the relative entropies \eqref{eq:alphazexplicit} and \eqref{eq:DT}. 

The surjective homomorphism with trivial kernel $\|\cdot\|:S\to\R_{>0}^2\cup\{(0,0)\}$ needed in the application of Theorem \ref{thm:Fritz2022} is simply given by $\|(A,B)\|=(\tr[A],\tr[B])$. The first property in \eqref{eq:surjectivehomomorphismproperties} follows from the trace preservation property of quantum channels, and the second property from the following majorization relations: 
\begin{equation}
    (A,B)\succeq(\tr[A],\tr[B])=(\tr[A'],\tr[B'])\preceq(A',B')
\end{equation}
for any $(A,B),(A',B')\in S_{\rm m.r.}$ satisfying $\|(A,B)\|=\|(A',B')\|$. 

Note that \eqref{eq:pu1} is automatically satisfied when condition \eqref{eq:LScondition2} is met: the latter implies that $\hat D^\T(\rho\|\sigma)>0$, which is equivalent to $\max_{i=1,\ldots,n}|\braket{\alpha_i}{\beta_i}|^{2}<1$, which in turn implies \eqref{eq:pu1}. By Proposition \ref{prop:pumr}, for $(\rho,\sigma)$ in the statement of the theorem to be power universal it thus suffices to only require it satisfies \eqref{eq:pu2}, besides the inequalities \eqref{eq:LScondition1}, \eqref{eq:LScondition2}. 

Our main result, Theorem \ref{thm:LS}, now follows immediately from Propositions \ref{prop:monhom} and \ref{prop:pumr}, and Theorem \ref{thm:Fritz2022}. 

\subsection{Asymptotic Majorization}\label{sec:approximate}

As a consequence of Theorem \ref{thm:LS}, we can formulate a similar result where the inequalities in terms of relative entropies are only required to be satisfied \emph{non-strictly}. In this case, large-sample and catalytic majorization only hold asymptotically, in the sense that we reach only one of the two states in the output pair exactly and the other state approximately, up to arbitrarily small error. This result is contained in Theorem \ref{thm:approximateLS}, which we prove here. 


\begin{proof}[Proof of Theorem \ref{thm:approximateLS}]
    First, note that satisfying the inequalities in part (i) of the theorem is equivalent to $\hat D_{\alpha,z}(\rho\|\sigma)\geq \hat D_{\alpha,z}(\rho'\|\sigma')$ for all $\alpha\in[0,1]$, $z\geq\max\{\alpha,1-\alpha\}$, and additionally $\hat D^\T(\rho\|\sigma)\geq \hat D^\T(\rho'\|\sigma')$. 

    Also, by Proposition \ref{prop:pumr}, $(\rho,\sigma)$ is power universal if it satisfies \eqref{eq:pu1} and \eqref{eq:pu2}. Note that Theorem \ref{thm:Fritz2022} only requires $(\rho,\sigma)$ to be power universal in the large-sample setting. This is why we require \eqref{eq:pu2} to hold for equivalence of all conditions (i), (ii), (iii). We require property \eqref{eq:pu1} to hold both in the large-sample and catalytic setting since our construction in this proof relies on it. However, since we will invoke Theorem \ref{thm:LS}, $(\rho,\sigma)$ has to satisfy \eqref{eq:pu1} anyway if condition \eqref{eq:LScondition2} is to be met, as explained in Section \ref{sec:exactproof}. 
    
    (i) $\Rightarrow$ (ii) and (iii): We write
    \begin{equation}
    \rho'=\bigoplus_{i=1}^{n'}p'_i\ketbrar{\alpha'_i},\quad\sigma'=\bigoplus_{i=1}^{n'}q'_i\ketbrar{\beta'_i}, 
    \end{equation}
    with all $\ket{\alpha'_i},\ket{\beta'_i}$ normalized or equal to $0$. Let $J\subseteq[n']$ be the subset of all $i$ where $0<|\braket{\alpha'_i}{\beta'_i}|<1$, which is non-empty since $\rho'$, $\sigma'$ do not commute. One shows analogously as in the discussion above \eqref{eq:pairstandardform} that each pair $(\ketbrar{\alpha'_i},\ketbrar{\beta'_i})$, $i\in J$, is equivalent (according to the equivalence relation $\approx$ defined in Section \ref{sec:definingthesemirings}) to the pair 
    \begin{align}
        \big(\cos^2\theta_i&\ketbrar{e_{i,1}}+\cos\theta_i\sin\theta_i\ketbra{e_{i,1}}{e_{i,2}}\\
        &+\cos\theta_i\sin\theta_i\ketbra{e_{i,2}}{e_{i,1}}+\sin^2\theta_i\ketbrar{e_{i,2}},\ketbrar{e_{i,1}}\big), 
    \end{align}
    for some angle $\theta_i\in(0,\pi/2)$, and $\ket{e_{i,1}},\ket{e_{i,2}}$ an orthonormal basis of a $2$-dimensional Hilbert space. Choose $\eta\in(0,\pi/2)$ such that $\eta<\theta_i$ for all $i\in J$ and $\cos\eta\geq\sqrt{1-\varepsilon}$. For each $i\in J$, we define 
    \begin{equation}
        \ket{\alpha'_{\varepsilon,i}}:=\cos(\theta_i-\eta)\ket{e_{i,1}}+\sin(\theta_i-\eta)\ket{e_{i,2}}, 
    \end{equation}
    and for each $i\in[n']\setminus J$, we define $\ket{\alpha'_{\varepsilon,i}}:=\ket{\alpha'_i}$. Consider 
    \begin{equation}
        \rho'_\varepsilon:=\bigoplus_{i=1}^{n'}p'_i\ketbrar{\alpha'_{\varepsilon,i}}. 
    \end{equation}
    One computes that $|\braket{\alpha'_{\varepsilon,i}}{\alpha'_i}|=\cos(\eta)\geq\sqrt{1-\varepsilon}$ for all $i\in J$, hence the fidelity between $\rho'_\varepsilon$ and $\rho'$ satisfies 
    \begin{equation}
        F(\rho'_\varepsilon,\rho')=\left(\sum_{i=1}^{n'} p'_i|\braket{\alpha'_{\varepsilon,i}}{\alpha'_i}|\right)^2\geq1-\varepsilon. 
    \end{equation}

    Also, for all $i\in J$ 
    \begin{equation}
        |\braket{\alpha'_i}{\beta'_i}|=\cos(\theta_i)<\cos(\theta_i-\eta)=|\braket{\alpha'_{\varepsilon,i}}{\beta'_i}|, 
    \end{equation}
    hence for all $\alpha\in[0,1]$, $z\geq\max\{\alpha,1-\alpha\}$ 
    \begin{equation}
        \hat D_{\alpha,z}(\rho\|\sigma)\geq \hat D_{\alpha,z}(\rho'\|\sigma')>\hat D_{\alpha,z}(\rho'_\varepsilon\|\sigma'). 
    \end{equation}
    Assume $\hat D^\T(\rho'\|\sigma')>0$, i.e. $\max_{i=1,\ldots,n'}|\braket{\alpha'_i}{\beta'_i}|^{2}<1$. Then this maximum is obtained by some subspace in $J$, from which it follows that $\max_{i=1,\ldots,n'}|\braket{\alpha'_i}{\beta'_i}|^{2}<\max_{i=1,\ldots,n'}|\braket{\alpha'_{\varepsilon,i}}{\beta'_i}|^{2}$, whence $\hat D^\T(\rho\|\sigma)\geq \hat D^\T(\rho'\|\sigma')>\hat D^\T(\rho'_\varepsilon\|\sigma')$. In case $\hat D^\T(\rho'\|\sigma')=0$, it follows immediately that $\hat D^\T(\rho\|\sigma)>\hat D^\T(\rho'_\varepsilon\|\sigma')$, since $\hat D^\T(\rho\|\sigma)>0$ by \eqref{eq:pu1}. 

    The pairs $(\rho,\sigma)$ and $(\rho'_\varepsilon,\sigma')$ thus satisfy the strict inequalities \eqref{eq:LScondition1} and \eqref{eq:LScondition2} of Theorem \ref{thm:LS}, hence (ii) and (iii) follow. 

    (ii) or (iii) $\Rightarrow$ (i): Using the tensor additivity and the DPI property of the $D_{\alpha,z}$, large-sample majorization as in (ii) implies that for all $\varepsilon>0$ and all $\alpha\in(0,1)$, $z>\max\{\alpha,1-\alpha\}$ 
    \begin{equation}\label{eq:approximateinequality}
        D_{\alpha,z}(\rho\|\sigma)\geq D_{\alpha,z}(\rho'_\varepsilon\|\sigma') 
    \end{equation}
    for a state $\rho'_\varepsilon$ satisfying $F(\rho'_\varepsilon,\rho')\geq1-\varepsilon$. Similarly, the same is implied by catalytic majorization as in (iii). 

    We may assume that $\rho'_\varepsilon$ resides in a Hilbert space $\Hil'_\varepsilon$ possibly larger than $\Hil'$ where $\rho'$ and $\sigma'$ operate. Let $P_\varepsilon$ be the projection of $\Hil'_\varepsilon$ onto $\Hil'$ and $P_\varepsilon^\perp$ be the projection onto the orthogonal complement of $\Hil'$ within $\Hil'_\varepsilon$. Set up the channel $\mc D_\varepsilon$, $\mc D_\varepsilon(\tau)=P_\varepsilon\tau P_\varepsilon+P_\varepsilon^\perp\tau P_\varepsilon^\perp$. Denoting $P_\varepsilon\rho'_\varepsilon P_\varepsilon=:{\rho'_\varepsilon}^0$ and $P_\varepsilon^\perp\rho'_\varepsilon P_\varepsilon^\perp=:{\rho'_\varepsilon}^1$, we have $\mc D_\varepsilon(\rho'_\varepsilon)={\rho'_\varepsilon}^0\oplus{\rho'_\varepsilon}^1$ and $\mc D_\varepsilon(\sigma')=\sigma'$ when we view $\sigma'$ as a state on $\Hil'_\varepsilon$.
    Note that $\lim_{\varepsilon\rightarrow0}({\rho'_\varepsilon}^0,\sigma')=(\rho',\sigma')$. 
    
    Next, for $\alpha\in(0,1)$ and $z>\max\{\alpha,1-\alpha\}$, consider the expression 
    \begin{equation}
        \Psi_{\alpha,z}(\tau,\omega):=\tr\left[\left(\omega^\frac{1-\alpha}{2z}\tau^\frac{\alpha}{z}\omega^\frac{1-\alpha}{2z}\right)^z\right] 
    \end{equation}
    inside the $\log$ of $D_{\alpha,z}$ as in \eqref{eq:alphazdivergence}. One checks that this is additive under $\oplus$, i.e. $\Psi_{\alpha,z}(\rho_1\oplus\rho_2,\sigma_1\oplus\sigma_2)=\Psi_{\alpha,z}(\rho_1,\sigma_1)+\Psi_{\alpha,z}(\rho_2,\sigma_2)$, and is non-decreasing when applying a quantum channel by the DPI of $D_{\alpha,z}$. Then, these properties and \eqref{eq:approximateinequality} imply 
    \begin{equation}
        \Psi_{\alpha,z}(\rho,\sigma)\leq\Psi_{\alpha,z}(\rho'_\varepsilon,\sigma')\leq\Psi_{\alpha,z}({\rho'_\varepsilon}^0,\sigma')+\underbrace{\Psi_{\alpha,z}({\rho'_\varepsilon}^1,0)}_{=0}. 
    \end{equation}
    Since the $\Psi_{\alpha,z}$ are continuous on pairs of fixed dimension, and $\lim_{\varepsilon\rightarrow0}({\rho'_\varepsilon}^0,\sigma')=(\rho',\sigma')$, we conclude that the inequalities in part (i) hold. 
\end{proof}


\subsection{Minimality of the Family of Relative Entropies}\label{sec:minimal}

Here we prove Theorem \ref{thm:minimal}, in much the same way as Theorem \ref{Mthm:minimal} on the minimal-restrictions matrix majorization semiring, using an idea from the proof of Corollary 2.7 in \cite{strassen1988}. 
\begin{proof}[Proof of Theorem \ref{thm:minimal}]
    Consider the set $\bar R$, containing the parameter range $R$, defined by 
    \begin{equation}\label{eq:Rbar}
        \bar R:=\big\{\,(\alpha,z)\,|\,\alpha\in[0,1],\ z\geq\max\{\alpha,1-\alpha\}\,\big\}\cup\big\{\,(\alpha,\infty)\,|\,\alpha\in[0,1]\,\big\}, 
    \end{equation}
    and define the injective map $\varphi:\bar R\rightarrow\R^2$ by 
    \begin{equation}
        \varphi(\alpha,z)=\left\{\begin{array}{ll}\big(\alpha,\frac{z}{z+1}\big)&\text{if }z\neq\infty\\
        (\alpha,1)&\text{if }z=\infty\end{array}\right..
    \end{equation}
    We define a topology on $\bar R$ by requiring that $U$ is open in $\bar R$ if and only if $U=\varphi^{-1}(V)$ for some open $V$ in $\R^2$ with the Euclidean topology. The subspace topology that $R$ inherits from $\bar R$ coincides with the Euclidean topology. One checks that the image $\varphi(\bar R)$ is a compact subset of $\R^2$, hence $\bar R$ is compact and Hausdorff. 

    For any $(A,B)\in S_{\rm m.r.}$, the evaluation function $f_{(A,B)}:\bar R\rightarrow\R$ is defined for $(\alpha,z)\in\bar R$ when $z\neq\infty$ by 
    \begin{equation}
        f_{(A,B)}(\alpha,z):=\Phi_{\alpha,z}(A,B)=\sum_{i=1}^n(p_i)^\alpha(q_i)^{1-\alpha}|\braket{\alpha_i}{\beta_i}|^{2z}
    \end{equation}
    and when $z=\infty$ by 
    \begin{equation}
        f_{(A,B)}(\alpha,\infty):=\lim_{z\rightarrow\infty}\Phi_{\alpha,z}(A,B)=\sum_{i\in[n]\,:\,|\braket{\alpha_i}{\beta_i}|=1}(p_i)^\alpha(q_i)^{1-\alpha}, 
    \end{equation} 
    where we used that all $|\braket{\alpha_i}{\beta_i}|\leq1$ to compute the limit. Observe that $f_{(A,B)}\circ\varphi^{-1}:\varphi(\bar R)\rightarrow\R$ is continuous since it is a composition of continuous functions on $\varphi(\bar R)\setminus([0,1]\times\{1\})$, and equals the limiting value on $[0,1]\times\{1\}$. Hence, $f_{(A,B)}$ is continuous. 

    Next, let $\mc A$ be the $\R$-algebra generated by multiplication of the functions $f_{(A,B)}$, $\sloppy{(A,B)\in S_{\rm m.r.}}$, by real numbers, and finite repetitions of point-wise addition and multiplication. Then $\mc A$ is a subalgebra of $C(\bar R,\R)$, the set of all continuous real-valued functions on $\bar R$. 
    
    Observe that $f_{(1,1)}(\alpha,z)=1$ for all $(\alpha,z)\in\bar R$, hence $\mc A$ contains a non-zero constant function. Clearly, for every two distinct $(\alpha,z),(\alpha',z')\in\bar R$ there exists a pair $\sloppy{(A,B)\in S_{\rm m.r.}}$ such that $f_{(A,B)}(\alpha,z)\neq f_{(A,B)}(\alpha',z')$, hence $\mc A$ separates points. Therefore, by the Stone-Weierstrass theorem, $\mc A$ is dense in $C(\bar R,\R)$. 

    The non-empty open subset $O$ in the statement of the theorem contains an open ball, and without loss of generality we assume $O$ to be equal to an open ball. Clearly, there exists a continuous function $f:\bar R\rightarrow\R$ such that $f(\alpha,z)=1$ for all $(\alpha,z)$ outside $O$, and $f(\alpha_0,z_0)=-1$ for some $(\alpha_0,z_0)$ in $O$. Since $\mc A$ is dense in $C(\bar R,\R)$, $f$ can be approximated by functions in $\mc A$ with arbitrary precision, from which it follows that there exist $r_1,\ldots,r_t\in\R$ and $(A_1,B_1),\ldots,(A_t,B_t)\in S_{\rm m.r.}$ such that 
    \begin{equation}\label{eq:SWapproximation}
        r_1\Phi_{\alpha,z}(A_1,B_1)+\ldots+r_t\Phi_{\alpha,z}(A_t,B_t)\left\{\begin{array}{ll}>0&\text{if }\ (\alpha,z)\in R\setminus O\\<0&\text{if }\ (\alpha,z)=(\alpha_0,z_0)\end{array}\right..
    \end{equation}
    Here we have used the multiplicativity of the homomorphisms, which implies that $\mc A$ is the span of the set of all $f_{(A,B)}$, $(A,B)\in S_{\rm m.r.}$, with real coefficients. 
    Define the pairs 
    \begin{align}
        &(A',B'):=\left(\bigoplus_{i\,:\,r_i>0}r_iA_i,\bigoplus_{i\,:\,r_i>0}r_iB_i\right),\\
        &(A,B):=\left(\bigoplus_{i\,:\,r_i<0}-r_iA_i,\bigoplus_{i\,:\,r_i<0}-r_iB_i\right). 
    \end{align}
    Then, by rearranging \eqref{eq:SWapproximation}, if $(\alpha,z)\in R\setminus O$, 
    \begin{equation}\label{eq:counterexample1}
        \Phi_{\alpha,z}(A',B')>\Phi_{\alpha,z}(A,B),  
    \end{equation}
    and 
    \begin{equation}\label{eq:counterexample2}
        \Phi_{\alpha_0,z_0}(A',B')<\Phi_{\alpha_0,z_0}(A,B). 
    \end{equation}
    Note that it follows in particular that both pairs $(A,B),(A',B')$ are non-zero. 
    
    If $(A,B)$ is not non-parallel, i.e. there are $i$ such that $|\braket{\alpha_i}{\beta_i}|=1$, then for such $i$ we may rotate $\ket{\alpha_i},\ket{\beta_i}$ to ensure that $|\braket{\alpha_i}{\beta_i}|<1$. By doing so, $\Phi_{\alpha,z}(A,B)$ may decrease, hence \eqref{eq:counterexample1} remains true. By ensuring that the rotations are sufficiently small, \eqref{eq:counterexample2} also remains true. We conclude that we may assume that $(A,B)$ is non-parallel. 

    If $(A',B')$ is a commuting pair, i.e. $|\braket{\alpha_i}{\beta_i}|=1$ for all $i$, then $\Phi_{\alpha_0,z}(A',B')$ does not depend on $z$. However, $\Phi_{\alpha_0,z}(A,B)$ is strictly decreasing in $z$. We may assume that $O$ is a sufficiently small open ball such that there exists $z_1\in\R$ with $(\alpha_0,z_1)\in R\setminus O$ and $z_1<z_0$. There also exists $z_2\in\R$ with $(\alpha_0,z_2)\in R\setminus O$ and $z_0<z_2$. Then \eqref{eq:counterexample1} implies that both $\Phi_{\alpha_0,z_1}(A,B)$ and $\Phi_{\alpha_0,z_2}(A,B)$ are strictly less than the constant $\Phi_{\alpha_0,z_0}(A',B')$. Since $\Phi_{\alpha_0,z_0}(A,B)>\Phi_{\alpha_0,z_0}(A',B')$, this contradicts the fact that $\Phi_{\alpha_0,z}(A,B)$ is strictly decreasing in $z$. In conclusion, $(A',B')$ is not a commuting pair. 
    
    We add an orthogonal component to both pairs as follows: 
    \begin{align}
        &(\tilde{A},\tilde{B}):=(A,B)\boxplus(a\ketbrar{0},b\ketbrar{1}),\\
        &(\tilde{A'},\tilde{B'}):=(A',B')\boxplus(a'\ketbrar{0},b'\ketbrar{1}), 
    \end{align}
    with $a,b>0$, $a',b'\geq0$ chosen such that 
    \begin{equation}
        T:=\tr{\tilde{A}}=\tr{\tilde{B}}=\tr{\tilde{A'}}=\tr{\tilde{B'}}. 
    \end{equation}
    The pairs $(\rho,\sigma):=(\tilde{A},\tilde{B})/T$ and $(\rho',\sigma'):=(\tilde{A'},\tilde{B'})/T$ are in $\mc F$, and the former is non-parallel and the latter is not commutative. Since $a,b>0$, $(\rho,\sigma)$ additionally satisfies \eqref{eq:pu2}. Note that 
    \begin{align}
        &\Phi_{\alpha,z}(\rho,\sigma)=T^{-1}\Phi_{\alpha,z}(\tilde{A},\tilde{B})=T^{-1}\Phi_{\alpha,z}(A,B),\\
        &\Phi_{\alpha,z}(\rho',\sigma')=T^{-1}\Phi_{\alpha,z}(\tilde{A'},\tilde{B'})=T^{-1}\Phi_{\alpha,z}(A',B')
    \end{align}    
    for all $(\alpha,z)\in R$. Hence, it follows from \eqref{eq:counterexample1}, \eqref{eq:counterexample2} that 
    \begin{equation}\label{eq:contradiction}
        \Phi_{\alpha,z}(\rho',\sigma')>\Phi_{\alpha,z}(\rho,\sigma) 
    \end{equation}
    for all $(\alpha,z)\in R\setminus O$ and the reverse inequality for $(\alpha,z)=(\alpha_0,z_0)$.
   
    Now, \eqref{eq:contradiction} for all $(\alpha,z)\in R\setminus O$ implies the conditions \eqref{eq:smallercondition}. However, (ii) and (iii) in Theorem \ref{thm:approximateLS} are both false: if either were true, then they would imply (i), which contradicts the fact that \eqref{eq:contradiction} with the reverse inequality is true for $(\alpha,z)=(\alpha_0,z_0)$. 
\end{proof}


\chapter{Outlook and Open Problems}\label{cha:open}

In this thesis, we study the convertibility of tuples of classical or quantum states by a single channel, in the setting of large-samples or when using catalyst states to make the conversion possible. Divergences, which are additive functions on the state-tuples that are monotone under the application of a channel, characterize large-sample and catalytic majorization: loosely speaking, if these divergences are always larger when evaluated on one tuple compared to another tuple, then the former tuple majorizes the latter in the large-sample and catalytic setting. Using the preordered semiring techniques developed in \cite{fritz2022}, we identified the precise family of these divergences in two different settings: one classical involving tuples of finite probability distributions that have different supports, and one quantum, involving pairs of cq-states with pure components. Additionally, these collections of divergences were shown to be minimal in some cases when it comes to characterizing large-sample and catalytic majorization. Our results also yielded optimal conversion rates between the tuples of states considered. Our results demonstrate that applying preordered semiring techniques to problems in (quantum) information theory can be very fruitful. We anticipate these techniques will yield more new and interesting results in the future. In the remainder of this section, we will discuss some open problems connected with the work in this thesis. 

For the minimal restrictions matrix majorization semiring and the minimal restrictions pure-component cq-state pairs semiring, we were able to show that the set of monotone homomorphisms characterizing large-sample and catalytic majorization was minimal: after removing a non-empty open subset from the set of all homomorphisms, the remaining set is too small to give sufficient conditions for majorization. See Theorems \ref{Mthm:minimal} and \ref{thm:minimal}. We suspect that a similar result is also true for the matrix majorization semirings of common support or of the dominating column types, using a similar proof technique involving the Stone-Weierstrass Theorem. However, in these semirings, the value of $\Phi(P)$ for a fixed matrix $P$ may diverge when the associated vector $\underline{\alpha}$ grows large. Instead of considering the set of homomorphisms and derivations, we could consider the set of all divergences associated with the semiring. This is known to be a compact set, hence an application of the Stone-Weierstrass Theorem would be possible. However, since in this case we would be dealing with divergences instead of homomorphisms, the proof strategy employed in Theorems \ref{Mthm:minimal} and \ref{thm:minimal} would have to be adapted. One possible solution could be to use the lattice version of Stone-Weierstrass instead, as in the proof of Lemma 7.14 in \cite{fritz2023}. It might be possible to adapt this proof to the general setting of preordered semirings considered in Chapter \ref{cha:VSS}. In order to accomplish this, one would have to find a way to deal with derivations, which were not considered in \cite{fritz2023}. We note that in all semirings discussed in this thesis, there were either no derivations, or the set of all derivations was nowhere dense inside the set of all divergences. This fact might aid in establishing a minimality result similar to Theorems \ref{Mthm:minimal} and \ref{thm:minimal} applicable to more general semirings. 

One issue with our results on sufficient conditions for the majorization of a matrix $P$ over another matrix $Q$ (both with normalized columns), is that in the large-sample or catalytic setting, and both in the exact and asymptotic regime, $P$ has to be a power universal of the preordered semiring that we are studying. In the case of exact large-sample or catalytic majorization (Theorems \ref{Mthm:majorization} and \ref{Mthm:majorizationdominatingcolumn}), the strict ordering of the values of the homomorphisms already imply that $P$ is power universal. However, in the case of asymptotic large-sample or catalytic majorization (Theorems \ref{Mthm:genasymptocatalytic} and \ref{Mthm:asymptocatalytic}), we have to require $P$ to be power universal. Being power universal is typically quite restrictive, as can be seen in Propositions \ref{Mlem:poweruniversalconditions} and \ref{Mpro:poweruniversaldcgeneral}. However, as shown in Section \ref{Msec:pusupports}, we can define a semiring depending on the precise relations the supports of the columns of $P$ have. If Conjecture \ref{con:puconditiongeneral} is correct, then this can be done in a way such that $P$ is power universal. Then, we could obtain results, similar to the other matrix majorization semirings studied, giving sufficient and (almost) necessary conditions for large-sample and catalytic majorization that apply to the chosen matrix $P$. 

A similar problem also appears in the case of pairs of pure-component cq-states: here, the input pair is required to be power universal for large-sample majorization in both the exact and asymptotic setting. For the semiring $S_{\rm m.r.}$, the conditions \eqref{eq:pu1} and \eqref{eq:pu2} are sufficient in order for a pair to be power universal. It seems to be possible to relax this condition, by removing the need for \eqref{eq:pu2} from the proof of Proposition \ref{prop:pumr} and/or a deeper study of the semiring $S_{\rm e.o.}$ of everywhere overlapping pairs. In Remark \ref{rmk:Seo} we found that the monotone homomorphisms associated to $S_{\rm e.o.}$ are those associated to $S_{\rm m.r.}$ as well as possibly additional ones for $\alpha=0,1$ and $z<1$, a tropical one associated with the limit $z\rightarrow-\infty$, and potentially also non-trivial derivations. Currently we are unsure whether these extra homomorphisms are monotone or not. However, the $\alpha$-$z$ relative entropies for $\alpha=0,1$, $z<1$ have been shown to violate DPI for general pairs of states. The counterexamples found in the literature do not involve pairs of states of the form \eqref{eq:pairform} that we study. Hence, the possibility remains that these homomorphisms are in fact monotone within our restricted setting. 

From a practical point of view, the notion of asymptotic convertibility of tuples of states (see Theorems \ref{Mthm:genasymptocatalytic}, \ref{Mthm:asymptocatalytic} and \ref{thm:approximateLS}), as opposed to exact majorization, is the most relevant case, both in the large-sample and catalytic setting. After all, it is often enough to know that large-sample or catalytic majorization only holds approximately with vanishing error. The mathematical methods applied in this thesis are geared towards identifying exact majorization in large samples or catalytically. The results in the asymptotic picture are then derived from these exact results. Potentially, this process might be introducing extra conditions which are not needed in the asymptotic picture, such as the requirement that the input tuple is power universal, as discussed above. Finding conditions directly for the desired asymptotic results without having to take a detour through the theory for exact majorization would obviously be of high interest. This would require developing a mathematical framework for preordered semirings that have a built-in topology from the start, and showing a result analogous to the Vergleichsstellensatz as in Theorem \ref{thm:Fritz2022} giving conditions for asymptotic large-sample and catalytic majorization. 

Obviously, generalizing our results on pairs of pure-component cq-states to the case of general pairs of quantum states would be of very high interest. This would involve identifying all divergences characterizing large-sample and catalytic majorization. Using the barycentric results from \cite{Haapasalo2025}, this would give a complete picture of all possible divergences for pairs of general quantum states. This has been an open problem that has proven to be very difficult to tackle. The methods we employ in identifying the divergences rely on reducing the problem to pairs (or tuples) of a smaller size. For the matrix majorization semirings, the problem reduces to studying monotones on single-row matrices, and in the setting of pure-component cq-state pairs, it reduces to studying monotones on pairs of pure states. For general pairs of quantum states, there is no obvious way to make this reduction. An intermediate step could be to study pairs of cq-states with components of dimension at most equal to some suitably small $d\geq2$. In such cases, identifying monotones reduces to studying pairs of dimension at most $d$. For qubit components, i.e. $d=2$, any pair of qubits is characterized, up to a change of basis, by just three real numbers. Hence a proof along the lines of that for Proposition \ref{prop:z}, where we used that pairs of pure states can be characterized by just their fidelity, might be tractable. 

Recall that Proposition \ref{propo:Uhlmann}, which gives necessary and sufficient conditions for one-shot majorization between two pairs of pure states, was used as an ingredient of our analysis of the semiring of pure-component cq-state pairs. There is a generalization of this proposition for one-shot majorization between tuples of $d\geq3$ pure states, also due to Uhlmann \cite{Uhlmann1985}. Namely, consider two tuples of $d$ pure states $A:=\big(\ketbrar{\alpha_1},\ldots,\ketbrar{\alpha_d}\big)$ and $B:=\big(\ketbrar{\beta_1},\ldots,\ketbrar{\beta_d}\big)$. The Gram matrix of $A$ is defined as $(G(A))_{i,j}:=\braket{\alpha_i}{\alpha_j}$, and similarly for $B$. Then, $A$ majorizes $B$ if and only if there exists a positive semi-definite matrix $M$ such that $G(B)\circ M=G(A)$, where $\circ$ is the Hadamard product (i.e. entry-wise multiplication). This result leads us to attempt to generalize the findings in Chapter \ref{cha:flat} on large-sample and catalytic majorization of pairs of pure-component cq-states as in \eqref{eq:pairform}, to tuples of $d\geq3$ states of this form. Namely, we can consider tuples $(\rho^{(1)},\ldots,\rho^{(d)})$, with $d\geq3$, where 
\begin{equation}\label{eq:dtuple}
\rho^{(k)}=\sum_{i=1}^n p_i^{(k)}\ketbrar{i}\otimes\big|\alpha_i^{(k)}\big\rangle\big\langle\alpha_i^{(k)}\big|
\end{equation}
are cq-states with pure components, for $k=1,\ldots,d$. One could try to use the results on matrix majorization involving $d$-tuples of finite probability distributions, in combination with Uhlmann's result on $d$-tuples of pure states stated above, to derive multivariate quantum relative entropies that give conditions for large-sample and catalytic majorization of $d$-tuples such as \eqref{eq:dtuple}. In this setting, the problem of identifying all monotone homomorphisms reduces to analyzing the values they take on tuples of $d$ pure states. In the $d=2$ case we have shown these to be $z$-th powers of the fidelity between the two pure states (see Proposition \ref{prop:z}). Our preliminary results suggest that in the general $d$ case these will be products of the fidelities between the states in each of the $d(d-1)/2$ pairs of pure states in the tuple, raised to a power that can be different for each pair.

\bibliographystyle{ultimate}
\bibliography{bibliography}

\end{document}